\pdfoutput=1

\documentclass[sigconf,authorversion]{acmart}

\usepackage{amsmath,amssymb,amsfonts}
\usepackage{algorithmic}
\usepackage{graphicx}
\usepackage{textcomp}
\usepackage{xcolor}
\usepackage{xspace}
\usepackage{float}
\usepackage{algorithm}  
\usepackage{enumerate}
\usepackage{url}
\usepackage{multirow}
\usepackage{pifont}
\usepackage{array}
\usepackage{ulem}
\usepackage{makecell}

\usepackage{scalerel}
\usepackage{cleveref}
\usepackage{caption}
\usepackage{subcaption}
\usepackage{stmaryrd}
\usepackage{tabularx}
\usepackage[shortlabels]{enumitem}

\def\BibTeX{{\rm B\kern-.05em{\sc i\kern-.025em b}\kern-.08em
    T\kern-.1667em\lower.7ex\hbox{E}\kern-.125emX}}

\def \ourtool{FASVERIF\xspace}

\newcommand{\kwl}[1]{\mathbf{#1}}
\newcommand{\kwf}[1]{\mathsf{#1}}
\newcommand{\kwc}[1]{\mathcal{#1}}
\newcommand{\kwp}[1]{\mathtt{#1}}
\newcommand{\kwt}[1]{\mathtt{#1}}
\newcommand{\kwe}[1]{\mathsf{#1}}

\ifdefined \VersionWithComments
\usepackage{marginnote}
\newcommand{\marginX}{\marginnote{\huge{\quad\quad\textbf{!}\quad\quad}}}
\newcommand{\jj}[1]{\mbox{}{\color{green!50!black}\marginX{}\textbf{[Jiao}: #1]}}
\newcommand{\ksl}[1]{\mbox{}{\color{green!50!blue}\marginX{}\textbf{[Shuanglong}: #1]}}
\newcommand{\lsw}[1]{\mbox{}{\color{orange}\marginX{}\textbf{[Shang-Wei}: #1]}}
\newcommand{\instructions}[1]{{\color{red}\marginX{}\textbf{[Instructions: ``#1'']}}}
\newcommand{\reviewer}[2]{\mbox{}{\color{red}\marginX{}\textbf{[Reviewer #1}: ``#2'']}}
\newcommand{\todo}[1]{\mbox{}{\color{blue}{\marginX{}\textbf{TODO}\ifx#1\\\else:\ \fi #1}}} 
\else
\newcommand{\instructions}[1]{}
\newcommand{\jj}[1]{}
\newcommand{\ksl}[1]{}
\newcommand{\lsw}[1]{}
\newcommand{\reviewer}[2]{}
\newcommand{\todo}[1]{}
\fi

\usepackage{listings}
\newcommand{\chuhao}{\fontsize{8pt}{\baselineskip}\selectfont}
\def\code#1{\texttt{\chuhao#1}}

\def\rulename#1{\textsc{\textcolor{blue}{#1}}}

\newcommand\RWrule[3]{
\scaleleftright[1ex]{<}{
\begin{tabular}{c}
#1\\
\cline{1-1}
#2\\
\end{tabular}\,$\ldots$}
{>}$_{\textit{#3}}$
}

\newcommand\RWruleN[3]{
\scaleleftright[1ex]{<}{
\begin{tabular}{c}
#1\\
\cline{1-1}
#2\\
\end{tabular}} 
{>}$_{\textit{#3}}$
}

\newcommand\RWruleM[3]{
\scaleleftright[1ex]{<}{
$\ldots$
\begin{tabular}{c}
#1\\
\cline{1-1}
#2\\
\end{tabular}\,$\ldots$}
{>}$_{\textit{#3}}$
}

\newcommand\Newrule[2]{
\begin{minipage}{1.0\textwidth}
~\\
{\small RULE \rulename{\textsc{#1}}}\\
{\footnotesize #2}\\
\end{minipage}
}

\newcommand\Rrule[2]{
\scaleleftright[1ex]{<}{$\ldots$\,
\begin{tabular}{c}
#1\\
\end{tabular}\,$\ldots$}
{>}$_{\textit{#2}}$
}

\newcommand\RruleM[2]{
\scaleleftright[1ex]{<}{
\begin{tabular}{c}
#1\\
\end{tabular}} 
{>}$_{\textit{#2}}$
}

\newcommand\RruleS[2]{
\scaleleftright[1ex]{<}{
\begin{tabular}{c}
#1\\
\end{tabular}\,$\ldots$}
{>}$_{\textit{#2}}$
}

\newcommand\Mrule[2]{
\langle #1\rangle_{\textit{#2}}
}

\newcommand\Mrules[2]{
\langle\_ \ #1 \ \_\rangle_{\textit{#2}}
}

\newcommand\karrowwx{
    \rightarrow_\mathcal{K}
}

\newcommand\karroww{
    \rightarrow_{\mathcal{K}_s}
}

\newcommand\karrowx[1]{
    \xrightarrow{\rho_{#1}}_\mathcal{K}
}

\newcommand\newkarrowx[1]{
    \xrightarrow{\rho_{#1}}_{\mathcal{K}_s}
}

\newcommand\karrow[1]{
    \xrightarrow{\rho_{#1}}_{\mathcal{K}_s}
}

\newcommand\csucc{ \kwl{succ}_c }

\newcommand\ksucc{ \kwl{succ}_k }

\newcommand\ks{ \mathcal{K}_s }
\newcommand\rs{R_s}

\newcommand{\bi}{\leftrightarrow_{v}}

\newcommand{\ins}{\in^\#}
\newcommand{\cups}{\cup^\#}
\newcommand{\typed}{\textit{typed}}
\newcommand{\myterms}{\textit{terms}}
\newcommand{\vars}{\textit{vars}}
\newcommand{\gvars}{\textit{gvars}}
\newcommand{\evars}{\textit{evars}}

\newcommand{\ind}{\textit{ind}}
\newcommand{\IND}{\textit{IND}}

\newtheorem{remark}{Remark}

\newcommand\smatch{\Mapsto}

\newcommand\tracek[1]{ \mathit{traces}^{\mathcal{K}_s}(#1)  }
\newcommand\traceks[1]{ \mathit{traces}^{\mathcal{K}_s}(#1)  }
\newcommand{\exe}[1]{\stackrel{#1}{\longrightarrow}{}_{R}}
\newcommand{\exer}[1]{\xrightarrow{#1}_{\rs}}
\newcommand{\exeri}[1]{\xrightarrow{#1}_{R_{inv}}}
\newcommand{\exere}[1]{\xrightarrow{#1}_{R_{equ}}}

\begin{document}

\title{An Automated Analyzer for Financial Security of Ethereum Smart Contracts
} 

\author{Wansen Wang\textsuperscript{1} \ \ \  Wenchao Huang\textsuperscript{1}\ \ \    Zhaoyi Meng\textsuperscript{2}\ \ \    Yan Xiong\textsuperscript{1}\ \ \    Fuyou Miao\textsuperscript{1}\ \ \ Xianjin Fang\textsuperscript{3}\ \ \  Caichang Tu\textsuperscript{1}\ \ \  Renjie Ji\textsuperscript{1}}
\affiliation{\textsuperscript{1}University of Science and Technology of China, Anhui, China}
\affiliation{\textsuperscript{2}Anhui University, Anhui, China}
\affiliation{\textsuperscript{3}Anhui University of Science and Technology, Anhui, China}

\begin{abstract}
At present, millions of Ethereum smart contracts are created per year and attract financially motivated attackers.
However, existing analyzers do not meet the need to precisely analyze the financial security of large numbers of contracts.
In this paper, we propose and implement \ourtool, an automated analyzer for fine-grained analysis of smart contracts' financial security.
On the one hand, \ourtool automatically generates models to be verified against security properties of smart contracts.
On the other hand, our analyzer automatically generates the security properties, which is different from existing formal verifiers for smart contracts.
As a result, \ourtool can automatically process source code of smart contracts, and uses formal methods whenever possible to simultaneously maximize its accuracy.

We evaluate \ourtool on a vulnerabilities dataset by comparing it with other automatic tools.
Our evaluation shows that \ourtool greatly outperforms the representative tools using different technologies, with respect to accuracy and coverage of types of vulnerabilities.
\end{abstract}




\maketitle

\section{Introduction} 
\label{sec:introduction}
Smart contracts on Ethereum have been applied in many fields such as financial industry \cite{application1}, and manage assets worth millions of dollars \cite{abs-2008-02712}, while the market cap of the Ethereum cryptocurrency, \textit{i.e.}, ethers, grows up to \$177 billions on July 27, 2022 \cite{ethercap}.
Unfortunately, this makes smart contracts become attractive targets for attackers.
The infamous vulnerability in the DAO contract led to losses of \$150M in June 2016 \cite{DAO}.
In July 2017, \$30M worth of ethers were stolen from Parity wallet due to a wrong function \cite{parity}.
Most recently, there were \$27M worth of ethers stolen from the Poly Network contract in August 2021~ \cite{poly}. 
It is therefore necessary to guarantee the financial security of smart contracts, 
{\color{black}\textit{i.e.}, the ethers and tokens of contracts are not lost in unexpected ways.}

Nevertheless, existing analyzers are not sufficient to analyze the financial security of numerous contracts accurately.
Current security analyzers for smart contracts can be divided into the following three categories: automated bug-finding tools, semi-automated verification frameworks, and automated verifiers. 
The bug-finding tools \cite{LuuCOSH16}\cite{mythril}\cite{0001LC18} support automated analysis on a great amount of smart contracts, motivated by the fact that 10.7 million contracts are created in 2020 \cite{etherscan2020}.
However, the analysis is based on pre-defined patterns and is not accurate enough~\cite{stephens2021smartpulse}.
The verification frameworks target to formally verify the correctness or security of smart contracts, with the requirement of manually defined properties~\cite{stephens2021smartpulse} or user assistance in verification~\cite{PermenevDTDV20}\cite{DBLP:conf/cpp/Concert}.
It is therefore difficult for these analyzers to analyze a large number of contracts.
The automated verifiers try to provide sound and automated verification of pre-defined properties for smart contracts.
To the best of our knowledge, there are three automated verifiers eThor \cite{DBLP:conf/ccs/SchneidewindGSM20}, SECURIFY \cite{TsankovDDGBV18} and ZEUS \cite{DBLP:conf/ndss/KalraGDS18}.
However, eThor does not aim for the financial security of smart contracts, and only detects reentrancy vulnerabilities \cite{SWC-reentrancy} and checks assertions automatically.
SECURIFY does not support solving numerical constraints and cannot detect numerical vulnerabilities, \textit{e.g.}, overflow.
ZEUS has soundness issues \cite{DBLP:conf/ccs/SchneidewindGSM20} in transforming contracts into IR and thus cannot analyze smart contracts accurately.

We propose and implement \ourtool, a system of automated inference~\cite{SagivRW02}\cite{ChinDNQ12}, \textit{i.e.}, a static reasoning mechanism where the properties are expected to be automatically derived, {\color{black}for achieving full automation on fine-grained financial security analysis of Ethereum smart contracts.
Firstly, \ourtool automatically generates two kinds of finance-related security properties along with the corresponding models for verification.
Secondly, \ourtool can verify these finance-related security properties automatically.}
Overall, the goal of \ourtool is to analyze the financial security of numerous contracts accurately, whereby the security properties are generated automatically based on our statistical analysis,
{\color{black} the soundness of modeling is proven and the verification is implemented using the formal tools Tamarin prover~\cite{MeierSCB13} and Z3~\cite{z3}.}
Moreover, \ourtool generates properties based on the financial losses caused by vulnerabilities instead of known vulnerability patterns, thus covering various vulnerabilities and suitable for the analysis of financial security.
We collect a vulnerabilities dataset consisting of 549 contracts from other works \cite{0001LC18}\cite{DBLP:conf/sp/SoLPLO20}\cite{DBLP:conf/issta/KolluriNSHS19}\cite{DBLP:conf/icse/dataset-icse}, and evaluate \ourtool on it with other automatic tools.
Our evaluation shows that \ourtool greatly outperforms the representative tools using different technologies, in which it achieves higher accuracy and F1 values in detection of various types of vulnerabilities. 
We also evaluate \ourtool on 1700 contracts randomly selected from a real-world dataset.
\ourtool finds 13 contracts deployed on Ethereum with exploitable bugs, 
including 10 contracts with vulnerabilities of transferMint~\cite{transfer} that can evade the detection of current automatic tools to the best of our knowledge.

In summary, this paper makes the following contributions:

{\color{black}

1) We propose a novel framework for achieving automated inference, where finance-related security properties and corresponding models are generated from the source code of a smart contract and used for automated verification.

2) We propose a method for property generation based on a statistical analysis of 30577 smart contracts. We design two types of properties, financial invariant properties and transactional equivalence properties, which correspond to various finance-related vulnerabilities such as transferMint \cite{transfer}, and we abbreviate them as invariant properties and equivalence properties, respectively.

3) We propose modeling methods for our invariant properties and equivalence properties and prove the soundness of verifying these two types of properties using our translated model based on a custom semantics of Solidity \cite{DBLP:conf/sp/JiaoK0S0020}.

4) We implement \ourtool for supporting property generation, modeling and verification, where we embed Z3 into Tamarin prover,  the state-of-the-art tool for verifying security protocols, to use trace properties of reachability and numerical constraint solving for verifying finance-related properties.

5) We evaluate the effectiveness of \ourtool and find 13 contracts with exploitable vulnerabilities using \ourtool.}

\section{Preliminaries} 
\label{sec:background}

\subsection{Smart contracts on Ethereum}
\label{subsec:ethereum}


Ethereum is a blockchain platform that supports two types of accounts: contract accounts, and external accounts.
Each account has an ether balance and a unique address.
A contract account is associated with a piece of code called a smart contract, which controls the behaviors of the account, and a storage that stores global variables denoting the state of the account.
External accounts are controlled by humans without associated code or global variables.

Functions in the smart contracts can be invoked by transactions sent by external accounts.
A transaction is packed into a block by the miner and when that block is published into the blockchain, the function invoked by the transaction is executed.
Functions can also be invoked by internal transactions sent by contract accounts and the sending of an internal transaction can only be triggered by another transaction or internal transaction.

\subsection{Solidity programming language}
\label{subsec:solidty}

\begin{figure}[ht]
\centering
\includegraphics[scale=0.95]{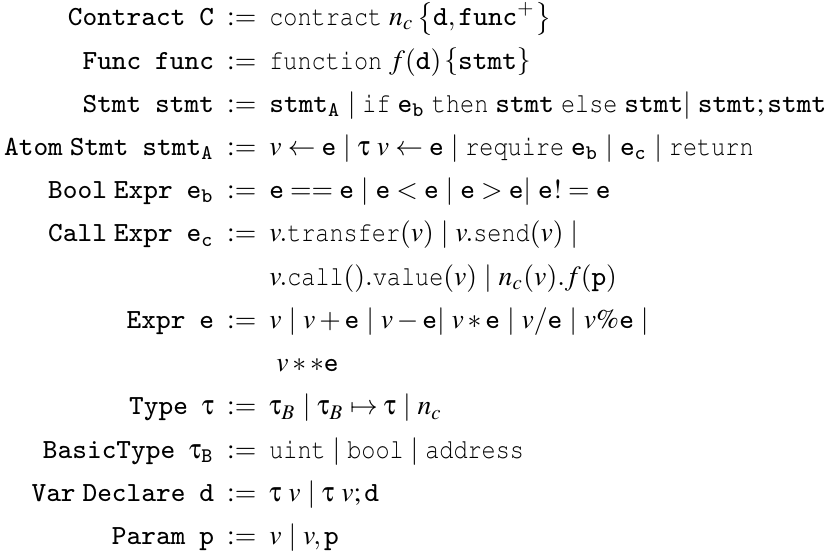}
\vspace{-0.1in}
\caption{Core subset of Solidity}
\label{fig:solidity}
\vspace{-0.08in}
\end{figure}

\begin{figure}[ht]
\centering
\includegraphics[scale=0.25]{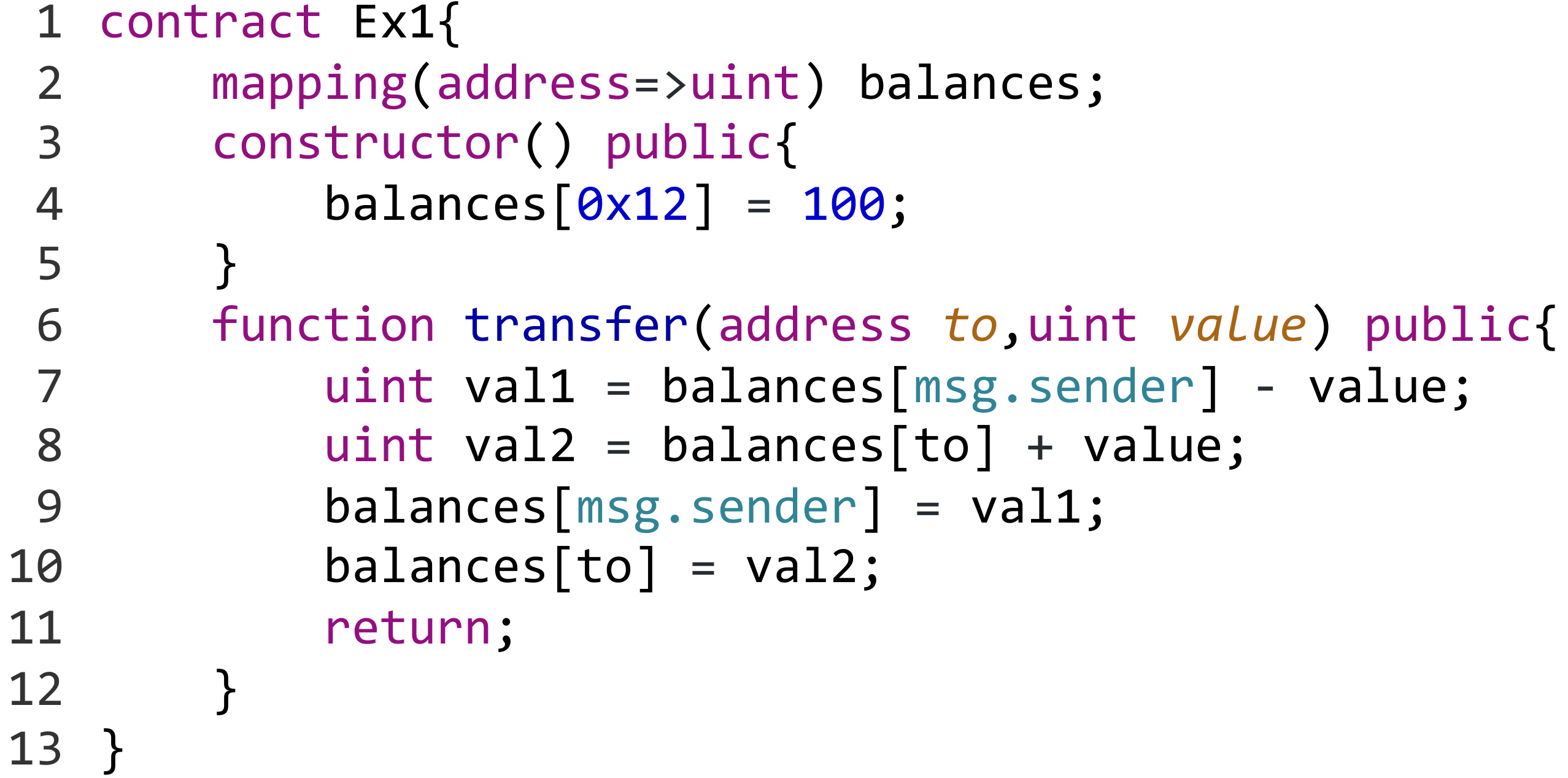}
\vspace{-0.1in}
\caption{Example contract Ex1.}
\label{fig:example1}
\vspace{-0.1in}
\end{figure}

The most popular programming language for Ethereum smart contracts is Solidity \cite{solidity}.
We take the smart contracts written in Solidity as the object of study in this paper.
For brevity, we focus on a core subset of Solidity as shown in Fig.~\ref{fig:solidity}. 
{\color{black}Taking the contract Ex1 in Fig.~\ref{fig:example1} for example, 
a contract consists of declarations of global variables (Line~2) and functions (Line~3 to 12).
Here, \texttt{constructor} is a special function used to initialize global variables.
The function bodies consist of atom statements $\kwp{stmt_{A}}$ and conditional statements.
Taking the function \texttt{transfer} as an example, $\kwp{stmt_{A}}$ can be a declaration statement on Line~7, an assignment statement on line~9, or a return statement on line~11, etc.
Specially, there is a kind of atom statements $\kwp{e_{c}}$ which are used to invoke official functions of Solidity or custom functions of contracts.}
The variables used in contracts fall into the following types: 1) basic types $\tau_{B}$. 2) $\tau_{B}\mapsto \tau$ denoting a mapping from variables of type $\tau_{B}$ to variables of type $\tau$, \textit{e.g.}, \texttt{balances} in Fig.~\ref{fig:example1}.  3) $n_c$ denoting a contract. 
Additionally, there are some special built-in variables of Solidity that cannot be assigned: 1) $\texttt{block.timestamp}$ denoting the timestamp of the block that contains the current transaction. 2) $c.\texttt{balance}$ expressing the ether balance of contract in address $c$. 3) $\texttt{msg.sender}$ denoting the address of the sender of the current transaction. 
Note that the functions of different visibilities are handled in similar ways, so we only introduce how to process public functions in this paper for brevity while \ourtool supports analysis of all kinds of them.

Currently, there is no official formal semantics of Solidity to the best of our knowledge.
{\color{black} Instead, we design \ourtool and prove the soundness of our translation based on a custom semantics of Solidity, named KSolidity \cite{DBLP:conf/sp/JiaoK0S0020}.}
KSolidity is defined using K-framework \cite{Rosu2010}, and the definition of KSolidity consists of 3 parts: Solidity syntax, the runtime configuration, and a set of rules constructed based on the syntax and the configuration.
Configurations form of cells that store information related to the executions of contracts, \textit{e.g.}, the variables of contracts.
The rules specify the transitions of configurations.

\subsection{Multiset rewriting system}
\label{subsec:msr_background}
\ourtool leverages the multiset rewriting system in Tamarin prover \cite{MeierSCB13} to model smart contracts and attackers.
Each state of a multiset rewriting system is a multiset of facts, denoted as $F(t_1, \dots, t_n)$, where $F$ is a fact symbol, and $t_1, \dots, t_n$ are terms.
The transitions of states are defined by labeled rewriting rules.
A labeled rewriting rule is denoted as $l -[a]\rightarrow r$, where $l$, $a$ and $r$ are three parts called premise, action, and conclusion, respectively.
The rule is applicable to state~$s$, {\color{black}if a ground instance $l\sigma$ (where $\sigma$ is a substitution~\cite{DBLP:phd/basesearch/Meier13}) to be a subset of $s$. 
To obtain the successor state $s'$, the ground instance $l\sigma$ is removed and $r\sigma$ is added.}
The action $a$ is also a multiset of facts representing the label of the rule.
Meanwhile, global restrictions on facts in $a$ can be made such that the execution of the protocol can be further restrained.
\section{OVERVIEW} 
\label{sec:overview}

\subsection{Design of \ourtool}
\label{subsec:design}

As shown in Fig. \ref{fig:design}, \ourtool contains 4 modules: 

\begin{figure}[ht]
\centering
\includegraphics[scale=0.34]{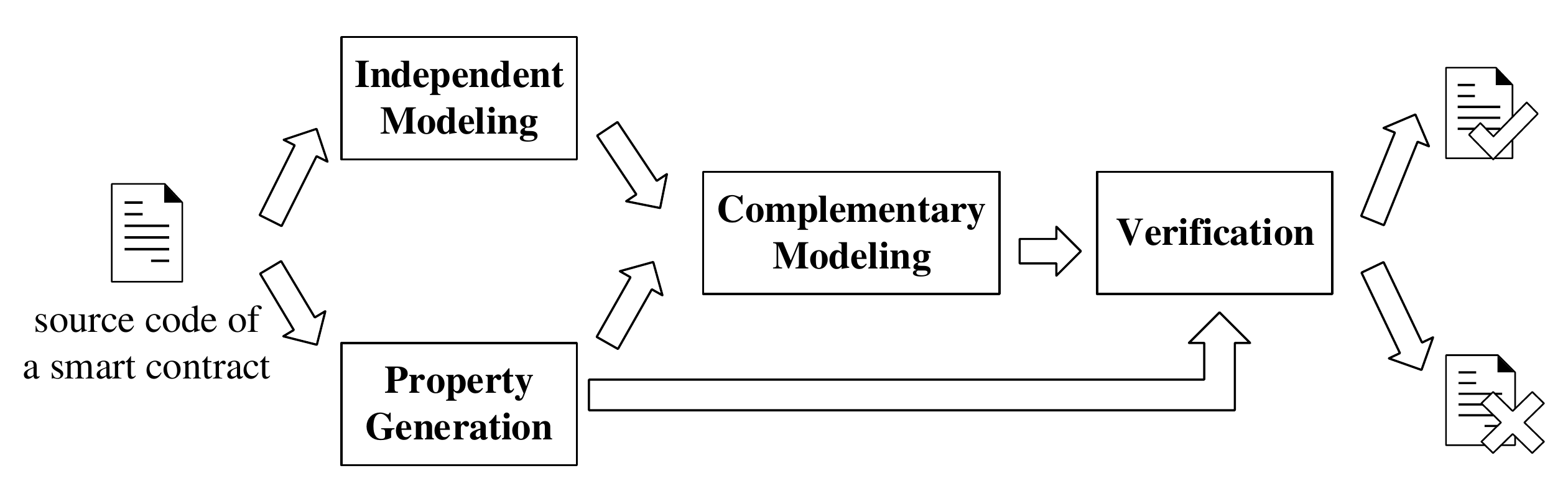}
\caption{Design of \ourtool.}
\label{fig:design}
\vspace{-0.05in}
\end{figure}

\textbf{Independent modeling}: given the source code of a smart contract as input, the module generates a partial model of the contract, which gives the initial state of the running contract and general rules for state transitions.
It translates the contract, as well as the possible behaviors of adversaries, into the model, which is independent of specific security properties.
Note that this partial model cannot be verified directly.

\textbf{Property generation}: \ourtool then generates a set of security properties that the smart contract should satisfy.

\textbf{Complementary modeling}: the module outputs additional rules for each property to complement the partial model, and tries to reduce the size of the model for different properties. 

\textbf{Verification}: we finally design the method of verification to determine whether the properties are valid. We also modify the code of Tamarin prover for supporting the verification where numerical constraint solving is additionally required.

\subsection{Adversary model}
\label{subsec:adversary}
We assume that the adversaries can launch attacks by leveraging the abilities of three types of entities: external accounts, contract accounts and miners.
The concerned attacks on a smart contract are processes that affect the variables related to the smart contract and thus the results of the smart contract executions.
The variables that can be changed by the adversary fall into two categories: some global variables of contracts and \texttt{block.timestamp}.
An external account or a contract account needs to invoke functions in victim contracts to change the values of their global variables, while a miner can manipulate \texttt{block.timestamp} in a range \cite{bestpractice}\cite{15rule}.
In summary, we assume that the adversary can perform the following operations: 
\textbf{C1}. Sending a transaction to invoke any function in victim contracts with any parameters.
\textbf{C2}. Implementing a fallback function to send an internal call message. This message can invoke any function in victim contracts with any parameters. 
\textbf{C3}. Increasing the timestamp of a block by up to 15 seconds~\cite{bestpractice}\cite{15rule}.
{\color{black}
Besides, the changes in exchange rates between tokens and ethers are not considered in \ourtool.}


\section{INDEPENDENT MODELING} 
\label{sec:independent modeling}

Given a smart contract, the module of independent modeling automatically outputs general rules for modeling the executions of the contract and the behaviors of external accounts and the adversaries.
The rules in the multiset rewriting system correspond to the sequences of transitions of the configurations of KSolidity. 
Therefore, we firstly define the terms used in the rules, and sequences using the terms. 
Then, we show the processes of modeling the behaviors using the terms.
Finally, a comprehensive example is given to illustrate the usage of the rules, and discussions are made on technical challenges of property generation and complementary modeling based on the independent modeling.

\begin{figure*}[!]
\centering
\setlength{\abovedisplayskip}{0pt} 
\setlength{\belowcaptionskip}{-0.15in}
\begin{footnotesize}
\begin{align*}
&\kwc{R}(\texttt{function}\ f(\texttt{d})\{\kwp{stmt}\},\varnothing,\omega_0) = \kwc{R}(\kwp{stmt},1,\llbracket\left\langle \sigma_a(f),\kwf{T_c},\kwf{R_o},\kwf{E_n} \right\rangle ,\left\langle \sigma_v(c_b),\kwf{T_v},\kwf{R_l},\kwf{E_n} \right\rangle, \left\langle \sigma_v(\textit{calltype}),\kwf{T_v},\kwf{R_l},\kwf{E_n} \right\rangle, \left\langle \sigma_v(\textit{depth},\kwf{T_v},\kwf{R_l},\kwf{E_n} \right\rangle\rrbracket\\&\qquad\qquad\qquad\qquad\qquad\qquad\qquad\ \  \cdot \omega_0 \cdot  seq(\texttt{d}))  \cup \{ [\kwf{Fr}(\sigma_v(c_b)),\kwf{FR}(\sigma(seq(\texttt{d})))]-[]\rightarrow  [\kwf{Call_e}( \llbracket\omega_0[1],\sigma_a(f),\sigma_v(c_b)\rrbracket\cdot \sigma(seq(\texttt{d})))],\qquad\quad\ \ \ (\kwt{ext\_call}) \\&\ [\kwf{Call_e}( \llbracket\omega_0[1],\sigma_a(f), \sigma_v(c_b)\rrbracket \cdot\sigma(seq(\texttt{d}))),\kwf{Evar}(e(\omega_0)),\kwf{Gvar}(\llbracket \omega_0[1] \rrbracket \cdot g(\omega_0)\backslash e(\omega_0))]-[]\rightarrow [\kwf{Var}_1( \llbracket\sigma_a(f),\sigma_v(c_b),\kwf{EXT},0\rrbracket\cdot\sigma(\omega_0)\cdot \sigma(seq(\texttt{d})))],\ \dots\}\quad\ \\&\qquad\qquad\qquad\qquad\qquad\qquad\qquad\qquad\qquad\qquad\qquad\qquad\qquad\qquad\qquad\qquad\qquad\qquad\qquad\qquad\qquad\qquad\qquad\qquad\qquad\qquad\qquad\qquad\ \  (\kwt{recv\_ext}) \\
&\kwc{R}(v_1\leftarrow v_2;\kwp{stmt},i,\omega) = \kwc{R}(\kwp{stmt},i\circ 1,\omega) \cup \{ [\kwf{Var}_i(\sigma(\omega))]-[]\rightarrow [\kwf{Var}_{i\circ 1}(\sigma(\omega)|\frac{\sigma_v(v_1)}{\sigma_v(v_2)})]\}\qquad\qquad\quad\qquad\qquad\quad\qquad\qquad\quad\qquad\quad\quad (\kwt{var\_assign})  \\
&\kwc{R}(\tau\ v_1\leftarrow v_2;\kwp{stmt},i,\omega) = \kwc{R}(\kwp{stmt},i\circ 1,\omega\cdot \llbracket\left\langle \sigma_v(v_1),\kwf{T_v},\kwf{R_l},\kwf{E_n} \right\rangle\rrbracket ) \cup \{[\kwf{Var}_i(\sigma(\omega))]-[]\rightarrow [\kwf{Var}_{i\circ 1}(\sigma(\omega)\cdot \llbracket\sigma_v(v_2)\rrbracket)]\}\qquad\qquad\qquad\qquad (\kwt{var\_declare})\\
&\kwc{R}(\texttt{return},i,\omega) = \{ [\kwf{Var}_i(\sigma(\omega))]  -[\kwf{Pred\_eq}(\omega[3],\kwf{EXT})]  \rightarrow  [\kwf{Gvar}(\llbracket \omega[5] \rrbracket \cdot g(\omega)\backslash e(\omega)),\kwf{Evar}(e(\omega))],\ \dots\}  \qquad\qquad\qquad\qquad\qquad\qquad\qquad(\kwt{ret\_ext}) \\
\end{align*}
\end{footnotesize}
\vspace{-0.42in}
\caption{Parts of the translation of functions and statements.}
\label{fig:translation}
\end{figure*}

\subsection{Terms and sequences}
\label{subsec:terms}

The terms in multiset rewriting system are translated from the names in Solidity language. 
There are two types of terms: constant terms and variable terms.
Correspondingly, as shown in Fig. \ref{fig:solidity}, a name $v$ in Solidity may represent a contract, a function, a variable, or a constant.
Therefore, given a name $v$, we compute a tuple $\left\langle name,type,range,ether \right\rangle$.
Here, $name$ is a term used in multiset rewriting system, which corresponds to $v$.
Term $type \in \{\kwf{T_v}, \kwf{T_c}\}$. 
If $v$ is a variable, $type=\kwf{T_v}$; otherwise, $type=\kwf{T_c}$.
Term $range \in \{\kwf{R_g},\kwf{R_l}, \kwf{R_o}\}$.
If $v$ is a global variable and a local variable, \textit{i.e.}, a variable defined inside a function, $range=\kwf{R_g}$ and $\kwf{R_l}$, respectively; otherwise, {\color{black} \textit{e.g.}, $v$ is a constant, $range=\kwf{R_o}$.}
If $v$ is a variable representing the ether balance of an account, $ether=\kwf{E_y}$; {\color{black} otherwise $ether=\kwf{E_n}$.
Note that we consider the variables denoting ether balances as global.}
Since value of $v$ is unchanged if $type=\kwf{T_c}$, in this case $name$ is assigned with the value of $v$; otherwise, $name=v$.

Denote $\llbracket e_1, e_2, ..., e_n \rrbracket$ as a sequence, where each element $e_i$ has the same type, \textit{i.e.}, a term, a name, or the aforementioned tuple.
$T_1 \cdot T_2$ represents the concatenation of sequence $T_1$ and $T_2$. 
$T|\frac{t}{t'}$ is a sequence obtained by replacing element $t$ of sequence $T$ with another element $t'$.
$T_1\backslash T_2$ represents a new sequence by removing all the elements in sequence $T_1$ that are the same as those in sequence $T_2$.
We additionally define operations for a tuple sequence $\omega$.
Here, $\omega[j]$ indicates $name$ of the $j$th tuple in $\omega$.
$\sigma(\omega)$ outputs a term sequence consisting of all $name$ in $\omega$.
$g(\omega)$, $e(\omega)$ outputs a term sequence by obtaining the $name$ of all tuples in $\omega$ whose $range=\kwf{R_g}$ and $ether=\kwf{E_y}$, respectively.
The order of terms in $\sigma(\omega), g(\omega), e(\omega),$ are in accordance of the order of terms in $\omega$.



Furthermore, to translate names into terms, we define and implement two functions $\sigma_v,\sigma_a$.
$\sigma_v$ translates a variable name into a variable term, and $\sigma_a$ translates a name that represents a contract, a function, or a constant into a constant term.
\subsection{Modeling the behaviors}
\label{subsec:translation}
Based on the above notations, we propose to model the initialization of contracts and transitions of configurations of KSolidity. 
Specifically, given a contract account of address $c$, we will introduce how to model the executions of functions in the contract codes of the account.
For brevity, we will refer to the account of address $c$ as account $c$ in the following paper.


\textbf{Modeling the initialization.}
 Assume that the contract of account $c$ is deployed on blockchain and the following data will be initialized in the corresponding configuration of KSolidity:
1) the ether balances of account $c$; 2) the global variables of account $c$. 
Besides, the ether balances of other accounts also need to be initialized since they may be modified during the executions of codes of account $c$.  
We use $\omega_0$ to model the configuration of KSolidity after initialization of account $c$.
There are three kinds of tuples in $\omega_0$ in order: 1) {\color{black}$\left\langle \sigma_a(c),\kwf{T_c},\kwf{R_o},\kwf{E_n} \right\rangle$} that represents the address of account $c$; 2) tuple sequence $g(\omega_0)\backslash e(\omega_0)$ denoting the global variables of account $c$ except the variable denoting the ether balance of $c$; {\color{black}3) tuple sequence $e(\omega_0)$ denoting the ether balance of account $c$ and the ether balances of all accounts who have ether exchanges with $c$.}
Therefore, $\omega_0[1] = \sigma_a(c)$.
The tuples in $\omega_0$ are then used to determine the order of parameters of facts in generated rules. Hence we define the following rules to model the initialization:


$[\kwf{FR}(e(\omega_0))]-[\kwf{Init_E}()]\rightarrow [\kwf{Evar}(e(\omega_0))]$
$(\kwt{init\_evars})$

$[\kwf{FR}(g(\omega_0)\backslash e(\omega_0) )]-[\kwf{Init_G}(\omega_0[1])]\rightarrow [\kwf{Gvar}(\llbracket \omega_0[1] \rrbracket \cdot g(\omega_0)\backslash e(\omega_0))]$
$\qquad\qquad\qquad\qquad\qquad\quad\ (\kwt{init\_gvars})$

Here, $\kwf{Evar}$ represents the current ether balances of all accounts on blockchain in initialization.
$\kwf{Gvar}$ represents the current global variables of account $c$. 
For brevity, we use $\kwf{FR}(e(\omega_0))$ to denote a sequence that consists of $\kwf{Fr}(t)$ for all elements $t$ in $e(\omega_0)$.
$\kwf{Fr}(t)$ here is a built-in fact of Tamarin prover \cite{MeierSCB13} that denotes a freshly generated name, 
we use it to denote that term $t$ is with arbitrary initial values.
In practice, the ether balances of all accounts can be initialized once and the global variables can be initialized once for every contract account.
Thus, the restrictions requiring that $\kwt{init\_evars}$ and $\kwt{init\_gvars}$ can be only applied once are added. 

\textbf{Translation of functions.}
After initialization, external accounts can send transactions to invoke any function in the contract of $c$.
To model the invocation of functions, we define $\kwc{R}$ partly shown in Fig. \ref{fig:translation} to recursively translate a function in the contract into rules.
Generally, in each recursive step, $\kwc{R}$ translates a fragment of codes into a rule or multiple rules and leaves the translation of the rest in the next steps. 
The first argument of $\kwc{R}$ represents the codes to be translated.
If the first argument is a sequence of statements, the second argument $i$ is a string encoding the position of the sequence in its function and $i\circ  a$ denotes a string obtained by concatenating $i$ and a string $a$;
otherwise, if the first argument is a function, the second argument is an empty string $\varnothing$.
The third argument is a tuple sequence $\omega$.

\begin{figure}[ht]
\centering
\includegraphics[scale=0.25]{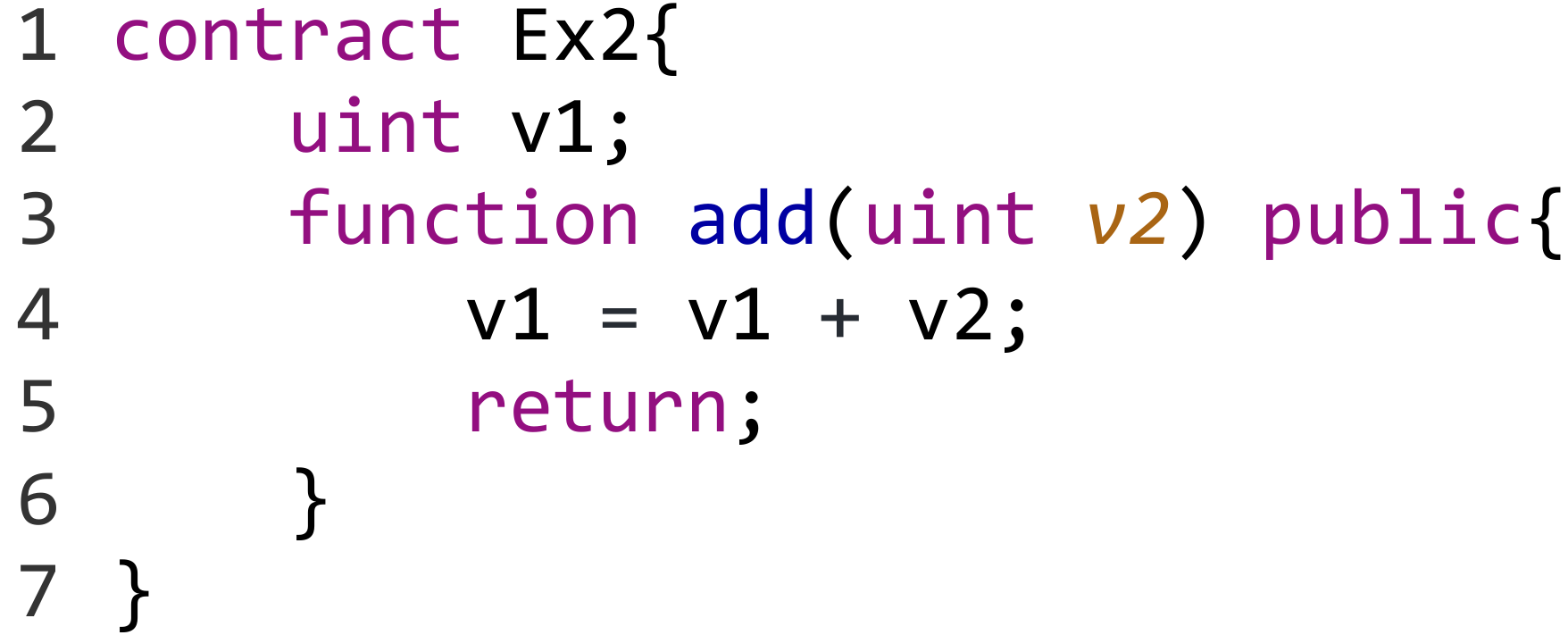}
\vspace{-0.1in}
\caption{Example contract Ex2.}
\label{fig:example2}
\vspace{-0.02in}
\end{figure}

{\color{black}
In the following, we introduce how $\kwc{R}$ translates a function into rules using the function $\texttt{add}$ in Fig.~\ref{fig:example2}  as an example.
Since function $\texttt{add}$ does not modify the ether balance of any account, we omit $\kwf{Evar}$ fact in the rules.

First, $\kwc{R}(\texttt{function}\ \texttt{add}(\texttt{uint}\ v2) \{\kwp{stmt}\},\varnothing,\omega_0)$ is applied and two rules are output, which correspond to  $\kwt{ext\_call}$ and $\kwt{recv\_ext}$ in  Fig. \ref{fig:translation}, respectively: 
\begin{multline}
 [\kwf{Fr}(\sigma_v(c_b)),\kwf{Fr}(\sigma_v(v2))]-[]\rightarrow  [\kwf{Call_e}( \sigma_a(c),\sigma_a(f), \sigma_v(c_b)\\ \qquad\qquad\qquad\qquad\qquad\qquad\qquad\qquad\qquad, \sigma_v(v2))]
\\ [\kwf{Call_e}( \sigma_a(c),\sigma_a(f),\sigma_v(c_b), \sigma_v(v2)),\kwf{Gvar}(\sigma_v(v1))]-[]\\ \rightarrow  [\kwf{Var}_1(\sigma_a(f), \sigma_v(c_b), \kwf{EXT}, 0, \sigma_a(c) \sigma_v(v1), \sigma_v(v2))]\nonumber
\end{multline}}
The first rule denotes an event that an external account $c_b$ sends a transaction to invoke $\texttt{add}$.
Here $seq(\texttt{d}) = \llbracket \sigma_v(v2) \rrbracket$  is a term sequence generated according to the parameter of $\texttt{add}$.
According to C1 in adversary model, $c_b$ and $\sigma_v(v2)$ are initialized by using $\kwf{Fr}$ facts.
The second rule denotes the reception of a transaction.
The $\kwf{Var_1}$ fact represents all the values required in executing $\texttt{add}$, whereby terms in $\kwf{Call_e}$, $\kwf{Gvar}$ are merged into terms in $\kwf{Var_1}$.
Therefore, $\kwc{R}$ also updates $\omega_0$ with a sequence of the corresponding tuples.
Here, $\textit{calltype}$ $\in \{\kwf{EXT}, \kwf{IN}\}$ indicates whether $c_b$ is an external account or a contract account and $\textit{depth}$ denotes current call depth.

Then, $\kwc{R}$ translates the statements in the function into rules for modeling the execution of the function $\texttt{add}$.
{\color{black}
The assignment statement in line 4 is translated into the following rule, which corresponds to $\kwt{var\_assign}$ in Fig. \ref{fig:translation}:
\begin{align}
 &[\kwf{Var}_1(\sigma_a(f), \sigma_v(c_b), \sigma_v(\textit{calltype}), \sigma_v(\textit{depth}), \sigma_a(c), \sigma_v(v1),\nonumber \\ & \sigma_v(v2))]
 -[]\rightarrow  [\kwf{Var}_{11}(\sigma_a(f), \sigma_v(c_b), \sigma_v(\textit{calltype}), \sigma_v(\textit{depth})\nonumber\\ &\qquad\qquad\qquad\qquad\qquad, \sigma_a(c), \sigma_v(v1)\oplus \sigma_v(v2), \sigma_v(v2))]\nonumber
\end{align}
The term $\sigma_v(v_1)$ is replaced by $\sigma_v(v_1)\oplus \sigma_v(v_2)$ when applying the rule.
Here $\oplus$ is translated from the operator $+$ and introduced in Appendix \ref{subsec:appendix1}.}

{\color{black}
Additionally, the return statement in line 5 is translated into the following rule corresponding to $\kwt{ret\_ext}$:
\begin{align}
  &[\kwf{Var}_{11}(\sigma_a(f), \sigma_v(c_b), \sigma_v(\textit{calltype}), \sigma_v(\textit{depth}), \sigma_a(c), \sigma_v(v1),\nonumber \\ &  \sigma_v(v2))]\ -[\kwf{Pred\_eq}(\sigma_v(\textit{calltype}),\kwf{EXT})]\rightarrow  [\kwf{Gvar}(\sigma_v(v1))]\nonumber
\end{align}}
The term $\sigma_v(v1)$ denoting the global variable of contract \texttt{Ex2} is put into $\kwf{Gvar}$ facts. 
The local variables will no longer be used and the corresponding terms will not be maintained.
Here, $\kwf{Pred\_eq}$ is a fact denoting equality between terms \cite{MeierSCB13}. 
We use it to determine whether $\sigma_v(\textit{calltype})$ is equal to $\kwf{EXT}$, corresponding to the case that the function is invoked by external accounts.
Similarly, this statement can be translated into a rule denoting the case that the function is invoked by contract accounts as shown in Appendix \ref{subsec:appendix1}.

\textbf{Adversaries.}
Here we introduce the modeling of the capability C1 and C2 of adversaries mentioned in Section~\ref{subsec:adversary}, and 
the modeling of C3 will be introduced in Appendix \ref{subsec:appendix2}.

\textbf{C1:} The operation that an adversary, besides normal participants, sends transactions can also be modeled by  $\kwt{ext\_call}$. Therefore, no additional rules for the operation are provided. 

\textbf{C2:} For each function $f$ in the contract of account $c$, multiple rules are generated to indicate that if the fallback function of the adversary is triggered by the execution of the contract of $c$, the adversary can send an internal transaction to invoke any function $f$ in the contract of $c$. The details of these rules are shown in Appendix~\ref{subsec:appendix1}.

\begin{figure}[]
\setlength{\belowcaptionskip}{-0.02in}
\setlength{\abovecaptionskip}{-0.02in}
\includegraphics[scale=0.33]{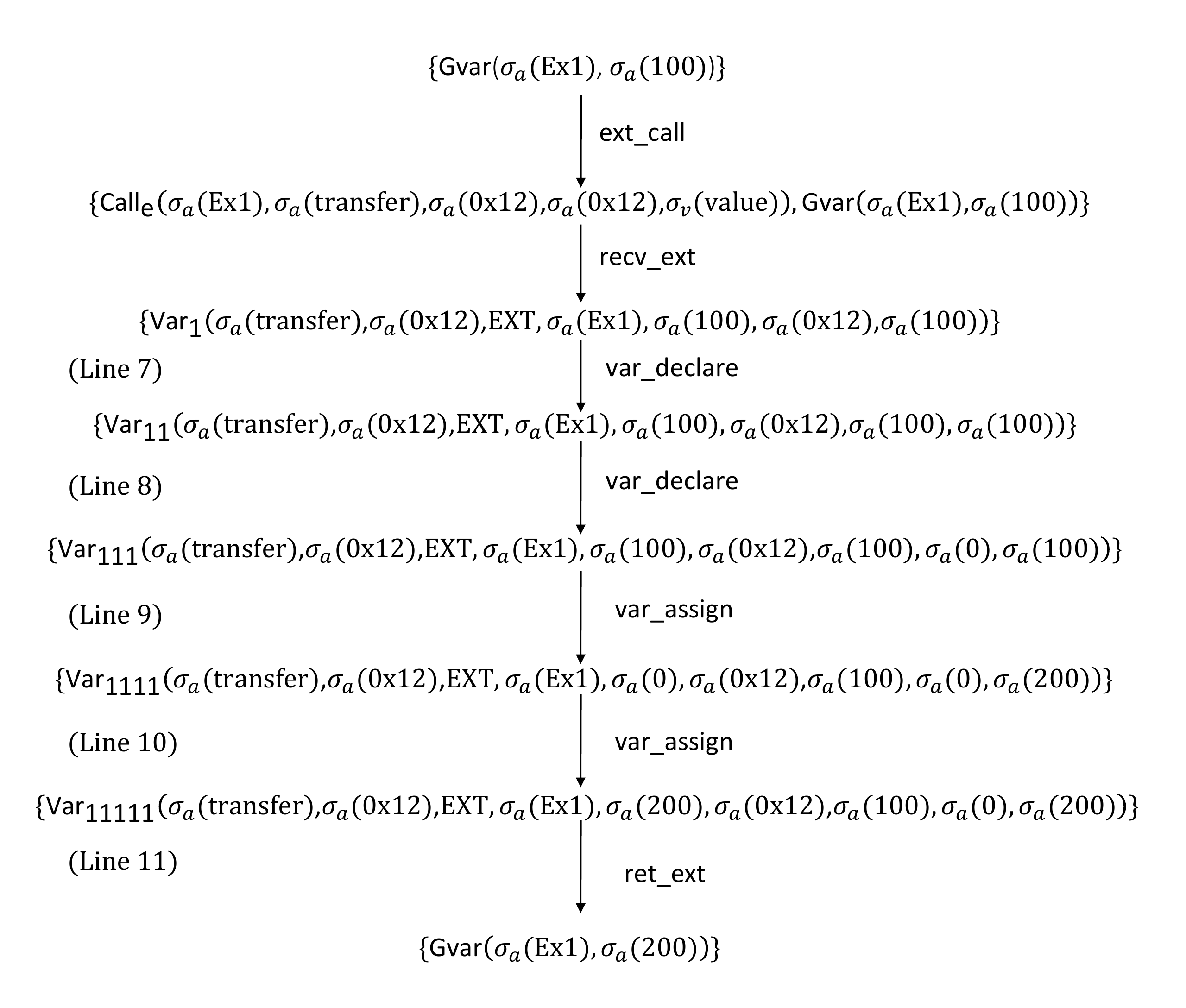}
\caption{The execution that models an attack on Ex1.}
\label{fig:path1}
\end{figure}


\subsection{An illustrative example}
\label{subsec:illu example}


Fig. \ref{fig:example1} shows a simplified version of a practical smart contract, which is with a vulnerability of transferMint \cite{transfer}. 
The global variable $\texttt{balances}$ denotes the token balances of accounts.
{\color{black} When the function \texttt{transfer} is invoked,  the token balance of \texttt{msg.sender} is supposed to decrease when the token balance of \texttt{to} increases.
However, assume that the account on address $0x12$ invokes \texttt{transfer} with the parameter \texttt{to} = $0x12$, $\texttt{balances}[0x12]$ will increase while the balance of no other account will decrease.
By exploiting this vulnerability, the account on address $0x12$ can mint tokens for profit or eventually make this type of tokens valueless through repeated attacks.}
An execution of the model that corresponds to the attack is shown in Fig. \ref{fig:path1}.
We use the contract name $\texttt{Ex1}$ to denote the address of the account who owns this contract.
Since function $\texttt{transfer}$ does not modify the ether balance of any account, we omit $\kwf{Evar}$ fact in the figure.
Hence, in the execution, the initial state is $\{ \kwf{Gvar}(\sigma_a(\texttt{Ex1}),\sigma_a(100)) \}$ where $\omega_0[1]  = \sigma_a(\texttt{Ex1})$ and $g(\omega_0)\backslash e(\omega_0)$ $= \llbracket \sigma_a(100) \rrbracket$.
Next, an external account invokes $\texttt{transfer}$ whereby the rule $\kwt{ext\_call}$ is applied such that $\kwf{Call_e}(\sigma_a(\texttt{Ex1}), \sigma_a(\texttt{transfer}),$ $\sigma_v(c_b), \sigma_v(\texttt{to}),\sigma_v(\texttt{value}))$ is added to the new state. 
Since $\sigma_v(c_b)$, $\sigma_v(\texttt{to})$ and $\sigma_v(\texttt{value})$ can be arbitrary values, 
in this execution, they can be instantiated as $\sigma_a(0x12)$, $\sigma_a(0x12)$ and $\sigma_c(100)$ respectively.
In the following steps, the state is updated in similar ways. 
When the transaction invoking $\texttt{transfer}$ finishes,  the state is $\kwf{Gvar}(\sigma_a(\texttt{Ex1}),\sigma_a(200))$, which implies that $\texttt{balances}[0x12]$ changes into 200 in an unexpected way.
Note that the numerical instantiation cannot be supported by the original Tamarin prover. Moreover, the independent model cannot be verified directly to find an attack as shown in Fig. \ref{fig:path1}, since several technical challenges need to be addressed.

\subsection{Technical challenges and main solutions}
\label{subsec:challenge}
Since the module of independent modeling only provides a framework that automatically generates models of smart contracts partially, 
we have to address the following technical challenges to complement the model for the verification.

\textbf{Challenge 1: recognizing security requirements.} 
Given an execution shown in Fig. \ref{fig:path1}, a corresponding property is still needed for the verifier to recognize this execution as an instance of some vulnerabilities.
However, there is no uniform standard for the security requirements of contracts in practical scenarios, which makes the precise generation of security properties difficult. 
There are automated bug-finding tools and verifiers defining patterns or properties according to known vulnerabilities \cite{TsankovDDGBV18}~\cite{0001LC18}~\cite{DBLP:conf/ndss/KalraGDS18}.
However, the vulnerabilities covered by these tools are limited to known ones, and a variant of a known vulnerability may evade their detection~\cite{RodlerLKD19}.


To address this challenge,
we perform statistical analysis on 30577 real-world smart contracts and obtain an observation: most of the smart contracts (91.11\%) are finance-related, \textit{i.e.}, the executions of these contracts may change the cryptocurrencies of themselves and others.
Therefore, we divide the smart contracts into different categories according to the cryptocurrencies that they use and propose security properties to check whether the cryptocurrencies may be lost unexpectedly.

\textbf{Challenge 2: contract-oriented automated reasoning.}
Given an independent model, the rule $\kwt{ext\_call}$ can be applied repeatedly, which is corresponding to the practical scenarios that a function can be invoked any times.
This may lead to non-termination of verification.
Besides, the independent model is insufficient for verifying 2-safety properties \cite{DBLP:conf/csfw/BartheDR04}.



We address the challenge based on the fact that a transaction is atomic and cannot be interfered by other transactions. Therefore, the independent model can be reduced for different types of properties: (1) the properties that should be maintained for a single transaction; (2) the properties that may be affected by other transactions. 
For the first type, we propose to automatically generate invariant properties and the corresponding reduced model that the behaviors of other transactions are ignored.
For the second type, since a transaction is atomic, the rest way to trigger an attack is to leverage different results of a sequence of transactions caused by different orders of the transactions or different block variables.  
Hence, we propose the equivalence properties and also the modeling method to achieve effective automated reasoning. 
We also modify the code of Tamarin for supporting the verification where numerical constraint solving is additionally required.

\section{PROPERTY GENERATION} 
\label{sec:property generation}



To address Challenge 1, we divide finance-related smart contracts into three categories according to the type of cryptocurrencies they use: ether-related, token-related, and indirect-related.
An ether-related contract may transfer or receive ethers, \textit{i.e.}, the official cryptocurrency of Ethereum.
Similarly, a token-related contract may send or receive tokens, \textit{i.e.}, the cryptocurrency implemented by the contract itself.  
An indirect-related contract is used in the former two contracts to provide additional functionality.
Hence, to check whether the cryptocurrencies may be lost in unexpected ways, we focus on generating security properties for ether-related and token-related contracts.
We propose to recognize the category and key variables related to cryptocurrencies from the codes of smart contracts and use the information to generate the security properties.
Note that we analyze the indirect-related contracts in an indirect way and do not generate properties for the indirect-related contracts.
For example, given a token contract $C_1$ and an indirect-related contract $C_2$, assume that $C_1$ implements authentication by invoking functions in $C_2$.
In this case, we can specify $C_1$ to be analyzed and then generate a model for it, which considers the interaction of $C_1$ and $C_2$.

\subsection{Recognizing categories and key variables}
\label{subsec:category of contracts}

\textbf{Ether-related contracts.} 
The ethers can be transferred by using official functions, \textit{e.g.}, $\texttt{transfer}$, $\texttt{send}$ and $\texttt{call}$.
The modifier $\texttt{payable}$ is only used in ether-related contracts for receiving ethers.
Therefore, \ourtool recognize an ether-related contract by determining if there are keywords, \textit{i.e.}, $\texttt{transfer}$, $\texttt{send}$, $\texttt{call}$ or $\texttt{payable}$ in the contract.
If a contract is recognized as an ether-related contract, then we use the built-in variable $\texttt{balance}$ as the key variable, which denotes the ether balance of an account.

\textbf{Token-related contracts.}
The token-related contracts can be divided into token contracts and token managing contracts.
A token contract is used to implement a kind of customized cryptocurrency, \textit{i.e.}, tokens, which can be traded and have financial value.
A token managing contract, \textit{e.g}, an ICO contract \cite{ico}, is used to manage the distribution or sale of tokens.

We propose a method to recognize token contracts based on another observation from our statistical result in Section~\ref{subsec:statistical}: developers tend to use similar variable names to represent the token balance of an account.
Therefore, a contract is identified as a token contract, if there is a variable of type $\texttt{mapping(address=>uint)}$ with a name similar to two commonly used names: $\texttt{balances}$ or $\texttt{ownedTokenCount}$.
Specifically, we calculate the similarity of names using Python package fuzzywuzzy\cite{fuzzywuzzy}. 
When the similarity is larger than 85, we consider two names similar. 
The threshold 85 is set based on our evaluation in Section \ref{subsec:statistical}.
For the contracts using uncommon names for token balances, \ourtool also supports the users to provide their own variable names.
In addition to $\texttt{balances}$, we observe that some token-related contracts define a variable of $\texttt{uint}$ type to record the total number of tokens.
Similarly, we use the most common variable name $\texttt{totalSupply}$ to match the variables representing the total amount of tokens. 
This kind of variables are not used to recognize the token contract, but rather for the subsequent generation of properties.
After the recognition of token contracts, we search for contracts instantiating token contracts and regard them as token managing contracts.

\subsection{Generating security properties}
\label{subsec:generating properties}
As mentioned in Section \ref{subsec:challenge}, we propose two kinds of properties: invariant properties and equivalence properties.




\textbf{Invariant properties.}
The invariant property requires that for any transaction a proposition (a statement that denotes the relationship between values of variables) $\phi$ holds when the transaction finishes, if $\phi$ holds when the transaction starts executing.
Since a transaction is atomic, 
\ourtool checks invariant properties in single transactions instead of the total executions to achieve effective automated reasoning.
Here, we design the invariants to ensure that the token balances in token-related contracts are calculated in an expected way. 
Note that we do not design invariant properties for ether-related contracts, since the calculation of ether balances is performed by the EVM and its correctness is guaranteed~\cite{evmgo}.




For a token contract with key variable $\kwe{balances}$, the following invariant is generated: 

$\qquad\quad\ \ \ \ \ \sum_{a\in A_1}\kwe{balances}(a) = C_1\qquad\qquad (\kwt{token\_inv})$

Since a transaction can only affect a limited number of accounts, $A_1$ is the set of addresses of the accounts whose token balances may be modified in the transaction. $C_1$ is an arbitrary constant value and the invariant implies that the sum of token balances of all accounts should be unchanged after a transaction.
If the invariant is broken, it indicates an error in the process of recording token balances, which would make this kind of tokens worthless \cite{tokenzero}.
Here, $\kwe{balances}$ can be replaced by any variable name denoting the token balances.
If there are multiple variables denoting token balances of different types, all of them will be used.
Specially, if there is a key variable $\kwe{totalSupply}$ denoting the total amount of tokens in the token contract, the constant $C_1$ in $\kwt{token\_inv}$ will be replaced by $\kwe{totalSupply}$.
For a token managing contract, the invariant $\kwt{token\_inv}$ is generated for the token contract that it manages.
\ourtool also supports the users to provide customized invariants to check the security of contracts.

\textbf{Equivalence properties.}
We define the equivalence property as follows: The equivalence of a global variable $v$ holds for a transaction sequence $T$, if the value of $v$ after $T$'s execution is always the same.
Here we study the equivalence of the token or ether balance of the adversary.
Given two sequences $T_A$ and $T_B$ that have the same transactions,
we propose the following property:

$\kwe{balances}_A(c_{adv}) = \kwe{balances}_B(c_{adv}) \land  $

$\quad\ \ \kwe{balance}_A(c_{adv}) = \kwe{balance}_B(c_{adv}) \qquad(\kwt{equivalence})$ 

Here, denote $\kwe{balances}_A(c_{adv})$ and $\kwe{balance}_A(c_{adv})$ as the token balance and ether balance of the adversary after execution of $T_A$, respectively.
Similarly, $\kwe{balances}_B(c_{adv})$ and $\kwe{balance}_B(c_{adv})$ represent the corresponding balances for $T_B$. 
$\kwt{equivalence}$ requires that the adversary cannot change its own balances by changing the orders of transactions or other conditions; otherwise, the difference of the balances may be the illegal profit of the adversary.

\subsection{Relationship between properties and common vulnerabilities}
\label{subsec:relationship}

\begin{figure}[ht]
\centering
\includegraphics[scale=0.25]{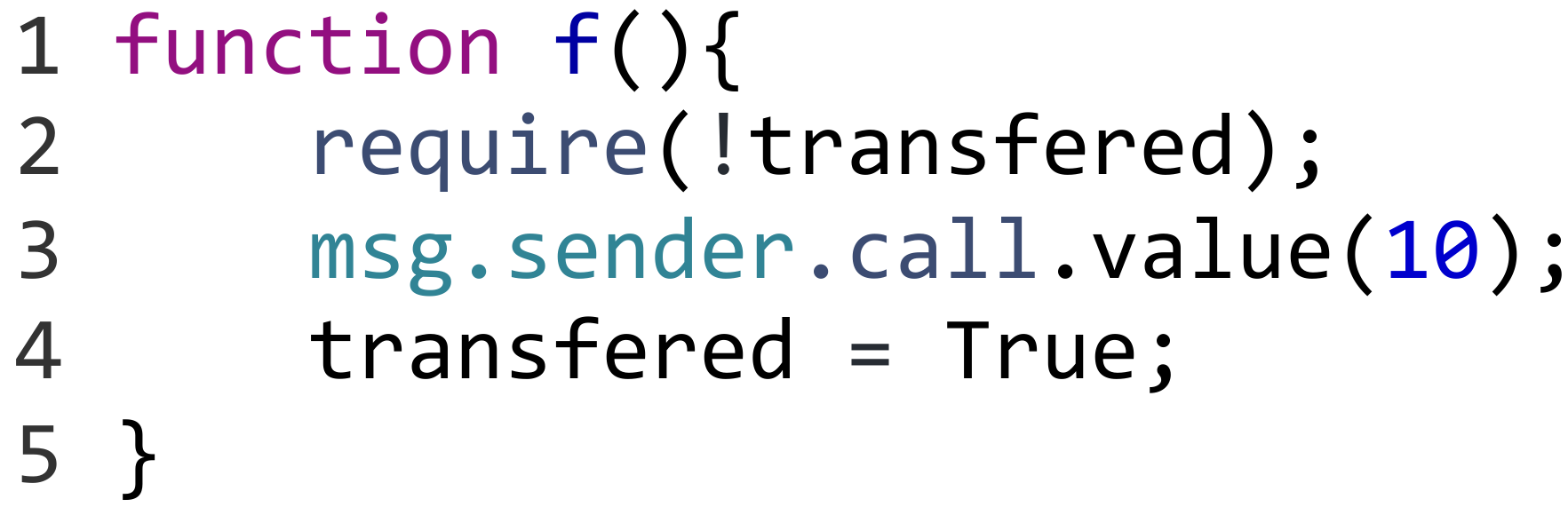}
\caption{An example with reentrancy vulnerability.}
\label{fig:exre}
\vspace{-0.05in}
\end{figure}

The properties of \ourtool are designed with a basic idea: leveraging the phenomenon that the loss of ethers and tokens is one of the popular intentions of attackers \cite{DBLP:journals/csur/ChenPNX20}.
\ourtool generates properties based on key variables denoting the token balances or ether balances. 
We aim to cover vulnerabilities causing financial losses.
As a result, \ourtool covers 6 types of vulnerabilities, including \textit{transferMint} not supported by existing automatic tools, through the two properties.
Note that these vulnerabilities do not necessarily cause financial loss and those that do not are ignored by \ourtool as they do not affect the financial security of contracts.

To explain the usage of our properties, we provide examples of several common vulnerabilities, detailing how contracts with these vulnerabilities violate the above two properties.

\textbf{Gasless send.}
During the executions of official functions \texttt{send} and \texttt{call}, if the gas is not enough, the transaction will not be reverted and a result will return.
If a contract does not check the execution results of \texttt{send} or \texttt{call}, it may mistakenly assume that the execution was successful.
Given two sequences of same transactions that invoke functions with gasless send vulnerability, one with sufficient gas and one without sufficient gas, the results of them will be different.
Therefore, the equivalence property is violated.


\textbf{Reentrancy.}
Taking the contract in Fig. \ref{fig:exre} as an example, suppose that the adversary sends a transaction to invoke \texttt{f} and the statement on line 4 sends ethers to the adversary.
According to Section \ref{subsec:adversary}, the adversary can then send an internal transaction through the fallback function to call \texttt{f} again, and since the code on line 5 is not executed, the check on line 3 will still be passed, allowing the adversary to get ethers one more time.
Assume that there are two sequences $T_1$ and $T_2$, $T_1$ consisting of two transactions invoking \texttt{f} and $T_2$ consisting of one transaction invoking \texttt{f} and one internal transaction invoking \texttt{f} through a reentrancy vulnerability.
We treat the internal transaction sent by the adversary as a transaction and consider $T_1$ and $T_2$ as consisting of the same transactions.
In this case, 
the ether balances of adversary after $T_1$ and $T_2$  are different, which means that the equivalence property is broken.

\textbf{TD\&TOD.}
When some statements are control dependent on the  \texttt{block.timestamp}, the adversary can control the execution of these statements by modifying \texttt{block.timestamp} in a range, which is called TD.
Given two same sequences of transactions that invoke a function with TD, the execution results may be different with different \texttt{block.timestamp}, which violates the equivalence property.
Similarly, the contracts with TOD vulnerability violate our equivalence property when the order of transactions changes.

\textbf{Overflow/underflow.}
Overflow/underflow is a kind of arithmetic error.
Since the goal of \ourtool is to analyze the financial security of contracts, \ourtool detects overflow/underflow vulnerabilities that can change the number of tokens.
For the remaining overflow/underflow vulnerabilities, 
\ourtool can also support them through custom invariants.

{\color{black}

Certain new vulnerability can be detected directly by \ourtool if it is covered by our properties, such as the transferMint vulnerability. 
If the vulnerability is not covered, we need to propose new properties or modify the rules in our models to support more features. 
For example, the airdrop hunting vulnerability~\cite{DBLP:conf/uss/ZhouYXCY020}, which is used by attackers to collect bonuses from airdrop contracts, is not currently supported by \ourtool. 
To extend \ourtool to cover airdrop hunting, we can propose a new invariant requiring the number of contract accounts to remain zero. 
However, it is challenging to model the identification of contract accounts.
We would like to study the extension of \ourtool in our future work.}
\section{COMPLEMENTARY MODELING AND VERIFICATION} 
\label{sec:complementary modeling and verification}

In this section, we introduce how we address Challenge 2. 
According to different properties of a contract, we propose the method of complementary modeling to generate customized models built upon the independent models with rules replaced or added.
Besides, we propose a solution to check whether a customized model satisfies the corresponding property. 

\subsection{Complementary modeling}
\label{subsec:complementary}
{\color{black}The goal of complementary modeling is to generate a customized model, which satisfies that the invariant property or equivalence property is not valid in the KSolidity Semantics, only if there exists an execution in the model that breaks the property. }
Besides, to support automated verification, the model is added with more constraints such that each execution that reaches a certain state breaks the property. 
Then, the property is not valid if and only if the state is reachable.
Hence, we design the method for invariant property and equivalence property as follows.

\textbf{Invariant properties.}
The generated model for invariant properties has the following features:

\textbf{i)} The invariant holds at the beginning of any execution. 
\textbf{ii)} An execution simulates the execution of one transaction.
\textbf{iii)} The invariant is assumed to be broken at the end of any execution, which corresponds to the state that breaks the property.


To make the generated model conform to feature i), we first replace the rule $\kwt{init\_gvars}$ with rule $\kwt{init\_gvars\_inv}$.
In rule $\kwt{init\_gvars\_inv}$, a fact $\theta_e(\phi)$ is added to denote that the invariant $\phi$ holds after the initialization. 
Here, $\phi$ is the invariant  $\kwt{token\_inv}$ in Section \ref{subsec:generating properties} and $\theta_e(e)$ is a function translating mathematical expressions into numerical facts in rules.
Numerical facts denote the relationships between numeric variables and are processed in the verification module. 
Similarly, we define $\theta_{ne}(e)$ to translate the negation of $e$.


Then we replace the rule $\kwt{ext\_call}$ with  $\kwt{ext\_call\_inv}$, added with an action and a restriction requiring that the rule can be applied only once, to achieve feature ii).


Finally, we modify the rule $\kwt{ret\_ext}$ into $\kwt{ret\_ext\_inv}$.
$\kwt{ret\_ext}$ $\kwt{\_inv}$ has additional facts $\theta_{ne}(\phi)$ and $\kwf{End}()$ compared to $\kwt{ret\_ext}$, which together achieve feature iii).
$\kwf{End}()$ serves as an indicator that an execution of the model reaches the end of the transaction if rule $\kwt{ret\_ext\_inv}$ is applied, and $\theta_{ne}(\phi)$  means that invariant $\phi$ is broken at the same time.

\textbf{Equivalence properties.}
The generated model for equivalence properties has the following features:
\textbf{i)} An execution of the model simulates the executions of two sequences $T_A$ and $T_B$ consisting of the same transactions but possibly with different orders.
\textbf{ii)} Before the executions of $T_A$ and $T_B$, the values of global variables and ether balances of all accounts are the same.
\textbf{iii)} The ether or token balances of the adversary are assumed to be different at the end of any execution, which corresponds to the state that breaks the equivalence property.


Firstly, to achieve feature ii), we replace  $\kwt{init\_evars}$ and $\kwt{init\_}$ $\kwt{gvars}$ with $\kwt{init\_evars\_AB}$ and $\kwt{init\_}$$\kwt{gvars\_AB}$, respectively. 
In $\kwt{init\_gvars\_AB}$, the $\kwf{Gvar}$ fact is duplicated into $\kwf{Gvar_A}$ and $\kwf{Gvar_B}$ facts, which indicates that the global variables are the same before $T_A$ and $T_B$.
Similarly, the $\kwf{Evar}$ fact is duplicated in $\kwt{init\_evars\_AB}$.

Then, we replace $\kwt{ext\_call}$ with $\kwt{ext\_call\_AB}$, in which $\kwf{Call_e}$ fact is duplicated, indicating that two transactions with same parameters and same sender are sent.



Except rules $\kwt{init\_evars}$, $\kwt{init\_gvars}$, $\kwt{ext\_call}$, each of the remaining rules in the model is replicated into two rules, and the facts of the two rules are added with different subscripts $A$ and $B$ to represent the execution of transactions in sequences $T_A$ and $T_B$, respectively. 
For example, the rule $\kwt{recv\_ext}$ is replaced with $\kwt{recv\_ext\_A}$ and $\kwt{recv\_ext\_B}$.
Specially, actions and restrictions are added into $\kwt{recv\_ext\_A}$ and $\kwt{recv\_ext\_B}$ to achieve feature i).
The complete form of the above rules is shown in Appendix \ref{subsec:appendix2}.

Finally, to achieve feature iii), we add a rule $\kwt{compare\_AB}$ to compare the ether balances and token balances of the adversary,
where $\theta_{ne}(\phi_{equ})$ and $\kwf{End}()$ are added for subsequent verification. 
$\phi_{equ}$ is property $\kwt{equivalence}$ in Section \ref{subsec:generating properties}.

\subsection{Verification}
\label{subsec:verification}
The verification module is implemented by modifying the source code of Tamarin prover \cite{MeierSCB13} to achieve modeling using multiset rewriting rules with additional support for numerical constraint solving by Z3~\cite{z3}.
Taking a generated property and the corresponding model as input, the workflow of this module is as follows:
1) Search for an execution that reaches $\kwf{End}()$ without considering the numerical constraints.
2) If the search fails, the module terminates and outputs that the property is valid; otherwise, go to step 3).
3) Collect the numerical constraints that the execution must satisfy and solve the constraints by Z3.
4) If the set of constraints is satisfied, which indicates that the execution that violates the property exists, the module terminates and outputs the execution as a counterexample;
otherwise, add a constraint to the model that the execution does not exist, and go to step 1). 

\subsection{Formal guarantee}
\label{subsec:proof_intro}
 {\color{black}We prove the soundness of translation from Solidity language to our models based on KSolidity \cite{DBLP:conf/sp/JiaoK0S0020},} which is claimed to fully cover the high-level core language features specified by the official Solidity documentation and be consistent with the official Solidity compiler.
 {\color{black}However, the completeness of our translation is not guaranteed due to two reasons: 1) the initialization of global variables and ether balances in rules $\kwt{init\_evars}$ and $\kwt{init\_gvars}$ assumes the initial values of global variables and ether balances to be arbitrary, which may over-approximate the range of values for these variables. 2) the specific values of the block timestamps are not considered.}
Specifically, we prove Theorem 1 (informal description). 
Note that Theorem 1 only holds for the contracts supported by \ourtool (See Section \ref{sec:limitations}).
The precise description of Theorem 1 is presented by Theorem 2 and Theorem 3 in Appendix \ref{sec:proof}.



\begin{theorem}(Soundness).
    \label{theorem: soundness}
    If an invariant property (or equivalence property) holds in the complementary model of \ourtool,  it holds in real-world transactions interpreted by KSolidity semantics.
\end{theorem} 

\begin{proof}
     See Appendix \ref{sec:proof}.
\end{proof}

\section{EVALUATION} 
\label{sec:evaluation}

\begin{table*}[]
 \setlength\tabcolsep{4pt}
 \centering 
\caption{A comparison of representative automated analyzers for smart contracts. (Acc and F1 outside brackets correpsond to the finance-vulnerable contracts, while those inside brackets correpsond to the vulnerable contracts, * denote automated verifiers)}  
\label{comp}
\vspace{-0.1in}
\resizebox{\linewidth}{!}{
\begin{tabular}{|c|cc|cc|cc|cc|cc|cc|cc|cc|cc|ccc|}
\hline
\multirow{2}{*}{\begin{tabular}[c]{@{}c@{}}Types of \\ Vulnerabilities\end{tabular}} & \multicolumn{2}{c|}{Osiris}                                            & \multicolumn{2}{c|}{SECURIFY*}                                          & \multicolumn{2}{c|}{Mythril}                                            & \multicolumn{2}{c|}{OYENTE}                                            & \multicolumn{2}{c|}{VERISMART}      & \multicolumn{2}{c|}{SmartCheck}                                        & \multicolumn{2}{c|}{Slither}                   & \multicolumn{2}{c|}{Manticore}                         & \multicolumn{2}{c|}{eThor*}                                             & \multicolumn{3}{c|}{\ourtool*}                                                                     \\ \cline{2-22} 
                                                                                     & \multicolumn{1}{c|}{Acc(\%)}                  & F1                     & \multicolumn{1}{c|}{Acc(\%)}                  & F1                     & \multicolumn{1}{c|}{Acc(\%)}                  & F1        & \multicolumn{1}{c|}{Acc(\%)}                  & F1              & \multicolumn{1}{c|}{Acc(\%)}                  & F1                     & \multicolumn{1}{c|}{Acc(\%)} & F1   & \multicolumn{1}{c|}{Acc(\%)}                  & F1                     & \multicolumn{1}{c|}{Acc(\%)}                  & F1                     & \multicolumn{1}{c|}{Acc(\%)}                  & F1                     & \multicolumn{1}{c|}{Acc(\%)}                  & \multicolumn{1}{c|}{F1}                     & U  \\ \hline
TOD-eth                                                                              & \multicolumn{1}{c|}{/}                        & /                      & \multicolumn{1}{c|}{96.43}                    & 0.98                   & \multicolumn{1}{c|}{/}                        & /                       & \multicolumn{1}{c|}{42.86}                    & 0.6                    & \multicolumn{1}{c|}{/}       & /    & \multicolumn{1}{c|}{/}                        & /      & \multicolumn{1}{c|}{/}                        & /                 & \multicolumn{1}{c|}{/}                        & /                      & \multicolumn{1}{c|}{/}                        & /                      & \multicolumn{1}{c|}{100}                      & \multicolumn{1}{c|}{1}                      & 10 \\ \hline
TOD-token                                                                            & \multicolumn{1}{c|}{/}                        & /                      & \multicolumn{1}{c|}{/}                        & /                      & \multicolumn{1}{c|}{/}                        & /                  & \multicolumn{1}{c|}{/}                        & /      & \multicolumn{1}{c|}{/}                        & /                      & \multicolumn{1}{c|}{/}       & /    & \multicolumn{1}{c|}{/}                        & /                      & \multicolumn{1}{c|}{/}                        & /                      & \multicolumn{1}{c|}{/}                        & /                      & \multicolumn{1}{c|}{100}                      & \multicolumn{1}{c|}{1}                      & 0  \\ \hline
TD                                                                                   & \multicolumn{1}{c|}{\makecell{71.60\\(70.37)}} & \makecell{0.83\\(0.82)} & \multicolumn{1}{c|}{/}                        & /                      & \multicolumn{1}{c|}{\makecell{45.68\\(44.44)}} & \makecell{0.62\\(0.62)} & \multicolumn{1}{c|}{\makecell{76.54\\(75.31)}} & \makecell{0.87\\(0.86)} & \multicolumn{1}{c|}{/}       & /    & \multicolumn{1}{c|}{/}                        & /                      & \multicolumn{1}{c|}{\makecell{16.05\\(14.81)}} & \makecell{0.26\\(0.25)} & \multicolumn{1}{c|}{\makecell{24.69\\(23.46)}} & \makecell{0.38\\(0.38)}  & \multicolumn{1}{c|}{/}                        & /                      & \multicolumn{1}{c|}{\makecell{95.06\\(93.83)}} & \multicolumn{1}{c|}{\makecell{0.97\\(0.96)}} & 33 \\ \hline
reentrancy                                                                           & \multicolumn{1}{c|}{\makecell{66.67\\(69.05)}} & \makecell{0.79\\(0.81)} & \multicolumn{1}{c|}{\makecell{78.57\\(76.19)}} & \makecell{0.85\\(0.84)} & \multicolumn{1}{c|}{\makecell{71.42\\(69.04)}} & \makecell{0.81\\(0.8)}   & \multicolumn{1}{c|}{\makecell{73.81\\(76.19)}} & \makecell{0.85\\(0.86)} & \multicolumn{1}{c|}{/}       & /    & \multicolumn{1}{c|}{\makecell{73.81\\(76.19)}} & \makecell{0.85\\(0.86)} & \multicolumn{1}{c|}{\makecell{85.71\\(83.33)}} & \makecell{0.91\\(0.90)} & \multicolumn{1}{c|}{\makecell{38.09\\(35.71)}} & \makecell{0.41\\(0.40)} & \multicolumn{1}{c|}{\makecell{83.72\\(86.05)}} & \makecell{0.92\\(0.93)} & \multicolumn{1}{c|}{\makecell{90.48\\(88.10)}} & \multicolumn{1}{c|}{\makecell{0.94\\(0.93)}} & 2  \\ \hline
gasless send                                                                         & \multicolumn{1}{c|}{/}                        & /                      & \multicolumn{1}{c|}{92.19}                    & 0.95                   & \multicolumn{1}{c|}{82.35}                    & 0.67                    & \multicolumn{1}{c|}{/}                        & /                      & \multicolumn{1}{c|}{/}       & /    & \multicolumn{1}{c|}{92.19}                    & 0.95                   & \multicolumn{1}{c|}{85.94}                    & 0.91         & \multicolumn{1}{c|}{29.69}                    & 0.26          & \multicolumn{1}{c|}{/}                        & /                      & \multicolumn{1}{c|}{100}                    & \multicolumn{1}{c|}{1}                   & 7  \\ \hline
overflow/underflow                                                                   & \multicolumn{1}{c|}{\makecell[c]{81.20\\(81.20)}}                    & \makecell[c]{0.89\\(0.89)}                   & \multicolumn{1}{c|}{/}                        & /                      & \multicolumn{1}{c|}{\makecell[c]{95.30\\(95.30)}}                    & \makecell[c]{0.97\\(0.97)}                    & \multicolumn{1}{c|}{\makecell[c]{90.27\\(90.27)}}                    & \makecell[c]{0.95\\(0.95)}                   & \multicolumn{1}{c|}{\makecell[c]{98.99\\(98.99)}}   & \makecell[c]{0.99\\(0.99)} & \multicolumn{1}{c|}{/}                    & \makecell[c]{/}                   & \multicolumn{1}{c|}{/}                  & \makecell[c]{/}     & \multicolumn{1}{c|}{\makecell[c]{19.40\\(19.40)}}                    & \multicolumn{1}{c|}{\makecell[c]{0.11\\(0.11)}}                & \multicolumn{1}{c|}{/}                        & /                      & \multicolumn{1}{c|}{\makecell[c]{99.33\\(99.33)}}                    & \multicolumn{1}{c|}{\makecell[c]{0.99\\(0.99)}}                   & 4  \\ \hline
transferMint                                                                         & \multicolumn{1}{c|}{/}                        & /                      & \multicolumn{1}{c|}{/}                        & /                      & \multicolumn{1}{c|}{/}                        & /                       & \multicolumn{1}{c|}{/}                        & /                      & \multicolumn{1}{c|}{/}       & /    & \multicolumn{1}{c|}{/}                        & /           & \multicolumn{1}{c|}{/}                        & /            & \multicolumn{1}{c|}{/}                        & /                      & \multicolumn{1}{c|}{/}                        & /                      & \multicolumn{1}{c|}{100}                      & \multicolumn{1}{c|}{1}                      & 0  \\ \hline
\end{tabular}}
\vspace{-0.1in}
\end{table*}

In this section, we firstly make preparations on the experimental setup, including the types of vulnerabilities, datasets, and representative tools that we choose. 
Then, we report the experimental results and analyze the effectiveness of \ourtool. 
Finally, we verify real-world contracts using \ourtool and demonstrate the exploitable bugs that \ourtool finds.

\subsection{Experimental setup}
\label{subsec:setup}
\textbf{Types of Vulnerabilities.}
First, we introduce the vulnerabilities that \ourtool currently targets.
We divide the 37 types of vulnerabilities in SWC Registry~\cite{swc}, a library consisting of smart contracts' vulnerabilities, into 
three categories: a) vulnerabilities that can be detected through syntax checking, \textit{e.g.}, outdated compiler version. 
b) vulnerabilities that do not have clear consequences, \textit{e.g.}, dangerous delegatecall.
c) vulnerabilities that can cause losses of ethers or tokens.
\ourtool targets the vulnerabilities in category c) as they can cause financial loss and are difficult to detect.
There are 6 types of vulnerabilities that \ourtool currently supports:
1) \textit{transaction order dependency (TOD)}; 2) \textit{timestamp dependency (TD)}; 3) \textit{reentrancy};  4) \textit{gasless send}; 5) \textit{overflow/underflow}; 6) \textit{transferMint} \cite{transfer}.
The relationship between these vulnerabilities and our properties has been mentioned in Section \ref{subsec:relationship}.
We divide the \textit{TOD} vulnerabilities into two groups: \textit{TOD-eth} changing ether balances of accounts, \textit{TOD-token} changing token balances of accounts, since SECURIFY and OYENTE only support the detection of the former.

\textbf{Datasets.}
We use two datasets \cite{dataset} of smart contracts to evaluate \ourtool. 
The first dataset, called \textit{vulnerability dataset}, is used to test the performance of \ourtool in detecting different types of vulnerabilities compared with other automated tools.
We collect 611 smart contracts with vulnerabilities in category c) mentioned above from public dataset of other works \cite{0001LC18}\cite{DBLP:conf/sp/SoLPLO20}\cite{DBLP:conf/issta/KolluriNSHS19}\cite{DBLP:conf/icse/dataset-icse}.
We filter out 6 smart contracts whose codes are incomplete and 56 smart contracts that \ourtool does not support.
We illustrate the number of contracts unsupported by \ourtool in the last column of Table \ref{comp}, and the reasons that \ourtool does not support them will be introduced in Section \ref{sec:limitations}.
Finally we get \textit{vulnerability dataset} with 549 contracts.
The second dataset, called \textit{real-world dataset}, is used to evaluate the effectiveness of \ourtool in detecting real-world smart contracts.
We crawl 46453 Solidity source code files from Etherscan \cite{etherscan}, and then filter the contracts to remove duplicates.
We calculate the similarity of two files using difflib \cite{difflib} package of Python, and considered two contracts as duplicates when their similarity is larger than 90\%.
Finally, we obtained 17648 Solidity files containing 30577 contracts as \textit{real-world dataset}.
We add 11 smart contracts with the vulnerability of \textit{transferMint} from the \textit{real-world dataset} to the \textit{vulnerability dataset}, since the previous datasets have not gathered this type of contracts.

\textbf{Tools.}
We compare \ourtool with the following representative automatic tools: OYENTE \cite{LuuCOSH16}, Mythril\cite{mythril}, SECURIFY (version~2)\cite{TsankovDDGBV18}, ContractFuzzer~\cite{0001LC18}, Osiris~\cite{DBLP:conf/acsac/osiris}, Slith\-er~\cite{DBLP:conf/icse/slither}, SmartCheck~\cite{DBLP:conf/icse/smartcheck}, VERISMART~\cite{DBLP:conf/sp/SoLPLO20}, Manticore~\cite{manticore} and eThor~\cite{DBLP:conf/ccs/SchneidewindGSM20}.




We do not compare \ourtool with ZEUS \cite{DBLP:conf/ndss/KalraGDS18}, another automated verifier, since it is not publicly available.
Besides, we do not compare \ourtool with semi-automated verification frameworks while they need manual input of properties, which require certain expertise and is labor-intensive when evaluating hundreds of contracts. Meanwhile, how to express the properties in different specification languages equivalently becomes a problem and may affect the fairness of comparison. 


\textbf{Experimental Environment.}
We experiment on a server with 2.50GHz CPU, 128G memory and 64-bit Ubuntu 16.04.

\subsection{Statistical analysis}
\label{subsec:statistical}
We first perform statistical analysis on \textit{real-world dataset}.
We manually classify the contract as finance-related or others taking the following parts of contracts into account and try our best to avoid misclassification:
1) Contract names. The usage of some contracts can be shown in their name.  
2) Contract annotations. The annotations of contracts can provide us some information, \textit{e.g.}, the contracts' usage.
3) Inheritance of contracts. The children of token contracts can possibly be token contracts.
4) Contract creation statements. The contracts creating token contracts can possibly be token managing contracts.
5) Ether transfer statements. The contracts transferring ethers are ether-related.
Note that the contracts that are difficult to distinguish their usage are classified as others.


After the above classification, we find 27858 finance-related contracts, including 6307 ether-related contracts (20.63\%), 7661 token-related contracts (25.05\%),  5994 contracts both ether-related and token-related (19.60\%) and 7896 indirect-related contracts in total (25.82\%).
The remaining contracts account for 8.89\%.
Hence, finance-related contracts make up a major portion (91.11\%) of the real-world contracts, which validates the goal of generating properties aiming to protect cryptocurrencies shown in Section \ref{subsec:generating properties}.

During the classification, we find that since the official ERC20 \cite{erc20} standard of Ethereum recommends using variable name $\texttt{balances}$ to denote token balances, most token contracts use names similar to $\texttt{balances}$ to denote token balances.
Besides, there are also token contracts using names similar to $\texttt{ownedTokenCount}$ due to ERC721~\cite{erc721} standard.


To validate our observation and evaluate the effectiveness of our methods to identify token contracts, we perform an evaluation on \textit{real-world dataset}.
We search for contracts with variables of type $\texttt{mapping(address=>uint)}$ that have names similar to $\texttt{balances}$ or $\texttt{ownedTokenCount}$, while the similarity of two names is calculated based on fuzzywuzzy\cite{fuzzywuzzy}, and the thresholds are set to 70, 75, 80, 85 and 90, respectively.
We collect the following data under different thresholds: 1) \textit{TP}: the number of token contracts correctly identified. 2) \textit{FN}: the number of token contracts that are missed. 3) \textit{FP}: the number of contracts misclassified as token contracts. 4) \textit{TN}: the number of contracts that are not token contracts correctly classified. 5) \textit{Accuracy}: $Acc = \frac{TP+TN}{TP+TN+FP+FN}$. 6) \textit{F1}: \textit{F1} $= \frac{2TP}{2TP+FP+FN}$.
\begin{table}[]
\centering 
\caption{The effectiveness of our method for identifying token contracts.}  
\label{token-finding}
\vspace{-0.1in}
\resizebox{2.8in}{!}{
\begin{tabular}{|c|c|c|c|c|c|}
\hline
threshold & 70    & 75    & 80    & 85    & 90    \\ \hline
Acc(\%)   & 98.31 & 98.32 & 98.32 & 98.50 & 98.46 \\ \hline
F1(\%)    & 98.13 & 98.14 & 98.14 & 98.31 & 98.27 \\ \hline
\end{tabular}}
\end{table}
We only show 5) and 6) in Table \ref{token-finding} due to the page limit.
According to Table \ref{token-finding}, our method achieves Accuracy and F1-score higher than 98\% under different thresholds.
We choose 85 as our threshold finally.



\subsection{Comparison}
\label{subsec:comparison}
Unlike other automatic tools, \ourtool detects the effect of the vulnerabilities, \textit{i.e.}, whether causing the financial loss.
To fairly compare \ourtool and other automatic tools, we run the tools on \textit{vulnerability dataset}, and collect two sets of results as shown in TABLE \ref{comp}.
Here we call a contract with a vulnerability as a vulnerable contract, and call a contract with a vulnerability causing financial loss as a finance-vulnerable contract.
We regard the number of contracts correctly recognized as a finance-vulnerable / vulnerable contract as \textit{TP}, and regard the number of contracts correctly recognized as a contract that is not finance-vulnerable / vulnerable as \textit{TN}.
The calculation formulas of accuracy and F1 are mentioned above.
Due to the page limit, we only show the accuracy and \textit{F1} of tools in TABLE~\ref{comp}.
Note that \textit{TOD-eth}, \textit{TOD-token}, \textit{gasless send} and \textit{transferMint} always cause financial loss, thus the two sets of results for them are the same.

Totally, \ourtool outperforms the representative tools that it achieves higher accuracy and F1 values in the detection of vulnerable and finance-vulnerable contracts in \textit{vulnerability dataset}.
Meanwhile, \ourtool is the only one that is able to detect all the types of vulnerabilities in TABLE \ref{comp} among the automated tools mentioned above.
Note that we fail to make ContractFuzzer report any findings.
Though being in contact with the authors, we are unable to fix the issue and both sides eventually give up.
SECURIFY can output alerts for all contracts using timestamp, but is not targeted to detect \textit{TD}, so we do not compare its ability to detect \textit{TD} with \ourtool.

We analyze the reason for the false results produced by different automatic tools shown in Table \ref{comp}.

\textbf{\textit{TOD-eth},\textit{TD}, \textit{gasless send}}: 
The above automatic tools detect these types of vulnerabilities based on their pre-defined patterns and their accuracy depends on the patterns.
On one hand, progressive patterns can result in false negatives. 
For example, SECURIFY decides if a contract is secure against \textit{gasless send} by matching the pattern whether each return value of \texttt{send} is checked. 
However, a contract checks the result of \texttt{send} but does not handle the exception, which evades the detection of SECURIFY.
On the other hand, conservative patterns can lead to false positives.
For example, OYENTE and SECURIFY detect \textit{TOD-eth} according to the pattern that when the transaction orders changes, the recipient of ethers may also change.
A contract returns ethers to their senders and the first sender will be the first receiver.
For this case, both OYENTE and SECURIFY falsely report \textit{TOD-eth} vulnerability.
However, all senders eventually receive ethers, \textit{i.e.}, the result is not changed with the transaction order, whereas our equivalence property holds.
Besides, the tools using symbolic execution, \textit{e.g.}, OYENTE and Mythril, may produce false negatives as they explore a subset of contracts' behaviors.

\textbf{\textit{reentrancy}}: 
EThor defines a property: an internal transaction can only be initiated by the execution of a \texttt{call} instruction, which over-approximates the property that a contract free from \textit{reentrancy} should satisfy.
Therefore, eThor gets more false positives than \ourtool in detection of \textit{reentrancy}.
The reasons for the false reports of the other tools in the detection of \textit{reentrancy} are still inaccurate patterns.

\textbf{\textit{Overflow/underflow}}: 
OYENTE, Mythril, Osiris assume that the values of all the variables are arbitrary and output FPs for this category.
Differently, FASVERIF and VERISMART consider additional constraints of variables, e.g., for the variables whose values are constant, their values should be equal to the initial values.
VERISMART outputs 2 false positives due to its assumption: every function can be accessed.

\ourtool also produces 9 false negatives due to the error of property generation.
Specifically, \ourtool fails to detect 2 contracts with \textit{overflow}.
In these two contracts, the variable \texttt{allowance} may overflow. 
We currently do not design the invariants for this variable. 
So we manually define a new invariant according to the two contracts and \ourtool successfully discovers the vulnerabilities.
\ourtool also misses 3 contracts with \textit{TD} and 4 contracts with \textit{reentrancy}.
These contracts use uncommon variable names to denote token balances.
We manually specify the key variable names and finally find out the missed vulnerabilities.

To compare the efficiency of the above tools, we calculate the average time taken by them to analyze one contract in \textit{vulnerability dataset} as follows:
Slither (2.16 s), SmartCheck (4.93 s), eThor (11.95 s), OYENTE (20.81 s), Mythril (55.00 s), VERISMART (63.45 s), Osiris (73.52 s), SECURIFY (222.99 s), \ourtool (829.61 s).

\subsection{Security analysis of real-world smart contracts}
\label{subsec:realworld}


To evaluate the effectiveness of \ourtool in real-world contracts, we conduct an experiment on randomly-selected 1700 contracts from \textit{real-world dataset}.
FASVERIF reports 15 contracts with vulnerabilities, of which 11 violates the invariant property and 4 violates the equivalence property.
{\color{black}We simulate attacks on these contracts on a private chain of Ethereum and check the exploitability of the vulnerabilities in them with on-chain states.
We eventually find that among the 15 contracts, there is one contract destroyed and another contract with non-exploitable vulnerabilities, whereas the vulnerabilities in the remaining 13 contracts are exploitable.}
Among the exploitable bugs, there are 10 of \textit{transferMint} vulnerabilities, which cannot be detected by existing automatic tools as shown in Table \ref{comp}. 
Considering the proportion of vulnerable contracts found and the vulnerabilities in them causing financial losses, we hope our work can raise security concerns.
{\color{black}
The unexploitable contract is a crowdsale contract selling tokens.
The contract specifies that users who buy tokens within a certain time frame can get bonuses.
However, the bonuses are no longer available after September 7, 2017, thus the vulnerability in this contract is not exploitable but misclassified due to the incompleteness of \ourtool.}

\textbf{Ethical Considerations.} 
{\color{black}
As Ethereum accounts are anonymous, we attempt to identify the owners of the vulnerable contracts by checking the contract code, the addresses of the contract creators, and 685 bug bounty programs~\cite{bug-bounty}. We also use a chat software \cite{blockscan} to send messages to the addresses of the contract creators but do not receive replies after waiting for 40 days.
To avoid the abuse of these vulnerabilities, we do not provide the addresses of the vulnerable contracts or open-source \ourtool. 
Instead, we present a simplified version of the destroyed contract and provide a website with an interface to use the restricted version of \ourtool~\cite{fasverif_online}.
Also, our tool is available upon request for researchers with validated identities for academic purposes.}



\textbf{Example.}
The contract Ex1  shown in Fig. \ref{fig:example2} is a contract with an exploitable bug.
Note that this contract is simplified. 
In practice there are conditional statements to avoid numerical operations causing \textit{overflow/underflow}.
\ourtool recognizes Ex1 as a token contract and chooses the invariant $\kwt{token\_inv}$.
Specifically, assume that the sum of token balances of $\texttt{to}$ and $\texttt{msg.sender}$ before the transaction, \textit{i.e.}, $\texttt{balances}_0[\texttt{to}]+\texttt{balances}_0[\texttt{msg.sender}]$ is value $C_1$, \ourtool checks whether the sum after the transaction, i.e., $\texttt{balances}_1[\texttt{to}]+\texttt{balances}_1[\texttt{msg.sender}]$ can be different from $C_1$.
In the verification, \ourtool finds an execution that reaches $\kwf{End}()$ and has a constraint $\kwf{Pred\_eq}(\texttt{to},\texttt{msg.sender})$.  
According to the constraint, \texttt{msg.sender} in all expressions are replaced with \texttt{to}.
Moreover, since the rules $\kwt{var\_declare}$ and $\kwt{var\_assign}$ are in the execution, $\texttt{balances}_1[\texttt{to}]$ is replaced by $\texttt{balances}_1[\texttt{to}]+\texttt{value}$. 
Hence, the constraints $\texttt{balances}_0[\texttt{to}]+\texttt{balances}_0[\texttt{to}] = C_1$, $\texttt{balances}_0[\texttt{to}]+\texttt{value}+\texttt{balances}_0[\texttt{to}]+\texttt{value} \neq C_1$ are added to Z3.
As a result, the constraints are satisfied with $\texttt{value} \neq 0$, which indicates the invariant $\kwt{token\_inv}$ is broken and 
\ourtool decides this contract as vulnerable.
Comparatively, SECURIFY, OYENTE and Mythril fail to detect this type of vulnerability with unknown patterns.
VERISMART cannot detect this vulnerability that does not cause \textit{overflow/underflow}.

{\color{black}
We set the verification timeout as 5 hours but 12 contracts cannot be verified within that time.
The remaining contracts take an average of 2 hours and 40 minutes to verify.
During the verification, we manually set variable names for 14 contracts in which \ourtool cannot find key variables.
}

{\color{black}
Besides, we compare TeEther  \cite{DBLP:conf/uss/KruppR18} with \ourtool on the 1700 real-world contracts. 
TeEther aims to reveal critical parts of code that can be abused to get ethers and assumes that if the attacker as an external account can obtain ethers from a contract, the contract is vulnerable. 
TeEther considers two contracts vulnerable, while \ourtool considers them non-vulnerable.
For the first contract, the attacker can destroy it whereby get ethers, but \ourtool cannot detect this vulnerability which is not covered by our properties. 
For the second contract, the attacker cannot disrupt its execution and can only get ethers in normal ways. 
\ourtool considers this contract safe since no ether or tokens will be lost unexpectedly.}


\section{RELATED WORK} 
\label{sec:related work}

\subsection{Automated bug-finding tools for contracts}
\label{subsec:autotool}
Automated bug-finding tools fall into two categories: tools using symbolic execution and tools using other technologies.
Among the tools using symbolic execution, 
OYENTE~\cite{LuuCOSH16} executes EVM bytecode symbolically and checks for vulnerability patterns in execution traces.
Mythril \cite{mythril} uses taint analysis and symbolic execution to find vulnerability patterns.
Osiris \cite{DBLP:conf/acsac/TorresSS18} is specially designed for detecting arithmetic bugs. 
In the tools using other technologies,
ContractFuzzer \cite{0001LC18} instruments EVM to search for executions that match patterns.
SmartCheck~\cite{DBLP:conf/icse/TikhomirovVITMA18} searches for specific patterns in the XML syntax trees of contracts.
VERISMART~\cite{DBLP:conf/sp/SoLPLO20} generates and checks invariants to find the overflow in smart contracts.

Compared with the above tools, there are differences between \ourtool and them:
1) \ourtool provides a proof of our translation and implements the verification using formal tools.
2) The vulnerabilities detected by these tools are in a particular category or dependent on pre-defined known patterns. Comparatively, \ourtool generates security properties on demand and covers various types of vulnerabilities.

\subsection{Verification frameworks for contracts}
\label{subsec:extool}
Verification frameworks formally verify the properties of contracts. 
SMARTPULSE \cite{stephens2021smartpulse} is used to check given temporal properties of smart contracts. 
Similarly, VerX \cite{PermenevDTDV20} performs a semi-automatic verification of temporal safety specifications. 
ConCert \cite{DBLP:conf/cpp/Concert} is a proof framework for functional smart contract languages. 
These tools can verify functional properties of contracts, which are not currently supported by \ourtool, 
but need human involvement to produce results.
Differently, \ourtool can generate and verify finance-related properties for contracts automatically.
Besides, according to their literature, the above tools cannot verify our equivalence properties.
{\color{black}
CFF \cite{DBLP:journals/corr/clockwork} is a formal verification framework for reasoning about the economic security properties of DeFi contracts.
CFF proposes \textit{extractable value} (EV), which is similar to our equivalence property.
Specifically, the equivalence property is used to check whether an adversary can obtain profits through operations such as reordering transactions, while EV is used to quantify the profits an adversary can obtain.
However, CFF takes into account more financial features, e.g., changes in exchange rates, which are not considered in \ourtool.

\subsection{Automated verifiers for contracts}
\label{subsec:hybridtool}
To the best of our knowledge, there are three automated verifiers for smart contracts: eThor \cite{DBLP:conf/ccs/SchneidewindGSM20}, SECURIFY~\cite{TsankovDDGBV18} and ZEUS \cite{DBLP:conf/ndss/KalraGDS18}.
eThor is a sound static analyzer that abstracts the semantics of EVM bytecode into Horn clauses. 
As the literature of eThor states, it can only detect \textit{reentrancy} or check assertions automatically.
In addition, eThor cannot verify our equivalence property.
SECURIFY detects specified patterns extracted from control flows of contracts.
SECURIFY cannot solve numerical constraints and thus cannot detect \textit{overflow} and \textit{transferMint}.
ZEUS transforms smart contracts into LLVM bitcode and uses existing symbolic model checkers. 
The transformations are claimed to be semantics preserving which however are refuted by \cite{DBLP:conf/ccs/SchneidewindGSM20}.
Besides, ZEUS uses pre-defined policies based on known patterns.
Thus, ZEUS may miss unknown vulnerabilities or variants of known vulnerabilities, \textit{e.g.},  \textit{transferMint} supported by \ourtool. 




\subsection{Generating properties for other verifiers}
\label{subsec:combine}
We investigate whether our properties can be used by other verifiers. 
We study the following verifiers that can verify properties automatically: ZEUS, VerX, SECURIFY, eThor and SMARTPULSE.
Among them, ZEUS is not publicly available, and VerX only provides a website that is no longer maintained.
SECURIFY cannot solve numerical constraints and thus cannot verify our properties. 
eThor analyzes the bytecode of contracts, ignoring semantic information like variable names, so it is non-trivial to convert our properties, which require variable names as part of them, into a form that eThor can verify.
Besides, we fail to make SMARTPULSE \cite{smartpulse-codes} work by following the instructions on its webpage.
\section{LIMITATIONS AND DISCUSSION} 
\label{sec:limitations}



\textbf{Limitations.} We summarize limitations of \ourtool as follows:
{\color{black}

\textbf{1)} The average time to analyze a contract using \ourtool is longer than the one using other automated tools. {\color{black} According to the experiment on \textit{vulnerability dataset}, \ourtool take an average of 829.61 seconds to analyze a contract, while the most time-consuming one of the other automated tools take an average of 222.99 seconds.}

\textbf{2)} \ourtool currently cannot detect vulnerabilities that do not cause financial losses, \textit{e.g.}, the overflow vulnerabilities that lead to DoS, which is supported by some automated tools.

\textbf{3)} {\color{black}
\ourtool can only support vulnerabilities that are covered by our properties under our assumptions. 
Specifically, we do not consider the exchange rates and focus only on vulnerabilities that result in abnormal token amounts or that allow attackers to gain differently with different transaction orders or block timestamps.
Thus, the economical security property (considering the exchange rates) proposed in \cite{DBLP:journals/corr/clockwork}, the airdrop hunting and self-destruction vulnerabilities (not covered by our properties) are unsupported by \ourtool.}

\textbf{4)} Solidity language is not fully supported.
Due to the Turing-completeness of Solidity \cite{solidity}, it is challenging to fully support its features.
Thus, we add the following restrictions to define a fragment of Solidity supported by \ourtool:
\begin{itemize}
\setlength{\itemsep}{0pt}
\setlength{\parsep}{0pt}
\setlength{\parskip}{0pt}

\item Loops. \ourtool supports unrolling of bounded loops, i.e., the execution times of loops are constant, where the loop statement is replaced by equivalent statements without loops.
The unbounded loops, whose execution times cannot be determined statically, are not supported.
{\color{black} We find 2988 contracts (9.77\%) with unbounded loops in \textit{real-world dataset} and omit them in our analysis.}

\item Revert. \ourtool verifies the properties under the assumption that all transactions can be executed to completion. For transactions where a revert occurs, we assume that the executions of the transactions do not result in the modification of any variables. 

\item Contract creation. \ourtool supports the case of static creation of contracts in the constructors.
To trade off efficiency and coverage for Solidity features, we omit the contracts creating contracts via function calls.
{\color{black} However, we only find 4.67\% (1428/30577) of contracts in \textit{real-world dataset} that create contracts via function calls.} 

\item Function call. Given a set of contracts with Solidity codes, \ourtool requires them not to invoke functions in contracts outside the set whose codes are unknown.
\ourtool can only analyze codes given beforehand, which is an inherent defect of static analyzers~\cite{DBLP:journals/tissec/OnwuzurikeMACRS19}.
{\color{black} We also count the contracts calling unknown codes in \textit{real-world dataset} and finally find 1754 contracts (5.74\%).}

\end{itemize}
{\color{black}
In summary, even with the above restrictions, \ourtool can still cover 82.41\% (25197/30577) real-world contracts. 

\textbf{5)} \ourtool may get incorrect key variables or invariants.
Though our method of identifying key variables achieves accuracy higher than 98\% in \textit{real-world dataset}, it still may misidentify some key variables.
Additionally, the correctness of the generated invariants is also not guaranteed.
As a result, incorrect variables or invariants can lead to legitimate contracts being ruled out.
Thus, we offer users the option to manually set invariants and key variables instead.}

{\color{black}
\textbf{6)} The incompleteness of \ourtool may lead to misclassifying safe contracts as vulnerable, \textit{e.g.}, the online contract that is unexploitable mentioned in Section \ref{subsec:realworld}.  }

{\color{black}
\textbf{Discussion.} We choose Tamarin due to its well-supported modeling of concurrent systems~\cite{DBLP:phd/basesearch/Meier13}.
Using Tamarin gives us the flexibility to add or modify rules in our models to verify hyperproperties~\cite{DBLP:conf/cav/hyper1} like the equivalence properties requiring simultaneous reasoning of multiple executions.
In comparison, using other tools may introduce more difficulties when modeling and verifying hyperproperties~\cite{DBLP:conf/cav/hyper1}\cite{DBLP:conf/uss/hyper2}.
However, our extensions to Tamarin are specific to finance-related properties and some features of Tamarin are not used.
It is interesting to further extend Tamarin in the future.}

\section{CONCLUSION} 
\label{sec:conclusion}

We propose and implement \ourtool, which can automatically generate finance-related properties and the corresponding models for smart contracts, and verify the properties automatically.
\ourtool outperforms other automatic tools in detecting finance-related vulnerabilities in accuracy and coverage of types of vulnerabilities, and it finds 13 contracts with exploitable bugs, including 10 contracts evading the detection of other automated tools to the best of our knowledge.

\bibliographystyle{ACM-Reference-Format}
\bibliography{main}


\begin{thebibliography}{00}


\ifx \showCODEN    \undefined \def \showCODEN     #1{\unskip}     \fi
\ifx \showDOI      \undefined \def \showDOI       #1{#1}\fi
\ifx \showISBNx    \undefined \def \showISBNx     #1{\unskip}     \fi
\ifx \showISBNxiii \undefined \def \showISBNxiii  #1{\unskip}     \fi
\ifx \showISSN     \undefined \def \showISSN      #1{\unskip}     \fi
\ifx \showLCCN     \undefined \def \showLCCN      #1{\unskip}     \fi
\ifx \shownote     \undefined \def \shownote      #1{#1}          \fi
\ifx \showarticletitle \undefined \def \showarticletitle #1{#1}   \fi
\ifx \showURL      \undefined \def \showURL       {\relax}        \fi
\providecommand\bibfield[2]{#2}
\providecommand\bibinfo[2]{#2}
\providecommand\natexlab[1]{#1}
\providecommand\showeprint[2][]{arXiv:#2}

\bibitem[\protect\citeauthoryear{??}{DAO}{2016}]%
        {DAO}
 \bibinfo{year}{2016}\natexlab{}.
\newblock \bibinfo{title}{The {DAO} hack}.
\newblock
  \bibinfo{howpublished}{\url{https://www.coindesk.com/learn/2016/06/25/understanding-the-dao-attack/}}.
    (\bibinfo{year}{2016}).
\newblock


\bibitem[\protect\citeauthoryear{??}{bes}{2016}]%
        {bestpractice}
 \bibinfo{year}{2016}\natexlab{}.
\newblock \bibinfo{title}{Ethereum Smart Contract Best Practices}.
\newblock
  \bibinfo{howpublished}{\url{https://consensys.github.io/smart-contract-best-practices/development-recommendations/solidity-specific/timestamp-dependence/}}.
    (\bibinfo{year}{2016}).
\newblock


\bibitem[\protect\citeauthoryear{??}{app}{2017}]%
        {application1}
 \bibinfo{year}{2017}\natexlab{}.
\newblock \bibinfo{title}{Blockchain is empowering the future of insurance}.
\newblock
  \bibinfo{howpublished}{\url{https://techcrunch.com/2016/10/29/blockchain-is-empowering-the-future-of-insurance/}}.
    (\bibinfo{year}{2017}).
\newblock


\bibitem[\protect\citeauthoryear{??}{par}{2017}]%
        {parity}
 \bibinfo{year}{2017}\natexlab{}.
\newblock \bibinfo{title}{The {Parity} bug}.
\newblock
  \bibinfo{howpublished}{\url{http://hackxingdistributed.com/2017/07/22/deep-dive-parity-bug}}.
    (\bibinfo{year}{2017}).
\newblock


\bibitem[\protect\citeauthoryear{??}{ico}{2018}]%
        {ico}
 \bibinfo{year}{2018}\natexlab{}.
\newblock \bibinfo{title}{{A Closer Look at ICO Smart Contracts}}.
\newblock
  \bibinfo{howpublished}{\url{https://tokeny.com/a-closer-look-at-ico-smart-contracts/}}.
    (\bibinfo{year}{2018}).
\newblock


\bibitem[\protect\citeauthoryear{??}{tok}{2018}]%
        {tokenzero}
 \bibinfo{year}{2018}\natexlab{}.
\newblock \bibinfo{title}{New multiOverflow Bug Identified in Multiple ERC20
  Smart Contracts}.
\newblock
  \bibinfo{howpublished}{\url{https://peckshield.medium.com/new-multioverflow-bug-identified-in-multiple-erc20-smart-contracts-cve-2018-10706-8e55946c252c}}.
    (\bibinfo{year}{2018}).
\newblock


\bibitem[\protect\citeauthoryear{??}{tra}{2019}]%
        {transfer}
 \bibinfo{year}{2019}\natexlab{}.
\newblock \bibinfo{title}{Fatal TransferMint Bug in Multiple TRC20 Smart
  Contracts}.
\newblock
  \bibinfo{howpublished}{\url{https://twitter.com/peckshield/status/1115226918855401479?cxt=HHwWjoCj5aq-ivoeAAAA}}.
    (\bibinfo{year}{2019}).
\newblock


\bibitem[\protect\citeauthoryear{??}{fuz}{2020}]%
        {fuzzywuzzy}
 \bibinfo{year}{2020}\natexlab{}.
\newblock \bibinfo{title}{fuzzywuzzy 0.18.0}.
\newblock \bibinfo{howpublished}{\url{https://pypi.org/project/fuzzywuzzy/}}.
  (\bibinfo{year}{2020}).
\newblock


\bibitem[\protect\citeauthoryear{??}{SWC}{2020}]%
        {SWC-reentrancy}
 \bibinfo{year}{2020}\natexlab{}.
\newblock \bibinfo{title}{SWC-107: Reentrancy}.
\newblock \bibinfo{howpublished}{\url{https://swcregistry.io/docs/SWC-107}}.
  (\bibinfo{year}{2020}).
\newblock


\bibitem[\protect\citeauthoryear{??}{swc}{2020}]%
        {swc}
 \bibinfo{year}{2020}\natexlab{}.
\newblock \bibinfo{title}{SWC Registry}.
\newblock \bibinfo{howpublished}{\url{https://swcregistry.io/}}.
  (\bibinfo{year}{2020}).
\newblock


\bibitem[\protect\citeauthoryear{??}{erc}{2021}]%
        {erc20}
 \bibinfo{year}{2021}\natexlab{}.
\newblock \bibinfo{title}{ERC20 standard}.
\newblock
  \bibinfo{howpublished}{\url{https://github.com/ethereum/EIPs/blob/master/EIPS/eip-20.md}}.
    (\bibinfo{year}{2021}).
\newblock


\bibitem[\protect\citeauthoryear{??}{sma}{2021}]%
        {smartpulse-codes}
 \bibinfo{year}{2021}\natexlab{}.
\newblock \bibinfo{title}{Main repository for SmartPulse}.
\newblock
  \bibinfo{howpublished}{\url{https://github.com/utopia-group/SmartPulseTool}}.
    (\bibinfo{year}{2021}).
\newblock


\bibitem[\protect\citeauthoryear{??}{pol}{2021}]%
        {poly}
 \bibinfo{year}{2021}\natexlab{}.
\newblock \bibinfo{title}{The {Poly} {Network} attack}.
\newblock
  \bibinfo{howpublished}{\url{https://en.wikipedia.org/wiki/Poly\_Network\_exploit}}.
    (\bibinfo{year}{2021}).
\newblock


\bibitem[\protect\citeauthoryear{??}{eth}{2022a}]%
        {ethercap}
 \bibinfo{year}{2022}\natexlab{a}.
\newblock \bibinfo{title}{AICoin — information about Ethereum}.
\newblock
  \bibinfo{howpublished}{\url{https://www.aicoin.com/currencies/ethereum.html?lang=en}}.
    (\bibinfo{year}{2022}).
\newblock


\bibitem[\protect\citeauthoryear{??}{blo}{2022}]%
        {blockscan}
 \bibinfo{year}{2022}\natexlab{}.
\newblock \bibinfo{title}{Blockscan Chat}.
\newblock \bibinfo{howpublished}{\url{https://chat.blockscan.com/start}}.
  (\bibinfo{year}{2022}).
\newblock


\bibitem[\protect\citeauthoryear{??}{bug}{2022}]%
        {bug-bounty}
 \bibinfo{year}{2022}\natexlab{}.
\newblock \bibinfo{title}{Bug Bounty Programs}.
\newblock
  \bibinfo{howpublished}{\url{https://consensys.github.io/smart-contract-best-practices/bug-bounty-programs/}}.
    (\bibinfo{year}{2022}).
\newblock


\bibitem[\protect\citeauthoryear{??}{dif}{2022}]%
        {difflib}
 \bibinfo{year}{2022}\natexlab{}.
\newblock \bibinfo{title}{difflib — Helpers for computing deltas}.
\newblock
  \bibinfo{howpublished}{\url{https://docs.python.org/3/library/difflib.html}}.
    (\bibinfo{year}{2022}).
\newblock


\bibitem[\protect\citeauthoryear{??}{erc}{2022}]%
        {erc721}
 \bibinfo{year}{2022}\natexlab{}.
\newblock \bibinfo{title}{ERC721 standard}.
\newblock
  \bibinfo{howpublished}{\url{https://github.com/ethereum/EIPs/blob/master/EIPS/eip-721.md}}.
    (\bibinfo{year}{2022}).
\newblock


\bibitem[\protect\citeauthoryear{??}{eth}{2022b}]%
        {etherscan}
 \bibinfo{year}{2022}\natexlab{b}.
\newblock \bibinfo{title}{Etherscan}.
\newblock \bibinfo{howpublished}{\url{https://etherscan.io/}}.
  (\bibinfo{year}{2022}).
\newblock


\bibitem[\protect\citeauthoryear{??}{dat}{2022}]%
        {dataset}
 \bibinfo{year}{2022}\natexlab{}.
\newblock \bibinfo{title}{FASVERIF-dataset}.
\newblock
  \bibinfo{howpublished}{\url{https://github.com/secwisf/FASVERIF-dataset/tree/master}}.
    (\bibinfo{year}{2022}).
\newblock


\bibitem[\protect\citeauthoryear{??}{fas}{2022}]%
        {fasverif_online}
 \bibinfo{year}{2022}\natexlab{}.
\newblock \bibinfo{title}{FASVERIF Online}.
\newblock \bibinfo{howpublished}{\url{https://101.200.87.174/}}.
  (\bibinfo{year}{2022}).
\newblock


\bibitem[\protect\citeauthoryear{??}{15r}{2022}]%
        {15rule}
 \bibinfo{year}{2022}\natexlab{}.
\newblock \bibinfo{title}{go-ethereum}.
\newblock
  \bibinfo{howpublished}{\url{https://github.com/ethereum/go-ethereum/blob/4e474c74dc2ac1d26b339c32064d0bac98775e77/consensus/ethash/consensus.go}}.
    (\bibinfo{year}{2022}).
\newblock


\bibitem[\protect\citeauthoryear{??}{man}{2022}]%
        {manticore}
 \bibinfo{year}{2022}\natexlab{}.
\newblock \bibinfo{title}{Manticore}.
\newblock
  \bibinfo{howpublished}{\url{https://github.com/trailofbits/manticore/}}.
  (\bibinfo{year}{2022}).
\newblock


\bibitem[\protect\citeauthoryear{??}{sol}{2022}]%
        {solidity}
 \bibinfo{year}{2022}\natexlab{}.
\newblock \bibinfo{title}{Solidity documentation}.
\newblock
  \bibinfo{howpublished}{\url{https://solidity.readthedocs.io/en/latest}}.
  (\bibinfo{year}{2022}).
\newblock


\bibitem[\protect\citeauthoryear{??}{evm}{2022}]%
        {evmgo}
 \bibinfo{year}{2022}\natexlab{}.
\newblock \bibinfo{title}{Source codes of EVM}.
\newblock
  \bibinfo{howpublished}{\url{https://github.com/ethereum/go-ethereum/tree/master/core/vm/evm.go}}.
    (\bibinfo{year}{2022}).
\newblock


\bibitem[\protect\citeauthoryear{??}{z3}{2022}]%
        {z3}
 \bibinfo{year}{2022}\natexlab{}.
\newblock \bibinfo{title}{Z3: An efficient SMT solver}.
\newblock
  \bibinfo{howpublished}{\url{https://www.microsoft.com/en-us/research/project/z3-3/}}.
    (\bibinfo{year}{2022}).
\newblock


\bibitem[\protect\citeauthoryear{Almeida, Barbosa, Barthe, Dupressoir, and
  Emmi}{Almeida et~al\mbox{.}}{2016}]%
        {DBLP:conf/uss/hyper2}
\bibfield{author}{\bibinfo{person}{Jos{\'{e}}~Bacelar Almeida},
  \bibinfo{person}{Manuel Barbosa}, \bibinfo{person}{Gilles Barthe},
  \bibinfo{person}{Fran{\c{c}}ois Dupressoir}, {and} \bibinfo{person}{Michael
  Emmi}.} \bibinfo{year}{2016}\natexlab{}.
\newblock \showarticletitle{Verifying Constant-Time Implementations}. In
  \bibinfo{booktitle}{{\em 25th {USENIX} Security Symposium, {USENIX} Security
  16, Austin, TX, USA, August 10-12, 2016}},
  \bibfield{editor}{\bibinfo{person}{Thorsten Holz} {and}
  \bibinfo{person}{Stefan Savage}} (Eds.). \bibinfo{publisher}{{USENIX}
  Association}, \bibinfo{pages}{53--70}.
\newblock
\showURL{%
\url{https://www.usenix.org/conference/usenixsecurity16/technical-sessions/presentation/almeida}}


\bibitem[\protect\citeauthoryear{Annenkov, Nielsen, and Spitters}{Annenkov
  et~al\mbox{.}}{2020}]%
        {DBLP:conf/cpp/Concert}
\bibfield{author}{\bibinfo{person}{Danil Annenkov},
  \bibinfo{person}{Jakob~Botsch Nielsen}, {and} \bibinfo{person}{Bas
  Spitters}.} \bibinfo{year}{2020}\natexlab{}.
\newblock \showarticletitle{ConCert: a smart contract certification framework
  in Coq}. In \bibinfo{booktitle}{{\em Proceedings of the 9th {ACM} {SIGPLAN}
  International Conference on Certified Programs and Proofs, {CPP} 2020, New
  Orleans, LA, USA, January 20-21, 2020}},
  \bibfield{editor}{\bibinfo{person}{Jasmin Blanchette} {and}
  \bibinfo{person}{Catalin Hritcu}} (Eds.). \bibinfo{publisher}{{ACM}},
  \bibinfo{pages}{215--228}.
\newblock
\showDOI{%
\url{https://doi.org/10.1145/3372885.3373829}}


\bibitem[\protect\citeauthoryear{Babel, Daian, Kelkar, and Juels}{Babel
  et~al\mbox{.}}{2021}]%
        {DBLP:journals/corr/clockwork}
\bibfield{author}{\bibinfo{person}{Kushal Babel}, \bibinfo{person}{Philip
  Daian}, \bibinfo{person}{Mahimna Kelkar}, {and} \bibinfo{person}{Ari Juels}.}
  \bibinfo{year}{2021}\natexlab{}.
\newblock \showarticletitle{Clockwork Finance: Automated Analysis of Economic
  Security in Smart Contracts}.
\newblock \bibinfo{journal}{{\em CoRR\/}}  \bibinfo{volume}{abs/2109.04347}
  (\bibinfo{year}{2021}).
\newblock
\showeprint{2109.04347}
\showURL{%
\url{https://arxiv.org/abs/2109.04347}}


\bibitem[\protect\citeauthoryear{Barthe, D'Argenio, and Rezk}{Barthe
  et~al\mbox{.}}{2004}]%
        {DBLP:conf/csfw/BartheDR04}
\bibfield{author}{\bibinfo{person}{Gilles Barthe}, \bibinfo{person}{Pedro~R.
  D'Argenio}, {and} \bibinfo{person}{Tamara Rezk}.}
  \bibinfo{year}{2004}\natexlab{}.
\newblock \showarticletitle{Secure Information Flow by Self-Composition}. In
  \bibinfo{booktitle}{{\em 17th {IEEE} Computer Security Foundations Workshop,
  {(CSFW-17} 2004), 28-30 June 2004, Pacific Grove, CA, {USA}}}.
  \bibinfo{publisher}{{IEEE} Computer Society}, \bibinfo{pages}{100--114}.
\newblock
\showDOI{%
\url{https://doi.org/10.1109/CSFW.2004.17}}


\bibitem[\protect\citeauthoryear{Baumeister, Coenen, Bonakdarpour, Finkbeiner,
  and S{\'{a}}nchez}{Baumeister et~al\mbox{.}}{2021}]%
        {DBLP:conf/cav/hyper1}
\bibfield{author}{\bibinfo{person}{Jan Baumeister}, \bibinfo{person}{Norine
  Coenen}, \bibinfo{person}{Borzoo Bonakdarpour}, \bibinfo{person}{Bernd
  Finkbeiner}, {and} \bibinfo{person}{C{\'{e}}sar S{\'{a}}nchez}.}
  \bibinfo{year}{2021}\natexlab{}.
\newblock \showarticletitle{A Temporal Logic for Asynchronous Hyperproperties}.
  In \bibinfo{booktitle}{{\em Computer Aided Verification - 33rd International
  Conference, {CAV} 2021, Virtual Event, July 20-23, 2021, Proceedings, Part
  {I}}} {\em (\bibinfo{series}{Lecture Notes in Computer Science})},
  \bibfield{editor}{\bibinfo{person}{Alexandra Silva} {and}
  \bibinfo{person}{K.~Rustan~M. Leino}} (Eds.), Vol.~\bibinfo{volume}{12759}.
  \bibinfo{publisher}{Springer}, \bibinfo{pages}{694--717}.
\newblock
\showDOI{%
\url{https://doi.org/10.1007/978-3-030-81685-8\_33}}


\bibitem[\protect\citeauthoryear{Chen, Pendleton, Njilla, and Xu}{Chen
  et~al\mbox{.}}{2020}]%
        {DBLP:journals/csur/ChenPNX20}
\bibfield{author}{\bibinfo{person}{Huashan Chen}, \bibinfo{person}{Marcus
  Pendleton}, \bibinfo{person}{Laurent Njilla}, {and} \bibinfo{person}{Shouhuai
  Xu}.} \bibinfo{year}{2020}\natexlab{}.
\newblock \showarticletitle{A Survey on Ethereum Systems Security:
  Vulnerabilities, Attacks, and Defenses}.
\newblock \bibinfo{journal}{{\em {ACM} Comput. Surv.\/}} \bibinfo{volume}{53},
  \bibinfo{number}{3} (\bibinfo{year}{2020}), \bibinfo{pages}{67:1--67:43}.
\newblock
\showDOI{%
\url{https://doi.org/10.1145/3391195}}


\bibitem[\protect\citeauthoryear{Chin, David, Nguyen, and Qin}{Chin
  et~al\mbox{.}}{2012}]%
        {ChinDNQ12}
\bibfield{author}{\bibinfo{person}{Wei{-}Ngan Chin}, \bibinfo{person}{Cristina
  David}, \bibinfo{person}{Huu~Hai Nguyen}, {and} \bibinfo{person}{Shengchao
  Qin}.} \bibinfo{year}{2012}\natexlab{}.
\newblock \showarticletitle{Automated verification of shape, size and bag
  properties via user-defined predicates in separation logic}.
\newblock \bibinfo{journal}{{\em Sci. Comput. Program.\/}}
  \bibinfo{volume}{77}, \bibinfo{number}{9} (\bibinfo{year}{2012}),
  \bibinfo{pages}{1006--1036}.
\newblock
\showDOI{%
\url{https://doi.org/10.1016/j.scico.2010.07.004}}


\bibitem[\protect\citeauthoryear{ConsenSys}{ConsenSys}{2022}]%
        {mythril}
\bibfield{author}{\bibinfo{person}{ConsenSys}.}
  \bibinfo{year}{2022}\natexlab{}.
\newblock \bibinfo{title}{Mythril}.
\newblock \bibinfo{howpublished}{\url{https://github.com/ConsenSys/mythril}}.
  (\bibinfo{year}{2022}).
\newblock


\bibitem[\protect\citeauthoryear{Durieux, Ferreira, Abreu, and Cruz}{Durieux
  et~al\mbox{.}}{2020}]%
        {DBLP:conf/icse/dataset-icse}
\bibfield{author}{\bibinfo{person}{Thomas Durieux},
  \bibinfo{person}{Jo{\~{a}}o~F. Ferreira}, \bibinfo{person}{Rui Abreu}, {and}
  \bibinfo{person}{Pedro Cruz}.} \bibinfo{year}{2020}\natexlab{}.
\newblock \showarticletitle{Empirical review of automated analysis tools on 47,
  587 Ethereum smart contracts}. In \bibinfo{booktitle}{{\em {ICSE} '20: 42nd
  International Conference on Software Engineering, Seoul, South Korea, 27 June
  - 19 July, 2020}}, \bibfield{editor}{\bibinfo{person}{Gregg Rothermel} {and}
  \bibinfo{person}{Doo{-}Hwan Bae}} (Eds.). \bibinfo{publisher}{{ACM}},
  \bibinfo{pages}{530--541}.
\newblock
\showDOI{%
\url{https://doi.org/10.1145/3377811.3380364}}


\bibitem[\protect\citeauthoryear{Feist, Grieco, and Groce}{Feist
  et~al\mbox{.}}{2019}]%
        {DBLP:conf/icse/slither}
\bibfield{author}{\bibinfo{person}{Josselin Feist}, \bibinfo{person}{Gustavo
  Grieco}, {and} \bibinfo{person}{Alex Groce}.}
  \bibinfo{year}{2019}\natexlab{}.
\newblock \showarticletitle{Slither: a static analysis framework for smart
  contracts}. In \bibinfo{booktitle}{{\em Proceedings of the 2nd International
  Workshop on Emerging Trends in Software Engineering for Blockchain,
  WETSEB@ICSE 2019, Montreal, QC, Canada, May 27, 2019}}.
  \bibinfo{publisher}{{IEEE} / {ACM}}, \bibinfo{pages}{8--15}.
\newblock
\showDOI{%
\url{https://doi.org/10.1109/WETSEB.2019.00008}}


\bibitem[\protect\citeauthoryear{Jiang, Liu, and Chan}{Jiang
  et~al\mbox{.}}{2018}]%
        {0001LC18}
\bibfield{author}{\bibinfo{person}{Bo Jiang}, \bibinfo{person}{Ye Liu}, {and}
  \bibinfo{person}{W.~K. Chan}.} \bibinfo{year}{2018}\natexlab{}.
\newblock \showarticletitle{ContractFuzzer: fuzzing smart contracts for
  vulnerability detection}. In \bibinfo{booktitle}{{\em Proceedings of the 33rd
  {ACM/IEEE} International Conference on Automated Software Engineering, {ASE}
  2018, Montpellier, France, September 3-7, 2018}},
  \bibfield{editor}{\bibinfo{person}{Marianne Huchard},
  \bibinfo{person}{Christian K{\"{a}}stner}, {and} \bibinfo{person}{Gordon
  Fraser}} (Eds.). \bibinfo{publisher}{{ACM}}, \bibinfo{pages}{259--269}.
\newblock
\showDOI{%
\url{https://doi.org/10.1145/3238147.3238177}}


\bibitem[\protect\citeauthoryear{Jiao, Kan, Lin, San{\'{a}}n, Liu, and
  Sun}{Jiao et~al\mbox{.}}{2020}]%
        {DBLP:conf/sp/JiaoK0S0020}
\bibfield{author}{\bibinfo{person}{Jiao Jiao}, \bibinfo{person}{Shuanglong
  Kan}, \bibinfo{person}{Shang{-}Wei Lin}, \bibinfo{person}{David San{\'{a}}n},
  \bibinfo{person}{Yang Liu}, {and} \bibinfo{person}{Jun Sun}.}
  \bibinfo{year}{2020}\natexlab{}.
\newblock \showarticletitle{Semantic Understanding of Smart Contracts:
  Executable Operational Semantics of Solidity}. In \bibinfo{booktitle}{{\em
  2020 {IEEE} Symposium on Security and Privacy, {SP} 2020, San Francisco, CA,
  USA, May 18-21, 2020}}. \bibinfo{publisher}{{IEEE}},
  \bibinfo{pages}{1695--1712}.
\newblock
\showDOI{%
\url{https://doi.org/10.1109/SP40000.2020.00066}}


\bibitem[\protect\citeauthoryear{Kalra, Goel, Dhawan, and Sharma}{Kalra
  et~al\mbox{.}}{2018}]%
        {DBLP:conf/ndss/KalraGDS18}
\bibfield{author}{\bibinfo{person}{Sukrit Kalra}, \bibinfo{person}{Seep Goel},
  \bibinfo{person}{Mohan Dhawan}, {and} \bibinfo{person}{Subodh Sharma}.}
  \bibinfo{year}{2018}\natexlab{}.
\newblock \showarticletitle{{ZEUS:} Analyzing Safety of Smart Contracts}. In
  \bibinfo{booktitle}{{\em 25th Annual Network and Distributed System Security
  Symposium, {NDSS} 2018, San Diego, California, USA, February 18-21, 2018}}.
  \bibinfo{publisher}{The Internet Society}.
\newblock
\showURL{%
\url{http://wp.internetsociety.org/ndss/wp-content/uploads/sites/25/2018/02/ndss2018\_09-1\_Kalra\_paper.pdf}}


\bibitem[\protect\citeauthoryear{Kamarul}{Kamarul}{2020}]%
        {etherscan2020}
\bibfield{author}{\bibinfo{person}{Harith Kamarul}.}
  \bibinfo{year}{2020}\natexlab{}.
\newblock \bibinfo{title}{Ethereum in 2020: the View from the Block Explorer}.
\newblock
  \bibinfo{howpublished}{\url{https://medium.com/etherscan-blog/ethereum-in-2020-the-view-from-the-block-explorer-2f9a1db2ee15}}.
    (\bibinfo{year}{2020}).
\newblock


\bibitem[\protect\citeauthoryear{Kolluri, Nikolic, Sergey, Hobor, and
  Saxena}{Kolluri et~al\mbox{.}}{2019}]%
        {DBLP:conf/issta/KolluriNSHS19}
\bibfield{author}{\bibinfo{person}{Aashish Kolluri}, \bibinfo{person}{Ivica
  Nikolic}, \bibinfo{person}{Ilya Sergey}, \bibinfo{person}{Aquinas Hobor},
  {and} \bibinfo{person}{Prateek Saxena}.} \bibinfo{year}{2019}\natexlab{}.
\newblock \showarticletitle{Exploiting the laws of order in smart contracts}.
  In \bibinfo{booktitle}{{\em Proceedings of the 28th {ACM} {SIGSOFT}
  International Symposium on Software Testing and Analysis, {ISSTA} 2019,
  Beijing, China, July 15-19, 2019}},
  \bibfield{editor}{\bibinfo{person}{Dongmei Zhang} {and}
  \bibinfo{person}{Anders M{\o}ller}} (Eds.). \bibinfo{publisher}{{ACM}},
  \bibinfo{pages}{363--373}.
\newblock
\showDOI{%
\url{https://doi.org/10.1145/3293882.3330560}}


\bibitem[\protect\citeauthoryear{Kremer and K{\"{u}}nnemann}{Kremer and
  K{\"{u}}nnemann}{2014}]%
        {KremerK14}
\bibfield{author}{\bibinfo{person}{Steve Kremer} {and} \bibinfo{person}{Robert
  K{\"{u}}nnemann}.} \bibinfo{year}{2014}\natexlab{}.
\newblock \showarticletitle{Automated Analysis of Security Protocols with
  Global State}. In \bibinfo{booktitle}{{\em 2014 {IEEE} Symposium on Security
  and Privacy, {SP} 2014, Berkeley, CA, USA, May 18-21, 2014}}.
  \bibinfo{publisher}{{IEEE} Computer Society}, \bibinfo{pages}{163--178}.
\newblock
\showDOI{%
\url{https://doi.org/10.1109/SP.2014.18}}


\bibitem[\protect\citeauthoryear{Krupp and Rossow}{Krupp and Rossow}{2018}]%
        {DBLP:conf/uss/KruppR18}
\bibfield{author}{\bibinfo{person}{Johannes Krupp} {and}
  \bibinfo{person}{Christian Rossow}.} \bibinfo{year}{2018}\natexlab{}.
\newblock \showarticletitle{teEther: Gnawing at Ethereum to Automatically
  Exploit Smart Contracts}. In \bibinfo{booktitle}{{\em 27th {USENIX} Security
  Symposium, {USENIX} Security 2018, Baltimore, MD, USA, August 15-17, 2018}},
  \bibfield{editor}{\bibinfo{person}{William Enck} {and}
  \bibinfo{person}{Adrienne~Porter Felt}} (Eds.). \bibinfo{publisher}{{USENIX}
  Association}, \bibinfo{pages}{1317--1333}.
\newblock
\showURL{%
\url{https://www.usenix.org/conference/usenixsecurity18/presentation/krupp}}


\bibitem[\protect\citeauthoryear{Luu, Chu, Olickel, Saxena, and Hobor}{Luu
  et~al\mbox{.}}{2016}]%
        {LuuCOSH16}
\bibfield{author}{\bibinfo{person}{Loi Luu}, \bibinfo{person}{Duc{-}Hiep Chu},
  \bibinfo{person}{Hrishi Olickel}, \bibinfo{person}{Prateek Saxena}, {and}
  \bibinfo{person}{Aquinas Hobor}.} \bibinfo{year}{2016}\natexlab{}.
\newblock \showarticletitle{Making Smart Contracts Smarter}. In
  \bibinfo{booktitle}{{\em Proceedings of the 2016 {ACM} {SIGSAC} Conference on
  Computer and Communications Security, Vienna, Austria, October 24-28, 2016}},
  \bibfield{editor}{\bibinfo{person}{Edgar~R. Weippl}, \bibinfo{person}{Stefan
  Katzenbeisser}, \bibinfo{person}{Christopher Kruegel},
  \bibinfo{person}{Andrew~C. Myers}, {and} \bibinfo{person}{Shai Halevi}}
  (Eds.). \bibinfo{publisher}{{ACM}}, \bibinfo{pages}{254--269}.
\newblock
\showDOI{%
\url{https://doi.org/10.1145/2976749.2978309}}


\bibitem[\protect\citeauthoryear{Meier}{Meier}{2013}]%
        {DBLP:phd/basesearch/Meier13}
\bibfield{author}{\bibinfo{person}{Simon Meier}.}
  \bibinfo{year}{2013}\natexlab{}.
\newblock {\em \bibinfo{title}{Advancing automated security protocol
  verification}}.
\newblock \bibinfo{thesistype}{Ph.D. Dissertation}. \bibinfo{school}{{ETH}
  Zurich, Z{\"{u}}rich, Switzerland}.
\newblock
\showDOI{%
\url{https://doi.org/10.3929/ethz-a-009790675}}


\bibitem[\protect\citeauthoryear{Meier, Schmidt, Cremers, and Basin}{Meier
  et~al\mbox{.}}{2013}]%
        {MeierSCB13}
\bibfield{author}{\bibinfo{person}{Simon Meier}, \bibinfo{person}{Benedikt
  Schmidt}, \bibinfo{person}{Cas Cremers}, {and} \bibinfo{person}{David~A.
  Basin}.} \bibinfo{year}{2013}\natexlab{}.
\newblock \showarticletitle{The {TAMARIN} Prover for the Symbolic Analysis of
  Security Protocols}. In \bibinfo{booktitle}{{\em Computer Aided Verification
  - 25th International Conference, {CAV} 2013, Saint Petersburg, Russia, July
  13-19, 2013. Proceedings}} {\em (\bibinfo{series}{Lecture Notes in Computer
  Science})}, \bibfield{editor}{\bibinfo{person}{Natasha Sharygina} {and}
  \bibinfo{person}{Helmut Veith}} (Eds.), Vol.~\bibinfo{volume}{8044}.
  \bibinfo{publisher}{Springer}, \bibinfo{pages}{696--701}.
\newblock
\showDOI{%
\url{https://doi.org/10.1007/978-3-642-39799-8\_48}}


\bibitem[\protect\citeauthoryear{Onwuzurike, Mariconti, Andriotis, Cristofaro,
  Ross, and Stringhini}{Onwuzurike et~al\mbox{.}}{2019}]%
        {DBLP:journals/tissec/OnwuzurikeMACRS19}
\bibfield{author}{\bibinfo{person}{Lucky Onwuzurike}, \bibinfo{person}{Enrico
  Mariconti}, \bibinfo{person}{Panagiotis Andriotis},
  \bibinfo{person}{Emiliano~De Cristofaro}, \bibinfo{person}{Gordon~J. Ross},
  {and} \bibinfo{person}{Gianluca Stringhini}.}
  \bibinfo{year}{2019}\natexlab{}.
\newblock \showarticletitle{MaMaDroid: Detecting Android Malware by Building
  Markov Chains of Behavioral Models (Extended Version)}.
\newblock \bibinfo{journal}{{\em {ACM} Trans. Priv. Secur.\/}}
  \bibinfo{volume}{22}, \bibinfo{number}{2} (\bibinfo{year}{2019}),
  \bibinfo{pages}{14:1--14:34}.
\newblock
\showDOI{%
\url{https://doi.org/10.1145/3313391}}


\bibitem[\protect\citeauthoryear{Permenev, Dimitrov, Tsankov,
  Drachsler{-}Cohen, and Vechev}{Permenev et~al\mbox{.}}{2020}]%
        {PermenevDTDV20}
\bibfield{author}{\bibinfo{person}{Anton Permenev}, \bibinfo{person}{Dimitar
  Dimitrov}, \bibinfo{person}{Petar Tsankov}, \bibinfo{person}{Dana
  Drachsler{-}Cohen}, {and} \bibinfo{person}{Martin~T. Vechev}.}
  \bibinfo{year}{2020}\natexlab{}.
\newblock \showarticletitle{VerX: Safety Verification of Smart Contracts}. In
  \bibinfo{booktitle}{{\em 2020 {IEEE} Symposium on Security and Privacy, {SP}
  2020, San Francisco, CA, USA, May 18-21, 2020}}. \bibinfo{publisher}{{IEEE}},
  \bibinfo{pages}{1661--1677}.
\newblock
\showDOI{%
\url{https://doi.org/10.1109/SP40000.2020.00024}}


\bibitem[\protect\citeauthoryear{Ro\c{s}u}{Ro\c{s}u}{2017}]%
        {Rosu2017}
\bibfield{author}{\bibinfo{person}{Grigore Ro\c{s}u}.}
  \bibinfo{year}{2017}\natexlab{}.
\newblock \showarticletitle{Matching logic}.
\newblock \bibinfo{journal}{{\em Logical Methods in Computer Science\/}}
  \bibinfo{volume}{13}, \bibinfo{number}{4} (\bibinfo{date}{December}
  \bibinfo{year}{2017}), \bibinfo{pages}{1--61}.
\newblock
\showDOI{%
\url{https://doi.org/abs/1705.06312}}


\bibitem[\protect\citeauthoryear{Ro\c{s}u and \c{S}tef\u{a}nescu}{Ro\c{s}u and
  \c{S}tef\u{a}nescu}{2012}]%
        {Rosu2012}
\bibfield{author}{\bibinfo{person}{Grigore Ro\c{s}u} {and}
  \bibinfo{person}{Andrei \c{S}tef\u{a}nescu}.}
  \bibinfo{year}{2012}\natexlab{}.
\newblock \showarticletitle{Checking Reachability using Matching Logic}. In
  \bibinfo{booktitle}{{\em Proceedings of the 27th Conference on
  Object-Oriented Programming, Systems, Languages, and Applications
  (OOPSLA'12)}}. \bibinfo{publisher}{ACM}, \bibinfo{pages}{555--574}.
\newblock
\showDOI{%
\url{https://doi.org/citation.cfm?doid=2384616.2384656}}


\bibitem[\protect\citeauthoryear{Rodler, Li, Karame, and Davi}{Rodler
  et~al\mbox{.}}{2019}]%
        {RodlerLKD19}
\bibfield{author}{\bibinfo{person}{Michael Rodler}, \bibinfo{person}{Wenting
  Li}, \bibinfo{person}{Ghassan~O. Karame}, {and} \bibinfo{person}{Lucas
  Davi}.} \bibinfo{year}{2019}\natexlab{}.
\newblock \showarticletitle{Sereum: Protecting Existing Smart Contracts Against
  Re-Entrancy Attacks}. In \bibinfo{booktitle}{{\em 26th Annual Network and
  Distributed System Security Symposium, {NDSS} 2019, San Diego, California,
  USA, February 24-27, 2019}}. \bibinfo{publisher}{The Internet Society}.
\newblock
\showURL{%
\url{https://www.ndss-symposium.org/ndss-paper/sereum-protecting-existing-smart-contracts-against-re-entrancy-attacks/}}


\bibitem[\protect\citeauthoryear{Rosu and Serbanuta}{Rosu and
  Serbanuta}{2010}]%
        {Rosu2010}
\bibfield{author}{\bibinfo{person}{Grigore Rosu} {and}
  \bibinfo{person}{Traian{-}Florin Serbanuta}.}
  \bibinfo{year}{2010}\natexlab{}.
\newblock \showarticletitle{An overview of the {K} semantic framework}.
\newblock \bibinfo{journal}{{\em J. Log. Algebraic Methods Program.\/}}
  \bibinfo{volume}{79}, \bibinfo{number}{6} (\bibinfo{year}{2010}),
  \bibinfo{pages}{397--434}.
\newblock
\showDOI{%
\url{https://doi.org/10.1016/j.jlap.2010.03.012}}


\bibitem[\protect\citeauthoryear{Sagiv, Reps, and Wilhelm}{Sagiv
  et~al\mbox{.}}{2002}]%
        {SagivRW02}
\bibfield{author}{\bibinfo{person}{Shmuel Sagiv}, \bibinfo{person}{Thomas~W.
  Reps}, {and} \bibinfo{person}{Reinhard Wilhelm}.}
  \bibinfo{year}{2002}\natexlab{}.
\newblock \showarticletitle{Parametric shape analysis via 3-valued logic}.
\newblock \bibinfo{journal}{{\em {ACM} Trans. Program. Lang. Syst.\/}}
  \bibinfo{volume}{24}, \bibinfo{number}{3} (\bibinfo{year}{2002}),
  \bibinfo{pages}{217--298}.
\newblock
\showDOI{%
\url{https://doi.org/10.1145/514188.514190}}


\bibitem[\protect\citeauthoryear{Schneidewind, Grishchenko, Scherer, and
  Maffei}{Schneidewind et~al\mbox{.}}{2020}]%
        {DBLP:conf/ccs/SchneidewindGSM20}
\bibfield{author}{\bibinfo{person}{Clara Schneidewind}, \bibinfo{person}{Ilya
  Grishchenko}, \bibinfo{person}{Markus Scherer}, {and} \bibinfo{person}{Matteo
  Maffei}.} \bibinfo{year}{2020}\natexlab{}.
\newblock \showarticletitle{eThor: Practical and Provably Sound Static Analysis
  of Ethereum Smart Contracts}. In \bibinfo{booktitle}{{\em {CCS} '20: 2020
  {ACM} {SIGSAC} Conference on Computer and Communications Security, Virtual
  Event, USA, November 9-13, 2020}}, \bibfield{editor}{\bibinfo{person}{Jay
  Ligatti}, \bibinfo{person}{Xinming Ou}, \bibinfo{person}{Jonathan Katz},
  {and} \bibinfo{person}{Giovanni Vigna}} (Eds.). \bibinfo{publisher}{{ACM}},
  \bibinfo{pages}{621--640}.
\newblock
\showDOI{%
\url{https://doi.org/10.1145/3372297.3417250}}


\bibitem[\protect\citeauthoryear{So, Lee, Park, Lee, and Oh}{So
  et~al\mbox{.}}{2020}]%
        {DBLP:conf/sp/SoLPLO20}
\bibfield{author}{\bibinfo{person}{Sunbeom So}, \bibinfo{person}{Myungho Lee},
  \bibinfo{person}{Jisu Park}, \bibinfo{person}{Heejo Lee}, {and}
  \bibinfo{person}{Hakjoo Oh}.} \bibinfo{year}{2020}\natexlab{}.
\newblock \showarticletitle{{VERISMART:} {A} Highly Precise Safety Verifier for
  Ethereum Smart Contracts}. In \bibinfo{booktitle}{{\em 2020 {IEEE} Symposium
  on Security and Privacy, {SP} 2020, San Francisco, CA, USA, May 18-21,
  2020}}. \bibinfo{publisher}{{IEEE}}, \bibinfo{pages}{1678--1694}.
\newblock
\showDOI{%
\url{https://doi.org/10.1109/SP40000.2020.00032}}


\bibitem[\protect\citeauthoryear{Stephens, Ferles, Mariano, Lahiri, and
  Dillig}{Stephens et~al\mbox{.}}{2021}]%
        {stephens2021smartpulse}
\bibfield{author}{\bibinfo{person}{Jon Stephens}, \bibinfo{person}{Kostas
  Ferles}, \bibinfo{person}{Benjamin Mariano}, \bibinfo{person}{Shuvendu
  Lahiri}, {and} \bibinfo{person}{Isil Dillig}.}
  \bibinfo{year}{2021}\natexlab{}.
\newblock \showarticletitle{SmartPulse: Automated Checking of Temporal
  Properties in Smart Contracts}. In \bibinfo{booktitle}{{\em 42nd IEEE
  Symposium on Security and Privacy}}. IEEE.
\newblock
\showURL{%
\url{https://www.microsoft.com/en-us/research/publication/smartpulse-automated-checking-of-temporal-properties-in-smart-contracts/}}


\bibitem[\protect\citeauthoryear{Tikhomirov, Voskresenskaya, Ivanitskiy,
  Takhaviev, Marchenko, and Alexandrov}{Tikhomirov et~al\mbox{.}}{2018a}]%
        {DBLP:conf/icse/smartcheck}
\bibfield{author}{\bibinfo{person}{Sergei Tikhomirov},
  \bibinfo{person}{Ekaterina Voskresenskaya}, \bibinfo{person}{Ivan
  Ivanitskiy}, \bibinfo{person}{Ramil Takhaviev}, \bibinfo{person}{Evgeny
  Marchenko}, {and} \bibinfo{person}{Yaroslav Alexandrov}.}
  \bibinfo{year}{2018}\natexlab{a}.
\newblock \showarticletitle{SmartCheck: Static Analysis of Ethereum Smart
  Contracts}. In \bibinfo{booktitle}{{\em 1st {IEEE/ACM} International Workshop
  on Emerging Trends in Software Engineering for Blockchain, WETSEB@ICSE 2018,
  Gothenburg, Sweden, May 27 - June 3, 2018}},
  \bibfield{editor}{\bibinfo{person}{Roberto Tonelli},
  \bibinfo{person}{Giuseppe Destefanis}, \bibinfo{person}{Steve Counsell},
  {and} \bibinfo{person}{Michele Marchesi}} (Eds.). \bibinfo{publisher}{{ACM}},
  \bibinfo{pages}{9--16}.
\newblock
\showDOI{%
\url{https://doi.org/10.1145/3194113.3194115}}


\bibitem[\protect\citeauthoryear{Tikhomirov, Voskresenskaya, Ivanitskiy,
  Takhaviev, Marchenko, and Alexandrov}{Tikhomirov et~al\mbox{.}}{2018b}]%
        {DBLP:conf/icse/TikhomirovVITMA18}
\bibfield{author}{\bibinfo{person}{Sergei Tikhomirov},
  \bibinfo{person}{Ekaterina Voskresenskaya}, \bibinfo{person}{Ivan
  Ivanitskiy}, \bibinfo{person}{Ramil Takhaviev}, \bibinfo{person}{Evgeny
  Marchenko}, {and} \bibinfo{person}{Yaroslav Alexandrov}.}
  \bibinfo{year}{2018}\natexlab{b}.
\newblock \showarticletitle{SmartCheck: Static Analysis of Ethereum Smart
  Contracts}. In \bibinfo{booktitle}{{\em 1st {IEEE/ACM} International Workshop
  on Emerging Trends in Software Engineering for Blockchain, WETSEB@ICSE 2018,
  Gothenburg, Sweden, May 27 - June 3, 2018}}. \bibinfo{publisher}{{ACM}},
  \bibinfo{pages}{9--16}.
\newblock
\showURL{%
\url{http://ieeexplore.ieee.org/document/8445052}}


\bibitem[\protect\citeauthoryear{Tolmach, Li, Lin, Liu, and Li}{Tolmach
  et~al\mbox{.}}{2020}]%
        {abs-2008-02712}
\bibfield{author}{\bibinfo{person}{Palina Tolmach}, \bibinfo{person}{Yi Li},
  \bibinfo{person}{Shang{-}Wei Lin}, \bibinfo{person}{Yang Liu}, {and}
  \bibinfo{person}{Zengxiang Li}.} \bibinfo{year}{2020}\natexlab{}.
\newblock \showarticletitle{A Survey of Smart Contract Formal Specification and
  Verification}.
\newblock \bibinfo{journal}{{\em CoRR\/}}  \bibinfo{volume}{abs/2008.02712}
  (\bibinfo{year}{2020}).
\newblock
\showeprint[arxiv]{2008.02712}
\showURL{%
\url{https://arxiv.org/abs/2008.02712}}


\bibitem[\protect\citeauthoryear{Torres, Sch{\"{u}}tte, and State}{Torres
  et~al\mbox{.}}{2018a}]%
        {DBLP:conf/acsac/osiris}
\bibfield{author}{\bibinfo{person}{Christof~Ferreira Torres},
  \bibinfo{person}{Julian Sch{\"{u}}tte}, {and} \bibinfo{person}{Radu State}.}
  \bibinfo{year}{2018}\natexlab{a}.
\newblock \showarticletitle{Osiris: Hunting for Integer Bugs in Ethereum Smart
  Contracts}. In \bibinfo{booktitle}{{\em Proceedings of the 34th Annual
  Computer Security Applications Conference, {ACSAC} 2018, San Juan, PR, USA,
  December 03-07, 2018}}. \bibinfo{publisher}{{ACM}},
  \bibinfo{pages}{664--676}.
\newblock
\showDOI{%
\url{https://doi.org/10.1145/3274694.3274737}}


\bibitem[\protect\citeauthoryear{Torres, Sch{\"{u}}tte, and State}{Torres
  et~al\mbox{.}}{2018b}]%
        {DBLP:conf/acsac/TorresSS18}
\bibfield{author}{\bibinfo{person}{Christof~Ferreira Torres},
  \bibinfo{person}{Julian Sch{\"{u}}tte}, {and} \bibinfo{person}{Radu State}.}
  \bibinfo{year}{2018}\natexlab{b}.
\newblock \showarticletitle{Osiris: Hunting for Integer Bugs in Ethereum Smart
  Contracts}. In \bibinfo{booktitle}{{\em Proceedings of the 34th Annual
  Computer Security Applications Conference, {ACSAC} 2018, San Juan, PR, USA,
  December 03-07, 2018}}. \bibinfo{publisher}{{ACM}},
  \bibinfo{pages}{664--676}.
\newblock
\showDOI{%
\url{https://doi.org/10.1145/3274694.3274737}}


\bibitem[\protect\citeauthoryear{Tsankov, Dan, Drachsler{-}Cohen, Gervais,
  B{\"{u}}nzli, and Vechev}{Tsankov et~al\mbox{.}}{2018}]%
        {TsankovDDGBV18}
\bibfield{author}{\bibinfo{person}{Petar Tsankov},
  \bibinfo{person}{Andrei~Marian Dan}, \bibinfo{person}{Dana
  Drachsler{-}Cohen}, \bibinfo{person}{Arthur Gervais},
  \bibinfo{person}{Florian B{\"{u}}nzli}, {and} \bibinfo{person}{Martin~T.
  Vechev}.} \bibinfo{year}{2018}\natexlab{}.
\newblock \showarticletitle{Securify: Practical Security Analysis of Smart
  Contracts}. In \bibinfo{booktitle}{{\em Proceedings of the 2018 {ACM}
  {SIGSAC} Conference on Computer and Communications Security, {CCS} 2018,
  Toronto, ON, Canada, October 15-19, 2018}},
  \bibfield{editor}{\bibinfo{person}{David Lie}, \bibinfo{person}{Mohammad
  Mannan}, \bibinfo{person}{Michael Backes}, {and} \bibinfo{person}{XiaoFeng
  Wang}} (Eds.). \bibinfo{publisher}{{ACM}}, \bibinfo{pages}{67--82}.
\newblock
\showDOI{%
\url{https://doi.org/10.1145/3243734.3243780}}


\bibitem[\protect\citeauthoryear{Zhou, Yang, Xiang, Cao, Yang, and Zhang}{Zhou
  et~al\mbox{.}}{2020}]%
        {DBLP:conf/uss/ZhouYXCY020}
\bibfield{author}{\bibinfo{person}{Shunfan Zhou}, \bibinfo{person}{Zhemin
  Yang}, \bibinfo{person}{Jie Xiang}, \bibinfo{person}{Yinzhi Cao},
  \bibinfo{person}{Min Yang}, {and} \bibinfo{person}{Yuan Zhang}.}
  \bibinfo{year}{2020}\natexlab{}.
\newblock \showarticletitle{An Ever-evolving Game: Evaluation of Real-world
  Attacks and Defenses in Ethereum Ecosystem}. In \bibinfo{booktitle}{{\em 29th
  {USENIX} Security Symposium, {USENIX} Security 2020, August 12-14, 2020}},
  \bibfield{editor}{\bibinfo{person}{Srdjan Capkun} {and}
  \bibinfo{person}{Franziska Roesner}} (Eds.). \bibinfo{publisher}{{USENIX}
  Association}, \bibinfo{pages}{2793--2810}.
\newblock
\showURL{%
\url{https://www.usenix.org/conference/usenixsecurity20/presentation/zhou-shunfan}}


\end{thebibliography}
\appendix
\section{APPENDIX} 
\label{sec:appendix}

\subsection{Complete definition of function $\kwc{R}$ and rules mentioned in Section \ref{subsec:translation}}
\label{subsec:appendix1}

\textbf{Definition of function $\kwc{R}$}

The function $\kwc{R}$ shown in Fig. \ref{fig:translation_a1} translates five categories of statements into rules.
Since the translation of assignment statements and return statements are introduced previously, we introduce the translation of statements in the remaining three categories here:

\begin{figure*}[!]
\centering
\setlength{\abovedisplayskip}{0pt} 
\setlength{\belowdisplayskip}{0pt}
\begin{small}
\begin{align*}
&\kwc{R}(\texttt{function}\ f(\texttt{d})\{\kwp{stmt}\},\varnothing,\omega_0) = \kwc{R}(\kwp{stmt},1,\llbracket\left\langle \sigma_a(f),\kwf{T_c},\kwf{R_o},\kwf{E_n} \right\rangle ,\left\langle \sigma_v(c_b),\kwf{T_v},\kwf{R_l},\kwf{E_n} \right\rangle, \left\langle \sigma_v(calltype),\kwf{T_v},\kwf{R_l},\kwf{E_n} \right\rangle, \left\langle \sigma_v(\kwf{depth}),\kwf{T_v},\kwf{R_l},\kwf{E_n} \right\rangle\rrbracket\\&\qquad\qquad\qquad\qquad\qquad\qquad\ \  \cdot \omega_0 \cdot  seq(\texttt{d}))  \cup \{ [\kwf{Fr}(\sigma_v(c_b)),\kwf{FR}(\sigma(seq(\texttt{d})))]-[]\rightarrow  [\kwf{Call_e}( \llbracket\omega_0[1],\sigma_a(f),\sigma_v(c_b)\rrbracket\cdot \sigma(seq(\texttt{d})))],\quad (\kwt{ext\_call}) \\&\qquad\qquad\quad\qquad\qquad\quad\qquad\qquad\quad [\kwf{Call_e}( \llbracket\omega_0[1],\sigma_a(f), \sigma_v(c_b)\rrbracket \cdot\sigma(seq(\texttt{d}))),\kwf{Evar}(e(\omega_0)),\kwf{Gvar}(\llbracket \omega_0[1] \rrbracket \cdot g(\omega_0)\backslash e(\omega_0))]-[]\rightarrow\\&\qquad\qquad\quad\qquad\qquad\quad\qquad\qquad\quad\qquad\qquad\quad\ \ \ [\kwf{Var}_1( \llbracket\sigma_a(f),\sigma_v(c_b),\kwf{EXT},0\rrbracket\cdot\sigma(\omega_0)\cdot \sigma(seq(\texttt{d})))],\qquad\qquad\qquad\qquad\qquad\ (\kwt{recv\_ext}) 
\\&\qquad\qquad\quad\qquad\qquad\quad\qquad\qquad\quad [\kwf{Call_{in}}(\llbracket\omega_0[1],\sigma_a(f),\sigma_v(c_b),\sigma_v(\kwf{depth})\rrbracket\cdot \sigma(seq(\texttt{d}))),\kwf{Evar}(e(\omega_0)),\kwf{Gvar}(\llbracket \omega_0[1] \rrbracket \cdot g(\omega_0)\backslash e(\omega_0))]-[]\rightarrow\\&\qquad\qquad\quad\qquad\qquad\quad\qquad\qquad\quad\qquad\qquad\quad\ \ \ [\kwf{Var}_1( \llbracket\sigma_a(f),\sigma_v(c_b),\kwf{IN},\sigma_v(\kwf{depth})\rrbracket\cdot\sigma(\omega_0)\cdot \sigma(seq(\texttt{d})))]\}\qquad\qquad\qquad\qquad\ (\kwt{recv\_in})  \\
&\kwc{R}(v_1\leftarrow v_2;\kwp{stmt},i,\omega) = \kwc{R}(\kwp{stmt},i\circ 1,\omega) \cup \{ [\kwf{Var}_i(\sigma(\omega))]-[]\rightarrow [\kwf{Var}_{i\circ 1}(\sigma(\omega)|\frac{\sigma_v(v_1)}{\sigma_v(v_2)})]\}\qquad\qquad\quad\qquad\qquad\quad\qquad\qquad\quad (\kwt{var\_assign})  \\
&\kwc{R}(\tau\ v_1\leftarrow v_2;\kwp{stmt},i,\omega) = \kwc{R}(\kwp{stmt},i\circ 1,\omega\cdot \llbracket\left\langle \sigma_v(v_1),\kwf{T_v},\kwf{R_l},\kwf{E_n} \right\rangle\rrbracket ) \cup \{[\kwf{Var}_i(\sigma(\omega))]-[]\rightarrow [\kwf{Var}_{i\circ 1}(\sigma(\omega)\cdot \llbracket\sigma_v(v_2)\rrbracket)]\}\qquad\qquad(\kwt{var\_declare})\\
&\kwc{R}(\texttt{if}\ e_{b}\ \texttt{then}\ \kwp{stmt}_1\ \texttt{else}\ \kwp{stmt}_2;\kwp{stmt}_3,i,\omega) = \kwc{R}(\kwp{stmt}_1;\kwp{stmt}_3,i\circ 1,\omega) \cup \kwc{R}(\kwp{stmt}_2;\kwp{stmt}_3,i\circ 2,\omega) \cup \\&\qquad\qquad\qquad\qquad\qquad\qquad\qquad\qquad\qquad\quad \{ [\kwf{Var}_i(\sigma(\omega))] -[\theta_e(e_b)]\rightarrow [\kwf{Var}_{i\circ 1}(\sigma(\omega))], \qquad\qquad\qquad\qquad\qquad\qquad \qquad  \qquad  (\kwt{if\_true})\\
& \qquad\qquad\qquad\qquad\qquad\qquad\qquad\qquad\qquad\quad\ \ [\kwf{Var}_i(\sigma(\omega))] -[\theta_{ne}(e_b)]\rightarrow [\kwf{Var}_{i\circ 2}(\sigma(\omega))]\} \qquad\qquad\qquad\qquad\qquad\qquad \qquad \quad\ \   (\kwt{if\_false})\\
&\kwc{R}(\texttt{require}\ e_{b}\;\kwp{stmt}_1,i,\omega) = \kwc{R}(\kwp{stmt}_1,i\circ 1,\omega) \cup \{ [\kwf{Var}_i(\sigma(\omega))] -[\theta_e(e_b)]\rightarrow [\kwf{Var}_{i\circ 1}(\sigma(\omega))], \qquad\qquad\qquad\qquad\qquad\qquad   (\kwt{require\_true})\\
& \qquad\qquad\qquad\qquad\qquad\qquad\qquad\qquad\quad\ \ [\kwf{Var}_i(\sigma(\omega))] -[\theta_{ne}(e_b)]\rightarrow []\} \qquad\qquad\qquad\qquad\qquad\qquad \qquad \qquad \qquad \qquad (\kwt{require\_false})\\
&\kwc{R}(\texttt{return},i,\omega) = \{ [\kwf{Var}_i(\sigma(\omega))]  -[\kwf{Pred\_eq}(\omega[3],\kwf{EXT})]  \rightarrow  [\kwf{Gvar}(\llbracket \omega[5] \rrbracket \cdot g(\omega)\backslash e(\omega)),\kwf{Evar}(e(\omega))],  \qquad\qquad\qquad\qquad\qquad\quad\ \ (\kwt{ret\_ext}) \\
&\qquad \qquad \qquad \qquad \quad \  [\kwf{Var}_i(\sigma(\omega))]  -[\kwf{Pred\_eq}(\omega[3],\kwf{IN})]  \rightarrow  [\kwf{Return}(\llbracket \omega[5],\omega[1],\omega[2],\omega[4]\rrbracket,\kwf{Gvar}(\llbracket \omega[5] \rrbracket \cdot g(\omega)\backslash e(\omega)),\kwf{Evar}(e(\omega)) )]\} \\& \qquad\qquad\qquad\qquad\qquad\qquad\qquad\qquad\qquad\qquad\qquad\qquad\qquad\qquad\qquad\qquad\qquad\qquad\qquad\qquad\qquad\qquad\qquad\qquad\qquad\qquad\quad\   (\kwt{ret\_in}) \\
&\kwc{R}(x(c_x).f_x(p);\kwp{stmt},i,\omega) = \kwc{R}(\kwp{stmt},i\circ 1 \circ 1,\omega) \cup \{ [\kwf{Var}_i(\sigma(\omega))] -[] \rightarrow  [\kwf{Call_{in}}( \llbracket\sigma_a(c_x),\sigma_a(f_x),\omega[5],\omega[4] \oplus 1\rrbracket\cdot \sigma_s(p)), \kwf{Var}_{i\circ 1}(\sigma(\omega)\backslash g(\omega)),\\&\qquad\qquad\qquad\qquad\qquad\qquad\qquad\qquad\qquad\qquad\qquad\qquad\qquad\qquad\qquad\qquad\qquad\qquad\qquad\qquad \kwf{Gvar}(\llbracket \omega[5] \rrbracket \cdot g(\omega)\backslash e(\omega)),\kwf{Evar}(e(\omega))],\\& \qquad\qquad\qquad\qquad\qquad\qquad\qquad\qquad\qquad\qquad\qquad\qquad\qquad\qquad\qquad\qquad\qquad\qquad\qquad\qquad\qquad\qquad\qquad\qquad\qquad\qquad\ \   (\kwt{in\_call}) \\
&\qquad \qquad \qquad \qquad \qquad \qquad [\kwf{Return}(\llbracket \sigma_a(c_x),\sigma_a(f_x),\omega[5],\sigma_v(\kwf{r\_depth})\rrbracket ),\kwf{Var}_{i\circ 1}(\sigma(\omega)\backslash g(\omega)),\kwf{Gvar}(\llbracket\omega[5] \rrbracket \cdot g(\omega)\backslash e(\omega)), \kwf{Evar}(e(\omega))]\\& \qquad \qquad \qquad \qquad \qquad \qquad \qquad \qquad \qquad \qquad \qquad \qquad \qquad \qquad \qquad \qquad \qquad -[\kwf{Pred\_eq}(\sigma_v(\kwf{r\_depth}),\omega[4] \oplus 1)] \rightarrow [\kwf{Var}_{i\circ 1 \circ 1}(\sigma(\omega))]\} \\&\qquad\qquad\qquad\qquad\qquad\qquad\qquad\qquad\qquad\qquad\qquad\qquad\qquad\qquad\qquad\qquad\qquad\qquad\qquad\qquad\qquad\qquad\qquad\qquad\qquad\qquad    (\kwt{recv\_ret}) \\
&\kwc{R}(c_x.\texttt{transfer}(v_1);\kwp{stmt},i,\omega) = \kwc{R}(\kwp{stmt},i\circ 1,\omega)   \cup\\&\ \ \{ [\kwf{Var}_i(\sigma(\omega))] - []\rightarrow  [\kwf{Var}_{i\circ 1}(\sigma(\omega)|\frac{\sigma_v(c_x)}{\sigma_v(c_x)\oplus \sigma_v(v_1)}|\frac{\sigma_v(c)}{\sigma_v(c)\ominus \sigma_v(v_1)})],\qquad\qquad\qquad\qquad\qquad\qquad\qquad\qquad\qquad\qquad\quad (\kwt{transfer\_succ})  \\
&\ \  [\kwf{Var}_i(\sigma(\omega))] - []\rightarrow [],\ \ \ \ \qquad \qquad \qquad \qquad \qquad \qquad \qquad \qquad\qquad \qquad\qquad\qquad\qquad\qquad \qquad\qquad\qquad\qquad\qquad\quad\ \ (\kwt{transfer\_fail})  \\
&\kwc{R}(c_x.\texttt{send}(v_1);\kwp{stmt},i,\omega) = \kwc{R}(\kwp{stmt},i\circ 1,\omega) \cup \kwc{R}(\kwp{stmt},i\circ 2,\omega)  \cup\\&\ \ \{ [\kwf{Var}_i(\sigma(\omega))] - []\rightarrow  [\kwf{Var}_{i\circ 1}(\sigma(\omega)|\frac{\sigma_v(c_x)}{\sigma_v(c_x)\oplus \sigma_v(v_1)}|\frac{\sigma_v(c)}{\sigma_v(c)\ominus \sigma_v(v_1)})],\qquad\qquad\qquad\qquad\qquad\qquad\qquad\qquad\qquad\qquad\qquad\quad (\kwt{send\_succ})  \\
&\ \  [\kwf{Var}_i(\sigma(\omega))] - []\rightarrow [\kwf{Var}_{i\circ 2}(\sigma(\omega))],\ \ \ \ \qquad \qquad \qquad \qquad \qquad \qquad \qquad \qquad\qquad \qquad\qquad\qquad\qquad\qquad \qquad\qquad\qquad\qquad (\kwt{send\_fail})  \\
&\kwc{R}(c_x.\texttt{call}().\texttt{value}(v_1);\kwp{stmt},i,\omega) = \kwc{R}(\kwp{stmt},i\circ 1,\omega) \cup \kwc{R}(\kwp{stmt},i\circ 2,\omega) \cup \kwc{R}(\kwp{stmt},i\circ 3\circ 1,\omega) \cup\\&\ \ \{ [\kwf{Var}_i(\sigma(\omega))] - []\rightarrow  [\kwf{Var}_{i\circ 1}(\sigma(\omega)|\frac{\sigma_v(c_x)}{\sigma_v(c_x)\oplus \sigma_v(v_1)}|\frac{\sigma_v(c)}{\sigma_v(c)\ominus \sigma_v(v_1)})],\qquad\qquad\qquad\qquad\qquad\qquad\qquad\qquad\qquad\qquad\quad\qquad  (\kwt{ether\_succ})  \\
&\ \  [\kwf{Var}_i(\sigma(\omega))] - []\rightarrow [\kwf{Var}_{i\circ 2}(\sigma(\omega))],\ \ \ \  \qquad \qquad \qquad \qquad \qquad \qquad \qquad\qquad \qquad\qquad\qquad\qquad\qquad \qquad\qquad\qquad\qquad\qquad (\kwt{ether\_fail})  \\
&\ \ [\kwf{Var}_i(\sigma(\omega))] - []\rightarrow [\kwf{Var}_{i\circ 3}(l(\omega)), \kwf{Fallback}(\llbracket\omega[5],\omega[1]\rrbracket), \kwf{Gvar}(\llbracket \omega[5] \rrbracket \cdot g(\omega)\backslash e(\omega)),\kwf{Evar}(e(\omega)|\frac{\sigma_v(c_x)}{\sigma_v(c_x)\oplus \sigma_v(v_1)}|\frac{\sigma_v(c)}{\sigma_v(c)\ominus \sigma_v(v_1)})],(\kwt{fb\_call}) \\
&\ \ [\kwf{ReturnFallback}(\llbracket\omega[5],\omega[1]\rrbracket ), \kwf{Gvar}(\llbracket \omega[5] \rrbracket \cdot g(\omega)\backslash e(\omega)),\kwf{Evar}(e(\omega)), \kwf{Var}_{i\circ 3}(l(\omega))] -[]\rightarrow[\kwf{Var}_{i\circ 3\circ 1}(\sigma(\omega))]\} \qquad \qquad \qquad\quad(\kwt{recv\_fb\_ret})\\
\end{align*}
\end{small}
\vspace{-0.35in}
\caption{Complete definition of function $\kwc{R}$.}
\label{fig:translation_a1}
\end{figure*}

\begin{figure*}[!]
\centering
\setlength{\abovedisplayskip}{0pt} 
\setlength{\belowdisplayskip}{0pt}
\begin{small}
\begin{align*}
&\kwc{R'}(\texttt{function}\ f(\texttt{d})\{\kwp{stmt}\},\varnothing,\omega_0,\varnothing) = \kwc{R'}(\kwp{stmt},1,\llbracket\left\langle \sigma_a(f),\kwf{T_c},\kwf{R_o},\kwf{E_n} \right\rangle ,\left\langle \sigma_v(c_b),\kwf{T_v},\kwf{R_l},\kwf{E_n} \right\rangle, \left\langle \sigma_v(calltype),\kwf{T_v},\kwf{R_l},\kwf{E_n} \right\rangle,\\&\qquad\qquad\qquad  \left\langle \sigma_v(\kwf{depth}),\kwf{T_v},\kwf{R_l},\kwf{E_n} \right\rangle\rrbracket\cdot \omega_0 \cdot  seq(\texttt{d}),\kwf{A}) \cup \kwc{R'}(\kwp{stmt},1,\llbracket\left\langle \sigma_a(f),\kwf{T_c},\kwf{R_o},\kwf{E_n} \right\rangle ,\left\langle \sigma_v(c_b),\kwf{T_v},\kwf{R_l},\kwf{E_n} \right\rangle, \left\langle \sigma_v(calltype),\kwf{T_v},\kwf{R_l},\kwf{E_n} \right\rangle, \\&\qquad\qquad\qquad \left\langle \sigma_v(\kwf{depth}),\kwf{T_v},\kwf{R_l},\kwf{E_n} \right\rangle\rrbracket\cdot \omega_0 \cdot  seq(\texttt{d}),\kwf{B})  \cup \{ [\kwf{Fr}(\sigma_v(c_b)),\kwf{FR}(\sigma(seq(\texttt{d})))]-[]\rightarrow  [\kwf{Call_{Ae}}( \llbracket\omega_0[1],\sigma_a(f),\sigma_v(c_b)\rrbracket\cdot \sigma(seq(\texttt{d}))),\\&\qquad\qquad\qquad\kwf{Call_{Be}}( \llbracket\omega_0[1],\sigma_a(f),\sigma_v(c_b)\rrbracket\cdot \sigma(seq(\texttt{d})))]\qquad\qquad\qquad\qquad\qquad\qquad\qquad\qquad\qquad\qquad\qquad\qquad\qquad\ \ \ \ (\kwt{ext\_call\_AB}) 
\\& [\kwf{Call_{Ae}}( \llbracket\omega_0[1],\sigma_a(f), \sigma_v(c_b)\rrbracket \cdot\sigma(seq(\texttt{d}))),\kwf{Evar_A}(e(\omega_0)),\kwf{Gvar_A}(\llbracket \omega_0[1] \rrbracket \cdot g(\omega_0)\backslash e(\omega_0))]-[\kwf{Exc_A}(\sigma_v(c_b)),\sigma_a(f))]\rightarrow \\& \qquad\qquad\qquad\qquad\qquad\qquad\qquad\qquad\qquad\qquad\qquad\qquad [\kwf{Var_{A1}}( \llbracket\sigma_a(f),\sigma_v(c_b),\kwf{EXT}\rrbracket\cdot\sigma(\omega_0)\cdot \sigma(seq(\texttt{d})))]\qquad\qquad\qquad\   (\kwt{recv\_ext\_A}) 
\\& [\kwf{Call_{Be}}( \llbracket\omega_0[1],\sigma_a(f), \sigma_v(c_b)\rrbracket \cdot\sigma(seq(\texttt{d}))),\kwf{Evar_B}(e(\omega_0)),\kwf{Gvar_B}(\llbracket \omega_0[1] \rrbracket \cdot g(\omega_0)\backslash e(\omega_0))]-[\kwf{Exc_B}(\sigma_v(c_b)),\sigma_a(f))]\rightarrow\\& \qquad\qquad\qquad\qquad\qquad\qquad\qquad\qquad\qquad\qquad [\kwf{Var_{B1}}( \llbracket\sigma_a(f),\sigma_v(c_b),\kwf{EXT}\rrbracket\cdot\sigma(\omega_0)\cdot \sigma(seq(\texttt{d})))] \qquad\qquad\qquad\qquad\qquad\   (\kwt{recv\_ext\_B}) 
\\&[\kwf{Call_{Ain}}(\llbracket\omega_0[1],\sigma_a(f),\sigma_v(c_b),\sigma_v(\kwf{depth})\rrbracket\cdot \sigma(seq(\texttt{d}))),\kwf{Evar_A}(e(\omega_0)),\kwf{Gvar_A}\llbracket \omega_0[1] \rrbracket \cdot g(\omega_0)\backslash e(\omega_0))]-[]\rightarrow\\&\qquad\qquad\quad\qquad\qquad\quad\qquad\qquad\quad\qquad\qquad\quad\ \ \ \ [\kwf{Var_{A1}}( \llbracket\sigma_a(f),\sigma_v(c_b),\kwf{IN},\sigma_v(\kwf{depth})\rrbracket\cdot\sigma(\omega_0)\cdot \sigma(seq(\texttt{d})))]\}\qquad\qquad\quad \ \ (\kwt{recv\_in\_A}) 
\\&[\kwf{Call_{Bin}}(\llbracket\omega_0[1],\sigma_a(f),\sigma_v(c_b),\sigma_v(\kwf{depth})\rrbracket\cdot \sigma(seq(\texttt{d}))),\kwf{Evar_B}(e(\omega_0)),\kwf{Gvar_B}\llbracket \omega_0[1] \rrbracket \cdot g(\omega_0)\backslash e(\omega_0))]-[]\rightarrow\\&\qquad\qquad\quad\qquad\qquad\quad\qquad\qquad\quad\qquad\qquad\quad\ \ \ \ [\kwf{Var_{B1}}( \llbracket\sigma_a(f),\sigma_v(c_b),\kwf{IN},\sigma_v(\kwf{depth})\rrbracket\cdot\sigma(\omega_0)\cdot \sigma(seq(\texttt{d})))]\}\qquad\qquad\quad \ \ (\kwt{recv\_in\_B})  \\
&\kwc{R'}(v_1\leftarrow v_2;\kwp{stmt},i,\omega,\kwf{A}) = \kwc{R'}(\kwp{stmt},i\circ 1,\omega, \kwf{A}) \cup \{ [\kwf{Var}_{\kwf{A} \circ i}(\sigma(\omega))]-[]\rightarrow [\kwf{Var}_{\kwf{A}\circ i\circ 1}(\sigma(\omega)|\frac{\sigma_v(v_1)}{\sigma_v(v_2)})]\}\quad\qquad\qquad\quad\qquad\ (\kwt{var\_assign\_A})  \\
&\kwc{R'}(\tau\ v_1\leftarrow v_2;\kwp{stmt},i,\omega,\kwf{A}) = \kwc{R'}(\kwp{stmt},i\circ 1,\omega\cdot \llbracket\left\langle \sigma_v(v_1),\kwf{T_v},\kwf{R_l},\kwf{E_n} \right\rangle\rrbracket,\kwf{A} ) \cup \{[\kwf{Var}_{\kwf{A} \circ i}(\sigma(\omega))]-[]\rightarrow [\kwf{Var}_{\kwf{A}\circ i\circ 1}(\sigma(\omega)\cdot \llbracket\sigma_v(v_2)\rrbracket)]\}\\& \qquad\qquad\qquad\qquad\qquad\qquad\qquad\qquad\qquad\qquad\qquad\qquad\qquad\qquad\qquad\qquad\qquad\qquad\qquad\qquad\qquad\qquad\qquad\qquad\quad  (\kwt{var\_declare\_A})\\
&\kwc{R'}(\texttt{if}\ e_{b}\ \texttt{then}\ \kwp{stmt}_1\ \texttt{else}\ \kwp{stmt}_2;\kwp{stmt}_3,i,\omega,\kwf{A}) = \kwc{R'}(\kwp{stmt}_1;\kwp{stmt}_3,i\circ 1,\omega,\kwf{A}) \cup \kwc{R'}(\kwp{stmt}_2;\kwp{stmt}_3,i\circ 2,\omega,\kwf{A}) \cup \\&\qquad\qquad\qquad\qquad\qquad\qquad\qquad\qquad\qquad\quad \{ [\kwf{Var}_{\kwf{A}\circ i}(\sigma(\omega))] -[\theta_e(e_b)]\rightarrow [\kwf{Var}_{\kwf{A}\circ i \circ 1}(\sigma(\omega))], \qquad\qquad\qquad\qquad \qquad \qquad   \ \  (\kwt{if\_true\_A})\\
& \qquad\qquad\qquad\qquad\qquad\qquad\qquad\qquad\qquad\quad\ \ [\kwf{Var}_{\kwf{A}\circ i}(\sigma(\omega))] -[\theta_{ne}(e_b)]\rightarrow [\kwf{Var}_{\kwf{A}\circ i\circ 2}(\sigma(\omega))]\} \qquad\qquad\qquad\qquad \qquad   \quad \ \ (\kwt{if\_false\_A})\\
&\kwc{R'}(\texttt{require}\ e_{b}\;\kwp{stmt}_1,i,\omega,\kwf{A}) = \kwc{R'}(\kwp{stmt}_1,i\circ 1,\omega,\kwf{A}) \cup \{ [\kwf{Var}_{\kwf{A}\circ i}(\sigma(\omega))] -[\theta_e(e_b)]\rightarrow [\kwf{Var}_{\kwf{A}\circ i \circ 1}(\sigma(\omega))], \qquad\qquad\quad  \quad \  (\kwt{require\_true\_A})\\
& \qquad\qquad\qquad\qquad\qquad\qquad\qquad\qquad\qquad\quad\ \ [\kwf{Var}_{\kwf{A}\circ i}(\sigma(\omega))] -[\theta_{ne}(e_b)]\rightarrow []\} \qquad\qquad\qquad \qquad \qquad \qquad \qquad \quad \ \ (\kwt{require\_false\_A})\\
&\kwc{R'}(\texttt{return},i,\omega,\kwf{A}) = \{ [\kwf{Var}_{\kwf{A}\circ i}(\sigma(\omega))]  -[\kwf{Pred\_eq}(\omega[3],\kwf{EXT})]  \rightarrow  [\kwf{Gvar}_{\kwf{A}}\llbracket \omega[5] \rrbracket \cdot g(\omega)\backslash e(\omega)),\kwf{Evar}_{\kwf{A}}(e(\omega))],  \qquad\qquad\qquad\quad \ (\kwt{ret\_ext\_A}) \\
&\qquad \qquad \qquad \qquad \quad \  [\kwf{Var}_{\kwf{A}\circ i}(\sigma(\omega))]  -[\kwf{Pred\_eq}(\omega[3],\kwf{IN})]  \rightarrow  [\kwf{Return}_{\kwf{A}}(\llbracket \omega[5],\omega[1],\omega[2],\omega[4]\rrbracket,\kwf{Gvar}_{\kwf{A}}\llbracket \omega[5] \rrbracket \cdot g(\omega)\backslash e(\omega)),\kwf{Evar}_{\kwf{A}}(e(\omega)) )]\} \\&\qquad\qquad\qquad\qquad\qquad\qquad\qquad\qquad\qquad\qquad\qquad\qquad\qquad\qquad\qquad\qquad\qquad\qquad\qquad\qquad\qquad\qquad\qquad\qquad\qquad \qquad (\kwt{ret\_in\_A}) \\
&\kwc{R'}(x(c_x).f_x(p);\kwp{stmt},i,\omega,\kwf{A}) = \kwc{R'}(\kwp{stmt},i\circ 1 \circ 1,\omega,\kwf{A}) \cup \{ [\kwf{Var}_{\kwf{A}\circ i}(\sigma(\omega))] -[] \rightarrow  [\kwf{Call}_{\kwf{Ain}}( \llbracket\sigma_a(c_x),\sigma_a(f_x),\omega[5],\omega[4] \oplus 1\rrbracket\cdot \sigma_s(p)),\\&\qquad\qquad\qquad\qquad\qquad\qquad\qquad\qquad\qquad\qquad\qquad\qquad\qquad\qquad \kwf{Var}_{\kwf{A}\circ i \circ 1}(\sigma(\omega)\backslash g(\omega)), \kwf{Gvar}_{\kwf{A}}\llbracket \omega[5] \rrbracket \cdot g(\omega)\backslash e(\omega)),\kwf{Evar}_{\kwf{A}}(e(\omega))],\\& \qquad\qquad\qquad\qquad\qquad\qquad\qquad\qquad\qquad\qquad\qquad\qquad\qquad\qquad\qquad\qquad\qquad\qquad\qquad\qquad\qquad\qquad\qquad\qquad\qquad \quad  (\kwt{in\_call\_A}) \\
&\qquad \qquad \qquad \qquad \qquad \qquad [\kwf{Return}_{\kwf{A}}(\llbracket \sigma_a(c_x),\sigma_a(f_x),\omega[5],\sigma_v(\kwf{r\_depth})\rrbracket ),\kwf{Var}_{\kwf{A}\circ i \circ 1}(\sigma(\omega)\backslash g(\omega)),\kwf{Gvar}_{\kwf{A}}\llbracket\omega[5] \rrbracket \cdot g(\omega)\backslash e(\omega)), \kwf{Evar}_{\kwf{A}}(e(\omega))]\\& \qquad \qquad \qquad \qquad \qquad \qquad \qquad \qquad \qquad \qquad \qquad \qquad \qquad \qquad \qquad  \qquad -[\kwf{Pred\_eq}(\sigma_v(\kwf{r\_depth}),\omega[4] \oplus 1)] \rightarrow [\kwf{Var}_{\kwf{A}\circ i\circ 1 \circ 1}(\sigma(\omega))]\} \\&\qquad\qquad\qquad\qquad\qquad\qquad\qquad\qquad\qquad\qquad\qquad\qquad\qquad\qquad\qquad\qquad\qquad\qquad\qquad\qquad\qquad\qquad\qquad\qquad\qquad\quad   (\kwt{recv\_ret\_A}) \\
&\kwc{R'}(c_x.\texttt{transfer}(v_1);\kwp{stmt},i,\omega,\kwf{A}) = \kwc{R'}(\kwp{stmt},i\circ 1,\omega,\kwf{A})   \cup \{ [\kwf{Var}_{\kwf{A}\circ i}(\sigma(\omega))] - []\rightarrow  [\kwf{Var}_{\kwf{A}\circ i \circ 1}(\sigma(\omega)|\frac{\sigma_v(c_x)}{\sigma_v(c_x)\oplus \sigma_v(v_1)}|\frac{\sigma_v(c)}{\sigma_v(c)\ominus \sigma_v(v_1)})],\\& \qquad\qquad\qquad\qquad\qquad\qquad\qquad\qquad\qquad\qquad\qquad\qquad\qquad\qquad\qquad\qquad\qquad\qquad\qquad\qquad\qquad\qquad\qquad\qquad\ \  (\kwt{transfer\_succ\_A})  \\
&\ \  [\kwf{Var}_{\kwf{A}\circ i}(\sigma(\omega))] - []\rightarrow [],\ \ \ \ \qquad \qquad \qquad \qquad \qquad \qquad \qquad\qquad \qquad\qquad\qquad\qquad\qquad \qquad\qquad\qquad\qquad\qquad\quad (\kwt{transfer\_fail\_A})  \\
&\kwc{R'}(c_x.\texttt{send}(v_1);\kwp{stmt},i,\omega,\kwf{A}) = \kwc{R'}(\kwp{stmt},i\circ 1,\omega,\kwf{A}) \cup \kwc{R'}(\kwp{stmt},i\circ 2,\omega,\kwf{A})  \cup\\&\ \ \{ [\kwf{Var}_{\kwf{A}\circ i}(\sigma(\omega))] - []\rightarrow  [\kwf{Var}_{\kwf{A}\circ i \circ 1}(\sigma(\omega)|\frac{\sigma_v(c_x)}{\sigma_v(c_x)\oplus \sigma_v(v_1)}|\frac{\sigma_v(c)}{\sigma_v(c)\ominus \sigma_v(v_1)})],\qquad\qquad\qquad\qquad\qquad\qquad\qquad\qquad\qquad\quad\quad (\kwt{send\_succ\_A})  \\
&\ \  [\kwf{Var}_{\kwf{A}\circ i}(\sigma(\omega))] - []\rightarrow [\kwf{Var}_{\kwf{A}\circ i\circ 2}(\sigma(\omega))],\ \ \ \ \qquad  \qquad \qquad \qquad \qquad \qquad \qquad\qquad \qquad\qquad\qquad\qquad\qquad \qquad\qquad\qquad\ \ (\kwt{send\_fail\_A})  \\
\end{align*}
\end{small}
\vspace{-0.35in}
\caption{Complete definition of function $\kwc{R'}$ (Part 1).}
\label{fig:translation_a1_A}
\end{figure*}

\begin{figure*}[!]
\centering
\setlength{\abovedisplayskip}{0pt} 
\setlength{\belowdisplayskip}{0pt}
\begin{small}
\begin{align*}
&\kwc{R'}(c_x.\texttt{call}().\texttt{value}(v_1);\kwp{stmt},i,\omega,\kwf{A}) = \kwc{R'}(\kwp{stmt},i\circ 1,\omega,\kwf{A}) \cup \kwc{R'}(\kwp{stmt},i\circ 2,\omega,\kwf{A}) \cup \kwc{R'}(\kwp{stmt},i\circ 3\circ 1,\omega,\kwf{A}) \cup\\&\ \ \{ [\kwf{Var}_{\kwf{A}\circ i}(\sigma(\omega))] - []\rightarrow  [\kwf{Var}_{\kwf{A}\circ i \circ 1}(\sigma(\omega)|\frac{\sigma_v(c_x)}{\sigma_v(c_x)\oplus \sigma_v(v_1)}|\frac{\sigma_v(c)}{\sigma_v(c)\ominus \sigma_v(v_1)})],\qquad\qquad\qquad\qquad\qquad\qquad\qquad\qquad\qquad\quad (\kwt{ether\_succ\_A})  \\
&\ \  [\kwf{Var}_{\kwf{A}\circ i}(\sigma(\omega))] - []\rightarrow [\kwf{Var}_{\kwf{A}\circ i\circ 2}(\sigma(\omega))],\ \ \ \ \qquad \qquad \qquad \qquad \qquad \qquad \qquad \qquad\qquad \qquad\qquad\qquad\qquad\qquad \qquad\quad\ \ \ \ (\kwt{ether\_fail\_A})  \\
&\ \ [\kwf{Var}_{\kwf{A}\circ i}(\sigma(\omega))] - []\rightarrow [\kwf{Var}_{\kwf{A}\circ i\circ 3}(l(\omega)), \kwf{Fallback}_\kwf{A}(\llbracket\omega[5],\omega[1]\rrbracket), \kwf{Gvar}_{\kwf{A}}\llbracket \omega[5] \rrbracket \cdot g(\omega)\backslash e(\omega)),\kwf{Evar}_{\kwf{A}}(e(\omega)|\frac{\sigma_v(c_x)}{\sigma_v(c_x)\oplus \sigma_v(v_1)}|\frac{\sigma_v(c)}{\sigma_v(c)\ominus \sigma_v(v_1)})], \\& \qquad\qquad\qquad\qquad\qquad\qquad\qquad\qquad\qquad\qquad\qquad\qquad\qquad\qquad\qquad\qquad\qquad\qquad\qquad\qquad\qquad\qquad\qquad\qquad\qquad\quad\ \ (\kwt{fb\_call\_A}) \\
&\ \ [\kwf{ReturnFallback}_\kwf{A}(\llbracket\omega[5],\omega[1]\rrbracket ), \kwf{Gvar}_{\kwf{A}}\llbracket \omega[5] \rrbracket \cdot g(\omega)\backslash e(\omega)),\kwf{Evar}_{\kwf{A}}(e(\omega)), \kwf{Var}_{\kwf{A}\circ i\circ 3}(l(\omega))] -[]\rightarrow[\kwf{Var}_{\kwf{A}\circ i\circ 3\circ 1}(\sigma(\omega))]\}  \quad(\kwt{recv\_fb\_ret\_A})\\
&\kwc{R'}(v_1\leftarrow v_2;\kwp{stmt},i,\omega,\kwf{B}) = \kwc{R'}(\kwp{stmt},i\circ 1,\omega, \kwf{B}) \cup \{ [\kwf{Var}_{\kwf{B} \circ i}(\sigma(\omega))]-[]\rightarrow [\kwf{Var}_{\kwf{B}\circ i\circ 1}(\sigma(\omega)|\frac{\sigma_v(v_1)}{\sigma_v(v_2)})]\}\quad\qquad\qquad\quad\qquad\ (\kwt{var\_assign\_B})  \\
&\kwc{R'}(\tau\ v_1\leftarrow v_2;\kwp{stmt},i,\omega,\kwf{B}) = \kwc{R'}(\kwp{stmt},i\circ 1,\omega\cdot \llbracket\left\langle \sigma_v(v_1),\kwf{T_v},\kwf{R_l},\kwf{E_n} \right\rangle\rrbracket,\kwf{B} ) \cup \{[\kwf{Var}_{\kwf{B} \circ i}(\sigma(\omega))]-[]\rightarrow [\kwf{Var}_{\kwf{B}\circ i\circ 1}(\sigma(\omega)\cdot \llbracket\sigma_v(v_2)\rrbracket)]\}\\& \qquad\qquad\qquad\qquad\qquad\qquad\qquad\qquad\qquad\qquad\qquad\qquad\qquad\qquad\qquad\qquad\qquad\qquad\qquad\qquad\qquad\qquad\qquad\qquad\quad  (\kwt{var\_declare\_B})\\
&\kwc{R'}(\texttt{if}\ e_{b}\ \texttt{then}\ \kwp{stmt}_1\ \texttt{else}\ \kwp{stmt}_2;\kwp{stmt}_3,i,\omega,\kwf{B}) = \kwc{R'}(\kwp{stmt}_1;\kwp{stmt}_3,i\circ 1,\omega,\kwf{B}) \cup \kwc{R'}(\kwp{stmt}_2;\kwp{stmt}_3,i\circ 2,\omega,\kwf{B}) \cup \\&\qquad\qquad\qquad\qquad\qquad\qquad\qquad\qquad\qquad\quad \{ [\kwf{Var}_{\kwf{B}\circ i}(\sigma(\omega))] -[\theta_e(e_b)]\rightarrow [\kwf{Var}_{\kwf{B}\circ i \circ 1}(\sigma(\omega))], \qquad\qquad\qquad\qquad \qquad \qquad   \ \  (\kwt{if\_true\_B})\\
& \qquad\qquad\qquad\qquad\qquad\qquad\qquad\qquad\qquad\quad\ \ [\kwf{Var}_{\kwf{B}\circ i}(\sigma(\omega))] -[\theta_{ne}(e_b)]\rightarrow [\kwf{Var}_{\kwf{B}\circ i\circ 2}(\sigma(\omega))]\} \qquad\qquad\qquad\qquad \qquad   \quad \ \ (\kwt{if\_false\_B})\\
&\kwc{R'}(\texttt{require}\ e_{b}\;\kwp{stmt}_1,i,\omega,\kwf{B}) = \kwc{R'}(\kwp{stmt}_1,i\circ 1,\omega,\kwf{B}) \cup \{ [\kwf{Var}_{\kwf{B}\circ i}(\sigma(\omega))] -[\theta_e(e_b)]\rightarrow [\kwf{Var}_{\kwf{B}\circ i \circ 1}(\sigma(\omega))], \qquad\qquad\quad  \quad \  (\kwt{require\_true\_B})\\
& \qquad\qquad\qquad\qquad\qquad\qquad\qquad\qquad\qquad\quad\ \ [\kwf{Var}_{\kwf{B}\circ i}(\sigma(\omega))] -[\theta_{ne}(e_b)]\rightarrow []\} \qquad\qquad\qquad \qquad \qquad \qquad \qquad \quad \ \ (\kwt{require\_false\_B})\\
&\kwc{R'}(\texttt{return},i,\omega,\kwf{B}) = \{ [\kwf{Var}_{\kwf{B}\circ i}(\sigma(\omega))]  -[\kwf{Pred\_eq}(\omega[3],\kwf{EXT})]  \rightarrow  [\kwf{Gvar}_{\kwf{B}}\llbracket \omega[5] \rrbracket \cdot g(\omega)\backslash e(\omega)),\kwf{Evar}_{\kwf{B}}(e(\omega))],  \qquad\qquad\qquad\quad \ (\kwt{ret\_ext\_B}) \\
&\qquad \qquad \qquad \qquad \quad \  [\kwf{Var}_{\kwf{B}\circ i}(\sigma(\omega))]  -[\kwf{Pred\_eq}(\omega[3],\kwf{IN})]  \rightarrow  [\kwf{Return}_{\kwf{B}}(\llbracket \omega[5],\omega[1],\omega[2],\omega[4]\rrbracket,\kwf{Gvar}_{\kwf{B}}\llbracket \omega[5] \rrbracket \cdot g(\omega)\backslash e(\omega)),\kwf{Evar}_{\kwf{B}}(e(\omega)) )]\} \\&\qquad\qquad\qquad\qquad\qquad\qquad\qquad\qquad\qquad\qquad\qquad\qquad\qquad\qquad\qquad\qquad\qquad\qquad\qquad\qquad\qquad\qquad\qquad\qquad\qquad \qquad (\kwt{ret\_in\_B}) \\
&\kwc{R'}(x(c_x).f_x(p);\kwp{stmt},i,\omega,\kwf{B}) = \kwc{R'}(\kwp{stmt},i\circ 1 \circ 1,\omega,\kwf{B}) \cup \{ [\kwf{Var}_{\kwf{B}\circ i}(\sigma(\omega))] -[] \rightarrow  [\kwf{Call}_{\kwf{Bin}}( \llbracket\sigma_a(c_x),\sigma_a(f_x),\omega[5],\omega[4] \oplus 1\rrbracket\cdot \sigma_s(p)),\\&\qquad\qquad\qquad\qquad\qquad\qquad\qquad\qquad\qquad\qquad\qquad\qquad\qquad\qquad \kwf{Var}_{\kwf{B}\circ i \circ 1}(\sigma(\omega)\backslash g(\omega)), \kwf{Gvar}_{\kwf{B}}\llbracket \omega[5] \rrbracket \cdot g(\omega)\backslash e(\omega)),\kwf{Evar}_{\kwf{B}}(e(\omega))],\\& \qquad\qquad\qquad\qquad\qquad\qquad\qquad\qquad\qquad\qquad\qquad\qquad\qquad\qquad\qquad\qquad\qquad\qquad\qquad\qquad\qquad\qquad\qquad\qquad\qquad \quad  (\kwt{in\_call\_B}) \\
&\qquad \qquad \qquad \qquad \qquad \qquad [\kwf{Return}_{\kwf{B}}(\llbracket \sigma_a(c_x),\sigma_a(f_x),\omega[5],\sigma_v(\kwf{r\_depth})\rrbracket ),\kwf{Var}_{\kwf{B}\circ i \circ 1}(\sigma(\omega)\backslash g(\omega)),\kwf{Gvar}_{\kwf{B}}\llbracket\omega[5] \rrbracket \cdot g(\omega)\backslash e(\omega)), \kwf{Evar}_{\kwf{B}}(e(\omega))]\\& \qquad \qquad \qquad \qquad \qquad \qquad \qquad \qquad \qquad \qquad \qquad \qquad \qquad \qquad \qquad  \qquad -[\kwf{Pred\_eq}(\sigma_v(\kwf{r\_depth}),\omega[4] \oplus 1)] \rightarrow [\kwf{Var}_{\kwf{B}\circ i\circ 1 \circ 1}(\sigma(\omega))]\} \\&\qquad\qquad\qquad\qquad\qquad\qquad\qquad\qquad\qquad\qquad\qquad\qquad\qquad\qquad\qquad\qquad\qquad\qquad\qquad\qquad\qquad\qquad\qquad\qquad\qquad\quad   (\kwt{recv\_ret\_B}) \\
&\kwc{R'}(c_x.\texttt{transfer}(v_1);\kwp{stmt},i,\omega,\kwf{B}) = \kwc{R'}(\kwp{stmt},i\circ 1,\omega,\kwf{B})   \cup \{ [\kwf{Var}_{\kwf{B}\circ i}(\sigma(\omega))] - []\rightarrow  [\kwf{Var}_{\kwf{B}\circ i \circ 1}(\sigma(\omega)|\frac{\sigma_v(c_x)}{\sigma_v(c_x)\oplus \sigma_v(v_1)}|\frac{\sigma_v(c)}{\sigma_v(c)\ominus \sigma_v(v_1)})],\\& \qquad\qquad\qquad\qquad\qquad\qquad\qquad\qquad\qquad\qquad\qquad\qquad\qquad\qquad\qquad\qquad\qquad\qquad\qquad\qquad\qquad\qquad\qquad\qquad\ \  (\kwt{transfer\_succ\_B})  \\
&\ \  [\kwf{Var}_{\kwf{B}\circ i}(\sigma(\omega))] - []\rightarrow [],\ \ \ \ \qquad \qquad \qquad \qquad \qquad \qquad \qquad\qquad \qquad\qquad\qquad\qquad\qquad \qquad\qquad\qquad\qquad\qquad\quad (\kwt{transfer\_fail\_B})  \\
&\kwc{R'}(c_x.\texttt{send}(v_1);\kwp{stmt},i,\omega,\kwf{B}) = \kwc{R'}(\kwp{stmt},i\circ 1,\omega,\kwf{B}) \cup \kwc{R'}(\kwp{stmt},i\circ 2,\omega,\kwf{B})  \cup\\&\ \ \{ [\kwf{Var}_{\kwf{B}\circ i}(\sigma(\omega))] - []\rightarrow  [\kwf{Var}_{\kwf{B}\circ i \circ 1}(\sigma(\omega)|\frac{\sigma_v(c_x)}{\sigma_v(c_x)\oplus \sigma_v(v_1)}|\frac{\sigma_v(c)}{\sigma_v(c)\ominus \sigma_v(v_1)})],\qquad\qquad\qquad\qquad\qquad\qquad\qquad\qquad\qquad\quad\quad (\kwt{send\_succ\_B})  \\
&\ \  [\kwf{Var}_{\kwf{B}\circ i}(\sigma(\omega))] - []\rightarrow [\kwf{Var}_{\kwf{B}\circ i\circ 2}(\sigma(\omega))],\ \ \ \ \qquad  \qquad \qquad \qquad \qquad \qquad \qquad\qquad \qquad\qquad\qquad\qquad\qquad \qquad\qquad\qquad\ \ (\kwt{send\_fail\_B})  \\
&\kwc{R'}(c_x.\texttt{call}().\texttt{value}(v_1);\kwp{stmt},i,\omega,\kwf{B}) = \kwc{R'}(\kwp{stmt},i\circ 1,\omega,\kwf{B}) \cup \kwc{R'}(\kwp{stmt},i\circ 2,\omega,\kwf{B}) \cup \kwc{R'}(\kwp{stmt},i\circ 3\circ 1,\omega,\kwf{B}) \cup\\&\ \ \{ [\kwf{Var}_{\kwf{B}\circ i}(\sigma(\omega))] - []\rightarrow  [\kwf{Var}_{\kwf{B}\circ i \circ 1}(\sigma(\omega)|\frac{\sigma_v(c_x)}{\sigma_v(c_x)\oplus \sigma_v(v_1)}|\frac{\sigma_v(c)}{\sigma_v(c)\ominus \sigma_v(v_1)})],\qquad\qquad\qquad\qquad\qquad\qquad\qquad\qquad\qquad\quad (\kwt{ether\_succ\_B})  \\
&\ \  [\kwf{Var}_{\kwf{B}\circ i}(\sigma(\omega))] - []\rightarrow [\kwf{Var}_{\kwf{B}\circ i\circ 2}(\sigma(\omega))],\ \ \ \ \qquad \qquad \qquad \qquad \qquad \qquad \qquad \qquad\qquad \qquad\qquad\qquad\qquad\qquad \qquad\quad\ \ \ \ (\kwt{ether\_fail\_B})  \\
&\ \ [\kwf{Var}_{\kwf{B}\circ i}(\sigma(\omega))] - []\rightarrow [\kwf{Var}_{\kwf{B}\circ i\circ 3}(l(\omega)), \kwf{Fallback}_\kwf{B}(\llbracket\omega[5],\omega[1]\rrbracket), \kwf{Gvar}_{\kwf{B}}\llbracket \omega[5] \rrbracket \cdot g(\omega)\backslash e(\omega)),\kwf{Evar}_{\kwf{B}}(e(\omega)|\frac{\sigma_v(c_x)}{\sigma_v(c_x)\oplus \sigma_v(v_1)}|\frac{\sigma_v(c)}{\sigma_v(c)\ominus \sigma_v(v_1)})], \\& \qquad\qquad\qquad\qquad\qquad\qquad\qquad\qquad\qquad\qquad\qquad\qquad\qquad\qquad\qquad\qquad\qquad\qquad\qquad\qquad\qquad\qquad\qquad\qquad\qquad\quad (\kwt{fb\_call\_B}) \\
&\ \ [\kwf{ReturnFallback}_\kwf{B}(\llbracket\omega[5],\omega[1]\rrbracket ), \kwf{Gvar}_{\kwf{B}}\llbracket \omega[5] \rrbracket \cdot g(\omega)\backslash e(\omega)),\kwf{Evar}_{\kwf{B}}(e(\omega)), \kwf{Var}_{\kwf{B}\circ i\circ 3}(l(\omega))] -[]\rightarrow[\kwf{Var}_{\kwf{B}\circ i\circ 3\circ 1}(\sigma(\omega))]\}  \quad\ \ (\kwt{recv\_fb\_ret\_B})\\
\end{align*}
\end{small}
\vspace{-0.35in}
\caption{Complete definition of function $\kwc{R'}$ (Part 2).}
\label{fig:translation_a1_B}
\end{figure*}

\textbf{1)} Conditional statements. 
Two rules $\kwt{if\_true}$ and $\kwt{if\_false}$ are generated for statement $\texttt{if}\ e_{b}\ \texttt{then}\ \kwp{stmt}_1\ \texttt{else}\ \kwp{stmt}_2$.
Here, $\theta_e(e_b)$ is a conditional fact representing that the value of $e_b$ is true.
Similarly, $\theta_{ne}(e_b)$ denotes that the value of $e_b$ is false.
$\kwf{Var}_{i\circ 1}$, $\kwf{Var}_{i\circ 2}$ correspond to the states when $\kwp{stmt}_1$, $\kwp{stmt}_2$ start to be executed, respectively.
The definition of the function $\theta_e(e)$ is as follows:
$$\theta_e(e) = \left\{
\begin{array}{lll}
\sigma_a(e)  & {\rm if}\ e\ {\rm is\ a\ constant} \\
\sigma_v(e)      & {\rm if}\ e\ {\rm is\ a\ variable} \\
\theta_o(e_1,e_2,\lozenge)  & {\rm if}\ e\ {\rm is}\ e_1\ \lozenge \ e_2\\
\kwf{EqNum}(\theta_e(e_1),\theta_e(e_2))  & {\rm if}\ e\ {\rm is}\ e_1 = e_2 \\
\kwf{NeNum}(\theta_e(e_1),\theta_e(e_2))  & {\rm if}\ e\ {\rm is}\ e_1 \neq e_2 \\
\kwf{LessNum}(\theta_e(e_1),\theta_e(e_2))  & {\rm if}\ e\ {\rm is}\ e_1 < e_2
\end{array}
\right.
$$
Here, numerical facts $\kwf{EqNum}, \kwf{NeqNum}, \kwf{LessNum}$ denote the relationships between numeric variables.
$\lozenge \in \{+,-,*,/,\%, **\}$ represents an operator in expressions.
$\theta_o(e_1,e_2,\lozenge)$ represents the term translated from $e_1 \lozenge e_2$.
As tamarin prover does not support numerical operators like $\lozenge$ and to avoid conflicts with existing operators in tamarin prover, 
we convert $\lozenge$ into special forms and modify the source code of tamarin prover to parse them, which are eventually passed into Z3 for processing.
For example, $+$ is translated in to $\oplus$ and $-$ is translated into $\ominus$.

\textbf{2)} Internal call statements. Besides external accounts, contract account $c$ can also invoke the function $f_x$ of another account $c_x$ by executing the statement $x(c_x).f_x(p)$, where $p$ is a sequence of parameters and $x$ is name of the contract of account $c_x$.
In this case, the execution of this statement can be divided into the following steps:
\textbf{a)} An internal transaction is sent by $c$ to invoke $f_x$. To denote this step, rule $\kwt{in\_call}$ is generated. 
Because the ether balances may be modified during executions of $f_x$, while local variables will not change, terms denoting local variables in rule $\kwt{in\_call}$ are maintained in $\kwf{Var}_{i\circ 1}$ fact, while terms representing the ether balances of all accounts are put into $\kwf{Evar}$ and terms denoting global variables of account $c$ are put into $\kwf{Gvar}$.
\textbf{b)} The codes in $f_x$ are executed, which have been recursively modeled as shown in Fig. \ref{fig:translation_a1}.
\textbf{c)} $f_x$ returns, which has been modeled in return statement category.
\textbf{d)} The next statement of $f_x$ is prepared to be executed. Rule $\kwt{recv\_ret}$ is generated for this step, indicating that $c$ receives the return message from $c_x$ and $f$ is ready to continue executing.

\textbf{3)} Ether transfer statements. 
The misuse of ether transfer statements using $\texttt{call}$ is one of the reasons that cause attacks.
Consider a statement $c_x.\texttt{call}().\texttt{value}(v_1)$, which means that the account $c$ who invokes the $\texttt{call}()$ is to transfer ether $v_1$ to the account $c_x$, where $v_1$ is assumed as a local variable.
There are three cases for the execution of the statement:
\textbf{a)} the transfer succeeds. The ether balance of  $c$ is reduced by $v_1$ and the ether balance of $c_x$ is increased by $v_1$.
\textbf{b)} the transfer fails and the ether balances of $c$ and $c_x$ are not modified.
\textbf{c)} the transfer succeeds with ether balances changed in the same way as case \textbf{a)}, but the fallback function is called probably in an unexpected way.
The rules $\kwt{ether\_succ}$, $\kwt{ether\_fail}$, $\kwt{fb\_call}$ are generated for the 3 cases respectively. 
In these rules, we use $\sigma_v(c)$ and $\sigma_v(c_x)$ to denote the ether balances of $c$ and $c_x$.
In rule $\kwt{fb\_call}$, fact $\kwf{Fallback}$ indicates that the fallback function is called.
According to assumption C2 in the adversary model, 
the global variables of $c$ and ether balances of all accounts may be modified due to execution of the fallback function,
thereby the terms denoting these variables are put into $\kwf{Gvar}$ and $\kwf{Evar}$ facts, while terms representing local variables are maintained in $\kwf{Var}_{i \circ 3}$.
Here $l(\omega)$ outputs a term sequence by obtaining the $name$ of all tuples in $\omega$ whose $range=\kwf{R_l}$.
In rule $\kwt{recv\_fb\_ret}$, $\kwf{ReturnFallback}$ implies that the fallback function finishes executing and a return message is sent to $c$.
The terms in $\kwf{Var}_{i \circ 3}$, $\kwf{Gvar}$ and $\kwf{Evar}$ are merged back into terms in $\kwf{Var}_{i \circ 3 \circ 1}$, which indicates that the function $f$ continues executing.


\textbf{Complete form of rules mentioned in Section \ref{subsec:translation}}



As mentioned in Section \ref{subsec:translation}, multiple rules are generated for each function $f$ in the contract of account $c$ to model the adversary model C2:
\begin{multline}
[\kwf{Fallback}(\llbracket\sigma_a(c),\sigma_a(f)\rrbracket)]-[]\rightarrow\\ [\kwf{Call_{in}}(\llbracket \sigma_a(c),\sigma_a(f'),\sigma_a(c_{adv})\rrbracket \cdot \sigma(seq(d')) )]\  (\kwt{fb\_in\_call})\nonumber
\vspace{-0.1in}
\end{multline}
\begin{multline}
[\kwf{Return}(\llbracket\sigma_a(c),\sigma_a(f'),\sigma_a(c_{adv})\rrbracket )]-[]\rightarrow\\ \qquad\qquad[\kwf{ReturnFallback}(\llbracket\sigma_a(c),\sigma_a(f)\rrbracket)]\qquad(\kwt{ret\_fb})\nonumber
\end{multline}
Here, $c_{adv}$ represents the address of the contract account owned by the adversary. 
$f'$ denotes an arbitrary function in the contract of account $c$, and $d'$ denotes the parameters of $f'$.
After the fallback function is triggered, an internal transaction is sent which invokes function $f'$.  
Therefore, the rule $\kwt{fb\_in\_call}$ indicates that the fallback function of account $c_{adv}$ is triggered by function $f$ in the contract of account $c$.
The rule $\kwt{ret\_fb}$ indicates that the adversary gets a return message after the execution of $f'$ and sends a message denoting that the fallback function finishes executing.

\subsection{Complete definition of rules mentioned in Section \ref{subsec:complementary} and function $\kwc{R}'$}
\label{subsec:appendix2}

\textbf{Definition of function $\kwc{R}'$}

We define a function $\kwc{R}'$ shown in Fig. \ref{fig:translation_a1_A} and \ref{fig:translation_a1_B} to generate rules of the complementary models for the invariant property and the equivalence property.
$\kwc{R}'$ is similar to $\kwc{R}$ and the difference is that $\kwc{R}'$ takes an additional argument compared to $\kwc{R}$.
This argument is a string denoting the subscript for facts, and the value of this argument is $\kwf{A}$ or $\kwf{B}$.
Given a rule $r$ translated from a sequence of statements by using $\kwc{R}$, we define a function $f_R(r,s)$ shown in Table \ref{table:rr_correspondence} to output the rule generated by using $\kwc{R}'$ from the same sequence.
Here $s$ denotes the subscript $\kwf{A}$ or $\kwf{B}$.
This function will be used in our subsequent proofs.

\begin{table}[]
\centering
    \caption{The correspondence between rules generated by $\mathcal{R}$ and $\mathcal{R'}$}
    \label{table:rr_correspondence}
\resizebox{\linewidth}{!}{
\begin{tabular}{|c|c|c|}
\hline
r         & $f_R(r,\kwf{A})$       & $f_R(r,\kwf{B})$       \\ \hline
 $\kwt{var\_assign}$ &  $\kwt{var\_assign\_A}$ & $\kwt{var\_assign\_B}$  \\ \hline
$\kwt{var\_declare}$ & $\kwt{var\_declare\_A}$ & $\kwt{var\_declare\_B}$ \\ \hline
$\kwt{if\_true}$ & $\kwt{if\_true\_A}$ & $\kwt{if\_true\_B}$ \\ \hline
$\kwt{if\_false}$ & $\kwt{if\_false\_A}$ & $\kwt{if\_false\_B}$ \\ \hline
 $\kwt{require\_true}$ &  $\kwt{require\_true\_A}$ &  $\kwt{require\_true\_B}$ \\ \hline
$\kwt{require\_false}$ & $\kwt{require\_false\_A}$ & $\kwt{require\_false\_B}$ \\ \hline
$\kwt{ret\_ext}$ & $\kwt{ret\_ext\_A}$ & $\kwt{ret\_ext\_B}$ \\ \hline
$\kwt{ret\_in}$ & $\kwt{ret\_in\_A}$ & $\kwt{ret\_in\_B}$ \\ \hline
 $\kwt{in\_call}$ &  $\kwt{in\_call\_A}$ &  $\kwt{in\_call\_B}$ \\ \hline
 $\kwt{recv\_ret}$ &  $\kwt{recv\_ret\_A}$ &  $\kwt{recv\_ret\_B}$ \\ \hline
$\kwt{transfer\_succ}$ & $\kwt{transfer\_succ\_A}$ & $\kwt{transfer\_succ\_B}$   \\ \hline
 $\kwt{transfer\_fail}$ & $\kwt{transfer\_fail\_A}$ & $\kwt{transfer\_fail\_B}$  \\ \hline 
 $\kwt{send\_succ}$ & $\kwt{send\_succ\_A}$ & $\kwt{send\_succ\_B}$   \\ \hline
$\kwt{send\_fail}$ & $\kwt{send\_fail\_A}$ & $\kwt{send\_fail\_B}$  \\ \hline
 $\kwt{ether\_succ}$ & $\kwt{ether\_succ\_A}$ & $\kwt{ether\_succ\_B}$  \\ \hline
 $\kwt{ether\_fail}$ & $\kwt{ether\_fail\_A}$  & $\kwt{ether\_fail\_B}$  \\ \hline
 $\kwt{fb\_call}$ & $\kwt{fb\_call\_A}$ & $\kwt{fb\_call\_B}$ \\ \hline
$\kwt{recv\_fb\_ret}$ & $\kwt{recv\_fb\_ret\_A}$ & $\kwt{recv\_fb\_ret\_B}$ \\ \hline
\end{tabular}}
\end{table}

\textbf{Complete definition of rules mentioned in Section \ref{subsec:complementary}}

\textbf{1) rule $\kwt{init\_gvars\_inv}$}
\setlength\multlinegap{0pt}
\begin{multline}
[\kwf{FR}(g(\omega_0)\backslash e(\omega_0) )]-[\kwf{Init_G}(\omega_0[1]),\theta_e(\phi)]\rightarrow\qquad\qquad\qquad\qquad\\ [\kwf{Gvar}(\llbracket \omega_0[1] \rrbracket \cdot g(\omega_0)\backslash e(\omega_0))]\qquad (\kwt{init\_gvars\_inv})\nonumber
\end{multline}
Here, $\phi$ can be any invariant in Section \ref{subsec:generating properties}.
The fact $\theta_e(\phi)$ denotes that the invariant $\phi$ holds after the initialization. 

\textbf{2) rule $\kwt{ext\_call\_inv}$}
\begin{multline}
[\kwf{Fr}(\sigma_v(c_b)),\kwf{FR}(\sigma(seq(\texttt{d})))]-[\kwf{Start}()]\rightarrow \\ [\kwf{Call_e}( \llbracket\omega_0[1],\sigma_a(f),\sigma_v(c_b)\rrbracket\cdot \sigma(seq(\texttt{d})))]\quad (\kwt{ext\_call\_inv}) \nonumber
\end{multline}
Different from rule $\kwt{ext\_call}$, an action $\kwf{Start}()$ is added into rule $\kwt{ext\_call\_inv}$.
The action is associated with the restriction:
$$All\ \#i\ \#j. \kwf{Start}()@i\ \&\ \kwf{Start}()@j\ =>\ i=j\quad (\kwt{start\_inv})$$
This restriction requires that $\kwf{Start}()$ occurs only once in an execution of the model. 

\textbf{3) rule $\kwt{ret\_ext\_inv}$}
\begin{multline}
[\kwf{Var}_i(\sigma(\omega))]  -[\kwf{Pred\_eq}(\omega[3],\kwf{EXT}),\theta_{ne}(\phi),\kwf{End}()]  \rightarrow\\  [\kwf{Gvar}(\llbracket \omega[4] \rrbracket \cdot g(\omega)\backslash e(\omega)),\kwf{Evar}(e(\omega))]  \quad (\kwt{ret\_ext\_inv}) \nonumber
\end{multline}

\textbf{4) rule $\kwt{init\_evars\_AB}$ and $\kwt{init\_gvars\_AB}$}
\setlength\multlinegap{0pt}
\begin{multline}
[\kwf{FR}(e(\omega_0))]-[\kwf{Init_E}()]\rightarrow [\kwf{Evar_A}(e(\omega_0)),\kwf{Evar_B}(e(\omega_0))]\\ \qquad\qquad\qquad\qquad\qquad\qquad\qquad\qquad\ \ (\kwt{init\_evars\_AB})\nonumber\\
[\kwf{FR}(g(\omega_0)\backslash e(\omega_0) )]-[\kwf{Init_G}(\omega_0[1])]\rightarrow [\kwf{Gvar_A}(\llbracket \omega_0[1] \rrbracket \cdot g(\omega_0)\\ \backslash e(\omega_0)) ,\kwf{Gvar_B}(\llbracket \omega_0[1] \rrbracket \cdot g(\omega_0)\backslash e(\omega_0))]\ \ (\kwt{init\_gvars\_AB})\nonumber
\end{multline}

\textbf{5) rule $\kwt{recv\_ext\_A}$ and $\kwt{recv\_ext\_B}$}

$\kwt{recv\_ext\_A}$ and $\kwt{recv\_ext\_B}$ are associated with the following restrictions:
\setlength\multlinegap{0pt}
\begin{multline}
All\ \#i\ c_b. \kwf{Exc_A}(\sigma_v(c_b)),\sigma_a(f))@i=>\\ \qquad\qquad\qquad\ \ \ Ex\ \#j. \kwf{Exc_B}(\sigma_v(c_b)),\sigma_a(f))@j\qquad (\kwt{exc\_A}) \ \ \\
All\ \#i\ c_b. \kwf{Exc_B}(\sigma_v(c_b)),\sigma_a(f))@i=>\qquad\qquad\quad\\ \qquad\qquad\qquad\ \ \  Ex\ \#j. \kwf{Exc_A}(\sigma_v(c_b)),\sigma_a(f))@j\qquad (\kwt{exc\_B})\ \ \nonumber
\end{multline}

\textbf{6) rule $\kwt{compare\_AB}$}

The following rule is used to compare the ether balances and token balances of the adversary:
\begin{multline}
[\kwf{Gvar_A}(\llbracket \omega[4] \rrbracket \cdot g(\omega)\backslash e(\omega)),\kwf{Evar_A}(e(\omega)),\kwf{Gvar_B}(\llbracket \omega[4] \rrbracket \cdot g(\omega)\\ \backslash e(\omega)), \kwf{Evar_B}(e(\omega))]- [\theta_{ne}(\phi_{equ}),\kwf{End}()]\rightarrow[\kwf{Gvar_A}(\llbracket \omega[4]\\ \rrbracket \cdot g(\omega)\backslash e(\omega)),\kwf{Evar_A}(e(\omega)),\kwf{Gvar_B}(\llbracket \omega[4] \rrbracket \cdot g(\omega)\backslash e(\omega)),\\  \kwf{Evar_B}(e(\omega))]\ (\kwt{compare\_AB})\nonumber
\end{multline}

\textbf{7) rule $\kwt{ext\_call\_bvar\_AB}$}

Specially, if a statement in function $f$ uses $\texttt{block.timestamp}$, rule $\kwt{ext\_call\_bvar\_AB}$ is generated instead of $\kwt{ext\_call\_AB}$.
$\kwf{Bvar_A}$ $(bt_A)$ and $\kwf{Bvar_B}(bt_B)$ facts are added in the conclusion of $\kwt{ext\_call\_}$ $\kwt{bvar\_AB}$, which indicates that the adversary can modify the timestamp of blocks containing transactions in $T_A$ and $T_B$, corresponding to C3 mentioned in Section \ref{subsec:adversary}.
\begin{multline}
[\kwf{Fr}(\sigma_v(c_b)),\kwf{FR}(\sigma(seq(\texttt{d}))),\kwf{Fr}(bt_A),\kwf{Fr}(bt_B)]-[]\rightarrow\\  [\kwf{Bvar_A}(bt_A),\kwf{Bvar_B}(bt_B),\kwf{Call_{Ae}}( \llbracket\omega_0[1],\sigma_a(f),\sigma_v(c_b)\rrbracket\\ \cdot \sigma(seq(\texttt{d}))), \kwf{Call_{Be}}( \llbracket\omega_0[1],\sigma_a(f),\sigma_v(c_b)\rrbracket\cdot \sigma(seq(\texttt{d})))]\\ \kwt{(ext\_call\_bvar\_AB)} \nonumber
\end{multline}
In addition, to process $\kwf{Bvar}$ facts,  \ourtool generates numerical constraints related to timestamps of blocks, which cannot be modeled in multiset rewriting system.  
Specifically, according to adversary model C3, 
if a contract uses $\texttt{block.timestamp}$, when verifying the equivalence property, \ourtool processes $\kwf{Bvar_A}(bt_A)$ and $\kwf{Bvar_B}(bt_B)$ in the execution with the following additional steps:
i) assign indices to each $bt_A$ and $bt_B$ in the order they appear in the execution, starting with 0;
ii) generate numerical constraint $bt_{Ai} > bt_{A(i-1)}$ and $bt_{Bi} > bt_{B(i-1)}$ for every index i except for 0.
ii) generate numerical constraint $bt_{Bi} <= bt_{Ai}+15$ for $bt_{Ai}$ and $bt_{Bi}$ with same index $i$. 
The constraints in step ii) restrict that the timestamp of blocks can only be increased and not decreased.
The constraints in step iii) indicate that the adversary can only increase the timestamp by a maximum of 15 seconds.

\subsection{An Example of vulnerable contract with \textit{TD}.}
\label{subsec:appendix_x}

\begin{figure}[ht]
\centering
\includegraphics[scale=0.26]{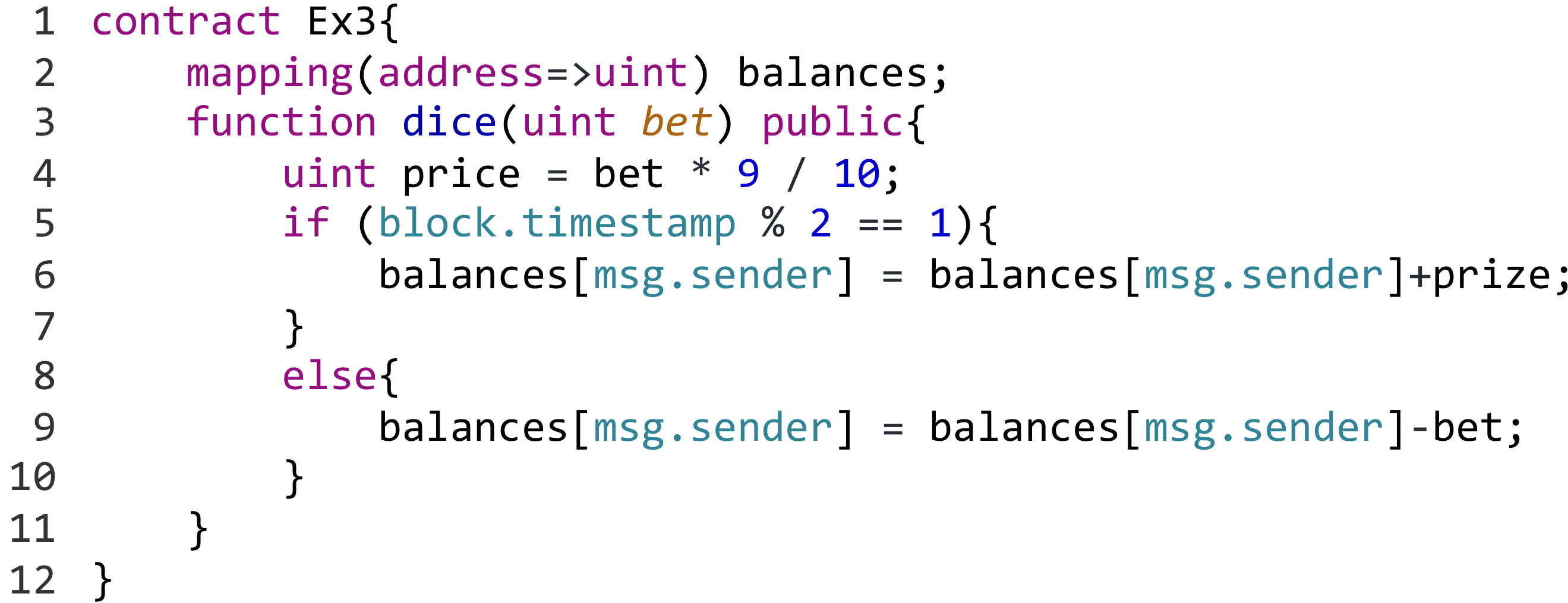}
\caption{Example Ex3: a vulnerable smart contract.}
\label{fig:example3}
\vspace{-0.06in}
\end{figure}

\textbf{Example.}
Fig. \ref{fig:example3} shows a simplified version of a practical contract that violates the equivalence property.
The function $\texttt{dice}$ is used to play a game.
If $\texttt{block.timestamp}$ is odd, $\texttt{msg.sender}$ will get a prize, \textit{i.e.}, his token balances will be increased.
The contract is recognized as a token contract since it defines a variable $\texttt{balances}$ of $\texttt{mapping(address=>uint)}$ type.
An invariant property and an equivalence property are both generated.
In verification of equivalence property, \ourtool can find an execution that simulates two sequences $T_A$ and $T_B$ consisting of a same transaction which invokes $\texttt{dice}$.
Correspondingly, since this two sequences only use $\texttt{block.timestamp}$ once, the terms $bt_A$ and $bt_B$ are transformed into $bt_{A0}$ and $bt_{B0}$ and a numerical constraint $bt_{B0} <= bt_{A0} + 15$ is generated as mentioned in Section \ref{subsec:verification}. 
The token balances after $T_A$ and $T_B$ are obviously different as $bt_A$ and $bt_B$ could be in different parity, which violates the equivalence property.
Therefore, Ex2 is regarded vulnerable by \ourtool.
The vulnerability in Ex2 is a kind of \textit{TD}. Although this is a known type of vulnerabilities, \ourtool still finds 3 contracts with \textit{TD} that are deployed on Ethereum.
\subsection{Soundness of the translations}
\label{sec:proof}    
\begin{figure*}
    \centering
    \input{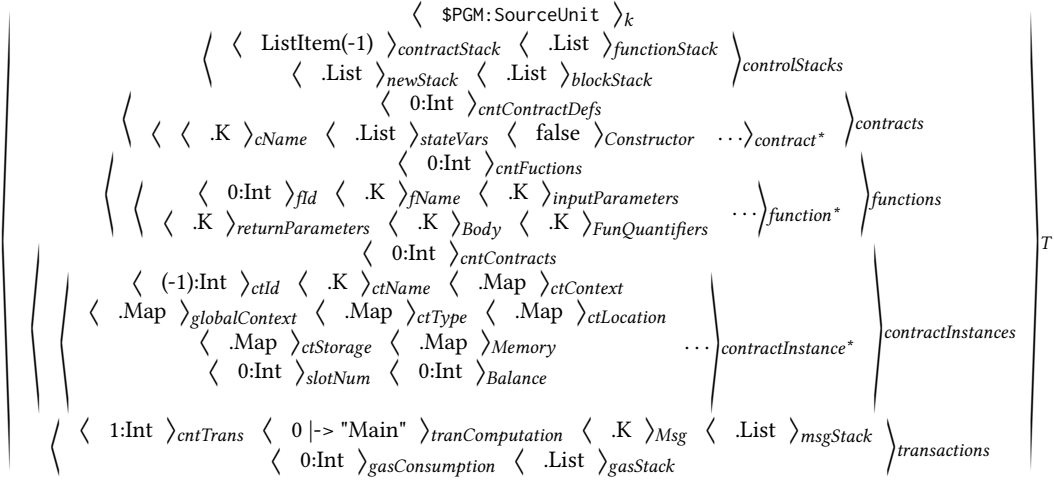}
    \centering
    \caption{\label{fig:Configuration}An example of configuration in KSolidity semantics}
\end{figure*}


 To prove the soundness of the translation from Solidity language to our models, we firstly introduce the custom semantics of Solidity~\cite{DBLP:conf/sp/JiaoK0S0020}, namely KSolidity.
We also briefly explain the relations between the modeling in FASVERIF and KSolidity, including the configurations and rules for the configurations.
Then, we explain the notations and adopted theories preparing for the proofs.
Finally, we prove the soundness of \ourtool.

\subsubsection{\textbf{A glance at K-framework and KSolidity \cite{DBLP:conf/sp/JiaoK0S0020}}}
\label{sec:Aglance}
KSolidity is defined using K-framework \cite{Rosu2010}, a rewrite-based executable semantic framework. 
The definition of a language consists of 3 parts: language syntax, the runtime configuration, and a set of rules constructed based on the syntax and the configuration.

Configurations organize the state in units called cells, which are labeled and can be nested.
The cells' contents can be various semantic data, such as trees, lists, maps, etc. 
As shown in Figure \ref{fig:Configuration}, the configuration of KSolidity is composed of six main cells: (1) \textit{k}, the rest of programs to be executed, (2) \textit{controlStacks}, a collection of runtime stacks, (3) \textit{contracts}, a set of contract definitions, (4) \textit{functions}, a set of function definitions (5) \textit{contractInstances}, a set of contract instances and (6) \textit{transactions}, information for the runtime transactions.


When an event occurs in the blockchain, e.g., a statement in a function of a smart contract is executed, the contents in the configuration are updated, i.e., a new configuration is yielded.
The rules in the K-framework describe how the configuration is yielded upon different events. 
For example, a configuration $c$ can be yielded to a new configuration
by applying rule \textsc{ReadAddress-LocalVariables}, 
if the first fragment of code in $c$'s \textit{k} cell  matches pattern \code{readAddress( Addr:Int,String2Id("Local"))} and the top of $c$ \textit{contractStack} records integer value \code{N}, which corresponds to \code{ctId} of $c$'s \textit{contractInstance} cell that records  \code{Addr |-> V} in subcell \textit{Memory}.  
As a result, the first fragment in $c$ is substituted by \code{gasCal(\#read,String2Id("Local")) $\curvearrowright$ V}, while the rest of the configuration stays unchanged.
In practice, the rule means that when \code{readAddress} is executed, the value $\code{V}$ is read from contract $\code{N}$'s local variable whose address is $\code{Addr}$, before which some gas has been consumed. The notation $\curvearrowright$ is a list constructor (read "followed by").
The detailed explanations of the configurations and the rules can be found in~\cite{DBLP:conf/sp/JiaoK0S0020}. 


\Newrule{ReadAddress-LocalVariables}{
\RWrule{
\code{readAddress(Addr:Int,String2Id("Local"))}
}{
\code{gasCal(\#read,String2Id("Local")) $\curvearrowright$ V:Value}
}{k}\\
\RruleS{
\code{ListItem(N:Int)}
}{contractStack}\\
\RruleS{
\RruleM{\code{N}}{ctId}\\
\Rrule{\code{Addr |-> V}}{Memory}
}{contractInstance}
}

\textbf{Relationship between facts and configurations}:  The configurations
In FASVERIF are mainly represented by the multiset of facts. We give examples of the relationship between the facts in FASVERIF and configurations in the K-framework, as well as the relationship between the corresponding rules, 
to facilitate the understanding of proofs in Appendix~\ref{subsec:proof}.

Suppose that a function is currently executing,
and the PC is currently at a certain point, denoted as index \ensuremath{i}.
As mentioned in Section \ref{sec:independent modeling},
fact $\mathsf{Var}_{i}$ represents the current configuration
in the model generated by FASVERIF. We demonstrate the corresponding
configuration in KSolidity as follows:
\begin{itemize}
\item The sequence of the first three terms in $\mathsf{Var}_{i}$, denoted
as $\llbracket\sigma_{a}(f),\sigma_{v}(c_{b}),\textit{calltype}\rrbracket,$
represents the current executed function $f$, contract account $c_{b}$
who invokes the function, and the type $calltype$ indicating whether
$c_{b}$ is an external account or a contract account. Correspondingly
in KSolidity, the name $f$ and $c_{b}$ is statically
recorded in sub-cell \textit{fName }and\textit{ cName}, respectively;
and the runtime information about $f$ and $c_{b}$ is stored in cell
\textit{controlStacks}. In detail, the \textit{fId} of $f$ is stored
on top of the stack\textit{ functionStack }, and the \textit{cId }of
$c_{b}$ is the second from the top of the stack \textit{contractStack}.
\item The 4th term in $\mathsf{Var}_{i}$, denoted as $\sigma_{a}(c)$,
represents the contract account whose function $f$ is being executed.
In KSolidity, the \textit{cId }of $c$ is top of the
stack \textit{contractStack}. 
\item The sequence after the 4th terms in $\mathsf{Var}_{i}$ represents
the concatenation of three sequences: the global variables of account
$c$, ether balances of all accounts, and the current local variables.
In KSolidity, global variables, balances, local variables
are stored in \textit{ctStorage}, \textit{Balance, and} \textit{Memory}, respectively. 
\end{itemize}

Consider the case of initializations in the independent model of FASVERIF, which prepares for a new transaction.
Assume that the contract account $c$'s function $f$ is to be invoked
by $c_{b}$ at the beginning of the transaction. Note that the contract
accounts, e.g., $c$, have been created, and it is also assumed that
some transactions may have finished execution before the initialization.
The facts $\mathsf{Evar}(e(\omega_{0}))$ and $\mathsf{Gvar}(\llbracket\sigma_{a}(c)\rrbracket\cdot g(\omega_{0})\backslash e(\omega_{0}))$
are generated by applying the rules \textsf{init\_evars} and \textsf{init\_gvars},
which means the initial values of balances and global variables of
account $c$ are set arbitrarily. The initial configuration can be
represented by $\llbracket\sigma_{a}(f), $ $\sigma_{v}  (c_{b}),$ $\textit{calltype=\ensuremath{\mathsf{EXT}}},\sigma_{a}(c)\rrbracket\cdot g(\omega_{0})\backslash e(\omega_{0})\cdot e(\omega_{0})$,
since there are no local variables. 

The first configuration in KSolidity corresponds to
the one when the blockchain starts running. For instance, the number
of created contract accounts is 0, as counted in \textit{cntContracts}.
Hence, the initial configuration in FASVERIF corresponds to a certain
configuration, not its first configuration, in KSolidity.
In the derivation of the configuration in the K-framework, contract
accounts have been created, i.e., rule \textsc{New-Contract-Instance-Creation}
has been applied for each creation, and transactions may also be executed,
i.e, multiple corresponding rules may be applied as well. Finally,
the application of the rules results in changing the value of the
cells, e.g., \textit{cntContracts}, in the configuration. 

\subsubsection{\textbf{Notations and theories}} 
We recall definitions and theories in Tamarin prover \cite{MeierSCB13,KremerK14}, K-framework \cite{Rosu2017} and matching logic~\cite{Rosu2017} necessarily needed for proofs in Appendix~\ref{subsec:proof} as follows.

\textbf{Tamarin \cite{MeierSCB13,KremerK14}}: 
Given a set $S$ we denote by $S^*$ the set of finite sequences of elements from $S$ and by $S^\#$ the set of finite multisets of elements from $S$.
We use the superscript $\#$ to annotate usual multiset operation, e.g., $S_1 \cup^\# S_2$ denotes the multiset union of multisets $S_1,S_2$, $\vert S\vert^{\#}$ denotes the number of elements in $S$.
Set membership modulo $E$ is denoted by $\in_E$ and defined as $e \in_E S$ if $\exists e' \in S.e' =_E e$.
Define the set of facts as the set $\mathcal{F}$ consisting of all facts $F(t_1,\dots,t_k)$, where $t_i$ are the terms.
Denote $\textit{names}(F)$ as the multiset of signature names of the facts in $\mathcal{F}$.
For a fact $f$ we denote by $\textit{ginsts}(f)$ the set of ground instances, i.e. instances that do not contain variables, of $f$.



\begin{definition}{(Multiset rewrite rule).}
    \textit{
    A labelled multiset rewrite rule $ri$ is a triple $(l,a,r)$, $l,a,r \in \mathcal{F}^*$, written $l -[a] \rightarrow r$. We call $l = \mathit{prems}(ri)$ the premises, $a = \mathit{actions}(ri)$ the actions, and $r = \mathit{conclusions}(ri)$ the conclusions of the rule.
    }
\end{definition}

\begin{definition}{(Labelled multiset rewriting system).} 
    \textit{
    A labelled multiset rewriting system is a set of labelled multiset rewrite rules $R$, such that each rule $l -[a] \rightarrow r \in R$ satisfies the following conditions:
    \begin{itemize}
        \item $ l, a, r$ do not contain fresh names
        \item $r$ does not contain $\kwf{Fr}$-facts 
    \end{itemize}
    }
\end{definition}

We define one distinguished rule $\kwt{Fresh}$ which is the only rule allowed to have $\kwf{Fr}$-facts on the right-hand side
\begin{equation}
  [] - [] \rightarrow [\kwf{Fr}(x : fresh)]  \tag{$\kwt{Fresh}$}
\end{equation}

\begin{definition}{(Labelled transition relation).}
    \label{lemma:trans}
    \textit{
    Given a multiset rewriting system $R$, define the labeled transition relation $\rightarrow_R \subseteq \mathcal{G}^\# \times \mathcal{P(G)}^\# \times \mathcal{G}^\# $ as
    $$S \xrightarrow{n}_R ((S \backslash ^\# \mathit{lfacts}(l)) \cup^\# r)$$
    if and only if $l -[a] \rightarrow r \in_E  \mathit{ginsts}(R \cup  \kwt{Fresh})$, $\mathit{lfacts}(l) \subseteq^\# S$ and $\mathit{pfacts}(l) \subseteq^\# S$.
    }
\end{definition}
Here, we denote $\mathcal{G}$ as the set of all ground facts, i.e.,  facts that do not contain variables, $n$ is the name of a rule in $R$. Given a sequence or set of facts $S$ we denote by $\textit{lfacts}(S)$ the multiset of all linear facts in $S$ and $\textit{pfacts}(S)$ the set of all persistent facts in $S$. Since the persistent facts are not used in FASVERIF, we trivially conclude the following lemma.

\begin{lemma}
    \label{lemma:simpletrans}
    [Simplified labelled transition relation]
    \textit{
    Given a multiset rewriting system $R$ that does not have persistent facts,
    $$S \xrightarrow{n}_R ((S \backslash ^\# l) \cup^\# r)$$
    where $n$ is the name of a rule in $R$, 
    if and only if $l -[a] \rightarrow r \in_E  \mathit{ginsts}(R \cup  \kwt{Fresh})$, $l \subseteq^\# S$.
    }
\end{lemma}


\begin{definition}{(MSR-executions)}
\textit{
    Given a multiset rewriting system $R$, define its set of executions as
}
    \begin{align*}
    \textit{exec}^{m s r}(R)=&\{\emptyset{\stackrel{r_{1}}{\longrightarrow}}_{R} \ldots{\stackrel{r_{n}}{\longrightarrow}}_{R} S_{n} \mid 
    \forall a, i, j: 0 \leq i \neq j<n .\\
   &(S_{i+1} \backslash^\# S_{i})=\{\operatorname{Fr}(a)\} \Rightarrow(S_{j+1} \backslash^{\#} S_{j}) \neq\{\operatorname{Fr}(a)\}\}
    \end{align*}
    The definition indicates that the rule $\kwt{Fresh}$ is at most fired once for each name in the transition sequence.
    $r_i$ is the name of a rule in $R$.
\end{definition}

\begin{definition}{(MSR-traces)}
\textit{
    The set of traces is defined as  
}
\begin{align*}
 \mathit{traces}^{msr}(R)=&\{[r_{1}, \ldots, r_{n}] \mid \forall 0 \leq i \leq n . r_{i} \neq \emptyset\\
 &\textit{and } \emptyset{\stackrel{r_{1}}{\longrightarrow}}_{R} \ldots{\stackrel{r_{n}}{\longrightarrow}}_{R} S_{n} \in \textit{exec}^{m s r}(R)\}
\end{align*}
\end{definition}

\textbf{K-framework \cite{Rosu2010} and matching logic \cite{Rosu2017}}: 

A signature $\Sigma$ is a pair $(S,F)$ where $S$ is a set of sorts and $F$ is a set of operations $f :w\rightarrow s$, where $f$ is an operation symbol, $w \in S^*$ is its arity, and $s \in S$ is its result sort.
If $w$ is the empty word $\epsilon$ then $f$ is a constant.
The universe of terms $T_\Sigma$ associated with a signature $\Sigma$ contains all the terms which can be formed by iteratively applying symbols in $F$ over existing terms (using constants as basic terms, to initiate the process), matching the arity of the symbol being applied with the result sorts of the terms it is applied to.
Given an S-sorted set of variables $\mathcal{X}$, the universe of terms $T_\Sigma(\mathcal{X})$  with operation symbols from $F$ and variables from $\mathcal{X}$ consists of all terms in  $T_\Sigma(\mathcal{X})$, where $\Sigma(\mathcal{X})$ is the signature obtained by adding the variables in $\mathcal{X}$ as constants to $\Sigma$, each to its corresponding sort. 
One can associate a signature with any context-free language (CFG), so that well-formed words in the  CFG language are associated with corresponding terms in the signature.

A \textit{substitution} is a mapping yielding terms (possibly with variables) for variables.
Any substitution $\psi : \mathcal{X} \rightarrow T_\Sigma(\mathcal{Y})$  naturally extends to terms, yielding a homonymous mapping $\psi : T_\Sigma(\mathcal{X}) \rightarrow T_\Sigma(\mathcal{Y})$.
When $\mathcal{X}$ is finite and small, the application of substitution $\psi :\{x_1,\dots,x_n\}  \rightarrow T_\Sigma(y)$ to term $t$ can be written as $t[\psi(x_1)/x_1,...,$ $\psi(x_n)/x_n] ::= \psi(t)$.
This notation allows one to use substitutions by need, without formally defining them.

Given an ordered set of variables, $\mathcal{W} = \{\square_1, \dots, \square_n\}$, named \textit{context variables}, or \textit{holes}, a $\mathcal{W}$-context over $\Sigma(\mathcal{X})$ (assume that $\mathcal{X} \cap\mathcal{W} = \emptyset$) is a term $ C \in T_\Sigma (\mathcal{X} \cup \mathcal{W})$ which is linear in $W$ (i.e., each hole appears exactly once).
The instantiation of a $\mathcal{W}$-context C with an n-tuple $\bar{t} =(t_1,\dots,t_n)$,written $C[\bar t]$ or $C[t_1,\dots,t_n]$, is the term $C[t_1/\square_1,\dots,t_n/\square_n]$.
One can alternatively regard $\bar t$ as a substitution $\bar t : \mathcal{W} \rightarrow T_\Sigma(X)$, defined by $\bar t(\square_i) = t_i$, in which case $C[\bar t] = \bar t(C)$.
A $\Sigma$-context is a $\mathcal{W}$ context over $\Sigma$ where $\mathcal{W}$ is a singleton.

A \textit{rewrite system} over a term universe $T_\Sigma$ consists of rewrite rules, which can be locally matched and applied at different positions in a $\Sigma$-term to gradually transform it.
 For simplicity, we only discuss unconditional rewrite rules. A $\Sigma$-rewrite rule is a triple $(X,l,r)$, written $(\forall \mathcal{X} ) l \rightarrow r$, where $\mathcal{X}$ is a set of variables and $l$ and $r$ are $T_\Sigma(X)$-terms, named the \textit{left-hand-side} \textit{(lhs)} and the \textit{right-hand-side} \textit{(rhs)} of the rule, respectively.
 A rewrite rule $(\forall \mathcal{X}) l \rightarrow r$ matches a $\Sigma$-term $t$ using $\mathcal{W}$-context $C$ and substitution $\theta$, iff $t = C[\theta[l]]$.
If that is the case, then the term t rewrites to $C[\theta[r]]$.
A ($\Sigma$-)rewrite-system $\mathcal{R} = (\Sigma, R)$ is a set $R$ of $\Sigma$-rewrite rules.

\begin{definition}{(K rule, K-system).}
    \label{def:k-rule-k-system}
\textit{
    A \textbf{K-rule} $\rho : (\forall \mathcal{X})p[\dfrac{L}{R}]$ over a signature $\Sigma=(S,F)$ is a tuple $(\mathcal{X},p,L,R)$, where: 
    \begin{itemize}
        \item  $\mathcal{X}$ is an $S$-sorted set, called the \textbf{variables} of the rule $\rho$;
        \item $p$ is a $\mathcal{W}$-context over $\Sigma(\mathcal{X})$, called the \textbf{rule pattern}, where $\mathcal{W}$ are the holes of $p$; $p$ can be thought of as the “read-only”  part of $\rho$;
        \item  $L, R : \mathcal{W} \rightarrow T_\Sigma(\mathcal{X} )$ associate to each hole in $\mathcal{W}$ the \textbf{original term} and its \textbf{replacement term}, respectively; $L$, $R$ can be thought of as the “read-write” part of $\rho$.
    \end{itemize}
    We may write $(\forall \mathcal{X})p[ \dfrac{l_1}{r_1},\dots,\dfrac{l_n}{r_n} ]$ instead of $\rho : (\forall \mathcal{X})p[\dfrac{L}{R}]$ whenever $\mathcal{W} =\{\square_1,\dots,\square_n\}$ and $L(\square_i)=l_i$ and $R(\square_i)=r_i$; this way, the holes are implicit and need not be mentioned.
    A set of K rules $\mathcal{K}$ is called a \textbf{K system}.
} 
\end{definition}

Recall rule \textsc{ReadAddress-LocalVariables} in Appendix~\ref{sec:Aglance}.
When formalizing the rule as $\rho_r : (\forall \mathcal{X}_r)p_r[\dfrac{L_r}{R_r}]$,
according to Definition~\ref{def:k-rule-k-system},
 we obtain 
 \begin{small}
 \begin{align*}
    \rho_r : \ & \Mrule{  \dfrac{\code{readAddress(Addr,String2Id("Local"))}}{\code{gasCal(\#read,String2Id("Local")) $\curvearrowright$ V}} \  \_ }{k} \\
    &\Mrule{\code{N} \ \_}{contractStack}\ 
    \Mrule{\ 
        \Mrule{\code{N}}{ctId}\ 
        \Mrule{\_ \ \code{Addr} \mapsto \code{V} \ \_}{Memory}\ 
    }{contractInstance}
 \end{align*}
  \end{small}

 If we want to identify the anonymous variables, the rule could be alternatively written as:
 \begin{small}
 \begin{align*}
    \rho_r : \ & \Mrule{  \dfrac{\code{readAddress(Addr,String2Id("Local"))}}{\code{gasCal(\#read,String2Id("Local")) $\curvearrowright$ V}} \  a }{k} \\
    &\Mrule{\code{N} \ b}{contractStack}\ 
    \Mrule{\ 
        \Mrule{\code{N}}{ctId}\ 
        \Mrule{c \ \code{Addr} \mapsto \code{V} \ d}{Memory}\ 
    }{contractInstance}
 \end{align*}
 \end{small}

Here, we have:
\begin{align*}
     \mathcal{X}_r&=\{ \code{Addr},  \code{V}, a, N, b,c,d \} \\
     \mathcal{W}_r&=\{\square\} \\
     p_r&=  
    \Mrule{ \square \  a }{k} \ 
    \Mrule{\code{N} \ b}{contractStack}\ \\
    &\ \ \ \ \   \Mrule{\ 
        \Mrule{\code{N}}{ctId}\ 
        \Mrule{c \ \code{Addr} \mapsto \code{V} \ d}{Memory}\ 
    }{contractInstance}\\
    L_r(\square) &= \code{readAddress(Addr,String2Id("Local"))} \\
    R_r(\square) &= \code{gasCal(\#read,String2Id("Local")) $\curvearrowright$ V}
\end{align*}

\Newrule{AllocateAddress-LocalVariables}{
\RWrule{
\code{allocateAddress(N:Int, Addr:Int,}\\ 
\code{String2Id("Local"), V:Value)}
}{
\code{gasCal(\#allocate,String2Id("Local")) $\curvearrowright$ V}
}{k}\\
\RruleS{
\RruleM{\code{N}}{ctId}\\
\RWruleN{
\code{MEMORY:Map}
}{
\code{MEMORY (Addr |-> V) }
}{Memory}
}{contractInstance}
}

Another example is rule \textsc{AllocateAddress-LocalVariables}.
The difference from the first example is that there are 2 cells to be subistituted in the configuration.
When formalizing the rule as $\rho_w : (\forall \mathcal{X}_w)p_w[\dfrac{L_w}{R_w}]$,
according to Definition~\ref{def:k-rule-k-system},
 we obtain 
\begin{align*}
    \rho_w :\  & 
    \Mrule{
        \dfrac{\code{allocateAddress(N, Addr,}
                \code{String2Id("Local"), V)}}
                {\code{gasCal(\#allocate,String2Id("Local")) $\curvearrowright$ V}}
                \ \_
    }{k} \\
    & \Mrule{
        \Mrule{\code{N}}{ctId}\ 
        \Mrule{
            \dfrac{\code{MEMORY}}{\code{MEMORY} \  (\code{Addr} \mapsto \code{V}) }\ \_
            }{Memory}
    }{contractInstance}
\end{align*}
 If we want to identify the anonymous variables, the rule could be alternatively written as:
\begin{align*}
    \rho_w :\  & 
    \Mrule{
        \dfrac{\code{allocateAddress(N, Addr,}
                \code{String2Id("Local"), V)}}
                {\code{gasCal(\#allocate,String2Id("Local")) $\curvearrowright$ V}}
                \ a
    }{k} \\
    & \Mrule{
        \Mrule{\code{N}}{ctId}\ 
        \Mrule{
            \dfrac{\code{MEMORY}}{\code{MEMORY}\  (\code{Addr} \mapsto \code{V}) }\ b
            }{Memory}
    }{contractInstance}
\end{align*}
Here, we have:
\begin{align*}
     \mathcal{X}_w&=\{\code{N}, \code{Addr},  \code{V}, a,  b, \code{MEMORY} \} \\
     \mathcal{W}_w&=\{\square_1, \square_2\} \\
     p_w&=  
    \Mrule{ \square_1 \  a }{k} \ 
    \Mrule{\ 
        \Mrule{\code{N}}{ctId}\ 
        \Mrule{\square_2 \  b}{Memory}\ 
    }{contractInstance}\\
    L_w(\square_1) &= \code{allocateAddress(N, Addr,}
                \code{String2Id("Local"), V)} \\
    R_w(\square_1) &=   \code{gasCal(\#allocate,String2Id("Local")) $\curvearrowright$ V} \\
    L_w(\square_2) &=  \code{MEMORY} \\
    R_w(\square_2) &= \code{MEMORY} \ (\code{Addr} \mapsto \code{V})
\end{align*}



The matching logic \cite{Rosu2017} can serve as a logic foundation of K system. 
We give some definitions and properties that we use as follows.

\begin{definition}{(Pattern).}
    \label{def:pattern}
\textit{
    A matching logic formula, or a \textbf{pattern}, is a first-order logic (FOL) formula.
    Let  $\mathcal{T}$ denote the elements of $T_\Sigma(\mathcal{X})$ of a distinguished sort, called configurations.
    Define satisfaction $ (\gamma, \psi) \vDash \varphi $ over configurations $\gamma \in \mathcal{T}$, valuations (can also be seen as substitutions) $\psi : T_\Sigma(\mathcal{X}) \rightarrow T_\Sigma(\mathcal{Y})$ and patterns $\varphi$ as follows (among the FOL constructs, we only show $\exists$):
    \begin{itemize}
        \item $(\gamma, \psi) \vDash \exists X \varphi $ iff $ (\gamma, \psi') \vDash \varphi $ for some $\psi' : T_\Sigma(\mathcal{X}) \rightarrow T_\Sigma(\mathcal{Y})$ with $\psi'(y)=\psi(y)$ for all $y \in T_\Sigma(\mathcal{X}) \backslash X$.
        \item $(\gamma, \psi) \vDash \pi$ iff $\gamma = \psi(\pi)$, where $\pi \in \mathcal{T} $.
    \end{itemize}
    We write $\vDash \varphi$ when $ (\gamma, \psi) \vDash \varphi $   for all $\gamma \in \mathcal{T}$  and all $\psi : T_\Sigma(\mathcal{X}) \rightarrow T_\Sigma(\mathcal{Y})$.
}
\end{definition}
\textit{Example}: in rule \textsc{AllocateAddress-LocalVariables}, $p_w$ is an abbreviated form of FOL logic formula:
\begin{equation*}
    \exists \square_1,a,N,\square_2,b.
    \Mrule{
    \Mrule{ \square_1 \  a }{k} \ 
    \Mrule{\ 
        \Mrule{\code{N}}{ctId}\ 
        \Mrule{\square_2 \  b}{Memory}\ 
    }{contractInstance}
    }{\textit{Cfg}}
\end{equation*}
Here, $\Mrule{\dots}{Cfg}$ represents a configuration pattern.

\begin{lemma}{(Structural Framing)}
    \label{lemma:structuralFraming}
     If $\sigma \in \Sigma_{s_1,\dots,s_n,s} $, 
     and $\varphi_i,\varphi_i' \in \mathrm{Pattern}_{s_i} $
        such that $\vDash \varphi_i $ $\rightarrow \varphi_i'$  for all $i \in 1\dots n$, then $\vDash \sigma(\varphi_1,\dots,\varphi_n)$ $ \rightarrow  \sigma(\varphi_1',\dots,\varphi_n') $.

    Let $T_{\Sigma,s}(\textit{Var})$  be the set of $\Sigma$-terms of sort $s$, and $\mathrm{Pattern}_s$ be the $s$-sorted set of patterns. Therefore, think of $\Sigma_{s_1,\dots,s_n,s}$ as the pattern $\sigma(x_1:s_1,\dots,x_n:s_n)$.
\end{lemma}

\textit{Example}: assume that
$$p= \Mrule{ 
    \code{gasCal(\#allocate,String2Id("Local"))  \_}
    }{k}\ 
    \Mrule{
        5
    }{ctId}
$$ 
Since  
$
\vDash
    \Mrule{
        5
    }{ctId}
    \rightarrow \top
$, 
where intuitively $\top$ is a pattern that is matched by all elements,
 we can get $\vDash p \rightarrow p'$ by Lemma~\ref{lemma:structuralFraming}, where
$$p'= \Mrule{ 
    \code{gasCal(\#allocate,String2Id("Local"))  \_}
    }{k}\ 
$$ 

\begin{lemma}
    \label{lemma:imply}
   If $ (\gamma, \psi) \vDash p $, and $\vDash p \rightarrow p'$, then $ (\gamma, \psi) \vDash p' $
\end{lemma}

Following from the definition of reachability system which is based on matching logic \cite{Rosu2012}, we  define the transition system in K systems:

\begin{definition}{(K-transition system)}
    \label{def:transitionsystem}
    \textit{
     The K system $\mathcal{K}$ induces a \textbf{K transition system} ($\mathcal{T},\karrowwx, \gamma_0$) on the configuration model. 
     Here  $\gamma_0 \in \mathcal{T}$ is the initial configuration.
      $\gamma \karrowwx \gamma'$ 
      for 
      $\gamma, \gamma' \in \mathcal{T}$
       iff there is a substitution
         $\psi : T_\Sigma(\mathcal{X}) \rightarrow T_\Sigma(\mathcal{Y})$
            and 
            $\rho : (\forall \mathcal{X})p[\dfrac{L}{R}]$
        in $\mathcal{K}$
          with
           $\gamma=\psi(p[L])$
            and
           $\gamma'=\psi(p[R])$
        also written as 
      $\gamma \karrowx{} \gamma'$.
    }
\end{definition}

\begin{definition}{(K-executions)}
\textit{
    Given a K-transition system ($\mathcal{T},\karrowwx, \gamma_0$) , define its set of executions as
}
    \begin{equation*}
    \mathit{exec}^\mathcal{K}(\gamma_0)=\{\gamma_0 \karrowx{1} \gamma_1 \karrowx{2} \dots \karrowx{n} \gamma_n
            \}
    \end{equation*}
\end{definition}



\begin{definition}{(K-traces).}
    \label{def:ktraces}
    Given a K transition system $(\mathcal{T},\karrowwx, \gamma_0)$, define the set of K traces as
    \begin{equation*}
        \mathit{traces}^\mathcal{K}(\gamma_0)=\{
            [\rho_1,\dots,\rho_n] \mid 
            \gamma_0 \karrowx{1} \gamma_1 \karrowx{2} \dots \karrowx{n} \gamma_n
            \}
    \end{equation*}
\end{definition}

\subsubsection{\textbf{Proof of Soundness}}
\label{subsec:proof}

In this section, we first introduce some definitions, propositions and lemmas that will be used in the subsequent proofs.
We then briefly describe the property we want to prove that \ourtool satisfies, that is, soundness.
Finally, we introduce and prove a theorem and show how it can be used to prove the soundness of \ourtool.
Note that the proofs can improve the faithfulness of \ourtool, but it still does not mean that the results of \ourtool are completely reliable, due to the informal part of \ourtool, \textit{i.e.}, the property generation, and the gap between Solidity and EVM bytecode.
Here we regard the property generation as an informal part since our properties are generated based on our statistical analysis instead of using formal methods.


\begin{table*}[ht]
    \centering
    \caption{Correspondence relationship of translated rules from Solidity codes. }
    \label{table:tableCorrespondence}
    \resizebox{\textwidth}{!}{




\begin{tabular}{|c|c|c|c|c|c|c|c|}
    \hline 
    \multirow{2}{*}{ID} & \multirow{2}{*}{States of solidity process in function $f_{c}$} & \multicolumn{4}{c|}{KSolidity} & \multicolumn{2}{c|}{FASVERIF}\tabularnewline
    \cline{3-8} \cline{4-8} \cline{5-8} \cline{6-8} \cline{7-8} \cline{8-8} 
     &  & Start or key rule & Command in rule & Correspondence & Valuation $\mathcal{E}$ & Key Rule & $F_{c}(i,id)$\tabularnewline
    \hline 
    \multirow{2}{*}{1} & $(\gamma,$ $\phi$, $\texttt{function}\ f_{c}(\texttt{d})\{\kwp{stmt}\}$) & \multirow{2}{*}{$\textsc{Function-Call}$} & $\code{functionCall(C:Int;R:Int;}$ & \multirow{2}{*}{$\code{Es}$ \textasciitilde{} d} & \multirow{2}{*}{$\mathcal{E}(\sigma(seq(\texttt{d})))=\texttt{d}$} & \multirow{2}{*}{$\kwt{recv\_ext}$} & $\{\kwf{Call}_{e},\kwf{GVar}$ \tabularnewline
     & $\Rightarrow(\gamma',1,\kwp{stmt})$  &  & $\code{F:Id;Es:Values;M:Msg)}$ &  &  &  & \multirow{1}{*}{$,\kwf{EVar}\}$}\tabularnewline
    \hline 
    \multirow{3}{*}{2} & local state: $(\gamma,i,x(c_{x}).f_{x}(p);\kwp{stmt}$$)\Rightarrow$ & \multirow{3}{*}{$\textsc{Function-Call}$} & $\code{functionCall(C:Int;R:Int;}$ &  &  & \multirow{3}{*}{$\kwt{in\_call}$} & \multirow{3}{*}{$\{\kwf{Var}_{i}\}$}\tabularnewline
     & Recipient state: $(\gamma',1,\kwp{stmt}_{x})$ &  & $\code{F:Id;Es:Values;M:Msg)}$ & $\code{Es}$ \textasciitilde{} p  & $\mathcal{E}(\sigma(seq(\texttt{p})))=\texttt{p}$ &  & \tabularnewline
     & local state: ...$\Rightarrow(\gamma'',i\circ1\circ1$, $\kwp{stmt}$$)$ &  &  &  &  &  & \tabularnewline
    \hline 
    \multirow{2}{*}{3} & $(\gamma,i,v_{1}\leftarrow v_{2};\kwp{stmt})$ & \multirow{2}{*}{..., $\textsc{{Write}}$} & \multirow{2}{*}{$\code{X:Id=V:Value}$} & $\code{X}\sim v_{1}$  & \multirow{2}{*}{$\mathcal{E}(\sigma_{v}(v_{2}))=v_{2}$} & \multirow{2}{*}{$\kwt{var\_assign}$} & $\{\kwf{Var}_{i}\}$\tabularnewline
    \cline{8-8} 
     & $\Rightarrow(\gamma',i\circ1,\kwp{stmt})$ &  &  & $\code{V}\sim v_{2}$ &  &  & $\{\kwf{Var}_{i}\}$\tabularnewline
    \hline 
    \multirow{2}{*}{4} & $(\gamma,i,\tau\ v_{1}\leftarrow v_{2};\kwp{stmt})$ & \multirow{2}{*}{..., $\textsc{Var-Declaration}$} & \multirow{2}{*}{$\code{T:EleType\ X:Id=V:Value}$} & $\code{X}\sim v_{1}$  & \multirow{2}{*}{$\mathcal{E}(\sigma_{v}(v_{2}))=v_{2}$} & \multirow{2}{*}{$\kwt{var\_declare}$} & \multirow{2}{*}{$\{\kwf{Var}_{i}\}$}\tabularnewline
     & $\Rightarrow(\gamma',i\circ1,\kwp{stmt})$ &  &  & $\code{V}\sim v_{2}$ &  &  & \tabularnewline
    \hline 
    \multirow{2}{*}{5} & $(\gamma,i,\texttt{if}\ e_{b}\ \texttt{then}\ \kwp{stmt}_{1}\ \texttt{else}\ \kwp{stmt}_{2};\kwp{stmt}_{3})$ & \multirow{2}{*}{..., $\textsc{{R5}}$} & $\code{if\ (true)\ S:Statement}$ & \multirow{2}{*}{$\code{true}\sim e_{b}$} & \multirow{2}{*}{$\mathcal{E}(\theta_{e}(e_{b}))=\code{true}$} & \multirow{2}{*}{$\kwt{if\_true}$} & \multirow{2}{*}{$\{\kwf{Var}_{i}\}$}\tabularnewline
     & $\Rightarrow(\gamma',i\circ1,\kwp{stmt}_{1};\kwp{stmt_{3}})$ &  & $\code{\ else\ S1:Statement}$ &  &  &  & \tabularnewline
    \hline 
    \multirow{2}{*}{6} & $(\gamma,i,\texttt{if}\ e_{b}\ \texttt{then}\ \kwp{stmt}_{1}\ \texttt{else}\ \kwp{stmt}_{2};\kwp{stmt}_{3})$ & \multirow{2}{*}{..., $\textsc{{R6}}$} & $\code{if\ (false)\ S:Statement}$ & \multirow{2}{*}{$\code{false}\sim e_{b}$} & \multirow{2}{*}{$\mathcal{E}(\theta_{e}(e_{b}))=\code{false}$} & \multirow{2}{*}{$\kwt{if\_false}$} & \multirow{2}{*}{$\{\kwf{Var}_{i}\}$}\tabularnewline
     & $\Rightarrow(\gamma',i\circ2,\kwp{stmt}_{2};\kwp{stmt_{3}})$ &  & $\code{\ else\ S1:Statement}$ &  &  &  & \tabularnewline
    \hline 
    7 & $(\gamma,i,\texttt{require}\ e_{b}$; $\kwp{stmt})$ & \multirow{2}{*}{..., $\textsc{{Require}}$} & $\code{require(true)}$ & $\code{true}\sim e_{b}$ & $\mathcal{E}(\theta_{e}(e_{b}))=\code{true}$ & $\kwt{require\_true}$ & $\{\kwf{Var}_{i}\}$\tabularnewline
    \cline{1-1} \cline{4-8} \cline{5-8} \cline{6-8} \cline{7-8} \cline{8-8} 
    8 & $\Rightarrow(\gamma',i\circ1,\kwp{stmt})$  &  & $\code{require(false)}$ & $\code{false}\sim e_{b}$ & $\mathcal{E}(\theta_{e}(e_{b}))=\code{false}$ & $\kwt{require\_false}$ & $\{\kwf{Var}_{i}\}$\tabularnewline
    \hline 
    9 & $(\gamma,i,\texttt{return}\ \code{i\}});$ $\Rightarrow(\gamma',\phi,\_)$ & $...,\textsc{Return-Value}$ & $\code{return\ E:Value}$ & $\code{E}\sim i$ &  & $\kwt{ret\_ext}$ & $\{\kwf{Var}_{i}\}$\tabularnewline
    \hline 
    10 & $(\gamma,i,\texttt{return}\ \code{i\}});$ $\Rightarrow...(\textit{caller states})$ & $...,\textsc{Return-Value}$ & $\code{return\ E:Value}$ & $\code{E}\sim i$ &  & $\kwt{ret\_in}$ & $\{\kwf{Var}_{i}\}$\tabularnewline
    \hline 
    11 & $(\gamma,i,c_{x}.\texttt{transfer}(v_{1});\kwp{stmt}$$)$  & \multirow{2}{*}{..., $\textsc{Transfer-Fund-Begin}$} & \multirow{1}{*}{$\code{\#memberAccess(R:Id,F:Id)}$} & $\code{MsgValue}\sim v_{1}$ & \multirow{2}{*}{$\mathcal{E}(\sigma_{v}(v_{1}))=v_{1}$} & $\kwt{transfer\_succ}$ & $\{\kwf{Var}_{i}\}$\tabularnewline
    \cline{1-1} \cline{7-8} \cline{8-8} 
    12 & $\Rightarrow(\gamma',i\circ1,\kwp{stmt})$  &  & $\code{\ensuremath{\curvearrowright}MsgValue:Int}$ & $\code{R}\sim c_{x}$  &  & $\kwt{transfer\_fail}$ & $\{\kwf{Var}_{i}\}$\tabularnewline
    \hline 
    13 & $(\gamma,i,c_{x}.\texttt{send}(v_{1});\kwp{stmt})$  & \multirow{2}{*}{..., $\textsc{Send-Fund-Begin}$} & \multirow{1}{*}{$\code{\#memberAccess(R:Id,F:Id)}$} & $\code{MsgValue}\sim v_{1}$ & \multirow{2}{*}{$\mathcal{E}(\sigma_{v}(v_{1}))=v_{1}$} & $\kwt{send\_succ}$ & $\{\kwf{Var}_{i}\}$\tabularnewline
    \cline{1-1} \cline{7-8} \cline{8-8} 
    14 & $\Rightarrow(\gamma',i\circ1,\kwp{stmt})$ or $\Rightarrow(\gamma',i\circ2,\kwp{stmt})$ &  & $\code{\ensuremath{\curvearrowright}MsgValue:Int}$ & $\code{R}\sim c_{x}$  &  & $\kwt{send\_fail}$ & $\{\kwf{Var}_{i}\}$\tabularnewline
    \hline 
    \end{tabular}

    }
\end{table*}

Given a set of source codes in Solidity language $S=\{c_1,c_2,\dots,c_{n}\}$, we define the K system of KSolidity $\ks$, the independent model of \ourtool $\rs$, the complementary model for the invariant property $R_{inv}$ and the complementary model for the equivalence property $R_{equ}$ as follows. 

\begin{definition}{($\ks$)}
    $\ks=\{$\textsc{Require, Out-of-gas},  $\dots\}
    $, \textit{as shown in Figure~\ref{fig:definerule1},\ref{fig:definerule2},\ref{fig:definerule3},\ref{fig:definerule4},\ref{fig:definerule5},\ref{fig:definerule6},\ref{fig:definerule7}.
    The K system $\ks$ induces a \textbf{K transition system} ($\mathcal{T},\karroww, \gamma_0$) according to Definition~\ref{def:transitionsystem}. Here the $k$ cell in $\gamma_0$ is ${\left\langle c_j \right\rangle}_{k} $ where $c_j \in S$.
    }
\end{definition}
Note that we include a subset of KSolidity rules. 
For instance, rules for $\kwt{while}$ statement and arrays are not included.
To improve readability, we also remove redundant and similar rules for proving.

\begin{definition}{$(R_s)$}
   \textit{
    $\rs= \mathcal{R}(c_1) \cup \mathcal{R}(c_2) \cup \dots \mathcal{R}(c_{n})\cup \{ \kwt{init\_evars,}$ $ \kwt{init\_gvars, fb\_in\_call,ret\_fb}\}$. 
    }
\end{definition}

\begin{definition}{$(R_{inv})$}
    \textit{
     $R_{inv}= \mathcal{R}(c_1) \cup \mathcal{R}(c_2) \cup \dots \mathcal{R}(c_{n})  \backslash$ $ \{ \kwt{ext\_call, ret\_ext}\} \cup \{ \kwt{init\_evars,ret\_fb, fb\_in\_call,init\_}$ $\kwt{gvars\_inv, ext\_call\_inv, ret\_ext\_}$ $\kwt{inv}\}$. 
     }
\end{definition}

\begin{definition}{$(R_{equ})$}
    \textit{
    $R_{equ}= \mathcal{R'}(c_1) \cup \mathcal{R'}(c_2) \cup \dots \mathcal{R'}(c_{n})$ $ \cup \{ \kwt{init\_evars\_AB,init\_gvars\_AB, ret\_fb\_A,ret\_fb\_B,compare\_AB}$ $\kwt{,fb\_in\_call\_A,}$ $\kwt{fb\_in\_call\_B}\}$.
    }
\end{definition}

\begin{table}[]
 \centering
    \caption{The code of state-typed configurations}
    \label{table:code}
\begin{tabular}{|c|}
\hline
code \\ \hline
$\code{functionCall(C:Int;R:Int;F:Id;Es:Values;M:Msg)}$   \\ \hline
$\code{X:Id=V:Value}$   \\ \hline
$\code{T:EleType\ X:Id=V:Value}$    \\ \hline
$\code{if\ (true)\ S:Statement}$   \\ \hline
$\code{if\ (false)\ S:Statement}$    \\ \hline
$\code{require(true)}$    \\ \hline
$\code{require(false)}$    \\ \hline
$\code{return\ E:Value}$ \\ \hline
$\code{\#memberAccess(R:Id,F:Id)}$ \\ \hline
\end{tabular}
\end{table}

\begin{table}[]
 \centering
    \caption{Transitions of states of Solidity process}
    \label{table:states}
\begin{tabular}{|c|}
\hline
transition \\ \hline
$(\gamma,\phi, \texttt{function}\ f_{c}(\texttt{d})\{\kwp{stmt}\})$\\ $\Rightarrow(\gamma',1,\kwp{stmt})$  \\ \hline
local state: $(\gamma,i,x(c_{x}).f_{x}(p);\kwp{stmt})$\\ $\Rightarrow$ Recipient state: $(\gamma',1,\kwp{stmt}_{x})$\\ local state: ...$\Rightarrow(\gamma'',i\circ1\circ1, \kwp{stmt})$ \\ \hline
$(\gamma,i,v_{1}\leftarrow v_{2};\kwp{stmt})\Rightarrow(\gamma',i\circ1,\kwp{stmt})$ \\ \hline
$(\gamma,i,\tau\ v_{1}\leftarrow v_{2};\kwp{stmt})\Rightarrow(\gamma',i\circ1,\kwp{stmt})$ \\ \hline
$(\gamma,i,\texttt{if}\ e_{b}\ \texttt{then}\ \kwp{stmt}_{1}\ \texttt{else}\ \kwp{stmt}_{2};\kwp{stmt}_{3})$\\ $\Rightarrow(\gamma',i\circ1,\kwp{stmt}_{1};\kwp{stmt_{3}})$ \\ \hline
$(\gamma,i,\texttt{if}\ e_{b}\ \texttt{then}\ \kwp{stmt}_{1}\ \texttt{else}\ \kwp{stmt}_{2};\kwp{stmt}_{3})$\\ $\Rightarrow(\gamma',i\circ2,\kwp{stmt}_{2};\kwp{stmt_{3}})$ \\ \hline
$(\gamma,i,\texttt{require}\ e_{b}; \kwp{stmt})$\\ $\Rightarrow(\gamma',i\circ1,\kwp{stmt})$ \\ \hline 
$(\gamma,i,\texttt{return}\ \code{i})\Rightarrow(\gamma',\phi,\_)$  \\ \hline
$(\gamma,i,\texttt{return}\ \code{i})\Rightarrow...(\textit{caller states})$ \\ \hline
$(\gamma,i,c_{x}.\texttt{transfer}(v_{1});\kwp{stmt})$\\ $\Rightarrow(\gamma',i\circ1,\kwp{stmt})$  \\ \hline
$(\gamma,i,c_{x}.\texttt{send}(v_{1});\kwp{stmt})$\\ $\Rightarrow(\gamma',i\circ1,\kwp{stmt})$ or $\Rightarrow(\gamma',i\circ2,\kwp{stmt})$ \\ \hline
\end{tabular}
\end{table}

Note that the semantics of function calls in KSolidity is designed from a general point of view, and the K-rules corresponding to $\kwt{ether\_succ}, \kwt{ether\_fail},$ $\kwt{fb\_call}$ are not implemented, so we omit the proof for the 3 rules.


Let a runtime process of Solidity be a state transition system, named \textit{Solidity process}, that the state is yielded iff a Solidity function is called or a statement in a Solidity function is executed, and the next transition is prepared to be executed.

\begin{definition}{(State-typed configurations)}
\label{def:stateconfig}
Define a configuration $\gamma$ as state-typed, writtern \typed($\gamma$), when the first fragment of code in \textit{k} cell of $\gamma$ is shown in Table \ref{table:code}.
\end{definition}


\begin{definition}{(States of Solidity process)}
\label{def:Soliditystate}
Given a state-typed configuration $\gamma$, define a state of Solidity process as $(\gamma, i, c)$.
Here, $i$ represents the position of a state on the syntax tree of its function, and $c$ is the sequence of codes (including statements and functions) of the current function to be executed. 
$i$ and $c$ can be computed from the configuration $\gamma$ and the correspondence between them is shown in Table~\ref{table:tableCorrespondence}.
\end{definition}

In Table~\ref{table:tableCorrespondence}, the first column represents the type ID we assign to the current state, and we denote $T(s)$ as type ID of state $s$.
Let function $F_c(i,id)$ be the key facts to be consumed when the state transits from the position $i$ and type $id$.
We manually assign the value of $F_c(i,id)$ in the table which is used for latter proving.
Here, the key facts are the ones that can be identified and used to build the correspondence relationship in the proof.

\begin{definition}{(Transition relations of states of Solidity process)}
Define the transition relation $(\gamma_n, i_n, c_n) \Rightarrow (\gamma_m, i_m, c_m)$ representing that a state $(\gamma_n, i_n, c_n)$ is yielded to state $(\gamma_m, i_m, c_m)$ when the following conditions hold:
\begin{enumerate}
\item $\exists [\rho_{n+1},\dots,\rho_m].\ \gamma_n \newkarrowx{n+1} \dots \newkarrowx{m} \gamma_m$
\item $\forall \gamma_j.\ n < j < m \rightarrow \neg \typed(\gamma_j)$.
\end{enumerate}
\end{definition}
The transitions of $i$ and $c$ during the transitions of states are shown in Table \ref{table:states}.
On the ninth row of Table \ref{table:states}, the notation $\_$ means the function of an external transaction to be executed is not concerned (can be an arbitrary one).

Assume the KSolidity semantic is correct and relatively complete as claimed, so the KSolidity rules shown in Table~\ref{table:tableCorrespondence} should correctly correspond to state transition during a transaction.
Hence we propose the following assumption:

\begin{proposition}
    \label{proposition:transitionsystem}
Given a K transition system $(\mathcal{T},\karroww, \gamma_0)$ that satisfies $\ks \vDash \alpha$.
For any 
 $[\rho_1,\dots,\rho_n] \in \tracek{\gamma_0}$, 
 $
            \lambda_0 \newkarrowx{1} \lambda_1 \newkarrowx{2} \dots \newkarrowx{n} \lambda_n
 $
 corresponds to a sequence of state transitions $s_0 \Rightarrow s_1\Rightarrow \dots \Rightarrow s_m$, and
 there exists a strictly monotonically increasing function $f: \{0,\dots, m\} \rightarrow \{0\dots, n\}$ such that $\gamma_j= \lambda_{f(j)}$, where $s_j=(\gamma_j,i_j,c_j)$.
\end{proposition}

\begin{definition}{(K Successors).}
\label{def:ksucc}
   \textit{
    Let $\rho \in \ks$,
     $\ksucc(\rho)$ represents the set of \textbf{K successors} of $\rho$ such that
    \begin{equation*}
        \ksucc(\rho)= \{ \rho_1 \in \ks  \mid  \exists \gamma_1,\gamma_2,\gamma_3. \gamma_1\newkarrowx{}\gamma_2\newkarrowx{1} \gamma_3  \}  
    \end{equation*}
   }
\end{definition}

\begin{definition}{(Code Successors).}
\label{def:csucc}
   \textit{
    Let $\rho \in \ks$,
     $\csucc(\rho)$ represents the set of \textbf{code successors} of $\rho$:
    \begin{align*}
        \csucc(\rho)= &\{ \rho_1 \in \ks \mid  
        \exists \psi,\psi_1.
        \psi(p[R])=\psi_1(p_1[L_1]) \}
    \end{align*}
    where $\rho : (\forall \mathcal{X})p[\dfrac{L}{R}]$ and $\rho_1 : (\forall \mathcal{X}_1)p_1[\dfrac{L_1}{R_1}]$
   }
\end{definition}

\begin{lemma}
    \label{lemma:subsetsucc}
    For each $\rho \in \ks$, we have 
    $$ \ksucc(\rho) \subseteq \csucc(\rho)$$
\end{lemma}
\begin{proof}
    We prove the lemma by the following sequence:
    \begin{enumerate}
        \item Assume $\rho_1 \in \ksucc(\rho)$.
        \item By Definition~\ref{def:ksucc}, $\exists \gamma_1,\gamma_2,\gamma_3. \gamma_1\karrow{}\gamma_2\karrow{1}\gamma_3$.
        \item Eliminate $\exists$ in (2), i.e.,  $\gamma_1\karrow{}\gamma_2\karrow{1}\gamma_3$ 
        \item By (3) and Definition~\ref{def:transitionsystem},  
        $$\exists \psi. \gamma_2=\psi(p[R])  $$
        \item Eliminate $\exists$ in (4), i.e., $\gamma_2=\psi(p[R])$
        \item By (3) and Definition~\ref{def:transitionsystem}, 
        $$\exists \psi_1.  \gamma_2=\psi_1(p_1[L]) $$
        \item By (5), (7), $$\psi(p[R])=\psi_1(p_1[L])$$
        \item By (7) and Definition~\ref{def:csucc}, $\rho_1 \in \csucc(\rho)$.
        \item By (1) and (8), proved.
    \end{enumerate}
\end{proof}


\begin{definition}{($\smatch$).}
    \label{def:simplematch}
    \textit{
    Let $\gamma$ be a configuration in $\ks$,  and $p$ be a pattern, we write  $ \gamma \smatch p $ iff there exists a substitution $\psi$, such that $( \gamma, \psi) \vDash p$.
    }
\end{definition}

\begin{definition}{(2-tuples and sets in $\ks$.)}
    \label{def:LGE}
\textit{
    In system $\ks$, let local/global variables be $\mathbb{L}(\gamma)/\mathbb{G}(\gamma)$, i.e., the \textit{Memory}/\textit{ctContext} cell of the contract on top of the \textit{contractStack} of configuration $\gamma$.
    Let balances be $\mathbb{E}(\gamma)$, i.e., union of \textit{Balance} cell of all \textit{contractInstance} cells. 
    Define $\mathbb{L}_b(\gamma)$/$\mathbb{G}_b(\gamma)$ which collects the local/global variables from cell \textit{globalContext}, respectively.
    Define $\mathbb{G}(\gamma,x)$/ $\mathbb{G}_g(\gamma,x)$ as global variables of contract $x$ collected by using cell \textit{ctContext}/\textit{globalContext}, respectively, in configuration $\gamma$.
    Define $\mathbb{S}(\gamma)$ as the collection of all contract ID in configuration $\gamma$.
    Define $\mathbb{M}(\gamma)$ as the current transaction of configuration $\gamma$.
     Formally, 
}
\begin{small}
\begin{align*}
    \mathbb{L}&(\gamma)=\{  (a,v) \mid \gamma \smatch 
      \Mrule{c \ \_}{contractStack}\ 
      \Mrule{\ 
      \Mrule{c}{ctId} \ 
      \Mrule{\_\  a \mapsto d\  \_}{ctContext} \ 
    \\
       &\Mrule{\_\ a \mapsto \textit{Local} \_}{ctLocation}\ 
        \Mrule{\_\ d\mapsto v \ \_}{Memory}
         \ 
    }
    {contractInstance}
     \}\\
    \mathbb{L}_b&(\gamma)=\{  (a,v) \mid \gamma \smatch 
      \Mrule{c \ \_}{contractStack}\ 
      \Mrule{\ 
      \Mrule{c}{ctId} \ 
    \\
       &  \Mrule{\_\  a \mapsto d\  \_}{globalContext} \ 
    \Mrule{\_\ a \mapsto \textit{Local} \_}{ctLocation}\ 
        \\
       &\Mrule{\_\ d\mapsto v \ \_}{Memory}
         \ 
    }
    {contractInstance}
     \}\\
     \mathbb{G}&(\gamma)=
      \{ (a,v) \mid \gamma \smatch 
       \Mrule{c \ \_}{contractStack}\ 
       \Mrule{\ 
        \Mrule{c}{ctId} \ 
        \Mrules{a \mapsto d}{ctContext} \ 
    \\&    \Mrules{a \mapsto \textit{Global}}{ctLocation}\ 
      \Mrule{\_\  d\mapsto v \ \_}{ctStorage}
         \ 
    }
    {contractInstance}
     \}\\
     \mathbb{G}&(\gamma,x)=
      \{ (a,v) \mid \gamma \smatch 
       \Mrule{\ 
        \Mrule{x}{ctId} \ 
        \Mrules{a \mapsto d}{ctContext} \ 
    \\&    \Mrules{a \mapsto \textit{Global}}{ctLocation}\ 
      \Mrule{\_\  d\mapsto v \ \_}{ctStorage}
         \ 
    }
    {contractInstance}
     \}\\
     \mathbb{G}_b&(\gamma)=
      \{ (a,v) \mid \gamma \smatch 
       \Mrule{c \ \_}{contractStack}\ 
       \Mrule{\ 
    \\
       &    \Mrule{c}{ctId} \ 
        \Mrules{a \mapsto d}{globalContext} \ 
       \Mrules{a \mapsto \textit{Global}}{ctLocation}\ 
     \\&  \Mrule{\_\  d\mapsto v \ \_}{ctStorage}
         \ 
    }
    {contractInstance}
     \}\\
     \mathbb{G}_g&(\gamma,x)=
      \{ (a,v) \mid \gamma \smatch 
       \Mrule{\ 
        \Mrule{x}{ctId} \ 
        \Mrules{a \mapsto d}{globalContext} \ 
    \\&    \Mrules{a \mapsto \textit{Global}}{ctLocation}\ 
      \Mrule{\_\  d\mapsto v \ \_}{ctStorage}
         \ 
    }
    {contractInstance}
     \}\\
     \mathbb{E}&(\gamma)= \{ (a , v) \mid \gamma \smatch 
     \Mrule{\ 
        \Mrule{a}{ctId} \ 
        \Mrule{v}{Balance}
        }{contractInstance}
        \}\\
        &\mathbb{S}(\gamma)=\{ x \mid \gamma \smatch \Mrule{ \Mrule{x}{ctId} }{contractInstance} \}\\
    \mathbb{M}&(\gamma)= (c,r,f,mem)\ \ \textit{iff}\ \ \gamma \smatch 
     \\&   \Mrule{\code{ListItem}(r) \code{ListItem}(c)}{contractStack}\ 
    \\&    \Mrule{\code{ListItem}(\#\code{state}(\_,f,\_,\_,\_)) }{functionStack}\ 
        \Mrule{0}{GasConsumption}\ 
    \\&  \Mrule{\ 
      \Mrule{r}{ctId} \  
        \Mrule{mem}{Memory}
         \ 
    }
    {contractInstance}
\end{align*}
\end{small}
\end{definition}

We also define information about the contexts in cell \textit{functionStack} as follows.
\begin{definition}(Contexts in function stacks).
    \label{def:contractStacks}
    \begin{align*}
    \mathrm{STACK}(\gamma)&=\{ s \mid \gamma \smatch\\& {\Mrule{\_ \ \code{ListItem}(\#\code{state}(s,\_,\_,\_,\_))  \_}{functionStack}}\}
    \end{align*}
\begin{align*}
    \mathbb{L}_s&(\gamma, s)=\{  (a,v) \mid \exists d. (a \mapsto d\ \in s) \land \gamma \smatch 
    \\
      \Mrule{\ 
          \ 
       &\Mrule{\_\ a \mapsto \textit{Local} \_}{ctLocation}\ 
        \Mrule{\_\ d\mapsto v \ \_}{Memory}
         \ 
    }
    {contractInstance}
     \}
\end{align*}
\begin{equation*}
    \mathbb{L}_A(\gamma)=\{ \mathbb{L}_s(\gamma, s)  \mid s \in \mathrm{STACK}(\gamma)  \}
\end{equation*}

\end{definition}
Here, an element in $\mathrm{STACK}(\gamma)$ is a copy of the cell \textit{ctContext}.
It represents the addressing information for both global variables and local variables, which is used for context switching.
Assume $\mathrm{STACK}(\gamma)=\emptyset$ just when a transaction starts.


\begin{definition}{(\myterms, \vars, \gvars, \evars)}
Given a fact $f$, define $\myterms(f)$ as the sequence of terms in $f$, define $\vars(f), \gvars(f), \evars(f) $ as the sequence of terms representing variables, global variables, and ether balances in $f$, respectively.
\end{definition}

We define formula $\alpha$ to set up the initial configuration and rules of $\ks$ to filter out the executions that we wish to discard. 
\begin{itemize}
    \item $\alpha_{\rho_1}$: $\rho_1$ starts a transaction. It also indicates that the environment, i.e., configuration, for executing a transaction has been prepared.
    For example, the contract accounts, i.e., $\Mrule{}{contractInstance}$, have been created and functions $\Mrule{}{function}$ have been loaded.
    \item $\alpha_{\rho_4}$:  The transaction is external.
    \item $\alpha_{\rho_e}$:  No error is in the trace.
    \item $\alpha_{\rho_n}$:  $\rho_n$  ends a transaction.
    \item $\alpha_{init}$: At the start of a transaction, local variables in all contract instances are empty, and the addressing information for global variables stored in \textit{globalContext} and \textit{ctContext} is the same.
\end{itemize}
\begin{align*}
     \alpha &\equiv \alpha_{\rho_1} \land \alpha_{\rho_4}\land \alpha_{\rho_e} \land \alpha_{\rho_n} \land \alpha_{init}   \\
     \alpha_{\rho_1} &\equiv  \forall [\rho_1,\dots,\rho_n] \in \tracek{\gamma_0}. \rho_1=\textsc{Function-Call}\\
     \alpha_{\rho_4} &\equiv  \forall [\rho_1,\dots,\rho_n] \in \tracek{\gamma_0}.
     \\& \quad  \rho_4=\textsc{Internal-Function-Call}\\
     \alpha_{\rho_e} &\equiv  \forall [\rho_1,\dots,\rho_n] \in \tracek{\gamma_0}. 
     \\& \quad  \forall \rho \in \{\rho_1,\dots,\rho_n\}. \rho \not= \textsc{Propagate-Exception-True}\\
     \alpha_{\rho_n} &\equiv  \forall [\rho_1,\dots,\rho_n] \in \tracek{\gamma_0}. 
     \\& \quad  \rho_n=\textsc{Propagate-Exception-False}\\
    \alpha_{init} &\equiv \forall x\in \mathbb{S}(\gamma_j).\mathbb{L}_g(\gamma_j,x)=\mathbb{L}(\gamma_j,x)=\emptyset \land\\& \mathbb{G}_g(\gamma_j,x)=\mathbb{G}(\gamma_j,x)
\end{align*}

We also define formula $\beta_{R_s}$ according to the restrictions of the independent model to filter out the executions of $R_s$ that we wish to discard.
\begin{itemize}
  \item $\beta_{r_n}$: $r_n$ ends a transaction.
  \item $\beta_{initE}$: The terms representing the ether balances are initialized only once.
  \item $\beta_{initG}$: The terms representing global variables are initialized only once for any contract. 
\end{itemize}




\begin{align*}
     \beta_{R_s}  &\equiv \beta_{r_n} \land \beta_{initE} \land \beta_{initG}   \\
     \beta_{r_n} &\equiv  \forall [r_1,\dots,r_n] \in \mathit{traces}^{msr}(R_s). r_n=\kwt{ret\_ext}\\
     \beta_{initE} &\equiv \forall \emptyset{\stackrel{r_{1}}{\longrightarrow}}_{R_s} \ldots{\stackrel{r_{n}}{\longrightarrow}}_{R_s} F_{n} \in \textit{exec}^{m s r}(R_s). \forall j \in \{1,\dots,n\}.\\&  \forall V,V' \ins F_{j}. 
     \  V \not= V' \land \textit{names}(\{V\})=\{\kwf{Evar} \}\ \rightarrow
     \\&  \textit{names}(\{V'\}) \not= \{\kwf{Evar} \}\\
     \beta_{initG} &\equiv \forall \emptyset{\stackrel{r_{1}}{\longrightarrow}}_{R_s} \ldots{\stackrel{r_{n}}{\longrightarrow}}_{R_s} F_{n} \in \textit{exec}^{m s r}(R_s). \forall j \in \{1,\dots,n\}.\\& \forall V,V' \ins F_{j}. 
        V \not= V' \land \textit{names}(\{V\})=\{\kwf{Gvar}^x \}\rightarrow\\& \textit{names}(\{V'\}) \not= \{\kwf{Gvar}^{x} \} \land \textit{names}(\{V'\}) \not= \{\kwf{Var}_i^{x} \}  \\
\end{align*}

Similarly, we define formulas $\beta_{R_{inv}}$ and  $\beta_{R_{equ}}$ according to the restrictions of the complementary models for the invariant property and the equivalence property, respectively.
\begin{itemize}
  \item $\beta_{r_n}'$: $r_n$ ends a transaction.
  \item $\beta_{initE}'$: The terms representing the ether balances are initialized only once.
  \item $\beta_{initG}'$: The terms representing global variables are initialized only once for any contract.
  \item $\beta_{start}$: Only one transaction is executed.
\end{itemize}
\begin{align*}
     \beta_{R_{inv}}  &\equiv \beta_{start} \land \beta_{r_n}' \land \beta_{initE}' \land \beta_{initG}'   \\
     \beta_{start} &\equiv  \forall [r_1,\dots,r_n] \in \mathit{traces}^{msr}(R_{inv}). \forall r,r' \in \{r_1,\dots,r_n\}\\& r\not=r' \land\ r=\kwt{ext\_call\_inv}\ \rightarrow\ r' \not= \kwt{ext\_call\_inv}\\
     \beta_{r_n}' &\equiv  \forall [r_1,\dots,r_n] \in \mathit{traces}^{msr}(R_{inv}). r_n=\kwt{ret\_ext\_inv}\\
     \beta_{initE}' &\equiv \forall \emptyset{\stackrel{r_{1}}{\longrightarrow}}_{R_{inv}} \ldots{\stackrel{r_{n}}{\longrightarrow}}_{R_{inv}} F_{n} \in \textit{exec}^{m s r}(R_{inv}).\\& \forall j \in \{1,\dots,n\}.  \forall V,V' \ins F_{j}. 
     \  V \not= V' \land \textit{names}(\{V\})=\\&\{\kwf{Evar} \}\ \rightarrow
       \textit{names}(\{V'\}) \not= \{\kwf{Evar} \}\\
     \beta_{initG}' &\equiv \forall \emptyset{\stackrel{r_{1}}{\longrightarrow}}_{R_{inv}} \ldots{\stackrel{r_{n}}{\longrightarrow}}_{R_{inv}} F_{n} \in \textit{exec}^{m s r}(R_{inv}).\\& \forall j \in \{1,\dots,n\}. \forall V,V' \ins F_{j}. 
        V \not= V' \land \textit{names}(\{V\})=\\&\{\kwf{Gvar}^x \}\rightarrow \textit{names}(\{V'\}) \not= \{\kwf{Gvar}^{x} \} \land\\& \textit{names}(\{V'\}) \not= \{\kwf{Var}_i^{x} \}  \\
\end{align*}

\begin{itemize}
  \item $\beta_{excAB}$: Two sequences $T_A$ and $T_B$ consisting of the same transactions are executed.
  \item $\beta_{comp}$: $r_n$ ends two sequences of transactions.
  \item $\beta_{initE_A}/\beta_{initE_B}$: The terms representing the ether balances are initialized only once.
  \item $\beta_{initG_A}/\beta_{initG_B}$: The terms representing global variables are initialized only once for any contract.
\end{itemize}
\begin{align*}
     \beta_{R_{equ}}  &\equiv \beta_{excAB} \land \beta_{comp} \land \beta_{initE_A} \land \beta_{initG_A}\land \beta_{initE_B} \land \beta_{initG_B}   \\
     \beta_{excAB} &\equiv \forall \emptyset{\stackrel{r_{1}}{\longrightarrow}}_{R_{equ}}F_1 \ldots{\stackrel{r_{n}}{\longrightarrow}}_{R_{equ}} F_{n} \in \textit{exec}^{m s r}(R_{equ}). 
     \\& \{\myterms(V) \mid V\ins F_n \land \textit{names}(\{V\})  = \{\kwf{Call_{Ae}}\}  \}^{\#} = \\& \{\myterms(V) \mid V\ins F_n  \land \textit{names}(\{V\}) = \{\kwf{Call_{Be}}\}  \}^{\#}\\
     \beta_{comp} &\equiv  \forall [r_1,\dots,r_n] \in \mathit{traces}^{msr}(R_{equ}). r_n=\kwt{comp\_AB}\\
     \beta_{initE_A} &\equiv \forall \emptyset{\stackrel{r_{1}}{\longrightarrow}}_{R_{equ}}F_1 \ldots{\stackrel{r_{n}}{\longrightarrow}}_{R_{equ}} F_{n} \in \textit{exec}^{m s r}(R_{equ}).\\& \forall j \in \{1,\dots,n\}.  \forall V,V' \ins F_{j}. 
     \  V \not= V' \land \textit{names}(\{V\})\\&= \{\kwf{Evar_A} \}\ \rightarrow
       \textit{names}(\{V'\}) \not= \{\kwf{Evar_A} \}\\
     \beta_{initG_A} &\equiv \forall \emptyset{\stackrel{r_{1}}{\longrightarrow}}_{R_{equ}}F_1 \ldots{\stackrel{r_{n}}{\longrightarrow}}_{R_{equ}} F_{n} \in \textit{exec}^{m s r}(R_{equ}).\\& \forall j \in \{1,\dots,n\}. \forall V,V' \ins F_{j}. 
        V \not= V' \land \textit{names}(\{V\})\\&=\{\kwf{Gvar_A}^x \}\rightarrow \textit{names}(\{V'\}) \not= \{\kwf{Gvar_A}^{x} \} \land\\& \textit{names}(\{V'\}) \not= \{\kwf{Var}_{\kwf{A}i}^{x} \}  \\
     \beta_{initE_B} &\equiv \forall \emptyset{\stackrel{r_{1}}{\longrightarrow}}_{R_{equ}}F_1 \ldots{\stackrel{r_{n}}{\longrightarrow}}_{R_{equ}} F_{n} \in \textit{exec}^{m s r}(R_{equ}).\\& \forall j \in \{1,\dots,n\}.  \forall V,V' \ins F_{j}. 
     \  V \not= V' \land \textit{names}(\{V\})\\&=\{\kwf{Evar_B} \}\ \rightarrow
       \textit{names}(\{V'\}) \not= \{\kwf{Evar_B} \}\\
     \beta_{initG_B} &\equiv \forall \emptyset{\stackrel{r_{1}}{\longrightarrow}}_{R_{equ}}F_1 \ldots{\stackrel{r_{n}}{\longrightarrow}}_{R_{equ}} F_{n} \in \textit{exec}^{m s r}(R_{equ}).\\& \forall j \in \{1,\dots,n\}. \forall V,V' \ins F_{j}. 
        V \not= V' \land \textit{names}(\{V\})\\&=\{\kwf{Gvar_B}^x \}\rightarrow \textit{names}(\{V'\}) \not= \{\kwf{Gvar_B}^{x} \} \land\\& \textit{names}(\{V'\}) \not= \{\kwf{Var}_{\kwf{B}i}^{x} \}  \\
\end{align*}

When a rule in $\rs$ is applied (except $\kwt{fresh}$ rule), the variable terms of facts on the right-hand side of the rule are substituted by constant terms or fresh names.
Therefore, the produced multiset of facts by the rule has no variables terms, which can correspond to a configuration in $\ks$, if the concrete terms and fresh names are assigned with values of data type in $\ks$.
Hence, we build the relationship between the terms and the corresponding values in $\ks$ as follows.
\begin{definition}(Valuation $\mathcal{E}$).
    \label{def:valuation}
\textit{
    Given a term $x$ in a fact, let valuation $\mathcal{E}(x)$ output $x$'s value of data types in $\ks$.
    If $x$ is a fresh name or constant term, the value is assigned beforehand;
    otherwise, if $x$ is a variable term, after the terms in its fact are substituted (i.e., a rule, where the fact is on the right-hand side, is applied), it is evaluated according to the substituted constant terms and fresh names.
}
\end{definition}
The valuation can be seen as a process of calculating the value of a variable term, based on the assignment of the constant terms and fresh names.

    

\begin{definition}{($\bi$).}
\textit{
Let $A$ be a 2-tuple derived from a configuration  $\gamma$ in $\ks$.
Let $B$ be a sequence of terms obtained from $\vars$ in $\rs$.
We write $A \bi B$, if there exists a bijection between $A$ and $B$ such that whenever $ (a,v) \in A$ is mapped to  $x \in B $, we have that
$$ v =\mathcal{E}(x)\ $$
When $A \bi B$, $ (a,v) \in A$ and  $x \in B $, we also write $(a,v) \bi x$ if this bijection maps $(a,v)$ to $x$.}
\end{definition}


\begin{remark}
    \label{remark:bi}
    Note that $\bi$ has the following properties. 
    \begin{itemize}
        \item If $A_1 \bi B_1$ and $A_2 \bi B_2$, then $A_1 \cup A_2 \bi B_1 \cup B_2$ 
        \item If $A_1 \bi B_1$, and $a \bi b$ for $a\in A_1$ and $b \in B_1$,then $A_1 \backslash \{a\} \bi B_1 \backslash \{b\}$ 
    \end{itemize}
    
\end{remark}

Recall that the first and fourth parameter in fact $\kwf{Gvar}$ and $\kwf{Var}$ represents its contract name, respectively.
To differentiate the facts, we denote $\kwf{Gvar}^x$ and $\kwf{Var}^x$ as the fact for the contract with \textit{ctId} $x$ (corresponding to the first and fourth parameter of the fact, respectively) in $K_s$; we also omit the tag $x$, i.e., denote the facts as  $\kwf{Gvar}$ and $\kwf{Var}$, if $x$ is the current running contract in the corresponding $K_s$ (in the following proofs, a state in $\ks$ always corresponds to a multiset of facts). 
Though in most cases of the proof the tag is not used, it is useful when analyzing the case related to context switching.

\begin{lemma}
\label{lemma:equalc}
Given  function (transaction) code  $c$,
define $C_r(c,i)$ as the first parameter of $\mathcal{R}$ when $\mathcal{R}$ is recursively applied on $c$ and $i$ is the second parameter of $\mathcal{R}$;
define $C_p(c,i)$ as the third element of a state when the state is at a transition running $c$ and $i$ is the second parameter of the state.

 Then we have 
 \begin{equation*}
    \vDash \forall i,c.  C_r(c,i)=C_p(c,i)
 \end{equation*}
\end{lemma}
\begin{proof}
    We proceed by induction over the number of state transitions.

If a function $c_0$ is running, and a sequence of state transitions on running $c_0$ is
$$s_0 \Rightarrow s_1 \Rightarrow \dots  $$
    where $s_j=(\gamma_j,i_j,c_j)$.

\textit{Base case}.
For $s_0$
\begin{enumerate}
    \item By Table~\ref{table:tableCorrespondence}, $T(s_0)=1$ and $i_0=\emptyset$.
    \item By Definition of $\mathcal{R}$, $C_r(c_0,i_0)=C_r(c_0,\emptyset)=c_0$
    \item By Definition of $C_p$ and Table~\ref{table:tableCorrespondence}, $C_p(c_0,\emptyset)=c_0$
    \item By (2),(3), $C_r(c_0,\emptyset)=C_p(c_0,\emptyset)$
\end{enumerate}

\textit{Inductive Step}. Assume the invariants hold for $i_j$. We have to show that the lemma holds for the successors of $i_j$.

\textbf{Case}: $T(s_{j})=1$, 
\begin{enumerate}
    \item  By Table~\ref{table:tableCorrespondence}, denote $C_p(c_0,i_{j})=\texttt{function}\ f(\texttt{d})\{\kwp{stmt}\}$
    \item By (1) and inductive hypothesis, 
    \begin{equation*}
      C_r(c_0,i_{j})=\texttt{function}\ f(\texttt{d})\{\kwp{stmt}\}
    \end{equation*}
    \item By (1), $i_j=\emptyset$
    \item By (1), (3), and Table~\ref{table:tableCorrespondence}, $C_p(c_0,1)=\kwp{stmt}$
    \item By (2) and Definition of $\mathcal{R}$, 
    \begin{equation*}
        C_r(c_0, 1)=\kwp{stmt}
    \end{equation*}
    \item By (4),(5), $C_r(c_0,1)=C_p(c_0,1)$
\end{enumerate}

\textbf{Case}: $T(s_{j})=2\dots 14$, the proof of the cases is similar to the case $T(s_{j})=1$, and we omit the proof for the cases.
\end{proof}
\vspace*{0.3cm}



\begin{lemma}
    \label{lemma:var_soundness}
    Let   $(\mathcal{T},\karroww, \gamma_0)$ be a K transition system that satisfies $\ks \vDash \alpha$. 
    Let $R_s \vDash \beta_{initE} \land \beta_{initG}$.
    If $$
            \lambda_0 \newkarrowx{1} \lambda_1 \newkarrowx{2} \dots \newkarrowx{n} \lambda_n
$$ 
    where $[\rho_1,\dots,\rho_n] \in \traceks{\gamma_0}$, 
    and it corresponds to a sequence of state transitions according to Proposition~\ref{proposition:transitionsystem}: 
    $$s_0 \Rightarrow s_1 \Rightarrow \dots \Rightarrow s_m $$
    where $s_j=(\gamma_j,i_j,c_j)$, then there are $(r_1,F_1), (r_2,F_2),\dots,(r_{m'},F_{m'}) $, such that
    \begin{equation*}
        \emptyset \exer{r_1} F_1 \exer{r_2} \dots \exer{r_{m'}} F_{m'} 
    \end{equation*}
    and there exists a valuation $\mathcal{E}$ and a monotonic, strictly increasing function $g: \{ 0,\dots,m\}$ $\rightarrow \{ 0,\dots,m'\} $ such that $g(m)=m'$ and for all $j\in \{ 0,\dots,m\} $
    \begin{enumerate}[label=(\alph*)]
        \item $F_c(i_j,T(s_j)) \subseteq  \textit{names}(F_{g(j)}) $
        \item $T(s_j)= 1  \rightarrow  \exists V\ins F_{g(j)}. \\
            \quad \quad \quad \quad  \mathbb{G}(\gamma_j) \bi \gvars(V) \land \textit{names}(\{V\})=\{\kwf{Gvar} \}  $
        \item  $  T(s_j)> 1  \rightarrow \exists i. \exists V\ins F_{g(j)}.  \\
            \quad \quad \quad \quad  \mathbb{G}(\gamma_j) \bi \gvars(V) \land \textit{names}(\{V\})=\{\kwf{Var}_i\}  $
        \item $T(s_j)= 1  \rightarrow  \exists V\ins F_{g(j)}. \\
            \quad \quad \quad \quad  \mathbb{E}(\gamma_j) \bi \evars(V) \land \textit{names}(\{V\})=\{\kwf{Evar} \}  $
        \item $T(s_j)> 1  \rightarrow \exists i. \exists V\ins F_{g(j)}.  \\
            \quad \quad \quad \quad  \mathbb{E}(\gamma_j) \bi \evars(V) \land \textit{names}(\{V\})=\{\kwf{Var}_i\}  $
        \item $T(s_j)= 1  \rightarrow \mathbb{L}(\gamma_j) = \emptyset $
        \item$T(s_j)> 1  \rightarrow \exists i. \exists V\ins F_{g(j)}.  \\
            \quad \quad \quad \quad  \mathbb{L}(\gamma_j) \bi \vars(V)\backslash\gvars(V)\backslash \land \textit{names}(\{V\})$ $=\{\kwf{Var}_i\}  $
        \item $\forall x\in \mathbb{S}(\gamma_j). \mathbb{G}_g(\gamma,x)=\mathbb{G}(\gamma,x)$
        \item $\forall x\in \mathbb{S}(\gamma_j).\mathbb{L}_g(\gamma_j,x)=\emptyset$
        \item $\forall l \in \mathbb{L}_A(\gamma_j). \exists i. \exists V\ins F_{g(j)}. \\ 
            \quad \quad \quad \quad  l \bi \vars(V)\backslash\gvars(V) \land \textit{names}(\{V\})=\{\kwf{Var}_i\}$
        \item $\forall x\in \mathbb{S}(\gamma_j). \exists V\ins F_{g(j)}. \mathbb{G}(\gamma_j,x) \bi \gvars(V) $.
    \end{enumerate}
\end{lemma}

\begin{proof}
    We proceed by induction over the number of state transitions $k$.
    Denote $\iota(\gamma)$ as the contract ID of on top of cell \textit{contractStack} of $\gamma$.
    Note that to achieve readability of the proof (besides strictness), we omit the details of a derivation if the derivation is similar to one that has been illustrated for another conclusion.

    \textit{Base case}. 
    For $k=0$, we let $g(0)=t+2$, if there are $t$ contracts translated by $\mathcal{R}$.  
    Choose a function of a contract to run and let $F_1,F_2,F_3$ be the multiset obtained by using the rule $\kwt{init\_evars}$, $\kwt{init\_gvars}$, $\kwt{ext\_call}$ in order respectively. 

    \begin{enumerate}
        \item By Definition~\ref{lemma:trans} and Lemma~\ref{lemma:simpletrans}
    \begin{align*}
        F_1=&\{ \kwf{Evar}(e(\omega_0))\}\\
        F_2=& F_1\cups \{\kwf{Gvar}(\llbracket \omega_0[1] \rrbracket \cdot g(\omega_0)\backslash e(\omega_0)) \} \\
        &=\{ \kwf{Evar}(e(\omega_0)), \kwf{Gvar}(\llbracket \omega_0[1] \rrbracket \cdot g(\omega_0)\backslash e(\omega_0)) \}\\
        F_3=& F_2 \cups \{ \kwf{Call_e}( \llbracket\omega_0[1],\sigma_a(f), \sigma_v(c_b)\rrbracket \cdot\sigma(seq(\texttt{d})))\\
        =&\{ \kwf{Evar}(e(\omega_0)), \kwf{Gvar}(\llbracket \omega_0[1] \rrbracket \cdot g(\omega_0)\backslash e(\omega_0)), 
        \\&\kwf{Call_e}( \llbracket\omega_0[1],\sigma_a(f), \sigma_v(c_b)\rrbracket \cdot\sigma(seq(\texttt{d})))
        \}
    \end{align*}
        Similarly, by the definition of $\beta_{initG}$, rule $\kwt{init\_gvars}$ is applied to other contracts for $t-1$ times. 
        For instance, for contract $x$, facts $\kwf{Gvar}^x$ is added to the current multiset of facts. 
        Finally, $F_{t+2}$ is generated.
    
    \item By Definition of $\textit{names}$
         $$\{\kwf{Evar}, \kwf{Gvar}, \kwf{Call_e} \} =names(F_3)\subseteq names(F_{t+2})$$
    \item By Table~\ref{table:tableCorrespondence}, $i_0=\emptyset, T(s_0)=1$,$$F_c(i_0,T(s_0))=\{\kwf{Evar}, \kwf{Gvar}, \kwf{Call_e}\}$$
    \item By (2),(3), $F_c(i_0,T(s_0)) \subseteq  names(F_{t+2}) $, \textit{condition (\textbf{a})} proved.
    \item By (3), $T(s_0) =1$, so \textit{condition (\textbf{c}), (\textbf{e}), (\textbf{g})} hold trivially.
    \item  By $\alpha_{init}$, \textit{Condition $(\textbf{b}),(\textbf{d}),(\textbf{f}), (\textbf{j}), (\textbf{k})$} hold by interpreting $\mathcal{E}$ such that
    \begin{itemize}
        \item $\mathbb{G}(\gamma_0) \bi \vars( \kwf{Gvar})$
        \item $\mathbb{E}(\gamma_0) \bi \vars(\kwf{Evar})$
        \item $\mathbb{L}(\gamma_0) \bi \emptyset$
        \item $\forall x\in \mathbb{S}(\gamma_0).\mathbb{L}(\gamma_0,x)=\emptyset \bi \emptyset$
        \item $\forall x\in \mathbb{S}(\gamma_0). \mathbb{G}(\gamma_0,x)\bi \gvars(\kwf{Gvar}^x)$
    \end{itemize}
    \item By $\alpha_{init}$, \textit{condition (\textbf{h}), (\textbf{i})} holds.
    \end{enumerate}
    
    \textit{Inductive step}. Assume the invariant holds for $k\ge 0$. We have to show that the lemma holds for $k+1$ transitions.
    $$s_0 \Rightarrow s_1 \Rightarrow \dots \Rightarrow s_{k}  \Rightarrow s_{k+1} $$

    \begin{enumerate}[label=(\Alph*)]
        \item By induction hypothesis, we have that there exists a monotonic increasing function $g$ and an execution 
            \begin{equation*}
                \emptyset \exer{r_1} F_1 \exer{r_2} \dots \exer{r_{k'}} F_{k'} 
            \end{equation*}
        such that the conditions hold and $g(k)=k'$.
        \item By Proposition~\ref{proposition:transitionsystem} and Definition \ref{def:Soliditystate}, 
        there exists a strictly monotonically increasing function $f$ such that $\gamma_j= \lambda_{f(j)}$ and $f(0)=0$.
        \item By eliminating $\exists$ on $g$ and $f$, we use the function $g_0$ and $f_0$, respectively.
        \item By (B), $\gamma_k=\lambda_{f_0(k)} $
        \item By (D) and Definition~\ref{def:ktraces}, in $\traceks{\gamma_0}$ the segment of traces  from $s_k$ to $s_{k+1}$ is  
        \begin{equation*}
            [\rho_{f_0(k)+1}, \dots, \rho_{f_0(k+1)}]
        \end{equation*}
        and
        \begin{equation*}
            \lambda_{f_0(k)} \karrowx{f_0(k)+1} \lambda_{f_0(k)+1} \karrowx{f_0(k)+2} \dots \karrowx{f_0(k+1)} \lambda_{f_0(k+1)}
        \end{equation*}
    \end{enumerate}
    
    We now proceed by case distinction over the type of transitions from $s_k$ to $s_{k+1}$.
    We will extend the previous executions by a number of steps, say \textit{step}, from $F_{k'}$ to some $F_{k'+\textit{step}}$, and prove that the conditions hold for $k+1$, and a function $g$, defined as follows:

    \begin{equation*}
    g(j):=
    \begin{cases}
        g_{0}(j) & \textrm{if } i\in \{0,...,k\} \\
        g_{0}(k)+\textit{step} & \textrm{if } i=k+1
    \end{cases}
    \end{equation*}
    
    \textbf{Case:} $T(s_k)=1$.

    \begin{enumerate}
        \item By Definition~\ref{def:csucc} and Lemma~\ref{lemma:subsetsucc}, there is only one possible sub-trace $[\rho_{f_0(k)+1},\dots,\rho_{f_0(k)+t+11}]$ for the state transition in (E) as follows:
        \begin{itemize}
            \item $\rho_{f_0(k)+1}= \textsc{Function-Call}$ 
            \item $\rho_{f_0(k)+2}= \textsc{Switch-Context}$ 
            \item $\rho_{f_0(k)+3}= \textsc{Create-Transaction}$ 
            \item $\rho_{f_0(k)+4}= \textsc{Internal-Function-Call}$ 
            \item $\rho_{f_0(k)+5}= \textsc{Save-Cur-Context}$ 
            \item $\rho_{f_0(k)+6}= \textsc{Call}$ 
            \item $\rho_{f_0(k)+7}= \textsc{Init-Fun-Params}$ 
            \item $\rho_{f_0(k)+8}= \textsc{Bind-Params}$ 
            \item $\rho_{f_0(k)+9}= \textsc{Bind-Params}$ 
            \item ...
            \item $\rho_{f_0(k)+t+9}= \textsc{Bind-Params-End}$ 
            \item $\rho_{f_0(k)+t+10}= \textsc{processFunQuantifiers}$ 
            \item $\rho_{f_0(k)+t+11}= \textsc{Call-Function-Body}$ 
            \item (---------End of state transition---------)
            \item $\rho_{f_0(k)+t+12}= \textsc{Function-Body}$ 
            \item ...
            \item (---------End of statements in function---------)
            \item $\rho_{f_0(k)+t+u+12}= \textsc{Update-Cur-Context}$ 
            \item $\rho_{f_0(k)+t+u+13}= \textsc{Return-Context}$ 
            \item $\rho_{f_0(k)+t+u+14}= \textsc{Clear-Recipient-Context}$ 
            \item $\rho_{f_0(k)+t+u+15}= \textsc{Clear-Caller-Context}$ 
            \item $\rho_{f_0(k)+t+u+16}= \textsc{Propagate-Exception-False}$ 
            \item (---------End of function---------)
        \end{itemize}
        Note that by $\alpha_{\rho_e}$, we exclude the possible trace such that $\rho_{f_0(k)+t+u+16}= \textsc{Propagate-Exception-True}$, which means there are exceptions and the transaction is reverted.
        \item By Proposition~\ref{proposition:transitionsystem}, in (E), the trace $[\rho_{f_0(k)+1},\dots,\rho_{f_0(k)+t+11}]$ in~(1) corresponds to the state transition from $s_{k}$ to $s_{k+1}$. 
        Here,  $t$ is assumed to be the number of parameters in the function. $u$ represents the number of steps on executing the statements and nested functions.
        \item By $T(s_k)$, $i_k=\emptyset$, and let $c_k=\texttt{function}\ f_{c}(\texttt{d})\{\kwp{stmt}\}$
        \item By (3) and Lemma~\ref{lemma:equalc},
            \begin{equation*}
                C_r(c_k,i_k)=C_p(c_k,i_k)=c_k=\texttt{function}\ f_{c}(\texttt{d})\{\kwp{stmt}\}
            \end{equation*}
        \item By (4), Lemma~\ref{lemma:equalc}  and Definition of $\mathcal{R}$, rule $\kwt{recv\_ext}$ is generated.
        \item By inductive hypothesis of condition (a), 
        \begin{small}
        \begin{equation*}
            \{\kwf{Call}_{e},\kwf{Gvar},\kwf{Evar}\} =F_c(\emptyset,1)=F_c(i_k,T(s_k)) \subseteq  \textit{names}(F_{g(k)}) 
        \end{equation*}
        \end{small}
        \item By (5), (6), construct $F_{g(k)+1}$ by applying rule $\kwt{recv\_ext}$, and by Definition of \textit{name}: (It is similar to condition (a).(1)(2) of Base case, and we omit the details.)
        \begin{equation*}
              \{\kwf{Var}_{1}\} \subseteq \textit{names}(F_{g(k)+1})
        \end{equation*}
        \item By (4) and Table~\ref{table:tableCorrespondence},
        \begin{equation*}
            c_{k+1}=\kwp{stmt}, i_{k+1}=1
        \end{equation*}
        \item By (8) and Table~\ref{table:tableCorrespondence}, $ T(s_{k+1})> 1$
        \item By (8), (9) and Table~\ref{table:tableCorrespondence}, $F_c(i_{k+1},T(s_{k+1}))=\{\kwp{Var_1}\}$.
        \item Let $\textit{step}=1$, i.e., $g(k+1)=g(k)+1$
        \item By (7), (10), (11)
        \begin{equation*}
            F_c(i_{k+1},T(s_{k+1})) = \{\kwp{Var_1}\} \subseteq \textit{names}(F_{g(k+1)})
        \end{equation*}
        \textit{Condition (\textbf{a}) proved.}
        \item By (B), $\mathbb{G}(\lambda_{f_0(k)})=\mathbb{G}(\gamma_k)$
        \item By $T(s_k)$, $T(s_k)=1$
        \item By (14) and inductive hypothesis, $\mathbb{G}(\gamma_k) \bi  \textit{gvars}(\kwp{Gvar})$ 
        \item By (5), $\textit{gvars}(\kwp{Gvar})= \textit{gvars}(\kwp{Var}_1)$ 
        \item By assumption of $f$, $\lambda_{f_0(k)+t+11}=\gamma_{k+1}$
        \item By (2), inductive hypothesis (h), and Definition~\ref{def:k-rule-k-system},\ref{def:LGE}
        \begin{small}
        \begin{equation*}
            \mathbb{G}(\lambda_{f_0(k)+t+11})=\dots=\mathbb{G}(\lambda_{f_0(k)+2})= \mathbb{G}_b(\lambda_{f_0(k)}) =\mathbb{G}(\lambda_{f_0(k)})
        \end{equation*}
        \end{small}
        \item By (15), (16), (18), $ \mathbb{G}(\lambda_{f_0(k)+t+11}) \bi \textit{gvars}(\kwp{Var}_1) $
        \item By (17), (19), $ \mathbb{G}(\gamma_{k+1}) \bi \textit{gvars}(\kwp{Var}_1) $
        \item By (9), (20), \textit{condition \textbf{(b), (c)} hold}.
        \item Similar to (15), $\mathbb{E}(\gamma_k) \bi  \textit{evars}(\kwp{Evar})$
        \item Similar to (16), $\textit{evars}(\kwp{Evar})= \textit{evars}(\kwp{Var}_1)$
        \item By Definition~\ref{def:k-rule-k-system},\ref{def:LGE},
        \begin{equation*}
            \mathbb{E}(\lambda_{f_0(k)+t+11})=\dots =\mathbb{E}(\lambda_{f_0(k)})
        \end{equation*}
        \item Similar to (21), \textit{condition (\textbf{d}),(\textbf{e}) hold}
        \item By Definition~\ref{def:k-rule-k-system}, and inductive hypothesis (i)
        \begin{equation*}
            \mathbb{L}(\lambda_{f_0(k)+7})=\dots=\mathbb{L}(\lambda_{f_0(k)+2})=\mathbb{L}_b(\lambda_{f_0(k)})=\emptyset
        \end{equation*}
        \item By Definition~\ref{def:valuation}, let valuation $\mathcal{E}(\sigma(seq(\texttt{d})))=\texttt{d}$
        \item By (26), (27) and Definition~\ref{def:valuation}, let 
        \begin{equation*}
            \mathbb{L}(\lambda_{f_0(k)+t+9})\bi \sigma(seq(\texttt{d}))
        \end{equation*}
        \item By Definition~\ref{def:k-rule-k-system}, 
            $\mathbb{L}(\lambda_{f_0(k)+t+11})=\dots=\mathbb{L}(\lambda_{f_0(k)+t+9})$
        \item By  (7), generated rule $\kwt{recv\_ext}$ satisfies:
         $$\vars(\kwp{Var_1})\backslash\gvars(\kwp{Var_1})=\sigma(seq(\texttt{d}))$$
        \item Similar to (21), by (28), (29), (30),  \textit{condition (\textbf{f}) (\textbf{g})} \textit{hold}.
        \item By Definition~\ref{def:k-rule-k-system}, \textit{condition (\textbf{h}), (\textbf{i})} \textit{hold} trivially.
        \item By Definition~\ref{def:contractStacks} and \ref{def:k-rule-k-system}, 
         \begin{small}
        \begin{equation*}
            \mathbb{L}_A(\lambda_{f_0(k)+t+11})=\dots=\mathbb{L}(\lambda_{f_0(k)+2})=\mathbb{L}_A(\lambda_{f_0(k)})\cup \{\emptyset\}=\{\emptyset\}
        \end{equation*}
         \end{small}
        \item By (2), (17), (33), let $\emptyset \bi \emptyset$, \textit{condition (\textbf{$j$}}) holds.
        \item By (12), (20) and inductive hypothesis (k), \textit{condition (\textbf{k})} hold.
    \end{enumerate}

    \textbf{Case:} $T(s_k)=2$.
    \begin{enumerate}
        \item Similar to (1)(2) in case $T(s_k)=1$, (abbreviated as $Case_1.1\sim2$), the sub-trace for the state transition in (E) is as follows:
        \begin{itemize}
            \item $\rho_{f_0(k)+1}= \textsc{Function-Call}$ 
            \item $\rho_{f_0(k)+2}= \textsc{Switch-Context}$ 
            \item $\rho_{f_0(k)+3}= \textsc{Create-Transaction}$ 
            \item $\rho_{f_0(k)+4}= \textsc{Nested-Function-Call}$ 
            \item $\rho_{f_0(k)+5}= \textsc{Call}$ 
            \item $\rho_{f_0(k)+6}= \textsc{Init-Fun-Params}$ 
            \item $\rho_{f_0(k)+7}= \textsc{Bind-Params}$ 
            \item $\rho_{f_0(k)+8}= \textsc{Bind-Params}$ 
            \item ...
            \item $\rho_{f_0(k)+t+8}= \textsc{Bind-Params-End}$ 
            \item $\rho_{f_0(k)+t+9}= \textsc{processFunQuantifiers}$ 
            \item $\rho_{f_0(k)+t+10}= \textsc{Call-Function-Body}$ 
            \item (---------End of state transition---------)
            \item $\rho_{f_0(k)+t+11}= \textsc{Function-Body}$ 
            \item ...
            \item (---------End of statements in function---------)
            \item $\rho_{f_0(k)+t+u+11}= \textsc{Return-Context}$ 
            \item $\rho_{f_0(k)+t+u+12}= \textsc{Clear-Recipient-Context}$ 
            \item $\rho_{f_0(k)+t+u+13}= \textsc{Clear-Caller-Context}$ 
            \item $\rho_{f_0(k)+t+u+14}= \textsc{Propagate-Exception-False}$ 
            \item (---------End of function---------)
        \end{itemize}

        \item Similar to $Case_1.2$,  $[\rho_{f_0(k)+1},\dots,\rho_{f_0(k)+t+10}]$ in~(1) corresponds to the state transition from $s_{k}$ to $s_{k+1}$.
        Here,  $t$ represents the number of parameters in the function. $u$ represents the number of steps on executing the statements.
        \item Similar to $Case_1.3$, let $c_k=x(c_{x}).f_{x}(p);\kwp{stmt}$
        \item Let the caller and recipient ID be $id_c$, $id_r$, respectively.
        \item Similar to $Case_1.4$, $C_r(c_k,i_k)=x(c_{x}).f_{x}(p);\kwp{stmt}$
        \item Similar to $Case_1.5$, rule  $\kwt{in\_call}$ is generated.
        \item By Definition of $\mathcal{R}$, the first rule for $f_x(p)$ is generated.
        \item Similar to $Case_1.6$, $\{\kwf{Var}_{i}^{id_c}\} \subseteq \textit{names}(F_{g(k)})$
        \item Similar to $Case_1.7$, $\{\kwf{Call}_{in}, \kwf{Var}_{i\circ 1}^{id_c}, \kwf{Gvar}^{id_c}\} \subseteq \textit{names}(F_{g(k)+1})$
        \item Similar to $Case_1.7$, by (7), $\{\kwf{Var}_{1}^{id_r}\} \subseteq \textit{names}(F_{g(k)+2})$
        \item Let \textit{step}=1, i.e., $g(k+1)=g(k)+2$
        \item Similar to $Case_1.12$, \textit{Condition (\textbf{a})} proved.
        \item By (B), $\gamma_k=\lambda_{f_0(k)} $
        \item Similar to $Case_1.15$, $\mathbb{G}(\gamma_k, id_c) \bi \textit{gvars}(\kwf{Var}_{i}^{id_c})$
        \item By inductive hypothesis (k), we have $V$, such that  $$\mathbb{G}(\gamma_k, id_r) \bi \textit{gvars}(V)$$
        \item By assumption of $f$, $\lambda_{f_0(k)+t+10}=\gamma_{k+1}$
        \item Similar to $Case_1.18$,  
         \begin{small}
        \begin{equation*}
            \mathbb{G}(\lambda_{f_0(k)+t+10})=\dots=\mathbb{G}(\lambda_{f_0(k)+2})= \mathbb{G}_g(\gamma_k, id_r)= \mathbb{G}(\gamma_k, id_r) 
        \end{equation*}
         \end{small}
        \item For caller $id_c$, by (5), (6),  $\mathbb{G}(\gamma_{k+1},id_c) \bi  \gvars(\kwf{Gvar}^{id_c}) $
        \item For recipient $id_r$, by  (10), (15), (17),  
        $$\mathbb{G}(\gamma_{k+1}, id_r)=\mathbb{G}(\gamma_{k+1})= \mathbb{G}(\lambda_{f_0(k)+t+10})$$
        $$\mathbb{G}(\gamma_{k+1}, id_r) \bi \kwf{Var}_{1}^{id_r}$$

        \item Similar to $Case_1.21$, by (19), \textit{condition \textbf{(b), (c)} hold}.
        \item Similar to $Case_1.25$, \textit{condition (\textbf{d}),(\textbf{e}) hold}.
        \item Similar to $Case_1.31$, \textit{condition (\textbf{f}) (\textbf{g})} \textit{hold}.
        \item By Definition~\ref{def:k-rule-k-system}, \textit{condition (\textbf{h}), (\textbf{i})} \textit{hold} trivially.
        \item Similar to $Case_1.33$, 
         \begin{small}
        \begin{equation*}
            \mathbb{L}_A(\lambda_{f_0(k)+t+10})=\dots=\mathbb{L}(\lambda_{f_0(k)+2})=\mathbb{L}_A(\lambda_{f_0(k)})\cup \{\mathbb{L}(\gamma_k)\}
        \end{equation*}
         \end{small}
        \item Similar to $Case_1.20$, by (8), 
        \begin{equation*}
            \mathbb{L}(\gamma_k) \bi \vars(\kwf{Var}_{i}^{id_c})\backslash\gvars(\kwf{Var}_{i}^{id_c})
        \end{equation*}
        \item By (5), (6), $\kwf{Var}_{i\circ1}^{id_c} \in \textit{names}(F_{g(k+1)}) $ and
        \begin{equation*}
            \vars(\kwf{Var}_{i}^{id_c})\backslash\gvars(\kwf{Var}_{i}^{id_c})= 
            \vars(\kwf{Var}_{i\circ 1}^{id_c})\backslash\gvars(\kwf{Var}_{i\circ 1}^{id_c}) 
        \end{equation*}
        \item By (16), (24), (25), (26), \textit{condition (\textbf{j})} holds.
        \item Similar to $Case_1.35$, by (10), (18), (19), \textit{condition (\textbf{k})} holds.
    \end{enumerate}

    \textbf{Case:} $T(s_k)=3$.
    \begin{enumerate}
        \item Similar to (1), (2) in case $T(s_k)=1$, the possible sub-traces for the state transition in (E) are as follows:

        Sub-trace 1:
        \begin{itemize}
            \item $\rho_{f_0(k)+1}= \textsc{Function-Body}$ 
            \item $\rho_{f_0(k)+2}= \textsc{Exe-Statement-Main-Contract}$ 
            \item $\rho_{f_0(k)+3}= \textsc{Write}$ 
            \item $\rho_{f_0(k)+4}= \textsc{WriteAddress-GlobalVariables}$ 
            \item $\rho_{f_0(k)+5}= \textsc{Gas-Cal}$ 
            \item (---------End of state transition---------)
            \item $\rho_{f_0(k)+6}= \textsc{Function-Body}$ 
        \end{itemize}

        Sub-trace 2:
        \begin{itemize}
            \item $\rho_{f_0(k)+1}= \textsc{Function-Body}$ 
            \item $\rho_{f_0(k)+2}= \textsc{Exe-Statement-Main-Contract}$ 
            \item $\rho_{f_0(k)+3}= \textsc{Write}$ 
            \item $\rho_{f_0(k)+4}= \textsc{WriteAddress-LocalVariables}$ 
            \item $\rho_{f_0(k)+5}= \textsc{Gas-Cal}$ 
            \item (---------End of state transition---------)
            \item $\rho_{f_0(k)+6}= \textsc{Function-Body}$ 
        \end{itemize}

        Sub-trace 3:
        \begin{itemize}
            \item $\rho_{f_0(k)+1}= \textsc{Function-Body}$ 
            \item $\rho_{f_0(k)+2}= \textsc{Exe-Statement}$ 
            \item $\rho_{f_0(k)+3}= \textsc{Write}$ 
            \item $\rho_{f_0(k)+4}= \textsc{WriteAddress-GlobalVariables}$ 
            \item $\rho_{f_0(k)+5}= \textsc{Gas-Cal}$ 
            \item (---------End of state transition---------)
            \item $\rho_{f_0(k)+6}= \textsc{Function-Body}$ 
        \end{itemize}

        Sub-trace 4:
        \begin{itemize}
            \item $\rho_{f_0(k)+1}= \textsc{Function-Body}$ 
            \item $\rho_{f_0(k)+2}= \textsc{Exe-Statement}$ 
            \item $\rho_{f_0(k)+3}= \textsc{Write}$ 
            \item $\rho_{f_0(k)+4}= \textsc{WriteAddress-LocalVariables}$ 
            \item $\rho_{f_0(k)+5}= \textsc{Gas-Cal}$ 
            \item (---------End of state transition---------)
            \item $\rho_{f_0(k)+6}= \textsc{Function-Body}$ 
        \end{itemize}

    The difference between $\textsc{Exe-Statement-Main-Contract}$ and  $\textsc{Exe-Statement}$ is that the first one occurs in a function of an external call and the second one is the function of an internal call, but the behaviors are the same.
    The difference between $\textsc{WriteAddress-GlobalVariables}$ and \\$\textsc{WriteAddress-LocalVariables}$ is that the first one writes the value to the cell \textit{ctContext} which stores global variables and second one writes the value to the cell \textit{Memory} which stores local variables.
    Therefore, the proofs for the above sub-traces are similar, and we choose \textit{sub-trace 1} for proving as follows.

        \item Similar to $Case_1.2$,  $[\rho_{f_0(k)+1},\dots,\rho_{f_0(k)+5}]$ in~(1) corresponds to the state transition from $s_{k}$ to $s_{k+1}$.
        \item Similar to $Case_1.3$, let $c_k=v_{1}\leftarrow v_{2};\kwp{stmt}$ (For readability, the value of left-hand of the name is assumed to be the ID of the variable in $\ks$)
        \item Similar to $Case_1.4$, $C_r(c_k,i_k)=v_{1}\leftarrow v_{2};\kwp{stmt}$
        \item Similar to $Case_1.5$, rule  $\kwt{var\_assign}$ is generated.
        \item Similar to $Case_1.6$, $\{\kwf{Var}_{i}\} \subseteq \textit{names}(F_{g(k)})$
        \item Similar to $Case_1.10$, $F_c(i_{k+1},T(s_{k+1}))=\{\kwp{Var}_{i\circ 1}\}$
        \item Let \textit{step}=1, i.e., $g(k+1)=g(k)+1$
        \item Similar to $Case_1.12$, \textit{Condition (\textbf{a})} proved.
        \item By (B), $\gamma_k=\lambda_{f_0(k)} $
        \item Similar to $Case_1.15$, $\mathbb{G}(\gamma_k) \bi \textit{gvars}(\kwf{Var}_{i})$
        \item By assumption of $f$, $\lambda_{f_0(k)+5}=\gamma_{k+1}$
        \item Similar to $Case_1.18$,  
            \begin{equation*}
                \mathbb{G}_{\lambda_{f_0(k)+5}} =\mathbb{G}_{\lambda_{f_0(k)+4}}  
            \end{equation*} 
        \item Compute value $x$ which satisfies 
        \begin{align*}
            \lambda_{f_0(k)} 
            & \smatch
                       \Mrule{\iota(\lambda_{f_0(k)}) \ \_}{contractStack}\ 
             \Mrule{
                       \Mrule{\iota(\lambda_{f_0(k)})}{ctId}\ \\&
                       \Mrule{v_1 \mapsto d}{ctContext} \ \ 
                       \Mrule{\_\  d\mapsto x \ \_}{ctStorage}
            }{contractInstance}
        \end{align*}
        \item  By (14), Definition~\ref{def:k-rule-k-system}, 
        \begin{equation*}
                \mathbb{G}_{\lambda_{f_0(k)+5}}=\mathbb{G}_{\lambda_{f_0(k)}} \cup \{(v_1,  v_2)\} \backslash \{(v_1, x)\}
        \end{equation*}
        \item By (4), (11), (15), 
        
        \begin{equation*}
            (v_1, x) \bi \sigma_v(v_1)
        \end{equation*}
        \item By (4), (7)
        \begin{equation*}
            gvars(\{\kwf{Var}_{i\circ 1}\})=
            gvars(\{\kwf{Var}_{i}\})
            \cup \{\sigma_v(v_1))\}
            \backslash \{\sigma_v(v_1))\}
        \end{equation*}
        \item Since $v_2=\mathcal{E}{(\sigma_v(v_2))}$, build mapping
        \begin{equation*} 
            (v_1, v_2) \bi \sigma_v(v_2)
        \end{equation*}
        \item By (10), (11), (15), (16), (17), (18), and Remark~\ref{remark:bi},
        \begin{equation*}
        \mathbb{G}(\gamma_{k+1}) \bi \textit{gvars}(\kwf{Var}_{i\circ 1})
        \end{equation*}
        \item Similar to $Case_1.10$, \textit{condition (\textbf{b}),(\textbf{c}) hold}.
        \item By Definition~\ref{def:k-rule-k-system}, \textit{condition (\textbf{d}), (\textbf{e}), (\textbf{f}), (\textbf{g}), (\textbf{h}), (\textbf{i}), (\textbf{j}), (\textbf{k}) hold} trivially.


    \end{enumerate}

    \textbf{Case:} $T(s_k)=4$.
    \begin{enumerate}
        \item Similar to (1), (2) in case $T(s_k)=1$, the possible sub-traces for the state transition in (E) are as follows:

        Sub-trace 1:
        \begin{itemize}
            \item $\rho_{f_0(k)+1}= \textsc{Function-Body}$ 
            \item $\rho_{f_0(k)+2}= \textsc{Exe-Statement-Main-Contract}$ 
            \item $\rho_{f_0(k)+3}= \textsc{Var-Declaration}$ 
            \item $\rho_{f_0(k)+4}= \textsc{Gas-Cal}$ 
            \item (---------End of state transition---------)
            \item $\rho_{f_0(k)+5}= \textsc{Function-Body}$ 
        \end{itemize}
        \vspace*{.2cm}

        Sub-trace 2:
        \begin{itemize}
            \item $\rho_{f_0(k)+1}= \textsc{Function-Body}$ 
            \item $\rho_{f_0(k)+2}= \textsc{Exe-Statement}$ 
            \item $\rho_{f_0(k)+3}= \textsc{Var-Declaration}$ 
            \item $\rho_{f_0(k)+4}= \textsc{Gas-Cal}$ 
            \item (---------End of state transition---------)
            \item $\rho_{f_0(k)+5}= \textsc{Function-Body}$ 
        \end{itemize}

        Similar to $Case.2$, we choose \textit{sub-trace 1} for proving as follows.
        \item Similar to $Case_1.2$,  $[\rho_{f_0(k)+1},\dots,\rho_{f_0(k)+4}]$ in~(1) corresponds to the state transition from $s_{k}$ to $s_{k+1}$.
        \item Similar to $Case_1.3$, let $ c_k= \tau \ v_{1}\leftarrow v_{2};\kwp{stmt}$
        \item Similar to $Case_1.4$, $C_r(c_k,i_k)=\tau \ v_{1}\leftarrow v_{2};\kwp{stmt}$
        \item Similar to $Case_1.5$, rule  $\kwt{var\_declare}$ is generated.
        \item Similar to $Case_1.6$,  $\{\kwf{Var}_{i}\} \subseteq \textit{names}(F_{g(k)})$
        \item Similar to $Case_1.10$, $F_c(i_{k+1},T(s_{k+1}))=\{\kwp{Var_{i\circ 1}}\}$
        \item Let \textit{step}=1, i.e., $g(k+1)=g(k)+1$
        \item Similar to $Case_1.12$, \textit{Condition (\textbf{a})} proved.
        \item By (B), $\gamma_k=\lambda_{f_0(k)} $
        \item Similar to $Case_1.15$, 
        \begin{equation*}
            \mathbb{L}(\gamma_k) \bi \vars(\kwp{Var_{i}})\backslash\gvars(\kwp{Var_{i}}) 
        \end{equation*}
        \item By assumption of $f$, $\lambda_{f_0(k)+4}=\gamma_{k+1}$
        \item By Definition~\ref{def:k-rule-k-system}
            \begin{equation*}
                \mathbb{L}_{\lambda_{f_0(k)+4}} =\mathbb{L}_{\lambda_{f_0(k)}} \cup \{(v_1,v_2)\}
            \end{equation*} 
        \item By (5), generated rule $\kwt{var\_declare}$ satisfies:
        \begin{align*}
            &\vars(\kwp{Var_{i\circ1}})\backslash\gvars(\kwp{Var_{i\circ 1}})= \\&
            \vars(\kwp{Var_{i}})\backslash\gvars(\kwp{Var_{i}}) \cup \llbracket\sigma_v(v_2)\rrbracket
        \end{align*}
        \item By (13),(14), since Since $v_2=\mathcal{E}{(\sigma_v(v_2))}$, build mapping
        \begin{equation*} 
            (v_1, v_2) \bi \sigma_v(v_2)
        \end{equation*}
        \item By (10), (11), (12), (13), (14),(15), and Remark~\ref{remark:bi},
        \begin{equation*}
            \mathbb{L}(\gamma_{k+1}) \bi \vars(\kwp{Var_{i+1}})\backslash\gvars(\kwp{Var_{i+1}}) 
        \end{equation*}
        \item Similar to $Case_1.10$, \textit{condition (\textbf{f}), (\textbf{g}) hold}.
        \item By Definition~\ref{def:k-rule-k-system}, \textit{condition (\textbf{b}), (\textbf{c}), (\textbf{d}), (\textbf{e}), (\textbf{h}), (\textbf{i}), (\textbf{j}), (\textbf{k}) hold} trivially.
    \end{enumerate}

    \textbf{Case:} $T(s_k)=5$.
    \begin{enumerate}
        \item Similar to (1), (2) in case $T(s_k)=1$, the possible sub-traces for the state transition in (E) are as follows:

        Sub-trace 1:
        \begin{itemize}
            \item $\rho_{f_0(k)+1}= \textsc{Function-Body}$ 
            \item $\rho_{f_0(k)+2}= \textsc{Exe-Statement-Main-Contract}$ 
            \item $\rho_{f_0(k)+3}= \textsc{R5}$ 
            \item (---------End of state transition---------)
            \item $\rho_{f_0(k)+4}= \textsc{Exe-Statement-Main-Contract}$ 
            \item 
            ...
            \item $\rho_{f_0(k)+i+4}= \textsc{Exe-Statement-Main-Contract}$ 
            \item (---------End of branch---------)
            \item $\rho_{f_0(k)+i+4}= \textsc{Function-Body}$ 
        \end{itemize}
        \vspace*{.2cm}

        Sub-trace 2:
        \begin{itemize}
            \item $\rho_{f_0(k)+1}= \textsc{Function-Body}$ 
            \item $\rho_{f_0(k)+2}= \textsc{Exe-Statement}$ 
            \item $\rho_{f_0(k)+3}= \textsc{R5}$ 
            \item (---------End of state transition---------)
            \item $\rho_{f_0(k)+4}= \textsc{Exe-Statement}$ 
            \item
            ...
            \item $\rho_{f_0(k)+i+4}= \textsc{Exe-Statement}$ 
            \item (---------End of branch---------)
            \item $\rho_{f_0(k)+i+4}= \textsc{Function-Body}$ 
        \end{itemize}

        Similar to $Case.2$, we choose \textit{sub-trace 1} for proving as follows.

        \item Similar to $Case_1.2$,  $[\rho_{f_0(k)+1},\dots,\rho_{f_0(k)+3}]$ in~(1) corresponds to the state transition from $s_{k}$ to $s_{k+1}$.
        \item Similar to $Case_1.3$, let 
        $$c_k=\texttt{if}\ e_{b}\ \texttt{then}\ \kwp{stmt}_{1}\ \texttt{else}\ \kwp{stmt}_{2};\kwp{stmt}_{3}$$
        \item Similar to $Case_1.4$, 
        \begin{equation*}
            C_r(c_k,i_k)=
            \texttt{if}\ e_{b}\ \texttt{then}\ \kwp{stmt}_{1}\ \texttt{else}\ \kwp{stmt}_{2};\kwp{stmt}_{3}
        \end{equation*}
        \item Similar to $Case_1.5$, since $e_b=\code{true}$, we get $\mathcal{E}(\theta_{e}(e_{b})) =\mathcal{E}(\sigma_v(e_{b}))=\code{true}$ and rule  $\kwt{if\_true}$ can be applied.
        \item Similar to $Case_1.6$, $\{\kwf{Var}_{i}\} \subseteq \textit{names}(F_{g(k)})$
        \item Similar to $Case_1.10$, $F_c(i_{k+1},T(s_{k+1}))=\{\kwp{Var_{i\circ 1}}\}$
        \item Let \textit{step}=1, i.e., $g(k+1)=g(k)+1$
        \item Similar to $Case_1.12$, \textit{Condition (\textbf{a})} proved.
        \item By Definition~\ref{def:k-rule-k-system}, \textit{condition (\textbf{b}), (\textbf{c}), (\textbf{d}), (\textbf{e}), (\textbf{f}), (\textbf{g}), (\textbf{h}), (\textbf{i}), (\textbf{j}), (\textbf{k}) hold} trivially.
    \end{enumerate}

    \textbf{Case:} $T(s_k)=6$.
        The proof of the case is similar to $Case_5$, and we omit the proof of this case.
    
    \textbf{Case:} $T(s_k)=7$.
    \begin{enumerate}
        \item Similar to (1), (2) in case $T(s_k)=1$, the possible sub-traces for the state transition in (E) are as follows:

        Sub-trace 1:
        \begin{itemize}
            \item $\rho_{f_0(k)+1}= \textsc{Function-Body}$ 
            \item $\rho_{f_0(k)+2}= \textsc{Exe-Statement-Main-Contract}$ 
            \item $\rho_{f_0(k)+3}= \textsc{Require}$ 
            \item (---------End of state transition---------)
            \item $\rho_{f_0(k)+4}= \textsc{Function-Body}$ 
        \end{itemize}
        \vspace*{.2cm}

        Sub-trace 2:
        \begin{itemize}
            \item $\rho_{f_0(k)+1}= \textsc{Function-Body}$ 
            \item $\rho_{f_0(k)+2}= \textsc{Exe-Statement}$ 
            \item $\rho_{f_0(k)+3}= \textsc{Require}$ 
            \item (---------End of state transition---------)
            \item $\rho_{f_0(k)+4}= \textsc{Function-Body}$ 
        \end{itemize}
        Similar to $Case.2$, we choose \textit{sub-trace 1} for proving as follows.

        \item Similar to $Case_1.2$,  $[\rho_{f_0(k)+1},\dots,\rho_{f_0(k)+3}]$ in~(1) corresponds to the state transition from $s_{k}$ to $s_{k+1}$.
        \item Similar to $Case_1.3$, let 
        $$c_k=\texttt{require}\ e_{b};  \kwp{stmt} $$
        \item Similar to $Case_1.4$, $C_r(c_k,i_k)=\texttt{require}\ e_{b};  \kwp{stmt}$
        \item Similar to $Case_1.5$, since $e_b=\code{true}$, we get $\mathcal{E}(\theta_{e}(e_{b})) =\mathcal{E}(\sigma_v(e_{b}))=\code{true}$ and rule  $\kwt{require\_true}$ can be applied.
        \item Similar to $Case_1.6$, $\{\kwf{Var}_{i}\} \subseteq \textit{names}(F_{g(k)})$
        \item Similar to $Case_1.10$, $F_c(i_{k+1},T(s_{k+1}))=\{\kwp{Var_{i\circ 1}}\}$
        \item Let \textit{step}=1, i.e., $g(k+1)=g(k)+1$
        \item Similar to $Case_1.12$, \textit{Condition (\textbf{a})} proved.
        \item By Definition~\ref{def:k-rule-k-system}, \textit{condition (\textbf{b}), (\textbf{c}), (\textbf{d}), (\textbf{e}), (\textbf{f}), (\textbf{g}), (\textbf{h}), (\textbf{i}), (\textbf{j}), (\textbf{k}) hold} trivially.

    \end{enumerate}

    \textbf{Case:} $T(s_k)=8$.
    \begin{enumerate}
        \item Similar to (1), (2) in case $T(s_k)=1$, the possible sub-traces for the state transition in (E) are as follows:

        Sub-trace 1:
        \begin{itemize}
            \item $\rho_{f_0(k)+1}= \textsc{Function-Body}$ 
            \item $\rho_{f_0(k)+2}= \textsc{Exe-Statement-Main-Contract}$ 
            \item $\rho_{f_0(k)+3}= \textsc{Require}$ 
            \item $\rho_{f_0(k)+4}= \textsc{Exception-Propagation}$ 
            \item $\rho_{f_0(k)+5}= \textsc{Update-Exception-State}$ 
            \item ...
            \item $\rho_{f_0(k)+t+5}= \textsc{Update-Cur-Context}$ 
            \item $\rho_{f_0(k)+t+6}= \textsc{Return-Context}$ 
            \item $\rho_{f_0(k)+t+7}= \textsc{Clear-Recipient-Context}$ 
            \item $\rho_{f_0(k)+t+8}= \textsc{Clear-Caller-Context}$ 
            \item $\rho_{f_0(k)+t+9}= \textsc{Propagate-Exception-True}$ 
            \item ...
        \end{itemize}
        \vspace*{.2cm}

        Sub-trace 2:
        \begin{itemize}
            \item $\rho_{f_0(k)+1}= \textsc{Function-Body}$ 
            \item $\rho_{f_0(k)+2}= \textsc{Exe-Statement}$ 
            \item $\rho_{f_0(k)+3}= \textsc{Require}$ 
            \item $\rho_{f_0(k)+4}= \textsc{Exception-Propagation}$ 
            \item $\rho_{f_0(k)+5}= \textsc{Update-Exception-State}$ 
            \item ...
            \item $\rho_{f_0(k)+t+6}= \textsc{Return-Context}$ 
            \item $\rho_{f_0(k)+t+7}= \textsc{Clear-Recipient-Context}$ 
            \item $\rho_{f_0(k)+t+8}= \textsc{Clear-Caller-Context}$ 
            \item $\rho_{f_0(k)+t+9}= \textsc{Propagate-Exception-True}$ 
            \item ...
        \end{itemize}
        Here, the application of rule $\textsc{Update-Exception-State}$ updates a 
        key value in cell \textit{contractStack} which finally results in the application of $\textsc{Propagate-Exception-True}$.
        
        Since both traces violates $\alpha_{\rho_e}$, the cases do not satisfy the precondition of the lemma, i.e., \textit{conditions of the case is proved}.

    \end{enumerate}
    \textbf{Case:} $T(s_k)=9$.
\begin{enumerate}
        \item Similar to (1), (2) in case $T(s_k)=1$, the possible trace for the state transition in (E) is as follows:
        \begin{itemize}
            \item $\rho_{f_0(k)+1}= \textsc{Function-Body}$ 
            \item $\rho_{f_0(k)+2}= \textsc{Exe-Statement-Main-Contract}$ 
            \item $\rho_{f_0(k)+3}= \textsc{Return-Value}$ 
            \item $\rho_{f_0(k)+4}= \textsc{Update-Cur-Context}$ 
            \item $\rho_{f_0(k)+5}= \textsc{Return-Context}$ 
            \item $\rho_{f_0(k)+6}= \textsc{Clear-Recipient-Context}$ 
            \item $\rho_{f_0(k)+7}= \textsc{Clear-Caller-Context}$ 
            \item $\rho_{f_0(k)+8}= \textsc{Propagate-Exception-False}$ 
            \item (---------End of state transition and function---------)
        \end{itemize}
    Note that we omit traces that does not satisfy $\alpha_{\rho_e}$. Hence, we prove the conditions hold for trace.
        \item Similar to $Case_1.2$,  $[\rho_{f_0(k)+1},\dots,\rho_{f_0(k)+8}]$ in~(1) corresponds to the state transition from $s_{k}$ to the end of the function.
        \item Similar to $Case_1.3$, let 
        $$c_k=\texttt{return}\ \code{i}$$
        \item Similar to $Case_1.4$, 
        \begin{equation*}
            C_r(c_k,i_k)=
            c_k=\texttt{return}\ \code{i}
        \end{equation*}

        \item Similar to $Case_1.5$, rule  $\kwt{ret\_ext}$ is generated.
        \item Similar to $Case_1.6$,  $\{\kwf{Var}_{i}\} \subseteq \textit{names}(F_{g(k)})$
        \item Similar to $Case_1.10$, $F_c(i_{k+1},T(s_{k+1}))=\{\kwf{Gvar},\kwf{Evar}\}$
        \item Let \textit{step}=1, i.e., $g(k+1)=g(k)+1$

        \item Similar to $Case_1.15$, 
        \begin{equation*}
            \mathbb{G}(\gamma_k) \bi \gvars(\kwp{Var_{i}}) 
        \end{equation*}
        \begin{equation*}
            \mathbb{E}(\gamma_k) \bi \evars(\kwp{Var_{i}}) 
        \end{equation*}

        \item Similar to $Case_1.16$,
        \begin{equation*}
            \gvars(\kwp{Var_{i}})=\gvars(\kwp{Gvar})
        \end{equation*}
        \begin{equation*}
            \evars(\kwp{Var_{i}})=\evars(\kwp{Evar})
        \end{equation*}

        \item Similar to $Case_1.17$, $\lambda_{f_0(k)+8}=\gamma_{k+1}$

        \item Similar to $Case_1.18$, 
         \begin{small}
        \begin{equation*}
            \mathbb{G}(\lambda_{f_0(k)+8})=\dots=\mathbb{G}(\lambda_{f_0(k)+6})=\mathbb{G}_b(\lambda_{f_0(k)})=\mathbb{G}(\lambda_{f_0(k)})
        \end{equation*}
         \end{small}
        \begin{equation*}
            \mathbb{L}(\lambda_{f_0(k)+8})=\dots=\mathbb{L}(\lambda_{f_0(k)+6})=\mathbb{G}_l(\lambda_{f_0(k)})= \emptyset
        \end{equation*}
        \begin{equation*}
            \mathbb{E}(\lambda_{f_0(k)+8})=\dots=\mathbb{E}(\lambda_{f_0(k)})
        \end{equation*}

        \item By (9), (10), (11), (12), condition \textit{(\textbf{b}), (\textbf{c}), (\textbf{d}), (\textbf{e})} hold.

        \item By (11), (12), condition \textit{(\textbf{f}), (\textbf{g})} hold.
        \item By Definition~\ref{def:k-rule-k-system}, \textit{(\textbf{h}), (\textbf{i})} hold trivially.
        \item Similar to $Case_1.33$,
        
        \begin{equation*}
            \mathbb{L}_A(\lambda_{f_0(k)+8})=\dots=\mathbb{L}(\lambda_{f_0(k)+3})=\mathbb{L}_A(\lambda_{f_0(k)})\backslash \ \{ \mathbb{L}_{some} \}
        \end{equation*}
         
        where $\mathbb{L}_{some}$ is some value. 
        \item By (16) and inductive hypothesis (j), \textit{condition (\textbf{j})} holds.
        \item Similar to $Case_1.35$, by (13), \textit{condition (\textbf{k})} holds.
\end{enumerate}

    \textbf{Case:} $T(s_k)=10$.

    \begin{enumerate}
        \item Similar to (1), (2) in case $T(s_k)=1$, the possible sub-traces for the state transition in (E) are as follows:

        Sub-trace 1:
        \begin{itemize}
            \item $\rho_{f_0(k)+1}= \textsc{Function-Body}$ 
            \item $\rho_{f_0(k)+2}= \textsc{Exe-Statement}$ 
            \item $\rho_{f_0(k)+3}= \textsc{Return-Value}$ 
            \item $\rho_{f_0(k)+4}= \textsc{Return-Context}$ 
            \item $\rho_{f_0(k)+5}= \textsc{Clear-Recipient-Context}$ 
            \item $\rho_{f_0(k)+6}= \textsc{Clear-Caller-Context}$ 
            \item $\rho_{f_0(k)+7}= \textsc{Propagate-Exception-False}$ 
            \item (---------End of state transition and function---------)
        \end{itemize}
        \vspace*{.2cm}

    Note that we omit traces that does not satisfy $\alpha_{\rho_e}$. Hence, we prove the conditions hold for trace.
        \item Let the ID of recipient and caller be $id_r, id_c$, respectively.
        \item Similar to $Case_1.2$,  $[\rho_{f_0(k)+1},\dots,\rho_{f_0(k)+7}]$ in~(1) corresponds to the state transition from $s_{k}$ to the end of the function.
        \item Similar to $Case_1.3$, let 
        $$c_k=\texttt{return}\ \code{i}$$
        \item Similar to $Case_1.4$, 
        \begin{equation*}
            C_r(c_k,i_k)=
            c_k=\texttt{return}\ \code{i}
        \end{equation*}

        \item Similar to $Case_1.5$, rule  $\kwt{ret\_in}$  is generated.
        \item By (6), both the recipient function and the caller function  have been applied by $R_s$, where the caller's rule  $\kwt{in\_call}$ and $\kwt{recv\_ret}$ has been generated.
        \item By (7), let the rule  $\kwt{in\_call}$ is generated at the corresponding state $$(\gamma_y, i_y, x(c_{x}).f_{x}(p);\kwp{stmt} )$$
        \item Similar to $Case_1.6$, by (8), $\{\kwf{Var}_{i_y}^{id_c}\} \subseteq \textit{names}(F_{g(k)})$
        \item Similar to $Case_1.10$, $F_c(i_{k+1},T(s_{k+1}))=\{\kwp{Var^{id_c}_{i_y\circ 1\circ 1}}\}$
        \item Let \textit{step}=2, i.e., $g(k+1)=g(k)+2$, and rule  $\kwt{in\_call}$ and $\kwt{recv\_ret}$ are applied in these steps.
        \item Similar to $Case_1.12$, \textit{condition (\textbf{a})} holds.

        \item Similar to $Case_1.17$, $\lambda_{f_0(k)+7}=\gamma_{k+1}$

        \item By inductive hypothesis (k), we have $V\ins F_{g(k)}$, such that  $$\mathbb{G}(\gamma_k, id_c) \bi \textit{gvars}(V)$$

        \item Similar to $Case_1.15$, 
            $$\mathbb{G}(\gamma_k, id_r) \bi \textit{gvars}(\kwf{Var}_i^{id_r})$$
        
        \item Similar to $Case_1.18$, 
         \begin{small}
        \begin{equation*}
            \mathbb{G}(\lambda_{f_0(k)+7})=\dots=\mathbb{G}(\lambda_{f_0(k)+6})=\mathbb{G}_g(\lambda_{f_0(k)},id_c) 
        \end{equation*}
        \begin{equation*}
            =\mathbb{G}(\lambda_{f_0(k)},id_c)
        \end{equation*}
        \begin{equation*}
            \mathbb{L}(\gamma_{k+1},id_r)=\mathbb{L}(\lambda_{f_0(k)+7},id_r)=\dots=\mathbb{L}(\lambda_{f_0(k)+5},id_r)= \emptyset
        \end{equation*}
        \begin{equation*}
            \mathbb{L}(\lambda_{f_0(k)+7},id_c)=\dots=\mathbb{L}(\lambda_{f_0(k)+6},id_c)= \mathbb{L}(\gamma_k, id_c) \in \mathbb{L}_A(\gamma_k)
        \end{equation*}
        \begin{equation*}
            \mathbb{E}(\lambda_{f_0(k)+7})=\dots=\mathbb{E}(\lambda_{f_0(k)})
        \end{equation*}
        \begin{equation*}
               \mathbb{G}(\gamma_{k+1}, id_c)=\mathbb{G}(\lambda_{f_0(k)+7})
        \end{equation*}
        \begin{equation*}
               \mathbb{G}(\gamma_{k+1}, id_r)=\mathbb{G}(\gamma_{k}, id_r)
        \end{equation*}
         \end{small}
        \item By (16) and inductive hypothesis (j),
        there exists $V_1$, such that   $     \mathbb{L}_A(\gamma_k) \bi \vars(V_1)\backslash\gvars(V_1)) $ and $V_1\ins  F_{g(k)}$
 
        \item By (11), (17), $V_1\ins F_{g(k+1)}$
        \item By (13), (16), (17), (18), conditions \textit{(\textbf{f}), (\textbf{g}) hold}.
        \item By (13), (14), (16), condition \textit{(\textbf{b}), (\textbf{c}) hold}.
        \item By (13), (16), condition \textit{(\textbf{d}), (\textbf{e}) hold}.
        \item By Definition~\ref{def:k-rule-k-system}, \textit{(\textbf{h}), (\textbf{i})} hold trivially.
        \item Similar to $Case_9.17$,  \textit{condition $(\textbf{j})$ holds} trivially.
        \item By (11), (14), $V\ins F_{g(k+1)}$.
        \item By (11), $\textit{gvars}(\kwf{Var}_i^{id_r}) = \textit{gvars}(\kwf{Gvar}^{id_r})$ and $\kwf{Gvar}^{id_r} \ins F_{g(k+1)}$
        \item Similar to $Case_1.35$, by (14), (16), (24), (25), \textit{condition (\textbf{k})} holds.

    \end{enumerate}

    \textbf{Case:} $T(s_k)=11$.

    \begin{enumerate}
        \item Similar to (1), (2) in case $T(s_k)=1$, the possible sub-traces for the state transition in (E) are as follows:

        Sub-trace 1:
        \begin{itemize}
            \item $\rho_{f_0(k)+1}= \textsc{Function-Body}$ 
            \item $\rho_{f_0(k)+2}= \textsc{Exe-Statement-Main-Contract}$ 
            \item $\rho_{f_0(k)+3}= \textsc{Transfer-Fund-Begin}$ 
            \item $\rho_{f_0(k)+4}= \textsc{Gas-Cal}$ 
            \item $\rho_{f_0(k)+5}= \textsc{Transfer-Fund}$ 
            \item (---------End of state transition---------)
            \item $\rho_{f_0(k)+6}= \textsc{Function-Body}$ 
        \end{itemize}
        \vspace*{.2cm}

        Sub-trace 2:
        \begin{itemize}
            \item $\rho_{f_0(k)+1}= \textsc{Function-Body}$ 
            \item $\rho_{f_0(k)+2}= \textsc{Exe-Statement}$ 
            \item $\rho_{f_0(k)+3}= \textsc{Transfer-Fund-Begin}$ 
            \item $\rho_{f_0(k)+4}= \textsc{Gas-Cal}$ 
            \item $\rho_{f_0(k)+5}= \textsc{Transfer-Fund}$ 
            \item (---------End of state transition---------)
            \item $\rho_{f_0(k)+6}= \textsc{Function-Body}$ 
        \end{itemize}
        Similar to $Case.2$, we choose \textit{sub-trace 1} for proving as follows.
            \item Similar to $Case_1.2$,  $[\rho_{f_0(k)+1},\dots,\rho_{f_0(k)+5}]$ in~(1) corresponds to the state transition from $s_{k}$ to $s_{k+1}$.
            \item Similar to $Case_1.3$, let $c_k=c_{x}.\texttt{transfer}(v_{1});\kwp{stmt} $ (For readability, assume in $\ks$, $c_x, c$ is the balance of the receiver and the account of $\gamma$, respectively)
            \item Similar to $Case_1.4$, $C_r(c_k,i_k)=c_{x}.\texttt{transfer}(v_{1});\kwp{stmt}$
            \item Similar to $Case_1.5$, rule  $\kwt{transfer\_succ}$ is generated.
            \item Similar to $Case_1.6$, $\{\kwf{Var}_{i}\} \subseteq \textit{names}(F_{g(k)})$
            \item Similar to $Case_1.10$, $F_c(i_{k+1},T(s_{k+1}))=\{\kwp{Var_{i\circ 1}}\}$
            \item Let \textit{step}=1, i.e., $g(k+1)=g(k)+1$
            \item Similar to $Case_1.12$, \textit{condition (\textbf{a})} holds.
            \item By (B), $\gamma_k=\lambda_{f_0(k)} $
            \item Similar to $Case_1.15$, $\mathbb{E}(\gamma_k) \bi \textit{evars}(\kwf{Var}_{i})$
            \item By assumption of $f$, $\lambda_{f_0(k)+5}=\gamma_{k+1}$
            \item Similar to $Case_1.18$,  
                \begin{equation*}
                    \mathbb{E}_{\gamma_k}=\mathbb{E}_{\lambda_{f_0(k)}}= \dots =\mathbb{E}_{\lambda_{f_0(k)+4}}  
                \end{equation*} 
            \item Let $(id_c,c) \in \mathbb{E}(\gamma_k)$ and $(id_x,c_x) \in \mathbb{E}(\gamma_k)$
            \item By (11),  let $(id_c,c) \bi  \sigma_v(c)$ and  $(id_x,c_x) \bi  \sigma_v(c_x)$
            \item By (12),(13),(14)
            \begin{align*}
                &\mathbb{E}_{\gamma_{k+1}}=\mathbb{E}_{\lambda_{f_0(k)+5}}=\\&=\mathbb{E}_{\gamma_{k}}\cup \{(id_c,c-v_1), (id_x,c_x+v_1)\} \backslash \{ (id_c,c), (id_x,c_x)\}
            \end{align*}
            \item By (15) and Definition~\ref{def:valuation}, 
           $$\mathcal{E}(\sigma_v(c))=c \land \mathcal{E}(\sigma_v(c_x))=c_x$$
           \item By (5), $\textit{evars}(\kwf{Var}_{i\circ 1})=$
           \begin{small}
           \begin{equation*}
            \textit{evars}(\kwf{Var}_{i})  \cup \{\sigma_v(c)\ominus\sigma_v(v_1), \sigma_v(c_x)\oplus \sigma_v(v_1) \} \backslash \{\sigma_v(c), \sigma_v(c_x) \}
           \end{equation*}
           \end{small}
           \item By (17) and Definition~\ref{def:valuation}, let $\mathcal{E}(\sigma_v(v_1))=v_1$, we get
                \begin{align*}
                    &\mathcal{E} (\sigma_v(c)\ominus \sigma_v(v_1))=c-v_1\\&
                    \mathcal{E} (\sigma_v(c_x)\oplus \sigma_v(v_1))=c_x+v_1
                \end{align*}
            \item By (19), let $(id_c,c-v_1) \bi  (\sigma_v(c)\oplus \sigma_v(v_1))$, and let 
        $(id_x,c_x\ominus \sigma_v(v_1)) \bi  \sigma_v(c_x)+v_1$
            \item By (11), (16), (18), (20), and Remark~\ref{remark:bi}, \textit{condition (\textbf{d}), (\textbf{e})} hold
            \item By Definition~\ref{def:k-rule-k-system}, \textit{condition (\textbf{b}), (\textbf{c}), (\textbf{f}), (\textbf{g}), (\textbf{h}), (\textbf{i}), \textbf{(j)}, \textbf{(k)} hold} trivially.
    \end{enumerate}

    \textbf{Case:} $T(s_k)=12$.

    \begin{enumerate}
        \item Similar to (1), (2) in case $T(s_k)=1$, the possible sub-traces for the state transition in (E) are as follows:

        Sub-trace 1:
        \begin{itemize}
            \item $\rho_{f_0(k)+1}= \textsc{Function-Body}$ 
            \item $\rho_{f_0(k)+2}= \textsc{Exe-Statement-Main-Contract}$ 
            \item $\rho_{f_0(k)+3}= \textsc{Transfer-Fund-Begin}$ 
            \item $\rho_{f_0(k)+4}= \textsc{Gas-Cal-Fail}$ 
            \item $\rho_{f_0(k)+5}= \textsc{Exception-Propagation}$ 
            \item $\rho_{f_0(k)+6}= \textsc{Update-Exception-State}$ 
            \item ...
            \item $\rho_{f_0(k)+t+7}= \textsc{Update-Cur-Context}$ 
            \item $\rho_{f_0(k)+t+8}= \textsc{Return-Context}$ 
            \item $\rho_{f_0(k)+t+9}= \textsc{Clear-Recipient-Context}$ 
            \item $\rho_{f_0(k)+t+10}= \textsc{Clear-Caller-Context}$ 
            \item $\rho_{f_0(k)+t+11}= \textsc{Propagate-Exception-True}$ 
            \item ...
        \end{itemize}
        \vspace*{.2cm}

        Sub-trace 2:
        \begin{itemize}
            \item $\rho_{f_0(k)+1}= \textsc{Function-Body}$ 
            \item $\rho_{f_0(k)+2}= \textsc{Exe-Statement}$ 
            \item $\rho_{f_0(k)+3}= \textsc{Transfer-Fund-Begin}$ 
            \item $\rho_{f_0(k)+4}= \textsc{Gas-Cal-Fail}$ 
            \item $\rho_{f_0(k)+5}= \textsc{Exception-Propagation}$ 
            \item $\rho_{f_0(k)+6}= \textsc{Update-Exception-State}$ 
            \item ...
            \item $\rho_{f_0(k)+t+7}= \textsc{Return-Context}$ 
            \item $\rho_{f_0(k)+t+8}= \textsc{Clear-Recipient-Context}$ 
            \item $\rho_{f_0(k)+t+9}= \textsc{Clear-Caller-Context}$ 
            \item $\rho_{f_0(k)+t+10}= \textsc{Propagate-Exception-True}$ 
            \item ...
        \end{itemize}
        Similar to $Case.8$, all conditions holds.
    \end{enumerate} 

    \textbf{Case:} $T(s_k)=13$.
    \begin{enumerate}
        \item Similar to (1), (2) in case $T(s_k)=1$, the possible sub-traces for the state transition in (E) are as follows:

        Sub-trace 1:
        \begin{itemize}
            \item $\rho_{f_0(k)+1}= \textsc{Function-Body}$ 
            \item $\rho_{f_0(k)+2}= \textsc{Exe-Statement-Main-Contract}$ 
            \item $\rho_{f_0(k)+3}= \textsc{Send-Fund-Begin}$ 
            \item $\rho_{f_0(k)+4}= \textsc{Gas-Cal}$ 
            \item $\rho_{f_0(k)+5}= \textsc{Send-Fund-Successful}$ 
            \item (---------End of state transition---------)
            \item $\rho_{f_0(k)+6}= \textsc{Function-Body}$ 
        \end{itemize}
        \vspace*{.2cm}

        Sub-trace 2:
        \begin{itemize}
            \item $\rho_{f_0(k)+1}= \textsc{Function-Body}$ 
            \item $\rho_{f_0(k)+2}= \textsc{Exe-Statement}$ 
            \item $\rho_{f_0(k)+3}= \textsc{Send-Fund-Begin}$ 
            \item $\rho_{f_0(k)+4}= \textsc{Gas-Cal}$ 
            \item $\rho_{f_0(k)+5}= \textsc{Send-Fund-Successful}$ 
            \item (---------End of state transition---------)
            \item $\rho_{f_0(k)+6}= \textsc{Function-Body}$ 
        \end{itemize}
    \end{enumerate} 
    The proof for the case is similar to the one for $Case_{11}$, and we omit the proof.

    \textbf{Case:} $T(s_k)=14$.
    \begin{enumerate}
        \item Similar to (1), (2) in case $T(s_k)=1$, the possible sub-traces for the state transition in (E) are as follows:

        Sub-trace 1:
        \begin{itemize}
            \item $\rho_{f_0(k)+1}= \textsc{Function-Body}$ 
            \item $\rho_{f_0(k)+2}= \textsc{Exe-Statement-Main-Contract}$ 
            \item $\rho_{f_0(k)+3}= \textsc{Send-Fund-Begin}$ 
            \item $\rho_{f_0(k)+4}= \textsc{Gas-Cal}$ 
            \item $\rho_{f_0(k)+5}= \textsc{Send-Fund-Failed}$ 
            \item (---------End of state transition---------)
            \item $\rho_{f_0(k)+6}= \textsc{Function-Body}$ 
        \end{itemize}
        \vspace*{.2cm}

        Sub-trace 2:
        \begin{itemize}
            \item $\rho_{f_0(k)+1}= \textsc{Function-Body}$ 
            \item $\rho_{f_0(k)+2}= \textsc{Exe-Statement}$ 
            \item $\rho_{f_0(k)+3}= \textsc{Send-Fund-Begin}$ 
            \item $\rho_{f_0(k)+4}= \textsc{Gas-Cal}$ 
            \item $\rho_{f_0(k)+5}= \textsc{Send-Fund-Failed}$ 
            \item (---------End of state transition---------)
            \item $\rho_{f_0(k)+6}= \textsc{Function-Body}$ 
        \end{itemize}
        \item Similar to $Case_1.12$, \textit{Condition (\textbf{a})} proved.
        \item By Definition~\ref{def:k-rule-k-system}, Condition \textit{(\textbf{b}), (\textbf{c}), (\textbf{d}), (\textbf{e}), (\textbf{f}), (\textbf{g}), (\textbf{h}), (\textbf{i}), (\textbf{j}), \textbf{(k)}  }\textit{hold} trivially.
    \end{enumerate}

\end{proof}


\begin{lemma}
    \label{lemma:var_soundness2}
    Let $(\mathcal{T},\karroww, \gamma_0)$ be a K transition system that satisfies $\ks \vDash \alpha$.
    Let $R_{s} \vDash \beta_{R_{s}}$.
    If $$
            \lambda_0 \newkarrowx{1} \lambda_1 \newkarrowx{2} \dots \newkarrowx{n} \lambda_n
    $$ 
    where $[\rho_1,\dots,\rho_n] \in \traceks{\gamma_0}$, and it corresponds to a sequence of state transitions according to Proposition~\ref{proposition:transitionsystem}: 
    $$s_0 \Rightarrow s_1 \Rightarrow \dots \Rightarrow s_m $$
    where $s_j=(\gamma_j,i_j,c_j)$, then there are
    \begin{equation*}
        \emptyset \exer{r_1} F_1 \exer{r_2} \dots \exer{r_{m'}} F_{m'} \in \textit{exec}^{m s r}(R_s)
    \end{equation*}
    and there exists a valuation $\mathcal{E}$ and a monotonic, strictly increasing function $g: \{ 0,\dots,m\}$ $\rightarrow \{ 0,\dots,m'\} $ such that 
    \begin{enumerate}[label=(\alph*)]
    \item $ \exists j \in \{1,\dots,m'\}. \forall x \in \mathbb{S}(\lambda_0).\ \exists V\ins F_{j}.\ \mathbb{G}(\lambda_0,x) \bi \gvars(V) $
    \item $ \forall x \in \mathbb{S}(\lambda_n).\ \exists V\ins F_{m'}.\ \mathbb{G}(\lambda_n,x) \bi \gvars(V)$
    \end{enumerate}
\end{lemma}

\begin{proof}
    \begin{enumerate}
        \item By Lemma \ref{lemma:var_soundness}, for $R_s \vDash \beta_{initE} \land \beta_{initG}$, there exists $\emptyset \exer{r_1} F_1 \exer{r_2} \dots \exer{r_{m''}} F_{m''}$  and there exists a valuation $\mathcal{E}$ and a monotonic, strictly increasing function $g: \{ 0,\dots,m\}$ $\rightarrow \{ 0,\dots,m''\} $ such that $g(m) = m''$and 
    \begin{enumerate}[label=(\alph*)]
    \item $ \forall x \in \mathbb{S}(\gamma_0). \exists V\ins F_{g(0)}. \mathbb{G}(\gamma_0,x) \bi \gvars(V)$
    \item $ \forall x \in \mathbb{S}(\gamma_m). \exists V\ins F_{m''}. \mathbb{G}(\gamma_m,x) \bi \gvars(V)$
    \end{enumerate}
        \item By Proposition~\ref{proposition:transitionsystem} and Definition \ref{def:Soliditystate}, there exists a strictly monotonically increasing function $f$ such that $\gamma_j= \lambda_{f(j)}$ and $f(0)= 0$.
        \item By (2), $\gamma_0$ = $\lambda_0$ and $ \forall x \in \mathbb{S}(\gamma_0). \mathbb{G}(\gamma_0,x) = \mathbb{G}(\lambda_0,x)$.
        \item By the definitions of $\alpha_{\rho_n}$ and $\ks$, the only possible trace for the transition $\lambda_{n-6} \newkarrowx{n-7} \lambda_{n-7} \dots \newkarrowx{n} \lambda_n$ is as follows:
        \begin{itemize}
            \item $\rho_{n-7}= \textsc{Function-Body}$ 
            \item $\rho_{n-6}= \textsc{Exe-Statement-Main-Contract}$ 
            \item $\rho_{n-5}= \textsc{Return-Value}$ 
            \item $\rho_{n-4}= \textsc{Update-Cur-Context}$ 
            \item $\rho_{n-3}= \textsc{Return-Context}$ 
            \item $\rho_{n-2}= \textsc{Clear-Recipient-Context}$ 
            \item $\rho_{n-1}= \textsc{Clear-Caller-Context}$ 
            \item $\rho_{n}= \textsc{Propagate-Exception-False}$ 
            \item (---------End of state transition and function---------)
        \end{itemize}
        \item By Definition \ref{def:Soliditystate}, \ref{def:k-rule-k-system}, \ref{def:LGE}, $\forall x \in \mathbb{S}(\gamma_m).\ \mathbb{G}(s_m,x) = \mathbb{G}(\lambda_{n-6},x) = \mathbb{G}(\lambda_n,x)$.
        \item By Table \ref{table:states} and Table \ref{table:tableCorrespondence}, $c_m = \texttt{return}\ \code{i}$ and $T(s_m) = 9$.
        \item By (6) and the definition of $\mathcal{R}$, rule $\kwt{ret\_ext}$ is generated.
        \item By (1), (7), there exists an execution $e_r = \emptyset \exer{r_1} F_1 \exer{r_2} \dots \exer{r_{m''}} F_{m''} \exer{\kwt{ret\_ext}} F_{m'} \in \textit{exec}^{m s r}(R_s)$ where $R_s \vDash \ \beta_{initE} \land \beta_{initG} \land \beta_{r_n}$.
        \item By the definition of rule $\kwt{ret\_ext}$, $\{\gvars(V) \mid \forall x \in \mathbb{S}(\gamma_m).\ \exists V\ins F_{m''}.\ \textit{names}(\{V\})=\{\kwf{Var}_i^x \}\} = \{\gvars(V) \mid \forall x \in \mathbb{S}(\gamma_m).\ \exists V\ins F_{m'}.\ \textit{names}(\{V\})=\{\kwf{Gvar}^x \}\}$.
        \item By (1), (5), (8), (9), $ \forall x \in \mathbb{S}(\lambda_n). \exists V\ins F_{m'}. \mathbb{G}(\lambda_n,x) \bi \gvars(V)$.
        \item By (1), (3), (8), condition (a) proved.
        \item By (8), (10), condition (b) proved.
    \end{enumerate}
\end{proof}

\begin{definition}{($\phi_{inv}^{\ks}$)}
    \label{definition: k_satisfy_inv}
    \textit{
    Let   $(\mathcal{T},\karroww, \lambda_0)$ be a K transition system that satisfies $\ks \vDash \alpha$.
    Given a set of global variables $\eta = \{\chi_1, \dots, \chi_m\}$ of contract $x$, define the invariant property for this K transition system $\phi_{inv}^{\ks}(\lambda_0,\eta,x)$ as follows:}


    \textit{
    For all executions $\lambda_0 \newkarrowx{1} \lambda_1 \newkarrowx{2} \dots \newkarrowx{n} \lambda_n \in \mathit{exec}^{\ks}(\gamma_0)$,
    we have that
    $$
       \sum_{v\in G_0}(v) = \sum_{v\in G_n}(v)
    $$
    where $G_i = \{v \mid \chi \in \eta\ \land\ (\chi,v) \in \mathbb{G}(\lambda_i,x) \}$
    }
\end{definition}

\begin{definition}{($\phi_{equ}^{\ks}$)}
    \label{definition: k_satisfy_equ}
    \textit{
    Let   $(\mathcal{T},\karroww, \lambda_0)$ be a K transition system that satisfies $\ks \vDash \alpha$.
    Given a global variable $\chi$ of contract $x$, define the equivalence property for this K transition system $\phi_{equ}^{\ks}(\lambda_0,\chi,x)$ as follows:}

   \textit{
    For any two executions $\lambda_0 \newkarrowx{A_1} \lambda_{A_1} \newkarrowx{A_2} \dots \newkarrowx{A_n} \lambda_{A_n} \in \mathit{exec}^{\ks}(\gamma_0)$,
    and $\lambda_0 \newkarrowx{B_1} \lambda_{B_1} \newkarrowx{B_2} \dots \newkarrowx{B_m} \lambda_{B_m} \in \mathit{exec}^{\ks}(\gamma_0)$,
    if $$
      \{\mathbb{M}(\lambda_{A_i}) \mid  i \in \{ 1, \dots, n\}\}^{\#}\ = \{\mathbb{M}(\lambda_{B_j}) \mid  j \in \{ 1, \dots, m\}\}^{\#} 
    $$
    then we have that
    $$
         v_A = v_B
    $$ where $(\chi,v_A) \in \mathbb{G}(\lambda_{A_n},x), (\chi,v_B) \in \mathbb{G}(\lambda_{B_m},x)$.}
\end{definition}

\begin{definition}(\ind, \IND)
 \textit{Given a variable $\chi$ and a sequence $V$ containing a term denoting $k$, we define \ind$(\chi,V)$ as the term denoting $\chi$ in $V$. Given a set $\eta$ of variables and a sequence $V$ containing terms denoting variables in $\eta$, we define \IND$(\eta,V) = \{\ind(\chi,V) \mid \chi \in \eta\}$.}
\end{definition}




\begin{definition}{($\phi_{inv}^{R_s}$)}
    \label{definition: r_satisfy_inv_ind}
    \textit{
    Let $R_s \vDash \beta_{R_s}$.
    Given a set of global variables $\eta$ of contract $x$ and a function $\IND$, define the invariant property $\phi_{inv}^{R_s}(\eta,\IND,x)$  for $R_s$ system as follows:}

    \textit{
    For all executions $F_0 \exer{r_1} F_1 \exer{r_2} \dots \exer{r_n} F_n \in \textit{exec}^{m s r}(R_s)$, 
    for all valuations $\mathcal{E}$,
    if there exists $j \in \{0,\dots,n\}$ and $$ F_{j} \exer{\kwt{init\_gvars}} F_{j+1}$$ and 
    $$\forall V \ins F_j.\ \textit{names}(\{V\})\not=\{\kwf{Gvar}^x \}\land $$ $$ \exists V' \ins F_{j+1}.\ \textit{names}(\{V'\})=\{\kwf{Gvar}^x \}$$
    then we have that
    $$
       \sum_{t\in \IND(\eta,\gvars(V_{j+1})) }(\mathcal{E}(t)) = \sum_{t\in \IND(\eta,\gvars(V_n)) }(\mathcal{E}(t))
    $$
    where  $V_i \ins F_i, \textit{names}(\{V_i\})=\{\kwf{Gvar}^x \}$.
    }
\end{definition}

\begin{definition}{($\phi^{R_{inv}}$)}
    \label{definition: r_satisfy_inv}
     \textit{
    Let $R_{inv} \vDash \beta_{R_{inv}}$.
    Given a set of global variables $\eta$ of contract $x$ and a function $\IND$, define the invariant property $\phi^{R_{inv}}(\eta,\IND,x)$  for $R_{inv}$ system as follows:}

    \textit{
    For all executions $F_0 \exeri{r_1} F_1 \exeri{r_2} \dots \exeri{r_n} F_n \in \textit{exec}^{m s r}(R_{inv})$,
    for all valuations $\mathcal{E}$, if there exists $j \in \{0,\dots,n\}$ and $$F_{j} \exeri{\kwt{init\_gvars\_inv}} F_{j+1}$$ and 
    $$\forall V \ins F_j.\ \textit{names}(\{V\})\not=\{\kwf{Gvar}^x \}\land $$ $$\exists V' \ins F_{j+1}.\ \textit{names}(\{V'\})=\{\kwf{Gvar}^x \}$$
    then we have that
    $$
       \sum_{t\in \IND(\eta,\gvars(V_{j+1})) }(\mathcal{E}(t)) = \sum_{t\in \IND(\eta,\gvars(V_n)) }(\mathcal{E}(t))
    $$
    where  $V_i \ins F_i, \textit{names}(\{V_i\})=\{\kwf{Gvar}^x \}$. }
\end{definition}

\begin{definition}{(Sub-executions)}
 Given a multiset rewriting system $R$, an execution $e_1 = F_0 \exe{r_1} F_1 \exe{r_2} \dots \exe{r_n} F_n \in \textit{exec}^{m s r}(R)$ and a sequence of transitions $e_2 = F_i \exe{r_{i+1}} \dots \exe{r_j} F_j$ where $0 \leq i < j \leq n$, define $e_2$ as a sub-execution of $e_1$, writtern $e_2 \sqsubseteq e_1$.
\end{definition}



\begin{definition}{($\phi_{equ}^{R_s}$)}
    \label{definition: r_satisfy_equ_ind}
    \textit{
    Let $R_s \vDash \beta_{R_s}$.
    Given a global variable $\chi$ of contract $x$ and a function $\ind$, define the equivalence property $\phi_{equ}^{R_s}(\chi,\ind,x)$ for $R_s$ system as follows:}

    \textit{
    Given any two executions $e_A = F_0 \exer{r_{A_1}} F_{A_1} \exer{r_{A_2}} \dots \exer{r_{A_n}} F_{A_n} \in \textit{exec}^{m s r}(R_s)$ and $e_B = F_0 \exer{r_{B_1}} F_{B_1} \exer{r_{B_2}} \dots \exer{r_{B_m}} F_{B_m} \in \textit{exec}^{m s r}(R_s)$, for all valuations $\mathcal{E}$, if the following conditions hold:
    \begin{enumerate}[label=(\alph*)]
    \item $\{V \mid \forall k \in \{1,\dots,n-1\}.\ V \ins F_{A_k} \land V \not\ins F_{A_{k+1}}\land  \textit{names}(\{V\})=\{\kwf{Call_e} \}    \}^{\#} = \{V \mid \forall k \in \{1,\dots,m-1\}.\ V \ins F_{B_k} \land V \not\ins F_{B_{k+1}}\land  \textit{names}(\{V\})=\{\kwf{Call_e} \}    \}^{\#}$
    \item  $\exists j \in \{1,\dots,n-1\}.\ \exists j' \in \{1,\dots,m-1\}.\ \exists V_A \ins F_{A_{j+1}}.\ \exists V_B \ins F_{B_{j'+1}}.\ V_A \not\ins F_{A_{j}} \land V_B \not\ins F_{B_{j'}} \land \textit{names}(\{V_A\})=\{\kwf{Gvar}^x \} \land \textit{names}(\{V_B\})=\{\kwf{Gvar}^x \} \rightarrow V_A = V_B $
    \end{enumerate}
    then we have that
    $$
       \mathcal{E}(\ind(\chi,\gvars(V_{A_{n}}))) = \mathcal{E}(\ind(\chi,\gvars(V_{B_{m}})))
    $$
    where $V_{A_{n}} \ins F_{A_n}, \textit{names}(\{V_{A_{n}}\})=\{\kwf{Gvar}^x \}, V_{B_{m}} \ins F_{B_m},$ $ \textit{names}(\{V_{B_{m}}\})=\{\kwf{Gvar}^x \}$.}
\end{definition}

\begin{definition}{($\phi^{R_{equ}}$)}
    \label{definition: r_satisfy_equ}
    \textit{
    Let $R_{equ} \vDash \beta_{R_{equ}}$.
    Given a global variable $\chi$ of contract $x$ and a function $\ind$, define the equivalence property $\phi^{R_{equ}}(\chi,\ind,x)$ for $R_s$ system as follows:}

    \textit{
    For all executions $e = F_0 \exere{r_1} F_1 \exere{r_2} \dots \exere{r_n} F_n \in \textit{exec}^{m s r}(R_{equ})$, 
    for all valuations $\mathcal{E}$, 
    $$
        \mathcal{E}(\ind(\chi,\gvars(V_{A_n}))) = \mathcal{E}(\ind(\chi,\gvars(V_{B_n})))
    $$
    where $V_{A_n},V_{B_n} \ins F_n, \textit{names}(\{V_{A_n}\})=\{\kwf{Gvar_A}^x \}, \textit{names}(\{V_{B_n}$ $\})=\{\kwf{Gvar_B}^x \}$.}
\end{definition}

\begin{lemma}
    \label{lemma: soundness}
    Let $R_{inv} \vDash \beta_{R_{inv}}$ and $R_{s} \vDash \beta_{R_{s}}$. 
    Given a set of global variables $\eta$ of contract $x$ and a function $\IND$, if $\phi^{R_{inv}}(\eta,\IND,x)$ holds, then $\phi_{inv}^{R_s}(\eta,\IND,x)$ holds.
\end{lemma}

\begin{proof}
    We prove this theorem by contradiction and firstly propose an assumption: $\phi^{R_{inv}}(\eta,\IND,x)$ holds and $\phi_{inv}^{R_s}(\eta,\IND,x)$ does not hold. 
    \begin{enumerate}
        \item By the above assumption and Definition \ref{definition: r_satisfy_inv_ind}, we have that there exists an execution 
               $e = F_0 \exer{r_1} F_1 \exer{r_2} \dots \exer{r_n} F_n \in \textit{exec}^{m s r}(R_s)$, and a valuation $\mathcal{E}$ and there exists $j \in \{0,\dots,n\}$ and $$ F_{j} \exer{\kwt{init\_gvars}} F_{j+1}$$ and   $$ F_{j} \exer{\kwt{init\_gvars}} F_{j+1}$$ and 
                $$\forall V \ins F_j.\ \textit{names}(\{V\})\not=\{\kwf{Gvar}^x \}\land $$ $$ \exists V' \ins F_{j+1}.\ \textit{names}(\{V'\})=\{\kwf{Gvar}^x \}$$
                and
                $$
                   \sum_{t\in \IND(\eta,\gvars(V_{j+1})) }(\mathcal{E}(t)) \not= \sum_{t\in \IND(\eta,\gvars(V_n)) }(\mathcal{E}(t))
                $$
                where  $V_i \ins F_i, \textit{names}(\{V_i\})=\{\kwf{Gvar}^x \}$.

        \item By the definition of $\beta_{r_n}$, $\exists r \in \mathit{traces}^{msr}(R_s). r = \kwt{ret\_ext}$.

        \item By (2) and the definition of rule $\kwt{ret\_ext}$, $\exists r \in \mathit{traces}^{msr}(R_s). r = \kwt{recv\_ext}$.

        \item By (1), (2), (3) and the definitions of rule $\kwt{ret\_ext}$ and $\kwt{recv\_ext}$, there exists a transition $e' = F_{k-1} \exer{r_{k}}  F_{k} \dots \exer{r_{k+u}} F_{k+u}$ where $r_{k} = \kwt{recv\_ext} \land r_{k+u} = \kwt{ret\_ext} \land \forall r \in \{r_{k+1},\dots,r_{k+u}\}. r\not= \kwt{recv\_ext}$ and $e' \sqsubseteq  e$.


        \item By the definitions of rules in $R_s$, $\forall r \ins \{r_1,\dots,r_n\}^{\#}. r \not\ins \{r_{k},\dots,r_{k+u}\}^{\#} \rightarrow r \in \{\kwt{init\_gvars},\kwt{init\_evars},\kwt{ext\_call}\}$. 


        \item By the definition of $\kwt{ext\_call}$, given a transition $F_{k-1} \exer{r_{k}}  F_{k}$, 
              $r_k = \kwt{ext\_call} \rightarrow $ 
              $ \gvars(V_{k-1}) = \gvars(V_{k}) $
              where $V_i \ins F_i, \textit{names}(\{V_i\})=\{\kwf{Gvar}^x \}$.

        \item By the definition of $\kwt{init\_evars}$, , given a transition $F_{k-1} \exer{r_{k}}  F_{k}$,
              $r_k = \kwt{init\_evars} \rightarrow $ 
              $ \gvars(V_{k-1}) = \gvars(V_{k}) $
              where $V_i \ins F_i, \textit{names}(\{V_i\})=\{\kwf{Gvar}^x \}$.

        \item By (1) and the definition of $\beta_{initG}$, given a a transition $F_{k-1} \exer{r_{k}}  F_{k} (j+1 \leq k \leq n)$,
              $r_k = \kwt{init\_evars} \rightarrow $ 
              $ \gvars(V_{k-1}) = \gvars(V_{k}) $
              where $V_i \ins F_i, \textit{names}(\{V_i\})=\{\kwf{Gvar}^x \}$. 

        \item By (5), (6), (7), (8),  $\gvars(V_{j+1}) = \gvars(V_{k-1})$ and $\gvars(V_{k+u}) = \gvars(V_{n})$  where $V_i \ins F_i, \textit{names}(\{V_i\})=\{\kwf{Gvar}^x \}, j+1 \leq k \leq n$.

        \item By (1), (9),
              $$ \sum_{t\in \IND(\eta,\gvars(V_{j+1})) }(\mathcal{E}(t)) = \sum_{t\in \IND(\eta,\gvars(V_{k-1})) }(\mathcal{E}(t)) \land $$
              $$ \sum_{t\in \IND(\eta,\gvars(V_{k-1})) }(\mathcal{E}(t)) \not= \sum_{t\in \IND(\eta,\gvars(V_{k+u})) }(\mathcal{E}(t)) $$
              where $V_i \ins F_i, \textit{names}(\{V_i\})=\{\kwf{Gvar}^x \}, k \geq j+2$.

        \item By (10), there exists an execution 
              $e'' = F_0 \exer{\kwt{init\_evars}} F_1' \exer{\kwt{init\_gvars}} F_2' \exer{r_3'} F_3' \dots \exer{r_{u+3}'} F_{u+3}' \in \textit{exec}^{m s r}(R_s)$
              where $r_i' = r_{i-k+3} $ and 
              $$ \sum_{t\in \IND(\eta,\gvars(V_{3})) }(\mathcal{E}(t)) \not= \sum_{t\in \IND(\eta,\gvars(V_{u+3})) }(\mathcal{E}(t)) $$
              where $V_i \ins F_i', \textit{names}(\{V_i\})=\{\kwf{Gvar}^x \}$.

        \item By the definitions of $\kwt{init\_gvars}$ and $\kwt{init\_gvars\_inv}$ and Lemma \ref{lemma:simpletrans}, $F \exer{\kwt{init\_gvars}}  F' \land F \exeri{\kwt{init\_gvars\_inv}}  F'' \rightarrow F' = F''$ .

        \item Similar to (12), $F \exer{\kwt{ext\_call}}  F' \land F \exeri{\kwt{ext\_call\_inv}}  F'' \rightarrow F' = F''$ .

        \item Similar to (12), $F \exer{\kwt{ret\_ext}}  F' \land F \exeri{\kwt{ret\_ext\_inv}}  F'' \rightarrow F' = F''$ .

        \item By (11), (12), (13), (14), there exists an execution 
              $e_{inv} = F_0 \exeri{\kwt{init\_evars}} F_1' \exeri{\kwt{init\_gvars\_inv}} F_2' \exeri{\kwt{ext\_call\_inv}} F_3' \dots \exeri{\kwt{ret\_ext\_inv}} F_{u+3}'$
              and we have that
              $$ \sum_{t\in \IND(\eta,\gvars(V_{3})) }(\mathcal{E}(t)) \not= \sum_{t\in \IND(\eta,\gvars(V_{u+3})) }(\mathcal{E}(t)) $$
              where $V_i \ins F_i', \textit{names}(\{V_i\})=\{\kwf{Gvar}^x \}$.

        \item By the definitions of $e_{inv}$ in (15) and $\beta_{R_{inv}}$, we have that $e_{inv} \in \textit{exec}^{m s r}(R_{inv})$ where $R_{inv} \vDash \beta_{R_{inv}}$.

        \item By (15),(16), $\phi^{R_{inv}}(\eta,\IND,x)$ does not hold, which is in contradiction to the assumption. Thus Lemma \ref{lemma: soundness} proved.

    \end{enumerate}

\end{proof}

\begin{lemma}
    \label{lemma: soundness2}
    Let $R_{equ} \vDash \beta_{R_{equ}}$ and $R_{s} \vDash \beta_{R_{s}}$. Given a global variable $\chi$ of contract $x$ and a function $\ind$, if $\phi^{R_{equ}}(\chi,\ind,x)$ holds, then $\phi_{equ}^{R_s}(\chi,\ind,x)$ holds.
\end{lemma}

\begin{proof}
    We prove this theorem by contradiction and firstly propose an assumption: $\phi^{R_{equ}}(\chi,\ind,x)$ holds and $\phi_{equ}^{R_s}(\chi,\ind,x)$ does not hold. 
    We also assume that there are $t$ contracts. 
    \begin{enumerate}
        \item By the above assumption and Definition \ref{definition: r_satisfy_equ_ind}, there exist two executions $e_A = F_0 \exer{r_{A_1}} F_{A_1} \exer{r_{A_2}} \dots \exer{r_{A_n}} F_{A_n} \in \textit{exec}^{m s r}(R_s)$ and $e_B = F_0 \exer{r_{B_1}} F_{B_1} \exer{r_{B_2}} \dots \exer{r_{B_m}} F_{B_m} \in \textit{exec}^{m s r}(R_s)$ and a valuation $\mathcal{E}$ and the following conditions hold:
        \begin{enumerate}[label=(\alph*)]
          \item $\{V \mid \forall k \in \{1,\dots,n-1\}.\ V \ins F_{A_k} \land V \not\ins F_{A_{k+1}}\land  \textit{names}(\{V\})=\{\kwf{Call_e} \}    \}^{\#} = \{V \mid \forall k \in \{1,\dots,n-1\}.\ V \ins F_{B_k} \land V \not\ins F_{B_{k+1}}\land  \textit{names}(\{V\})=\{\kwf{Call_e} \}    \}^{\#}$
    
          \item  $\exists j,j' \in \{1,\dots,n-1\}.\ \exists V_A \ins F_{A_{j+1}}.\ \exists V_B \ins F_{B_{j'+1}}.\ V_A \not\ins F_{A_{j}} \land V_B \not\ins F_{B_{j'}} \land \textit{names}(\{V_A\})=\{\kwf{Gvar}^x \} \land \textit{names}(\{V_B\})=\{\kwf{Gvar}^x \} \rightarrow V_A = V_B $
         
         \item   $
                    \mathcal{E}(\ind(\chi,\gvars(V_{A_{n}}))) \not= \mathcal{E}(\ind(\chi,\gvars(V_{B_{m}})))
                  $

                where $V_{A_{n}} \ins F_{A_n}, \textit{names}(\{V_{A_{n}}\})=\{\kwf{Gvar}^x \}, V_{B_{m}} \ins F_{B_m},$ $ \textit{names}(\{V_{B_{m}}\})=\{\kwf{Gvar}^x \}$.
    \end{enumerate}

    \item By the definition of rule $\kwt{init\_gvars}$, given a transition $F \exer{\kwt{init\_gvars}} F'$, $\exists V \ins F.\ \textit{names}(\{V\})=\{\kwf{Call_e} \} \rightarrow \exists V' \ins F'.\ \textit{names}(\{V'\})=\{\kwf{Call_e} \} \land V = V'$.

    \item By the definition of rule $\kwt{init\_gvars}$, $\forall F_0 \exer{r_1} F_1 \dots \exer{\kwt{init\_gvars}} F_{k} \dots \exer{r_{k+u}} F_{k+u} \in \textit{exec}^{m s r}(R_s)$. $\exists F_0 \exer{\kwt{init\_gvars}} F_1' \dots \exer{r_{k-1}} F_{k}' \dots \exer{r_{k+u}} F_{k+u}' \in \textit{exec}^{m s r}(R_s)$. $F_{k} = F_{k}' \land F_{k+u} = F_{k+u}'$.

    \item By the condition (a) in (1), and the definition of rule $\kwt{recv\_ext}$, $\{r \mid r \ins \{r_{A_1},\dots, r_{A_n}\}^{\#} \land r = \kwt{recv\_ext}\}^{\#} = \{r \mid r \ins \{r_{B_1},\dots, r_{B_m}\}^{\#} \land r = \kwt{recv\_ext}\}^{\#}$.

    \item By (4), $\beta_{initG}$ and the definition of rule $\kwt{init\_gvars}$, $\vert \{r \mid r \ins \{r_{A_1},\dots, r_{A_n}\}^{\#} \land r = \kwt{init\_gvars}\}^{\#} \vert^{\#}= \vert\{r \mid r \ins \{r_{B_1},\dots, r_{B_m}\}^{\#} \land r = \kwt{init\_gvars}\}^{\#}\vert^{\#} = t$.

    \item By (5), there exist two executions $e_A' = F_0 \exer{r_{A_1}'} F_{A_1}' \exer{r_{A_2}'} \dots \exer{r_{A_n}'} F_{A_n}' \in \textit{exec}^{m s r}(R_s)$ and $e_B' = F_0 \exer{r_{B_1}'} F_{B_1}' \exer{r_{B_2}'} \dots \exer{r_{B_m}'} F_{B_m}' \in \textit{exec}^{m s r}(R_s)$, where $\forall i \in \{1,\dots,t\}. r_{A_i}' = r_{B_i}' = \kwt{init\_gvars} $ and $ [r_{A_{t+1}}',\dots,r_{A_{n}}'] = [r_{A_1},\dots,r_{A_n}]\backslash [r_{A_1}',\dots,r_{A_t}']$ and $ [r_{B_{t+1}}',\dots,r_{B_{m}}'] =$ 

    $ [r_{B_1},\dots,r_{B_m}]\backslash [r_{B_1}',\dots,r_{B_t}']$.

    \item By (1), (2), (5),  $\{V \mid \forall k \in \{1,\dots,n-1\}.\ V \ins F_{A_k}' \land V \not\ins F_{A_{k+1}}' \land  \textit{names}(\{V\})=\{\kwf{Call_e} \}    \}^{\#} = \{V \mid \forall k \in \{1,\dots,n-1\}.\ V \ins F_{B_k}'  \land V \not\ins F_{B_{k+1}}' \land  \textit{names}(\{V\})=\{\kwf{Call_e} \}    \}^{\#}$. 

    \item By (1), (3), (5), $\exists j,j' \in \{1,\dots,t\}.\ \exists V_A \ins F_{A_{j+1}}'.\ \exists V_B \ins F_{B_{j'+1}}'.\ V_A \not\ins F_{A_{j}}' \land V_B \not\ins F_{B_{j'}}' \land \textit{names}(\{V_A\})=\{\kwf{Gvar}^x \} \land \textit{names}(\{V_B\})=\{\kwf{Gvar}^x \} \rightarrow V_A = V_B $. 

    \item By (1), (3), (5),   $
                    \mathcal{E}(\ind(\chi,\gvars(V_{A_{n}}'))) \not= \mathcal{E}(\ind(\chi,\gvars(V_{B_{m}}')))
                  $

                where $V_{A_{n}}' \ins F_{A_n}', \textit{names}(\{V_{A_{n}}'\})=\{\kwf{Gvar}^x \}, V_{B_{m}}' \ins F_{B_m}',$ $ \textit{names}(\{V_{B_{m}}'\})=\{\kwf{Gvar}^x \}$.

    \item By (6), (8), $\forall V_A \ins F_{A_{t}}'.\ \textit{names}(\{V_{A}\})=\{\kwf{Gvar}^x \} \rightarrow \exists V_B \ins F_{B_{t}}'.\ V_A = V_B$ and $\forall V_B \ins F_{B_{t}}'.\ \textit{names}(\{V_{B}\})=\{\kwf{Gvar}^x \} \rightarrow \exists V_A \ins F_{A_{t}}'.\ V_A = V_B$.

    \item By (6), (10) and the definition of $\kwt{init\_gvars\_AB}$, there exists an execution $F_0  \exere{r_1} F_1 \dots \exere{r_{t}} F_{t}$ where $\forall r \in \{r_1,\dots,r_{t}\}.\ r = \kwt{init\_gvars\_AB}$ such that $$\forall V \ins F_{t}. \textit{names}(\{V\})=\{\kwf{Gvar_A}^x \} \rightarrow \exists V_A \ins F_{A_t}'. $$ $$\textit{names}(\{V_A\})=\{\kwf{Gvar}^x \} \land \myterms(V) = \myterms(V_A)$$ and $$\forall V \ins F_{t}. \textit{names}(\{V\})=\{\kwf{Gvar_B}^x \} \rightarrow \exists V_B \ins F_{B_t}'.$$ $$ \textit{names}(\{V_B\})=\{\kwf{Gvar}^x \} \land \myterms(V) = \myterms(V_B)$$

    \item By the definition of R' and Table \ref{table:rr_correspondence}, given transitions $F \exer{r} F' $ and $F_{A} \exere{F_R(r,\kwf{A})} F_{A}'$, $\forall V \ins F.\ \exists V_A \ins F_A.\ \myterms(V) = \myterms(V_A) \rightarrow \forall V' \ins F'.\ \exists V_A' \ins F_A'.\ \myterms(V') = \myterms(V_A')$.

    \item   Similar to (12), given transitions $F \exer{r} F' $ and $F_{B} \exere{f_R(r,\kwf{B})} F_{B}'$, $\forall V \ins F.\ \exists V_B \ins F_B.\ \myterms(V) = \myterms(V_B) \rightarrow \forall V' \ins F'.\ \exists V_B' \ins F_B'.\ \myterms(V') = \myterms(V_B')$.

    \item By (11), there exists an execution $e_{equ} = F_0  \exere{r_1} F_1 \dots \exere{r_{t}} F_{t} \dots \exere{r_{m+n-t}} F_{m+n-t} \exere{\kwt{compare\_AB}} F_{m+n-t+1} $ where  
    \begin{equation*}
    r_k:=
    \begin{cases}
        f_R(r_{A_k},\kwf{A}) & \textrm{if } k\in \{t+1,...,n\} \\
        f_R(r_{B_{k-n+t}},\kwf{B}) & \textrm{if } k\in \{n+1,...,m+n-t\}
    \end{cases}
    \end{equation*}

    \item By (11), (12), (13), (14), 
    $$
        \mathcal{E}(\ind(\chi,\gvars(V_{A_{m+n-t}}))) = \mathcal{E}(\ind(\chi,\gvars(V_{A_{n}}'))) 
    $$
    and
    $$
     \mathcal{E}(\ind(\chi,\gvars(V_{B_{m}}')))= \mathcal{E}(\ind(\chi,\gvars(V_{B_{m+n-t}})))
    $$
    where $V_{A_{m+n-t}},V_{B_{m+n-t}} \ins F_{m+n-t}, \textit{names}(\{V_{A_{m+n-t}}\})=\{\kwf{Gvar_A}^x \}, \textit{names}(\{V_{B_{m+n-t}}\})=\{\kwf{Gvar_B}^x \}$ and $V_{A_{n}}' \ins F_{A_n}', \textit{names}(\{V_{A_{n}}'\})=\{\kwf{Gvar}^x \}, V_{B_{m}}' \ins F_{B_m}',$ $ \textit{names}(\{V_{B_{m}}'\})=\{\kwf{Gvar}^x \}$.

    \item By (9), (15), 
    $$
         \mathcal{E}(\ind(\chi,\gvars(V_{A_{m+n-t}}))) \not= \mathcal{E}(\ind(\chi,\gvars(V_{B_{m+n-t}})))
    $$
    where $V_{A_{m+n-t}},V_{B_{m+n-t}} \ins F_{m+n-t}, \textit{names}(\{V_{A_{m+n-t}}\})=\{\kwf{Gvar_A}^x \}, \textit{names}(\{V_{B_{m+n-t}}$ $\})=\{\kwf{Gvar_B}^x \}$.

    \item By (16) and the definition of $\kwt{compare\_AB}$, 
    $$
        \mathcal{E}(\ind(\chi,\gvars(V_{A_{m+n-t+1}}))) \not= \mathcal{E}(\ind(\chi,\gvars(V_{B_{m+n-t+1}})))
    $$
    where $V_{A_{m+n-t+1}},V_{B_{m+n-t+1}} \ins F_{m+n-t+1},$ $ \textit{names}(\{V_{A_{m+n-t+1}}\})$ $=\{\kwf{Gvar_A}^x \}, \textit{names}(\{V_{B_{m+n-t+1}}$ $\})=\{\kwf{Gvar_B}^x \}$.

    \item By the definition of $\kwt{ext\_call\_AB}$, $\{\myterms(V) \mid V\ins F_{t} \land \textit{names}(\{V\})  = \{\kwf{Call_{Ae}}\}  \}^{\#} =  \{\myterms(V) \mid V\ins F_{t}  \land \textit{names}(\{V\}) = \{\kwf{Call_{Be}}\}  \}^{\#}$

    \item By (8), (14), (18), 
    $\{\myterms(V) \mid V\ins F_{m+n-t+1} \land \textit{names}(\{V\})  = \{\kwf{Call_{Ae}}\}  \}^{\#} =  \{\myterms(V) \mid V\ins F_{m+n-t+1}  \land \textit{names}(\{V\}) = \{\kwf{Call_{Be}}\}  \}^{\#}$

    \item By the definitions of $e_{equ}$ in (14) and (19), we have that $e_{equ} \in \textit{exec}^{m s r}(R_{equ})$ where $R_{equ} \vDash \beta_{R_{equ}}$.

    \item By (11), (17), (20), $\phi^{R_{equ}}(\chi\ind,x)$ does not hold, which is in contradiction to the assumption. Thus Lemma \ref{lemma: soundness2} proved.
    \end{enumerate}
\end{proof}


\begin{lemma}
    \label{lemma: soundness3}
    Let   $(\mathcal{T},\karroww, \gamma_0)$ be a K transition system that satisfies $\ks \vDash \alpha$.
    Let $R_{s} \vDash \beta_{R_{s}}$. Given a set of global variables $\eta$ of contract $x$ and a function $\IND$,
    if  $\phi_{inv}^{R_s}(\eta,\IND,x)$ holds, then $\phi_{inv}^{\ks}(\lambda_0,\eta,x)$ holds.
\end{lemma}

\begin{proof}
\label{proof:soundness3}
    We prove this theorem by contradiction and firstly propose an assumption: there exists a function $\IND$ such that $\phi_{inv}^{R_s}(\IND,\eta,x)$ holds and $\phi_{inv}^{\ks}(\lambda_0,\eta,x)$ does not hold. 
    We also assume that there are $t$ contracts $C = \{x_1,\dots,x_t\}$ and $x \in C$. 
    \begin{enumerate}
        \item By the assumption and Definition \ref{definition: k_satisfy_inv}, there exists an execution $e = \lambda_0 \newkarrowx{1} \lambda_1 \newkarrowx{2} \dots \newkarrowx{n} \lambda_n \in \mathit{exec}^{\ks}(\gamma_0)$,  we have that
        $$
          \sum_{v\in G_0}(v) \not= \sum_{v\in G_n}(v)
        $$
        where $G_i = \{v \mid \chi \in \eta\ \land\ (\chi,v) \in \mathbb{G}(\lambda_i,x) \}$.
        \item By Lemma \ref{lemma:var_soundness2}, there exist $e_r = \emptyset \exer{r_1} F_1 \exer{r_2} \dots \exer{r_{m'}} F_{m'} \in \textit{exec}^{m s r}(R_s)$ and a valuation $\mathcal{E}$ and 
            \begin{enumerate}[label=(\alph*)]
                \item $ \exists j \in \{1,\dots,n\}.\ \forall x \in \mathbb{S}(\lambda_0).\ \exists V\ins F_{j}.\ \mathbb{G}(\lambda_0,x) \bi \gvars(V) $
                \item $ \forall x \in \mathbb{S}(\lambda_n).\ \exists V\ins F_{m'}.\ \mathbb{G}(\lambda_n,x) \bi \gvars(V)$
            \end{enumerate}
        \item By the definition of rule $\kwt{init\_gvars}$, $\forall F_0 \exer{r_1} F_1 \dots$ $\exer{\kwt{init\_gvars}} F_{k} \dots \exer{r_{k+u}} F_{k+u} \in \textit{exec}^{m s r}(R_s)$. 

        $\exists F_0 \exer{\kwt{init\_gvars}} F_1' \dots \exer{r_{k-1}} F_{k}' \dots \exer{r_{k+u}} F_{k+u}' \in \textit{exec}^{m s r}(R_s)$. $F_{k} = F_{k}' \land F_{k+u} = F_{k+u}'$.
        \item By (2), (3), there exists an execution $e_r' = \emptyset \exer{r_1'} F_1' \exer{r_2'} \dots \exer{r_{m'}'} F_{m'}' \in \textit{exec}^{m s r}(R_s)$ where $\forall r \in \{r_1',\dots,r_t'\}.\ r = \kwt{init\_gvars}$. And there exists a valuation $\mathcal{E}$ such that 
            \begin{enumerate}[label=(\alph*)]
                \item $ \exists j \in \{1,\dots,n\}.\ \forall x \in \mathbb{S}(\lambda_0).\ \exists V\ins F_{j}'.\ \mathbb{G}(\lambda_0,x) \bi \gvars(V) $
                \item $ \forall x \in \mathbb{S}(\lambda_n).\ \exists V\ins F_{m'}'.\ \mathbb{G}(\lambda_n,x) \bi \gvars(V)$
            \end{enumerate}
        \item By eliminating $\exists$ on $j$, we let $j=t$.
        \item By the definition of $\beta_{initG}$, $\exists k.\ 0 \leq k \leq t-1$  $$ F_{k}'' \exer{\kwt{init\_gvars}} F_{k+1}'' $$ and 
             $$\forall V_k \ins F_k'.\ \textit{names}(\{V_k\})\not=\{\kwf{Gvar}^x \}\land $$ $$ \exists V_{k+1} \ins F_{k+1}'.\exists V_t \ins F_{t}'.\ \textit{names}(\{V_{k+1}\})=\textit{names}(\{V_t\})=$$ $$\{\kwf{Gvar}^x \} \land V_{k+1} = V_t$$
        \item By (4), $\forall V \ins F_{k+1}'.\ \exists x' \in C. \ \textit{names}(\{V\} = \{\kwf{Gvar}^{x'} \}$.
        \item By (6), (7),  $ \forall x \in \mathbb{S}(\lambda_0).\ \exists V\ins F_{k+1}'.\ \mathbb{G}(\lambda_0,x) \bi \gvars(V) \land \textit{names}(\{V\} = \{\kwf{Gvar}^{x} \}  $.
        \item By the definition of $\bi$, there exists a function $\ind$ mapping from $\{\chi \mid (\chi,v) \in \mathbb{G}(\lambda_0,x)\}$ to $\{t \mid t \in \gvars(V)\land V\ins F_{k+1}' \land \textit{names}(\{V\} = \{\kwf{Gvar}^{x} \} \}$ and a function \IND$(\eta) = \{\ind(\chi) \mid \chi \in \eta\}$.
        \item By Definition \ref{def:LGE} and the definition of $\ks$, $\{\chi \mid (\chi,v) \in \mathbb{G}(\lambda_0,x)\} = \{\chi \mid (\chi,v) \in \mathbb{G}(\lambda_n,x)\}$.
        \item By (2), (8), (9), (10) and the definition of $\mathcal{R}$, there exists a function $\IND$ such that
        $$ \sum_{v\in G_0}(v) = \sum_{b \in \IND(\eta,\gvars(V_{k+1})) }(\mathcal{E}(b)) \land $$
        $$ \sum_{v\in G_n}(v) = \sum_{b \in \IND(\eta,\gvars(V_{m'})) }(\mathcal{E}(b)) $$
        where  $G_i = \{v \mid \chi \in \eta\ \land\ (\chi,v) \in \mathbb{G}(\lambda_i,x) \}$ and $V_i \ins F_i', \textit{names}(\{V_i\})=\{\kwf{Gvar}^x \}$.
        \item By (11),
        $$ \sum_{b \in \IND(\eta,\gvars(V_{k+1})) }(\mathcal{E}(b)) \not= \sum_{b \in \IND(\eta,\gvars(V_{m'})) }(\mathcal{E}(b)) $$
        where  $V_i \ins F_i', \textit{names}(\{V_i\})=\{\kwf{Gvar}^x \}$.
        \item By (6), (12),  $\phi_{inv}^{R_s}(\IND(\eta),x)$ does not hold, which is in contradiction to the assumption. Thus Lemma \ref{lemma: soundness3} proved.

    \end{enumerate}
\end{proof}

\begin{lemma}
    \label{lemma: soundness4}
    Let   $(\mathcal{T},\karroww, \gamma_0)$ be a K transition system that satisfies $\ks \vDash \alpha$.
    Let $R_{s} \vDash \beta_{R_{s}}$. Given a global variable $\chi$ of contract $x$ and a function $\ind$, if $\phi_{equ}^{R_s}(\chi,\ind,x)$ holds, then $\phi_{equ}^{\ks}(\lambda_0,\chi,x)$ holds.
\end{lemma}

\begin{proof}
    We prove this theorem by contradiction and firstly propose an assumption: there exists a function $\ind$ such that $\phi_{equ}^{R_s}(\chi,\ind,x)$ holds and $\phi_{equ}^{\ks}(\lambda_0,\chi,x)$ does not hold. 
    We also assume that there are $t$ contracts $C = \{x_1,\dots,x_t\}$ and $x \in C$. 
  \begin{enumerate}
    \item By the assumption, there exist two executions $e_A = \lambda_0 \newkarrowx{A_1} \lambda_{A_1} \newkarrowx{A_2} \dots \newkarrowx{A_n} \lambda_{A_n} \in \mathit{exec}^{\ks}(\gamma_0)$,
    and $e_B = \lambda_0 \newkarrowx{B_1} \lambda_{B_1} \newkarrowx{B_2} \dots \newkarrowx{B_m} \lambda_{B_m} \in \mathit{exec}^{\ks}(\gamma_0)$ and 
    $$
      \{\mathbb{M}(\lambda_{A_i}) \mid  i \in \{ 1, \dots, n\}\}\ = \{\mathbb{M}(\lambda_{B_j}) \mid  j \in \{ 1, \dots, m\}\}\land
    $$
    $$
         v_A \not= v_B
    $$ where $(\chi,v_A) \in \mathbb{G}(\lambda_{A_n},x), (\chi,v_B) \in \mathbb{G}(\lambda_{B_m},x)$.

    \item By Lemma \ref{lemma:var_soundness2}, there exist $e_{Ar} = \emptyset \exer{r_{A_1}} F_{A_1} \exer{r_{A_2}} \dots \exer{r_{A_{n'}}} F_{A_{n'}} \in \textit{exec}^{m s r}(R_s)$ and a valuation $\mathcal{E}$ and 
            \begin{enumerate}[label=(\alph*)]
                \item $ \exists j \in \{1,\dots,n\}. \forall x \in \mathbb{S}(\lambda_0).\ \exists V\ins F_{A_j}.\ \mathbb{G}(\lambda_0,x) \bi \gvars(V) $.
                \item $ \forall x \in \mathbb{S}(\lambda_{A_n}).\ \exists V\ins F_{A_{n'}}.\ \mathbb{G}(\lambda_{A_n},x) \bi \gvars(V)$.
            \end{enumerate}

    \item By Lemma \ref{lemma:var_soundness2}, there exist $e_{Br} = \emptyset \exer{r_{B_1}} F_{B_1} \exer{r_{B_2}} \dots \exer{r_{B_{m'}}} F_{B_{m'}} \in \textit{exec}^{m s r}(R_s)$ and a valuation $\mathcal{E}$ and 
            \begin{enumerate}[label=(\alph*)]
                \item $ \exists j \in \{1,\dots,m\}.  \forall x \in \mathbb{S}(\lambda_0).\ \exists V\ins F_{B_j}.\ \mathbb{G}(\lambda_0,x) \bi \gvars(V) $.
                \item $ \forall x \in \mathbb{S}(\lambda_{B_m}).\ \exists V\ins F_{B_{m'}}.\ \mathbb{G}(\lambda_{B_m},x) \bi \gvars(V)$.
            \end{enumerate}

    \item By the definition of rule $\kwt{init\_gvars}$, $\forall F_0 \exer{r_1} F_1 \dots \exer{\kwt{init\_gvars}} F_{k} \dots \exer{r_{k+u}} F_{k+u} \in \textit{exec}^{m s r}(R_s)$. $\exists F_0 \exer{\kwt{init\_gvars}} F_1' \dots \exer{r_{k-1}} F_{k}' \dots \exer{r_{k+u}} F_{k+u}' \in \textit{exec}^{m s r}(R_s)$. $F_{k} = F_{k}' \land F_{k+u} = F_{k+u}'$.
        
    \item By (2), (3), (4), there exist executions $e_{Ar}' = \emptyset \exer{r_{A_1}'} F_{A_1}' \exer{r_{A_2}'} \dots \exer{r_{A_{n'}}'} F_{A_{n'}}' \in \textit{exec}^{m s r}(R_s)$ where $\forall r \in \{r_{A_1}',\dots,r_{A_t}'\}.\ r = \kwt{init\_gvars}$ and $e_{Br}' = \emptyset \exer{r_{B_1}'} F_{B_1}' \exer{r_{B_2}'} \dots \exer{r_{B_{m'}}'} F_{B_{m'}}' \in \textit{exec}^{m s r}(R_s)$ where $\forall r \in \{r_{B_1}',\dots,r_{B_t}'\}.\ r = \kwt{init\_gvars}$. And there exists a valuation $\mathcal{E}$ such that 
            \begin{enumerate}[label=(\alph*)]
                \item $ \exists j \in \{1,\dots,n\}. \forall x \in \mathbb{S}(\lambda_0).\ \exists V\ins F_{A_j}'.\ \mathbb{G}(\lambda_0,x) \bi \gvars(V) $.
                \item $ \forall x \in \mathbb{S}(\lambda_{A_n}).\ \exists V\ins F_{A_{n'}}'.\ \mathbb{G}(\lambda_{A_n},x) \bi \gvars(V)$.
                \item $ \exists j \in \{1,\dots,m\}. \forall x \in \mathbb{S}(\lambda_0).\ \exists V\ins F_{B_j}'.\ \mathbb{G}(\lambda_0,x) \bi \gvars(V) $.
                \item $ \forall x \in \mathbb{S}(\lambda_{B_m}).\ \exists V\ins F_{B_{m'}}'.\ \mathbb{G}(\lambda_{B_m},x) \bi \gvars(V)$.
            \end{enumerate}
    
    \item By eliminating $\exists$ on $j$, we let $j=t$.

    \item By (5), $\forall V \ins F_{A_t}'.\ \exists x' \in C. \ \textit{names}(\{V\} = \{\kwf{Gvar}^{x'} \}$.

    \item By (5), $\forall V \ins F_{B_t}'.\ \exists x' \in C. \ \textit{names}(\{V\} = \{\kwf{Gvar}^{x'} \}$.

    \item By (5), (6), (7) and the definition of $\bi$, there exists a valuation $\mathcal{E}$ and $\forall V_A \ins F_{A_t}'.\ \exists V_B \ins F_{B_t}'.\  \textit{names}(\{V_A\} )=  \textit{names}(\{V_B\}) =  \{\kwf{Gvar}^{x} \} \land \forall b_A \in V_A.\ \exists b_B \in V_B.\ \mathcal{E}(b_A) = \mathcal{E}(b_B)$.

    \item Similar to (9), there exists a valuation $\mathcal{E}$ and $\forall V_B \ins F_{B_t}'.\ \exists V_A \ins F_{A_t}'.\  \textit{names}(\{V_B\} )=  \textit{names}(\{V_A\}) =  \{\kwf{Gvar}^{x} \} \land \forall b_B \in V_B.\ \exists b_A \in V_A.\ \mathcal{E}(b_B) = \mathcal{E}(b_A)$.

    \item By (9), (10) and the definition of $\kwt{init\_gvars}$, $\forall V_A \ins F_{A_t}'.\  \textit{names}(\{V_A\}) =  \{\kwf{Gvar}^{x} \}  \rightarrow \exists V_B \ins F_{B_t}'.\ V_A = V_B.$

    \item By the definition of $\bi$, there exists a function $\ind$ mapping from $\{a \mid (a,v) \in \mathbb{G}(\lambda_{A_n},x)\}$ to $\{b \mid b \in \gvars(V)\land V\ins F_{A_t}' \land \textit{names}(\{V\} = \{\kwf{Gvar}^{x} \} \}$.

    \item Similar to (12), there exists a function $\ind$ mapping from $\{a \mid (a,v) \in \mathbb{G}(\lambda_{B_m},x)\}$ to $\{b \mid b \in \gvars(V)\land V\ins F_{B_t}' \land \textit{names}(\{V\} = \{\kwf{Gvar}^{x} \} \}$.

    \item By (12), (13) and the definition of $\mathcal{R}$, there exists a function $\ind$ such that
         $$
           v_A = \mathcal{E}(\ind(\chi,\gvars(V_{A_{n}}))) \not= \mathcal{E}(\ind(\chi,\gvars(V_{B_{m}}))) = v_B
        $$
    where $V_{A_{n}} \ins F_{A_{n'}}', \textit{names}(\{V_{A_{n}}\})=\{\kwf{Gvar}^x \}, V_{B_{m}} \ins F_{B_{m'}}',$ $ \textit{names}(\{V_{B_{m}}\})=\{\kwf{Gvar}^x \}$.

    \item By the definitions of $\mathbb{M}$ and $\ks$, if $\mathbb{M}(\lambda_i) \not= \emptyset$, the only possible trace for the transition $\lambda_{i-2} \newkarrowx{i-1} \lambda_{i-1} \newkarrowx{i} \lambda_{i} \newkarrowx{i+1} \dots \newkarrowx{i+u+8} \lambda_{i+u+8}$ is as follows:
        \begin{itemize}
            \item $\rho_{i-2}= \textsc{Function-Call}$ 
            \item $\rho_{i-1}= \textsc{Switch-Context}$ 
            \item $\rho_{i}= \textsc{Create-Transaction}$ 
            \item $\rho_{i+1}= \textsc{Internal-Function-Call}$ 
            \item $\rho_{i+2}= \textsc{Save-Cur-Context}$ 
            \item $\rho_{i+3}= \textsc{Call}$ 
            \item $\rho_{i+4}= \textsc{Init-Fun-Params}$ 
            \item $\rho_{i+5}= \textsc{Bind-Params}$ 
            \item $\rho_{i+6}= \textsc{Bind-Params}$ 
            \item ...
            \item $\rho_{i+u+6}= \textsc{Bind-Params-End}$ 
            \item $\rho_{i+u+7}= \textsc{processFunQuantifiers}$ 
            \item $\rho_{i+u+8}= \textsc{Call-Function-Body}$ 
        \end{itemize}
    \item  By Proposition~\ref{proposition:transitionsystem} , $[\rho_{i-2},\dots,\rho_{i+u+8}]$ in~(15) corresponds to a state transition from $s_{k}$ to $s_{k+1}$.
        Here,  $u$ represents the number of parameters in the function.

    \item By (16) and the definition of $\ks$ and Table \ref{table:tableCorrespondence}, $T(s_k) = 1$ and $c_k = \texttt{function}\ f_{c}(\texttt{d})\{\kwp{stmt}\}$.

    \item By (17) and the definition of $\mathcal{R}$, rule $\kwt{recv\_ext}$ is generated.

    \item By Proposition~\ref{proposition:transitionsystem} , $e_A$ corresponds to a sequence of state transitions: $s_{A_0} \Rightarrow s_{A_1} \dots \Rightarrow s_{A_{n''}}$ and $e_B$ corresponds to a sequence of state transitions: $s_{B_0} \Rightarrow s_{B_1} \dots \Rightarrow s_{B_{m''}}$.

    \item By (1), (17), $\vert \{s_{A_i} \mid i \in \{1,\dots,n''\} \land T(s_{A_i}) = 1\}^{\#} \vert^{\#} = \vert \{s_{B_i} \mid i \in \{1,\dots,m''\} \land T(s_{B_i}) = 1\}^{\#} \vert^{\#}$

    \item By (17), (20), and the definition of rule $\kwt{recv\_ext}$ ,
     $\{V \mid \forall k \in \{1,\dots,n'-1\}.\ V \ins F_{A_k}' \land V \not\ins F_{A_{k+1}}'\land  \textit{names}(\{V\})=\{\kwf{Call_e} \}    \}^{\#} = \{V \mid \forall k \in \{1,\dots,m'-1\}.\ V \ins F_{B_k}' \land V \not\ins F_{B_{k+1}}'\land  \textit{names}(\{V\})=\{\kwf{Call_e} \}    \}^{\#}$.

    \item By (11), (14), (21), $\phi_{equ}^{R_s}(\chi,\ind,x)$ does not hold, which is in contradiction to the assumption. Thus Lemma \ref{lemma: soundness4} proved. 
    
   \end{enumerate}
\end{proof}

\begin{theorem}{(Soundness for the invariant property)}
 \label{theorem: final_soundness}
    Let   $(\mathcal{T},\karroww, \gamma_0)$ be a K transition system that satisfies $\ks \vDash \alpha$.
    Let $R_{inv} \vDash \beta_{R_{inv}}$. Given a set of global variables $\eta$ of contract $x$ and a function $\IND$,
    if  $\phi^{R_{inv}}(\eta,\IND,x)$ holds, then $\phi_{inv}^{\ks}(\lambda_0,\eta,x)$ holds.
\end{theorem}

\begin{proof}
 \begin{enumerate}
    \item By Lemma \ref{lemma: soundness}, if $\phi^{R_{inv}}(\eta,\IND,x)$ holds, then $\phi_{inv}^{R_s}(\eta,\IND,x)$ holds.
    \item By Lemma \ref{lemma: soundness3}, if  $\phi_{inv}^{R_s}(\eta,\IND,x)$ holds, then $\phi_{inv}^{\ks}(\lambda_0,\eta,x)$ holds .
    \item By (1), (2), Theorem \ref{theorem: final_soundness} proved.
\end{enumerate}
\end{proof}

\begin{theorem}{(Soundness for the equivalence property)}
 \label{theorem: final_soundness2}
    Let   $(\mathcal{T},\karroww, \gamma_0)$ be a K transition system that satisfies $\ks \vDash \alpha$.
    Let $R_{equ} \vDash \beta_{R_{equ}}$.  Given a global variable $\chi$ of contract $x$ and a function $\ind$, if $\phi^{R_{equ}}(\chi,\ind,x)$ holds, then $\phi_{equ}^{\ks}(\lambda_0,\chi,x)$ holds.
\end{theorem}

\begin{proof}
 \begin{enumerate}
    \item By Lemma \ref{lemma: soundness2}, if $\phi^{R_{equ}}(\chi,\ind,x)$ holds, then $\phi_{equ}^{R_s}(\chi,\ind,x)$ holds.
    \item By Lemma \ref{lemma: soundness4}, if $\phi_{equ}^{R_s}(\chi,\ind,x)$ holds, then $\phi_{equ}^{\ks}(\lambda_0,\chi,x)$ holds.
    \item By (1), (2), Theorem \ref{theorem: final_soundness2} proved.
\end{enumerate}
\end{proof}



\begin{figure*}[p]
    \centering
    \begin{subfigure}{.47\textwidth}
        \input{rule/r-Require.tex}
        \input{rule/r-Out-of-Gas.tex}
        \input{rule/r-Revert.tex}
        \input{rule/r-Assert.tex}
        \input{rule/r-Exception-Propagation.tex}
        \input{rule/r-Transaction-Reversion.tex}
        \input{rule/r-Update-Exception-State.tex}
        \input{rule/r-Revert-State.tex}
        \input{rule/r-Transfer-Fund-Begin.tex}
        \input{rule/r-Transfer-Fund.tex}
        \\
    \end{subfigure}
    \begin{subfigure}{.47\textwidth}
        \input{rule/r-Revert-In-Contracts.tex}
        \input{rule/r-Delete-New-Contracts.tex}
        \input{rule/r-Send-Fund-Successful.tex}
        \input{rule/r-Send-Fund-Failed.tex}
    \end{subfigure}
    \caption{Definition of $\mathcal{K}_s$ (Part 1).}
    \label{fig:definerule1}
\end{figure*}

\begin{figure*}[p]
    \centering
    \begin{subfigure}{.47\textwidth}
        \input{rule/r-Send-Fund-Begin.tex}
        \input{rule/r-Write.tex}
        \input{rule/r-WriteAddress-GlobalVariables.tex}
        \input{rule/r-WriteAddress-LocalVariables.tex}
        \input{rule/r-Read.tex}
        \input{rule/r-ReadAddress-GlobalVariables.tex}
        \input{rule/r-ReadAddress-LocalVariables.tex}
        \input{rule/r-AllocateStateVariables.tex}
    \end{subfigure}
    \begin{subfigure}{.47\textwidth}
        \input{rule/r-New-Contract-Instance-Creation.tex}
        \input{rule/r-UpdateState-Main-Contract.tex}
        \input{rule/r-UpdateState-Function-Call.tex}
        \input{rule/r-AllocateStorage.tex}
        \input{rule/r-InitInstance-NoConstructor.tex}
        \input{rule/r-InitInstance-WithConstructor.tex}
        \input{rule/r-Decompose-Solidity-Call.tex}
        \input{rule/r-Function-Call.tex}
    \end{subfigure}
    \caption{Definition of $\mathcal{K}_s$ (Part 2).}
    \label{fig:definerule2}
\end{figure*}

\begin{figure*}[p]
    \centering
    \begin{subfigure}{.47\textwidth}
        \input{rule/r-AllocateStateVariables-End.tex}
        \input{rule/r-Allocate.tex}
        \input{rule/r-AllocateAddress-GlobalVariables.tex}
        \input{rule/r-AllocateAddress-LocalVariables.tex}
        \input{rule/r-Switch-Context.tex}
        \input{rule/r-Return-Context.tex}
    \end{subfigure}
    \begin{subfigure}{.47\textwidth}
        \input{rule/r-Internal-Function-Call.tex}
        \input{rule/r-Nested-Function-Call.tex}
        \input{rule/r-Clear-Recipient-Context.tex}
        \input{rule/r-Save-Cur-Context.tex}
        \\
    \end{subfigure}
    \caption{Definition of $\mathcal{K}_s$ (Part 3).}
    \label{fig:definerule3}
\end{figure*}

\begin{figure*}[p]
    \centering
    \begin{subfigure}{.47\textwidth}
        \input{rule/r-Update-Cur-Context.tex}
        \input{rule/r-Create-Transaction.tex}
        \input{rule/r-Propagate-Exception-True.tex}
        \input{rule/r-Propagate-Exception-False.tex}
    \end{subfigure}
    \begin{subfigure}{.47\textwidth}
        \input{rule/r-Clear-Caller-Context.tex}
        \input{rule/r-Call.tex}
        \input{rule/r-Init-Fun-Params.tex}
        \input{rule/r-Call-Function-Body.tex}
        \input{rule/r-Bind-Params.tex}
        \input{rule/r-Bind-Params-End.tex}
        \input{rule/r-Function-Body.tex}
        \input{rule/r1.tex}
        \input{rule/r5.tex}
        \input{rule/r6.tex}
        \input{rule/r-Exe-Statement.tex}
    \end{subfigure}
    \caption{Definition of $\mathcal{K}_s$ (Part 4).}
    \label{fig:definerule4}
\end{figure*}

\begin{figure*}[p]
    \centering
    \begin{subfigure}{.47\textwidth}
        \input{rule/r-Exe-Statement-Main-Contract.tex}
        \input{rule/r-Exe-Statement-End.tex}
        \input{rule/r-Less-GlobalVariables.tex}
        \input{rule/r-Less-LocalVariables.tex}
        \input{rule/r-Equal-GlobalVariables.tex}
        \input{rule/r-Equal-LocalVariables.tex}
        \\
    \end{subfigure}
    \begin{subfigure}{.47\textwidth}
        \input{rule/r-More-GlobalVariables.tex}
        \input{rule/r-More-LocalVariables.tex}
        \input{rule/r-Var-Declaration.tex}
        \input{rule/r-Var-Assignment.tex}
        \input{rule/r-Return-Value.tex}
        \input{rule/r-Return.tex}
    \end{subfigure}
    \caption{Definition of $\mathcal{K}_s$ (Part 5).}
    \label{fig:definerule5}
\end{figure*}

\begin{figure*}[p]
    \centering
    \begin{subfigure}{.47\textwidth}
        \input{rule/r-Gas-Cal.tex}
        \input{rule/r-Add-GlobalVariables.tex}
        \input{rule/r-Add-LocalVariables.tex}
        \input{rule/r-Sub-GlobalVariables.tex}
    \end{subfigure}
    \begin{subfigure}{.47\textwidth}
        \input{rule/r-Gas-Cal-Fail.tex}
        \input{rule/r-Sub-LocalVariables.tex}
        \input{rule/r-Mul-GlobalVariables.tex}
        \input{rule/r-Mul-LocalVariables.tex}
        \\
    \end{subfigure}
    \caption{Definition of $\mathcal{K}_s$ (Part 6).}
    \label{fig:definerule6}
\end{figure*}

\begin{figure*}[p]
    \centering
    \begin{subfigure}{.47\textwidth}
        \input{rule/r-Div-GlobalVariables.tex}
    \end{subfigure}
    \begin{subfigure}{.47\textwidth}
        \input{rule/r-Div-LocalVariables.tex}
    \end{subfigure}
    \caption{Definition of $\mathcal{K}_s$ (Part 7).}
    \label{fig:definerule7}
\end{figure*}



\clearpage

\end{document}